\documentclass[11pt]{article}
\pdfoutput=1

\usepackage{fourier}
\usepackage[utf8]{inputenc}
\usepackage[T1]{fontenc}
\usepackage{microtype}

\usepackage[paper=a4paper]{geometry}
\usepackage{graphicx}
\usepackage{tikz}
\usepackage[nottoc,notlot,notlof]{tocbibind}
\usepackage[margin=5mm,font=small,labelfont=bf,labelsep=period]{caption}
\usepackage{float}

\usepackage{amsmath}
\usepackage{amsthm}
\usepackage{amssymb}
\numberwithin{equation}{section}

\theoremstyle{plain}
\newtheorem{proposition}{Proposition}[section]

\newtheorem{lemma}[proposition]{Lemma}
\newtheorem{theorem}[proposition]{Theorem}

\theoremstyle{definition}
\newtheorem{definition}[proposition]{Definition}
\newtheorem{example}[proposition]{Example}
\newtheorem*{example*}{Example}
\newtheorem{remark}[proposition]{Remark}

\newcommand{\abs}[1]{\left\lvert #1 \right\rvert}
\newcommand{\R}{\mathbf{R}}
\renewcommand{\epsilon}{\varepsilon}
\newcommand{\twin}{\widetilde}
\newcommand{\detJ}{\,\operatorname{J}}
\DeclareMathOperator{\sgn}{sgn}

\newcommand{\X}{\,\begin{tikzpicture} \draw[fill] (0,0) circle (1.1mm); \end{tikzpicture}\,}
\newcommand{\arrowX}{\,\begin{tikzpicture} \draw[fill] (0,0) circle (1.1mm); \draw[thick,->] (0,0.5) -- (0,0.2); \end{tikzpicture}\,}
\newcommand{\Y}{\,\begin{tikzpicture} \draw (0,0) circle (1.1mm); \end{tikzpicture}\,}

\newcommand{\Z}{\phantom{\X}}

\usepackage[colorlinks,linkcolor=blue,citecolor=blue,urlcolor=blue]{hyperref}
\newcommand{\ourtitle}{Non-interlacing peakon solutions of the Geng--Xue equation}
\hypersetup{pdfauthor={Budor Shuaib, Hans Lundmark},pdftitle={\ourtitle}}

\begin{document}

\title{\ourtitle}
\author{%
  Budor Shuaib\thanks{\textbf{budor.shuaib@liu.se}, \textbf{hans.lundmark@liu.se}, Department of Mathematics, Link\"oping University, SE-581 83 Link\"oping, Sweden.}
  \and
  Hans Lundmark\footnotemark[1]
}

\date{December 20, 2018}
\maketitle

\begin{abstract}
  The aim of the present paper is to derive explicit formulas for
  arbitrary peakon solutions of the Geng--Xue equation, a two-component
  generalization of Novikov's cubically nonlinear Camassa--Holm type equation.
  By performing limiting procedures on the previosly known formulas for so-called
  interlacing peakon solutions,
  where the peakons in the two component occur alternatingly,
  we turn some of the peakons into zero-amplitude ``ghostpeakons'',
  in such a way that the remaining ordinary peak\-ons
  occur in any desired configuration.
  We also study the large-time asymptotics of these solutions.
\end{abstract}

\tableofcontents
\label{toc}

\section{Introduction}
\label{sec:intro}

In 2009, Geng and Xue~\cite{geng-xue:2009:GX-peakon-equation-cubic-nonlinearity}
derived a coupled integrable equation with cubic nonlinearity:
\begin{equation}
  \label{eq:GX}
  \begin{gathered}
    m_t + (m_xu + 3mu_x)v = 0
    ,\\
    n_t + (n_xv + 3nv_x)u = 0
    ,\\
    m = u - u_{xx}
    ,\qquad
    n = v - v_{xx}
    .
  \end{gathered}
\end{equation}
This equation is known as the Geng--Xue (GX) equation and also as the
two-component Novikov equation, since with $u=v$ the system
\eqref{eq:GX} reduces to two copies of the Novikov equation from~2008
\cite{novikov:2009:generalizations-of-CH, hone-wang:2008:cubic-nonlinearity, hone-lundmark-szmigielski:2009:novikov},
\begin{equation}
  \label{eq:Novikov}
  m_t + u (m_x u + 3m u_x) = 0
  ,\qquad
  m = u - u_{xx}
  .
\end{equation}
The GX and Novikov  equations  are mathematical relatives of
the Camassa--Holm (CH) equation
\begin{equation}
  \label{eq:CH}
  m_t + m_x u + 2 m u_x = 0
  ,\qquad
  m = u - u_{xx}
  ,
\end{equation}
which was derived as model for shallow water waves
in~1993
\cite{camassa-holm:1993:CH-orginal-paper, camassa-holm-hyman:1994:CH-new-integrable},
and of the Degasperis--Procesi (DP)
equation found in~1998
\cite{degasperis-procesi:1999:asymptotic-integrability, degasperis-holm-hone:2002:new-integrable-equation-DP},
\begin{equation}
  \label{eq:DP}
  m_t + m_x u + 3 m u_x = 0
  ,\qquad
  m = u - u_{xx}
  .
\end{equation}
Some reasons that these partial differential equations are of interest
are that they are integrable systems and admit \emph{multipeakon}
solutions, which are weak solutions formed by superposition of
peak-shaped waves:
\begin{equation}
  u(x,t) = \sum_{k=1}^N m_k(t) \, e^{-\abs{x - x_k(t)}}
  .
\end{equation}
The positions $x_k(t)$ and the amplitudes $m_k(t)$ are governed by a
system of ODEs, whose precise form depends on which equation we are studying.
The general solution of these peakon ODEs is known for these three
equations (CH, DP and Novikov) in the case that all the amplitudes are
nonzero ($m_k(t)\neq 0$)
\cite{beals-sattinger-szmigielski:2000:moment,lundmark-szmigielski:2005:DPlong,hone-lundmark-szmigielski:2009:novikov}.

The case when some amplitudes are zero might seem uninteresting, since
if $m_k(0)=0$, then $m_k(t)=0$ for all~$t$, and then peakon number~$k$
is simply absent from the solution. However, the equation for $x_k(t)$
in the peakon ODEs is still nontrivial and its solution describes the
trajectory of a zero-amplitude \emph{ghostpeakon} which is influenced
by the other peakons but does not influence them. These ghostpeakon
trajectories are of interest, since they are in fact the
\emph{characteristic curves} associated with the multipeakon solution,
and we recently derived explicit formulas for the ghostpeakons for the
CH, DP and Novikov equations~\cite{lundmark-shuaib:2018p:ghostpeakons}.
Moreover, as we will demonstrate in this paper, the methods used for
studying ghostpeakons are also very useful for studying ordinary
(non-ghost) peakons for the GX equation.

The multipeakon solutions of the Geng--Xue equation,
with $N_1$ peakons in~$u$ and $N_2$ peakons in~$v$,
have the form
\begin{equation}
  \label{eq:GX-multipeakons}
  \begin{split}
    u(x,t) &= \sum_{k=1}^{N_1} m_k(t) \, e^{-\abs{x - x_k(t)}}
    ,\\[1ex]
    v(x,t) &= \sum_{k=1}^{N_2} n_k(t) \, e^{-\abs{x - y_k(t)}}
    ,
  \end{split}
\end{equation}
where the positions $x_k(t)$ and~$y_k(t)$ and the amplitudes
$m_k(t)$ and~$n_k(t)$ in~\eqref{eq:GX-multipeakons} have to satisfy the ODEs
\begin{equation}
  \label{eq:GX-peakon-ode}
  \begin{aligned}
    \dot x_k &= u(x_k) \, v(x_k)
    ,\\
    \dot y_k &= u(y_k) \, v(y_k)
    ,\\
    \dot m_k &= m_k \bigl(u(x_k) \, v_x(x_k) - 2 u_x(x_k) v(x_k) \bigr)
    ,\\
    \dot n_k &= n_k \bigl(u_x(y_k) \, v(y_k) - 2 u(y_k) v_x(y_k) \bigr)
    ,
  \end{aligned}
\end{equation}
in order for~\eqref{eq:GX-multipeakons} to be a weak solution of the Geng--Xue equation.
Here the notation
\begin{equation}
  \label{eq:shorthand-notation}
  u(x_k) = \sum_{i=1}^{N_1} m_i e^{-|x_k-x_i|}
  ,\qquad
  u_x(x_k) = - \sum_{i=1}^{N_1} m_i \sgn(x_i-x_k) \, e^{-|x_k-x_i|}
\end{equation}
is shorthand for the expressions obtained by substituting $x=x_k$
into the peakon ansatz~\eqref{eq:GX-multipeakons} and its derivative,
and similarly for the other quantities appearing in~\eqref{eq:GX-peakon-ode}.
Actually, $u_x(x_k)$ denotes the average of the left and right derivatives,
since $u$ is not differentiable at $x=x_k$.
The convention $\sgn(0)=0$ is used here.

We will assume throughout the paper that all amplitudes are \emph{positive} ($m_k>0$ and $n_k>0$),
the so-called \emph{pure peakon} case
where there are no  ``antipeakons''
with negative amplitude,
and also that the peakons are \emph{non-overlapping},
meaning that $x_i \neq y_j$ for all $i$ and~$j$.

For an explanation of these requirements,
see  Lundmark and Szmigielski~\cite{lundmark-szmigielski:2016:GX-inverse-problem, lundmark-szmigielski:2017:GX-dynamics-interlacing},
who solved the ODEs~\eqref{eq:GX-peakon-ode}
and studied the dynamics of the solutions in the
\emph{interlacing} case where $N_1=N_2=K$ and
\begin{equation}
  \label{eq:interlacing-assumption}
  x_1 < y_1 < x_2 < y_2 < \dots < x_K < y_K
  .
\end{equation}
In Section~\ref{sec:interlacing-review} we will recall their explicit
solution formulas for this interlacing peakon configuration
with an even number of peakons.
There is also the case of an interlacing configuration with an odd number of peakons
($N_1=K+1$ and $N_2=K$),
\begin{equation}
  \label{eq:interlacing-assumption-odd}
  x_1 < y_1 < x_2 < y_2 < \dots < x_K < y_K < x_{K+1}
  ,
\end{equation}
which was not studied by Lundmark and Szmigielski.
The solution for the odd case~\eqref{eq:interlacing-assumption-odd}
could be derived by a slight modification of the inverse spectral procedure
used by them in the even interlacing case~\eqref{eq:interlacing-assumption},
but we will instead obtain it here as a special case of
our more general results.

The purpose of this paper is to derive the solution formulas for an
\emph{arbitrary} (pure) peakon configuration, and study the
asymptotics as $t \to \pm\infty$. The general solution will depend on
$2(N_1+N_2)$ arbitrary parameters, which can in principle be
determined from initial data for the $N_1+N_2$ positions and $N_1+N_2$
amplitudes at time $t=0$, say. However, we will not go into the
details of how to do that, since it would require introducing even
more notation than we already have. Our philosophy is instead to study
the global behaviour of the whole solution for given values of the
parameters.

We can label the peakons in such a way that
\begin{equation}
  x_1 < x_2 < \dots < x_{N_1}
  ,\qquad
  y_1 < y_2 < \dots < y_{N_2}
  ,
\end{equation}
where we also assume that $x_i \neq y_j$ for all $i$ and~$j$,
but we do not impose the condition
\eqref{eq:interlacing-assumption}
or~\eqref{eq:interlacing-assumption-odd}.
In this case, which is in general \emph{non-interlacing},
it is more convenient to use a different notation where we gather
adjacent peakons into ``$X$-groups'' and ``$Y$-groups'' and label
them with two indices, as follows:
\begin{equation}
  \label{eq:non-interlacing-notation}
  \begin{split}
    &
    \underbrace{x_{1,1} < x_{1,2} < \dots < x_{1,N_1^X}}_{\text{First $X$-group}}
    <
    \underbrace{y_{1,1} < y_{1,2} < \dots < y_{1,N_1^Y}}_{\text{First $Y$-group}}
    <
    \dotsb
    \\
    <
    &
    \underbrace{x_{j,1} < x_{j,2} < \dots < x_{j,N_j^X}}_{j\text{th $X$-group}}
    <
    \underbrace{y_{j,1} < y_{j,2} < \dots < y_{j,N_j^Y}}_{j\text{th $Y$-group}}
    <
    \dotsb
    \\
    <
    &
    \underbrace{x_{K,1} < x_{K,2} < \dots < x_{K,N_K^X}}_{\text{Last $X$-group}}
    <
    \underbrace{y_{K,1} < y_{K,2} < \dots < y_{K,N_K^Y}}_{\text{Last $Y$-group}}
    ,
  \end{split}
\end{equation}
and similarly for the amplitudes $m_{j,i}$ and $n_{j,i}$.
Here $N_j^X$ and~$N_j^Y$ denote the number of peakons in the $j$th
$X$-group and $Y$-group, respectively.
Thus, the peakon solutions~\eqref{eq:GX-multipeakons} in the
notation~\eqref{eq:non-interlacing-notation} take the form
\begin{equation}
  \label{eq:non-interlacing-peakons}
  \begin{aligned}
    u(x,t) &= \sum_{k=1}^K \left( \sum_{i=1}^{N^X_k} m_{k,i}(t) \, e^{-\abs{x - x_{k,i}(t)}} \right)
    ,
    \\
    v(x,t) &= \sum_{k=1}^K \left( \sum_{i=1}^{N^Y_k} n_{k,i}(t) \, e^{-\abs{x - y_{k,i}(t)}} \right)
    .
  \end{aligned}
\end{equation}
Since $u$ and $v$ play the same role in the Geng--Xue equation
\eqref{eq:GX}, we will always assume (without loss of generality) that
the first group is an $X$-group.
The last group may then be a $Y$-group as in
\eqref{eq:non-interlacing-notation} and
\eqref{eq:non-interlacing-peakons}, in which case the total number of
groups is even ($K+K$), or it may be an $X$-group, so that there is an
odd number of groups ($(K+1)+K$) and we have $u(x,t)=\sum_{k=1}^{K+1}$
instead in~\eqref{eq:non-interlacing-peakons}.
We will need to treat these even and odd cases separately,
and in fact it will turn out that the asymptotics in the odd case shows some
interesting differences from the even case.
For singletons (groups containing only a single peakon),
we will often write just $x_j$ instead of $x_{j,1}$, and similarly for the other variables.
We may also assume that there is at least one group of each kind;
otherwise the dynamics is trivial (everything is constant).

A simple example with $2+2$ groups is shown in
Figure~\ref{fig:noninterlacing-peakons-K=2}, where there is one
$X$-group with two members, and the other groups are singletons;
we will consider this particular configuration in detail in Example~\ref{ex:proof-technique}.

\begin{figure}
  \centering
  \begin{tikzpicture}
    \draw[->] (-1, 0) -- (9.5, 0) node[below] {$x$};
    \draw[->] (0, -0.5) -- (0, 2.5);
    
    \def\uformula{plot (\noexpand\x,{2.5*exp(-abs(\noexpand\x-1))+1*exp(-abs(\noexpand\x-3))+0.5*exp(-abs(\noexpand\x-5))})}
    \def\vformula{plot (\noexpand\x,{2*exp(-abs(\noexpand\x-2))+1.2*exp(-abs(\noexpand\x-8))})}
    
    \draw[blue] (1, 0.1) -- +(0, -0.2) node[below] {\small $x_1$};
    \draw[red]  (2, 0.1) -- +(0, -0.2) node[below] {\small $y_1$};
    \draw[blue] (3, 0.1) -- +(0, -0.2) node[below] {\small $x_{2,1}$};
    \draw[blue] (5, 0.1) -- +(0, -0.2) node[below] {\small $x_{2,2}$};
    \draw[red]  (8, 0.1) -- +(0, -0.2) node[below] {\small $y_2$};
    
    \begin{scope}[blue,very thick]
      \draw [domain = -1 : 1] \uformula;
      \draw [domain = 1 : 3] \uformula;
      \draw [domain = 3 : 7] \uformula;
      \draw [domain = 7 : 9] \uformula;
    \end{scope}
    \begin{scope}[red,very thick]
      \draw [domain = -1 : 2] \vformula;
      \draw [domain = 2 : 5] \vformula;
      \draw [domain = 5 : 8] \vformula;
      \draw [domain = 8 : 9] \vformula;
    \end{scope}
    \draw (0.3,3) node[right] {$u(x,t) = m_1 \, e^{-\abs{x-x_1}} + m_{2,1} \, e^{-\abs{x-x_{2,1}}} + m_{2,2} \, e^{-\abs{x-x_{2,2}}}$};
    \draw (2,2) node[right] {$v(x,t) = n_1 \, e^{-\abs{x-y_1}} + n_2 \, e^{-\abs{x-y_2}}$};
  \end{tikzpicture}
  \caption{A \textbf{non-interlacing} peakon configuration
    with two groups in each component,
    where all the groups are singletons,
    except the second $X$-group which contains two peakons.}
  \label{fig:noninterlacing-peakons-K=2}
\end{figure}
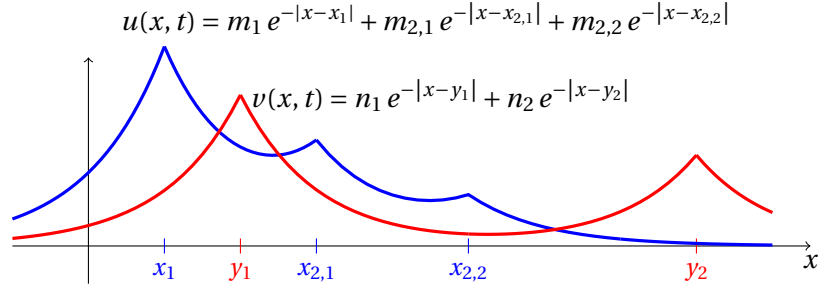

In the notation~\eqref{eq:non-interlacing-notation},
the ODEs~\eqref{eq:GX-peakon-ode} become
\begin{equation}
  \label{eq:GX-peakon-ode-new-notation}
  \begin{split}
    \dot x_{k,i} &= u(x_{k,i}) \, v(x_{k,i})
    ,\\
    \dot y_{j,i} &= u(y_{j,i}) \, v(y_{j,i})
    ,\\
    \dot m_{k,i} &= m_{k,i} \bigl(u(x_{k,i}) \, v_x(x_{k,i}) - 2 u_x(x_{k,i}) v(x_{k,i}) \bigr)
    ,\\
    \dot n_{j,i} &= n_{j,i} \bigl(u_x(y_{j,i}) \, v(y_{j,i}) - 2 u(y_{j,i}) v_x(y_{j,i}) \bigr)
    .
  \end{split}
\end{equation}
To solve these equations, one approach would be to adapt the inverse
spectral technique used by Lundmark and Smigielski in the interlacing
case, but we are not going to do that, for two reasons. Firstly, it is
quite a long and complicated procedure, and secondly, it is not even
clear that it is going to work, since the Lax pairs for the Geng--Xue
equation do not seem to provide sufficiently many constants of motion
in the non-interlacing case; see Section~\ref{sec:effective}.
Motivated by our previous paper~\cite{lundmark-shuaib:2018p:ghostpeakons},
we will instead use the idea of ghostpeakons:
starting from an \emph{interlacing} peakon solution with
a larger number of peakons, we perform appropriate limiting procedures
on the spectral data to ``kill off'' selected peakons,
i.e., we turn them into ghostpeakons by making their amplitudes tend
to zero. After discarding the ghostpeakons and relabeling the variables,
what remains is the desired peakon configuration.

\begin{remark}
  \label{rem:dong-zhou}
  In 2015, Xia and Qiao~\cite{xia-qiao:2015:two-component-CH-with-peakons}
  introduced the integrable two-component Camassa--Holm-type equation
  \begin{equation}
    \label{eq:two-component-CH}
    \begin{aligned}
      m_t &= \tfrac12 \bigl( m (u v - u_x v_x) \bigr)_x - \tfrac12 m (u v_x - u_x v)
      ,\\
      n_t &= \tfrac12 \bigl( n (u v - u_x v_x) \bigr)_x + \tfrac12 n (u v_x - u_x v)
      ,\\
      m &= u - u_{xx}
      ,\quad
      n = v - v_{xx}
      ,
    \end{aligned}
  \end{equation}
  which admits non-overlapping peakon solutions, just like the Geng--Xue equation.
  The solution formulas for the interlacing case have recently been computed by
  Dong and Zhou~\cite{dong-zhou:2018:interlacing-peakons-two-component-CH}
  using inverse spectral techniques.
  It may be interesting to investigate whether the general (non-interlacing)
  solution of~\eqref{eq:two-component-CH}
  can be found in a manner similar to what we do for the GX equation here.
\end{remark}

The outline of this article is as follows.
In section \ref{sec:interlacing-review}, notation and solution formulas for
the Geng--Xue equation in the (even) interlacing peakon case will be
recalled; see in particular Theorem~\ref{thm:interlacing-solution}.
In section \ref{sec:more-notation}, we present some additional notation
which will be used in this paper.
Section~\ref{sec:examples} contains plenty of examples,
both of the proof technique and of the behaviour of the solutions, illustrated in figures.
In Sections \ref{sec:solutions-even} and~\ref{sec:solutions-odd}
we give the complete list of solution formulas for the
general peakon solution of the Geng--Xue equation
in the even case ($2K$ groups) and odd case ($2K+1$ groups),
respectively.
The proofs of these formulas are given in
Sections \ref{sec:proofs-even} and~\ref{sec:proofs-odd},
and we study the asymptotics for the even and odd cases in
Sections \ref{sec:asymptotics-even} and~\ref{sec:asymptotics-odd}, respectively.
In Section~\ref{sec:effective}, it is shown that for each group of peakons one can define
an ``effective position'' and an ``effective amplitude'', which have the same time dependence
as the single peakon in the group would have if the group were a singleton,
and we use this to explain why the solution formulas for singleton groups in non-interlacing
configurations are identical to the previously known formulas from the interlacing case.
In Section~\ref{sec:no-collisions} we prove that for pure peakon solutions there are no collisions
(which implies that these solutions are globally defined in time).
And finally, in Section~\ref{sec:characteristic-curves} we give the formulas for the
characteristic curves associated with an arbitrary peakon solution;
these formulas may have some intrinsic interest, and they also clarify certain aspects
of the structure of the formulas for the peakon solution.

\section{Review of notation and solution formulas for the inter\-lacing case}
\label{sec:interlacing-review}

In this section we will state the solution formulas for
the $K+K$ interlacing peakon solutions~\eqref{eq:interlacing-assumption}
of the GX equation.
These formulas, which were derived by Lundmark and
Szmigielski~\cite{lundmark-szmigielski:2016:GX-inverse-problem,lundmark-szmigielski:2017:GX-dynamics-interlacing},
will be our starting point when we derive the
general non-interlacing solution later.

First we need to define some notation.
The general solution for the $2K$ positions $x_k$ and~$y_k$
and the $2K$ amplitudes $m_k$ and~$n_k$
(where $1 \le k \le K$)
depends on $4K$ constant parameters,
collectively referred to as the \emph{spectral data}.
These parameters are all \emph{positive} in the pure peakon case.
Firstly, there are two sets of eigenvalues $\lambda_i$ and~$\mu_j$,
\begin{equation}
  \label{eq:eigenvalues}
  0 < \lambda_1 < \lambda_2 < \dots < \lambda_K
  ,\qquad
  0 < \mu_1 < \mu_2 < \dots < \mu_{K-1}
  ,
\end{equation}
coming from boundary value problems
associated with the two Lax pairs of the Geng--Xue
equation~\cite{lundmark-szmigielski:2016:GX-inverse-problem}.
These $2K-1$ eigenvalues are accompanied by
$2K-1$ residues of the corresponding Weyl functions (at time $t=0$),
\begin{equation}
  \label{eq:residues}
  a_1(0), \, a_2(0), \, \dots, a_K(0) \in \R_+
  ,\qquad
  b_1(0), \, b_2(0), \, \dots, b_{K-1}(0) \in \R_+
  ,
\end{equation}
and finally there are two additional constants,
\begin{equation}
  \label{eq:CD}
  C, \, D  \in \R_+
  ,
\end{equation}
also related to the Lax pairs.

\begin{remark}
  Lundmark and Szmigielski~\cite{lundmark-szmigielski:2016:GX-inverse-problem,lundmark-szmigielski:2017:GX-dynamics-interlacing}
  used the equivalent parameters
  \begin{equation*}
    b_\infty=D
    \qquad\text{and}\qquad
    b_\infty^* = \frac{C \mu_1 \cdots \mu_{K-1}}{2 \lambda_1 \cdots \lambda_{K}}
  \end{equation*}
  instead of our $C$ and~$D$.
  We have also found it more convenient here to write
  $x_{j}$, $y_{j}$,
  $m_{j}$, $n_{j}$
  instead of their
  $x_{2j-1}$, $x_{2j}$,
  $m_{2j-1}$,~$n_{2j}$.
\end{remark}

The solution formulas also contain the time-dependent quantities
\begin{equation}
  \label{eq:residues-time-dependence}
  a_i(t) = a_i(0) \, e^{t/\lambda_i}
  ,\qquad
  b_j(t) = b_j(0) \, e^{t/\mu_j}
  ,
\end{equation}
which will mostly be denoted simply by $a_i$ and~$b_j$.

\begin{definition}[Integer interval]
  Let $[m,n]=[m,m+1,m+2,\dots,n-1,n]$ for integers $m \le n$. If $m>n$, let $[m,n]= \emptyset$.
  We will sometimes also use the shorter form $[n]$ for~$[1,n]$.
\end{definition}

\begin{definition}
  For $i \ge 0$, let $\binom{[1,K]}{i}$ denote the set of $i$-element subsets
  $I = \{ i_1 < i_2 < \dots < i_n \}$ of the integer interval $[1,K]= \{1,2,\dots, K\}$.
  For  $I \in \binom{[1,A]}{k}$ and $J \in \binom{[1,B]}{l}$, let
  \begin{equation}
    \label{eq:DeltaI}
    \begin{split}
      \Delta_I^2 &
      = \Delta(\lambda_{i_1}, \dots, \lambda_{i_k})^2
      = \prod_{\substack{a,b \in I \\ a < b}} (\lambda_{a} - \lambda_{b})^2
      ,
      \\
      \twin\Delta_J^2 &
      = \Delta(\mu_{j_1}, \dots, \mu_{j_l})^2
      = \prod_{\substack{a,b \in J \\ a < b}} (\mu_{a} - \mu_{b})^2
      ,
      \\
      \Gamma_{IJ} &
      = \Gamma(\lambda_{i_1}, \dots, \lambda_{i_k}; \mu_{j_1}, \dots, \mu_{j_l})
      = \prod_{n \in I, \, m \in J} (\lambda_n + \mu_m)
      ,
    \end{split}
  \end{equation}
  with the special cases
  \begin{equation*}
    \Delta_{\emptyset}^2
    = \Delta_{\{ n \}}^2
    = \twin\Delta_{\emptyset}^2
    = \twin\Delta_{\{ m \}}^2
    = \Gamma_{I\emptyset}
    = \Gamma_{\emptyset J}
    = 1
    .
  \end{equation*}
\end{definition}

\begin{definition}
  \label{def:heineintegral}
  Using the abbreviations
  \begin{equation}
    \label{eq:product-IJ}
    \Psi_{IJ} = \frac{\Delta_I^2 \twin\Delta_J^2}{\Gamma_{IJ}}
    ,\qquad
    \lambda_I^r a_I =\Bigl( \prod_{i \in I} \lambda_i^r a_i \Bigr)
    ,\qquad
    \mu_J^s b_J =\Bigl( \prod_{j \in J} \mu_j^s b_j \Bigr)
    ,
  \end{equation}
  let
  \begin{equation}
    \label{eq:heine-integral-as-sum}
    \begin{split}
      \detJ_{ij}^{rs} &
      = \detJ[A,B,r,s,i,j]
      \\ &
      =
      \begin{cases}
        \displaystyle
        \sum_{I \in \binom{[1,A]}{i}} \sum_{J \in \binom{[1,B]}{j}}
        \Psi_{IJ} \, \lambda_I^r a_I \, \mu_J^s b_J
        ,
        &
        \text{if $0 \le i \le A$ and $0 \le j \le B$}
        ,\\
        0
        ,
        &
        \text{otherwise}
        .
      \end{cases}
    \end{split}
  \end{equation}
  Usually, if the values of $A$ and $B$ are clear from the context,
  we will just write $\detJ_{ij}^{rs}$
  rather than $\detJ[A,B,r,s,i,j]$,
  but in a few places the longer notation is needed for precision.
\end{definition}

\begin{remark}
  \label{rem:J-bimoment-determinant}
  We will sometimes refer to the quantities~$\detJ_{ij}^{rs}$ as ``determinants'',
  since they are originally determinants of bimoments, coming from the theory of
  Cauchy biorthogonal polynomials,
  and the expressions~\eqref{eq:heine-integral-as-sum}
  arise in the case when the two measures involved are
  finite linear combinations of Dirac deltas~\cite[Sections A.3 and~A.4]{lundmark-szmigielski:2016:GX-inverse-problem}.
\end{remark}

\begin{remark}
  Since we are assuming that all spectral data are positive,
  we will have
  \begin{equation}
    \detJ[A,B,r,s,i,j] > 0
  \end{equation}
  if $0 \le i \le A$ and $0 \le j \le B$.
\end{remark}

\begin{example}
  If $A=3$ and $B=2$, then
  \begin{equation}
    \begin{split}
      \detJ_{21}^{01} &
      = \detJ[3,2,0,1,2,1]
      \\ &
      =
      \frac{(\lambda_1 - \lambda_2)^2 \, \mu_1}{(\lambda_1 + \mu_1) (\lambda_2 + \mu_1)} a_1 a_2 b_1
      + \frac{(\lambda_1 - \lambda_3)^2 \, \mu_1}{(\lambda_1 + \mu_1) (\lambda_3 + \mu_1)} a_1 a_3 b_1
      \\ &
      + \frac{(\lambda_2 - \lambda_3)^2 \, \mu_1}{(\lambda_2 + \mu_1) (\lambda_3 + \mu_1)} a_2 a_3 b_1
      + \frac{(\lambda_1 - \lambda_2)^2 \, \mu_2}{(\lambda_1 + \mu_2) (\lambda_2 + \mu_2)} a_1 a_2 b_2
      \\
      &
      + \frac{(\lambda_1 - \lambda_3)^2 \, \mu_2}{(\lambda_1 + \mu_2) (\lambda_3 + \mu_2)} a_1 a_3 b_2
      + \frac{(\lambda_2 - \lambda_3)^2 \, \mu_2}{(\lambda_2 + \mu_2) (\lambda_3 + \mu_2)} a_2 a_3 b_2
      .
    \end{split}
  \end{equation}
  For a list of all nonzero $\detJ[3,2,0,0,i,j]$, see
  Example~A.2 in~\cite{lundmark-szmigielski:2016:GX-inverse-problem}.
\end{example}

The formulas in the following theorem provide a one-to-one
correspondence between the interlacing pure peakon sector
\begin{equation}
  \label{eq:interlacing-assumption-pure}
  \begin{split}
    &
    x_1 < y_1 < x_2 < y_2 < \dots < x_K < y_K
    ,
    \\ &
    m_1, m_2, \dots, m_K \in \R_+
    ,
    \\ &
    n_1, n_2, \dots, n_K \in \R_+
    ,
  \end{split}
\end{equation}
and the set of spectral variables \eqref{eq:eigenvalues}, \eqref{eq:residues} and~\eqref{eq:CD},
for $K \ge 2$.
Together with the time dependence~\eqref{eq:residues-time-dependence},
this provides the interlacing peakon solution $x_j(t)$, $y_j(t)$, $m_j(t)$, $n_j(t)$.
Regarding the case $K=1$, see Remark~\ref{rem:interlacing-K1} below.

\begin{theorem}
  \label{thm:interlacing-solution}
  Let $K \ge 2$.
  In terms of the abbreviations
  \begin{equation}
    \label{eq:XYQP}
    X_k= \frac12 \exp{2x_k}
    ,\quad
    Y_k= \frac12 \exp{2y_k}
    ,\quad
    Q_k = 2 m_k \, e^{-x_k}
    ,\quad
    P_k = 2 n_k \, e^{-y_k}
    ,
  \end{equation}
  \begin{equation}
    \label{eq:LM}
    L = \prod_{i=1}^K \lambda_i
    ,\qquad
    M = \prod_{j=1}^{K-1} \mu_j
    ,
  \end{equation}
  \begin{equation}
    j'= K+1-j
    ,
  \end{equation}
  and with
  \begin{equation*}
    \detJ_{ij}^{rs} = \detJ[K,K-1,r,s,i,j]
    ,
  \end{equation*}
  the general solution of the peakon ODEs~\eqref{eq:GX-peakon-ode}
  in the $K+K$ interlacing pure peakon case~\eqref{eq:interlacing-assumption-pure} is
  \begin{equation}
    \label{eq:interlacing-solution-positions}
    \begin{split}
      X_{j'}
      =
      X_{K+1-j}
      &=
      \begin{cases}
        \displaystyle
        \frac{\detJ_{jj}^{00}}{\detJ_{j-1,j-1}^{11}}
        ,
        &
        1 \le j \le K-1
        ,
        \\[1.5em]
        \displaystyle
        \frac{\detJ_{K,K-1}^{00}}{\detJ_{K-1,K-2}^{11} + C \, \detJ_{K-1,K-1}^{10}}
        ,
        &
        j=K
        ,
      \end{cases}
      \\[1ex]
      Y_{j'}
      =
      Y_{K+1-j}
      & =
      \begin{cases}
        \detJ_{11}^{00} + D \detJ_{10}^{00}
        \,
        ,
        &
        j = 1
        ,
        \\[1.5ex]
        \displaystyle
        \frac{\detJ_{j,j-1}^{00}}{\detJ_{j-1,j-2}^{11}}
        ,
        &
        2 \le j \le K
        ,
      \end{cases}
    \end{split}
  \end{equation}
  and
  \begin{equation}
    \label{eq:interlacing-solution-amplitudes}
    \begin{split}
      Q_{j'}
      =
      Q_{K+1-j}
      &=
      \begin{cases}
        \displaystyle
        \frac{\detJ_{j-1,j-1}^{11} \detJ_{j,j-1}^{01}}{\detJ_{jj}^{10} \detJ_{j-1,j-1}^{10}}
        ,&
        1 \le j \le K-1
        ,
        \\[1.5em]
        \displaystyle
        \frac{M}{L} \left( \frac{\detJ_{K-1,K-2}^{11}}{\detJ_{K-1,K-1}^{10}} + C \right)
        ,
        &
        j=K
        ,
      \end{cases}
      \\[1ex]
      P_{j'}
      =
      P_{K+1-j}
      &=
      \begin{cases}
        \displaystyle
        \frac{1}{\detJ_{10}^{00}}
        ,
        &
        j = 1
        ,
        \\[1.5em]
        \displaystyle
        \frac{\detJ_{j-1,j-2}^{11} \detJ_{j-1,j-1}^{10}}{\detJ_{j-1,j-2}^{01} \detJ_{j,j-1}^{01}}
        ,
        &
        2 \le j \le K
        ,
      \end{cases}
    \end{split}
  \end{equation}
  where the time dependence is given by
  \begin{equation}
    a_i(t) = a_i(0) \, e^{t/\lambda_i}
    ,\qquad
    b_j(t) = b_j(0) \, e^{t/\mu_j}
    .
  \end{equation}
\end{theorem}

\begin{remark}
  \label{rem:actual-positions-amplitudes}
  From~\eqref{eq:XYQP} it is clear that
  we can obtain the actual peakon variables via
  \begin{equation}
    x_k = \frac{1}{2} \ln(2 X_k)
    ,\qquad
    y_k = \frac{1}{2} \ln(2 Y_k)
    ,\qquad
    m_k = \frac{\sqrt{X_k} \, Q_k}{\sqrt{2}}
    ,\qquad
    n_k = \frac{\sqrt{Y_k} \, P_k}{\sqrt{2}}
    ,
  \end{equation}
  but the formulas are less complicated when expressed in terms of
  $X_k$, $Y_k$, $Q_k$ and~$P_k$.
\end{remark}

\begin{remark}
  \label{rem:interlacing-K1}
  The formulas in Theorem~\ref{thm:interlacing-solution}
  also work in the case $K=1$,
  provided that one adds the additional constraint $CD>1$.
  In this case, one should disregard the formulas
  stated for $1 \le j \le K-1$ or $2 \le j \le K$,
  and use the other ones.
\end{remark}

\begin{example}[The $3+3$ interlacing case]
  \label{ex:GX-3+3-interlacing}
  When $K=3$ we obtain the solution formulas for the
  $3+3$ interlacing solution.
  We will study this example quite carefully, since we are going to compare it
  with the examples of non-interlacing solutions later,
  in Section~\ref{sec:examples}.

  With $\detJ_{ij}^{rs} = \detJ[3,2,r,s,i,j]$,
  the solution formulas for the positions are
  \begin{subequations}
    \label{eq:GX-3+3-interlacing-joint}
    \begin{equation}
      \label{eq:GX-3+3-interlacing-positions}
      \begin{aligned}
        X_1 = \tfrac12 e^{2x_1}
        &
        = \frac{\detJ_{32}^{00}}{\detJ_{21}^{11} + C \detJ_{22}^{10}}
        ,\qquad
        &
        Y_1 = \tfrac12 e^{2y_1}
        &
        = \frac{\detJ_{32}^{00}}{\detJ_{21}^{11}}
        ,
        \\[1ex]
        X_2 = \tfrac12 e^{2x_2}
        &
        = \frac{ \detJ_{22}^{00}}{\detJ_{11}^{11}}
        ,
        &
        Y_2 = \tfrac12 e^{2y_2}
        &
        = \frac{ \detJ_{21}^{00}}{\detJ_{10}^{11}}
        ,
        \\[1ex]
        X_3 = \tfrac12 e^{2x_3}
        &
        = \detJ_{11}^{00}
        ,
        &
        Y_3 = \tfrac12 e^{2y_3}
        &
        = \detJ_{11}^{00} + D \detJ_{10}^{00}
        ,
      \end{aligned}
    \end{equation}
    while the amplitudes are obtained from
    \begin{equation}
      \label{eq:GX-3+3-interlacing-amplitudes}
      \begin{aligned}
        Q_1 = 2 m_1 e^{-x_1}
        &
        = \frac{\mu_1 \mu_2}{\lambda_1 \lambda_2 \lambda_3 } \,
        \left( \frac{\detJ_{21}^{11}}{\detJ_{22}^{10}} + C \right)
        ,
        \qquad
        &
        P_1 = 2 n_1 e^{-y_1}
        &= \frac{\detJ_{21}^{11} \detJ_{22}^{10}}{\detJ_{21}^{01} \detJ_{32}^{01}}
        ,
        \\[1ex]
        Q_2 = 2 m_2 e^{-x_2}
        &
        = \frac{\detJ_{11}^{11} \detJ_{21}^{01}}{\detJ_{11}^{10} \detJ_{22}^{10}}
        ,
        &
        P_2 = 2 n_2 e^{-y_2}
        &
        = \frac{\detJ_{10}^{11} \detJ_{11}^{10}}{\detJ_{10}^{01} \detJ_{21}^{01}}
        ,
        \\[1ex]
        Q_3 = 2 m_3 e^{-x_3}
        &
        =\frac{\detJ_{10}^{01}}{\detJ_{11}^{10}}
        ,
        &
        P_3 = 2 n_3 e^{-y_3}
        &
        = \frac{1}{\detJ_{10}^{00}}
        .
      \end{aligned}
    \end{equation}
  \end{subequations}
  Plots of the positions $x = x_k(t)$ and $x = y_k(t)$
  are shown in Figure~\ref{fig:GX-3+3-interlacing-positions-all},
  with the parameter values
  \begin{equation}
    \label{eq:GX-3+3-interlacing-spectral-data}
    \begin{gathered}
      \lambda_1 = \frac{1}{5}
      ,\quad
      \lambda_2 = 1
      ,\quad
      \lambda_3 = 2
      ,\qquad
      \mu_1 = \frac{1}{3}
      ,\quad
      \mu_2 = 4
      ,
      \\
      a_1(0) = 10^{-4}
      ,\quad
      a_2(0) = 10^{1}
      ,\quad
      a_3(0) = 10^{3}
      ,\qquad
      b_1(0) = 10^{-6}
      ,\quad
      b_2(0) = 10^{2}
      ,
      \\
      C = 10^{20}
      ,\quad
      D = 10^{18}
      .
    \end{gathered}
  \end{equation}
  Asymptotically, as $t \to \pm \infty$,
  these curves will approach certain straight lines $x = ct + d$,
  whose coefficients $c$ and~$d$ are given by the formulas in  
  Theorem~\ref{thm:asymptotics-singletons-even}
  (due to Lundmark and Szmigielski~\cite{lundmark-szmigielski:2017:GX-dynamics-interlacing}).
  For details about the constant terms~$d$,
  see Theorem~\ref{thm:asymptotics-singletons-even} or Example~\ref{ex:GX-3+3-typicalY2} below.
  Here we only list the $t$-coefficients $c$ that occur,
  the \emph{asymptotic velocities} of the peakons:
  \begin{equation}
    \label{eq:asymptotic-velocities-example}
    \begin{aligned}
      \frac12 \left( \frac{1}{\lambda_1} + \frac{1}{\mu_1} \right) &= 4
      \quad (\text{twice})
      ,\\
      \frac12 \left( \frac{1}{\lambda_2} + \frac{1}{\mu_1} \right) &= 2
      ,\\
      \frac12 \left( \frac{1}{\lambda_2} + \frac{1}{\mu_2} \right) &= \frac{5}{8}
      ,\\
      \frac12 \left( \frac{1}{\lambda_3} + \frac{1}{\mu_2} \right) &= \frac{3}{8}
      ,\\
      \frac12 \frac{1}{\lambda_3} &= \frac{1}{4}
      .
    \end{aligned}
  \end{equation}
  The meaning of the word ``twice'' here is that, as $t \to -\infty$,
  the curves $x = x_1(t)$ and $x = y_1(t)$ both approach the same line,
  while the other curves approach distinct lines:
  \begin{itemize}
  \item $x = x_1(t)$ and $x = y_1(t)$ both approach the same line $x = 4t + \text{constant}$,
  \item $x = x_2(t)$ approaches a line $x = 2t + \text{constant}$,
  \item $x = y_2(t)$ approaches a line $x = \tfrac58 t + \text{constant}$,
  \item $x = x_3(t)$ approaches a line $x = \tfrac38 t + \text{constant}$,
  \item $x = y_3(t)$ approaches a line $x = \tfrac14 t + \text{constant}$.
  \end{itemize}
  And as $t \to +\infty$, the same asymptotic velocities~$c$ appear, but in the opposite order:
  \begin{itemize}
  \item $x = x_1(t)$ approaches a line $x = \tfrac14 t + \text{constant}$.
  \item $x = y_1(t)$ approaches a line $x = \tfrac38 t + \text{constant}$,
  \item $x = x_2(t)$ approaches a line $x = \tfrac58 t + \text{constant}$,
  \item $x = y_2(t)$ approaches a line $x = 2t + \text{constant}$,
  \item $x = x_3(t)$ and $x = y_3(t)$ both approach the same line $x = 4t + \text{constant}$.
  \end{itemize}
  Only the velocities~$c$ are the same;
  the constant terms~$d$ for the lines as $t \to +\infty$ are not the same as for the corresponding lines
  as $t \to -\infty$.
  The shifts in the constant terms can be
  computed~\cite[Corollary~9.5]{lundmark-szmigielski:2017:GX-dynamics-interlacing},
  and turn out to depend only on the parameters $\lambda_i$, $\mu_j$, $C$ and~$D$.
  
  It is not very meaningful to plot the amplitudes $m_k(t)$ and~$n_k(t)$ directly,
  since they exhibit exponential growth or decay as $t \to \pm \infty$.
  Instead, their asymptotic features are most clearly displayed by plotting the curves
  $s = \ln m_k(t)$ and $s = -\ln n_k(t)$,
  as in Figure~\ref{fig:GX-3+3-interlacing-amplitudes-all}.
  These curves will approach lines $s = ct + d$,
  where the slope~$c$ takes on the following values (in order):
  \begin{equation}
    \label{eq:asymptotic-slopes-example}
    \begin{aligned}
      \frac12 \left( \frac{1}{\lambda_1} - \frac{1}{\mu_1} \right) &= 1
      \quad (\text{twice})
      ,\\
      \frac12 \left( \frac{1}{\lambda_2} - \frac{1}{\mu_1} \right) &= -1
      ,\\
      \frac12 \left( \frac{1}{\lambda_2} - \frac{1}{\mu_2} \right) &= \frac{3}{8}
      ,\\
      \frac12 \left( \frac{1}{\lambda_3} - \frac{1}{\mu_2} \right) &= \frac{1}{8}
      ,\\
      \frac12 \frac{1}{\lambda_3} &= \frac{1}{4}
      .
    \end{aligned}
  \end{equation}
  More precisely: as $t \to -\infty$,
  \begin{itemize}
  \item $s = \ln m_1(t)$ and $s = -\ln n_1(t)$ approach a pair of parallel lines $s = t + \text{constant}$,
  \item $s = \ln m_2(t)$ approaches a line $s = -t + \text{constant}$,
  \item $s = -\ln n_2(t)$ approaches a line $s = \tfrac38 t + \text{constant}$,
  \item $s = \ln m_3(t)$ approaches a line $s = \tfrac18 t + \text{constant}$,
  \item $s = -\ln n_3(t)$ approaches a line $s = \tfrac14 t + \text{constant}$,
  \end{itemize}
  and as $t \to +\infty$,
  \begin{itemize}
  \item $s = \ln m_1(t)$ approaches a line $s = \tfrac14 t + \text{constant}$,
  \item $s = -\ln n_1(t)$ approaches a line $s = \tfrac18 t + \text{constant}$,
  \item $s = \ln m_2(t)$ approaches a line $s = \tfrac38 t + \text{constant}$,
  \item $s = -\ln n_2(t)$ approaches a line $s = -t + \text{constant}$,
  \item $s = \ln m_3(t)$ and $s = -\ln n_3(t)$ approach a pair of parallel lines $s = t + \text{constant}$.
  \end{itemize}
  Again, the expressions for the constant terms can be found in
  Theorem~\ref{thm:asymptotics-singletons-even},
  and the shifts when comparing corresponding lines as $t \to \pm\infty$
  depend only on the parameters $\lambda_i$, $\mu_j$, $C$ and~$D$
  \cite[Corollary~9.9]{lundmark-szmigielski:2017:GX-dynamics-interlacing}.
\end{example}

\begin{figure}[H]
  \centering
  \includegraphics[width=13cm]{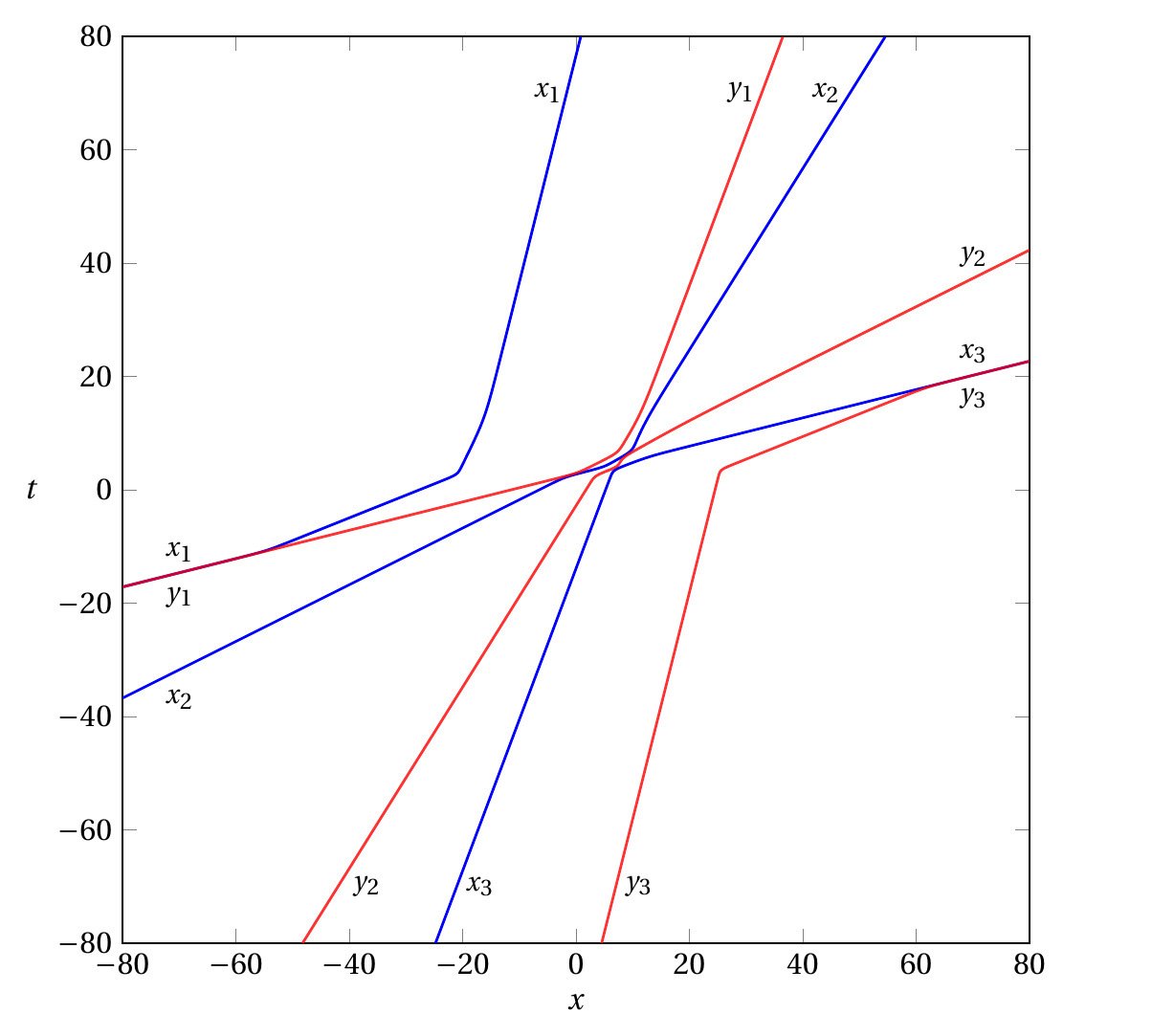}
  \caption{\textbf{Positions in the (even) interlacing case.}
    Spacetime plot, i.e., a plot in the $(x,t)$ plane,
    of the positions $x=x_k(t)$ (blue)
    and $x=y_k(t)$ (red)
    for the 3+3 interlacing peakon solution~\eqref{eq:GX-3+3-interlacing-joint}
    given by the parameter values~\eqref{eq:GX-3+3-interlacing-spectral-data}
    in Example~\ref{ex:GX-3+3-interlacing}.
    Like all the graphics in this article,
    the curves are computed from the exact solution formulas.
    The curves actually never cross or touch,
    since $x_1(t) < y_1(t) < x_2(t) < y_2(t) < x_3(t) < y_3(t)$
    is known to hold for all~$t$.
    As $t \to \pm\infty$, the curves approach certain straight lines $x = ct + d$;
    in this example,
    the asymptotic velocities that occur are
    $c \in \bigl\{ 4,2,\tfrac58,\tfrac38,\tfrac14 \bigr\}$;
    see~\eqref{eq:asymptotic-velocities-example}.
    As $t \to -\infty$, the curves
    $x = x_1(t)$ and $x = y_1(t)$ approach the same straight line
    (with the fastest velocity, $c=4$),
    while each of the other peakons has its own asymptotic velocity.
    Similarly for the curves $x = x_3(t)$ and $x = y_3(t)$ as $t \to +\infty$.
  }
  \label{fig:GX-3+3-interlacing-positions-all}
\end{figure}

\begin{figure}[H]
  \centering
  \includegraphics[width=13cm]{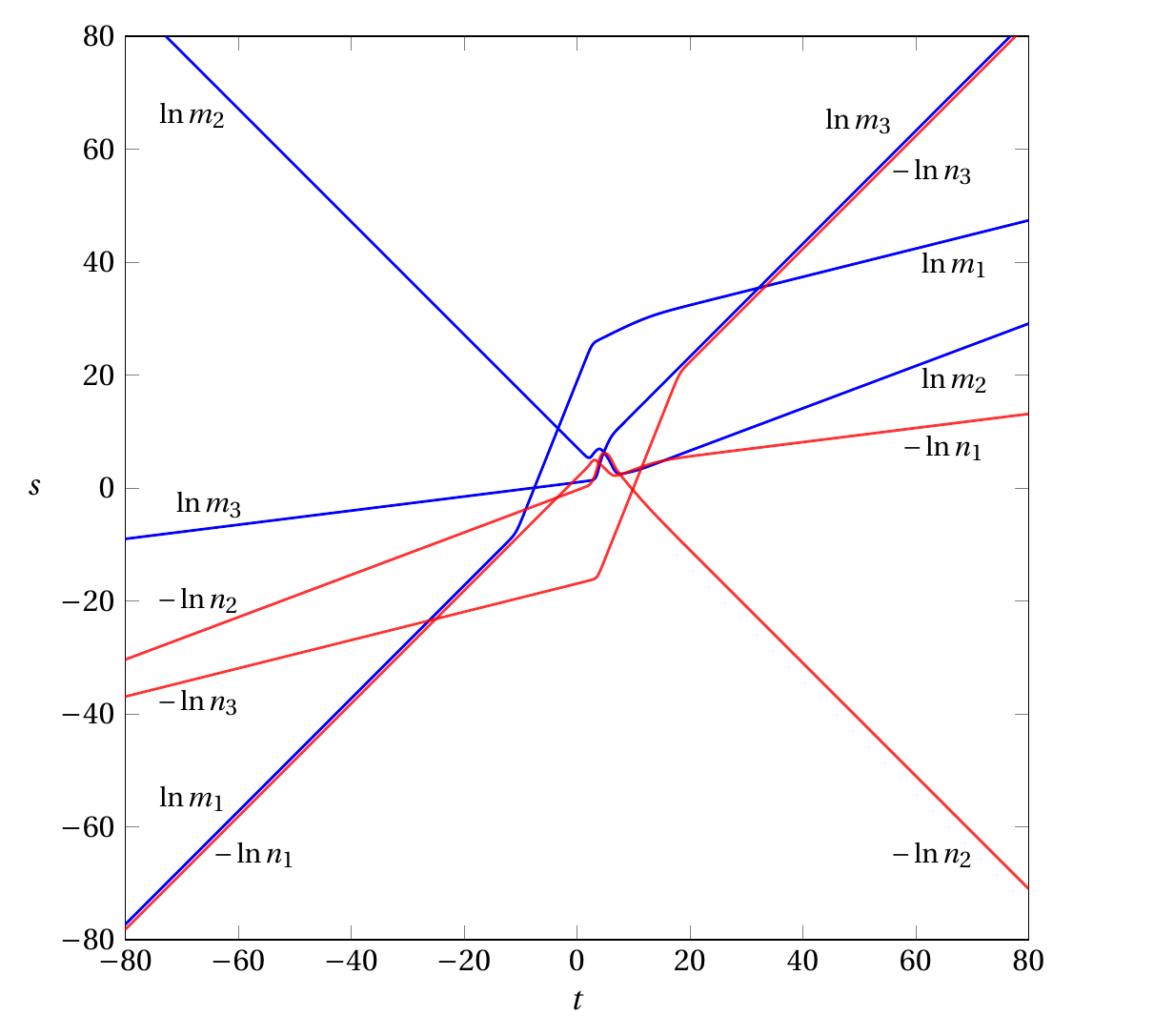}
  \caption{\textbf{Amplitudes in the (even) interlacing case.}
    Plot in the $(t,s)$ plane of the curves $s = \ln m_k(t)$ (blue)
    and $s = -\ln n_k(t)$ (red)
    for the 3+3 interlacing peakon solution
    whose positions were shown in Figure~\ref{fig:GX-3+3-interlacing-positions-all}.
    The peakon amplitudes $m_k(t)$ and~$n_k(t)$ asymptotically have exponential grow or decay,
    so by plotting their logarithms we obtain curves that approach
    certain straight lines $s = ct + d$;
    in this example,
    the slopes that occur are
    $c \in \bigl\{ 1,-1,\tfrac38,\tfrac18,\tfrac14 \bigr\}$;
    see~\eqref{eq:asymptotic-slopes-example}.
    As $t \to -\infty$, the curves
    $s = \ln m_1(t)$ and $s = -\ln n_1(t)$ approach parallel lines
    (with slope~$1$ in this example),
    and likewise for the curves $s = \ln m_3(t)$ and $s = -\ln n_3(t)$ as $t \to +\infty$.
    The other curves approach lines whose slopes in the generic case
    (like here) are all distinct,
    although some slopes may happen to coincide
    for certain choices of the eigenvalues $\lambda_i$ and~$\mu_j$.
  }
  \label{fig:GX-3+3-interlacing-amplitudes-all}
\end{figure}

\section{Additional notation for the non-interlacing case}
\label{sec:more-notation}

The solution formulas for the $K+K$ interlacing case in
Section~\ref{sec:interlacing-review} contained $4K$ spectral variables with
$K+(K-1)$ eigenvalues,
and these parameters will also appear in our solution formulas for the
even non-interlacing case with $K+K$ groups:
\begin{equation}
  \label{eq:spectral-parameters-even}
  \begin{gathered}
    0 < \lambda_1 < \lambda_2 < \dots < \lambda_K
    ,
    \qquad
    0 < \mu_1 < \mu_2 < \dots < \mu_{K-1}
    ;
    \\
    a_1,a_2,\dots,a_K \in \R_+
    ,
    \qquad
    b_1,b_2,\dots,b_{K-1} \in \R_+
    ,
    \qquad
    C, D  \in \R_+
    .
  \end{gathered}
\end{equation}
In the odd case, with $(K+1)+K$ groups,
there will be
$4K+2$ spectral variables with $K+K$ eigenvalues:
\begin{equation}
  \label{eq:spectral-parameters-odd}
  \begin{gathered}
    0 < \lambda_1 < \lambda_2 < \dots < \lambda_K
    ,
    \qquad
    0 < \mu_1 < \mu_2 < \dots < \mu_{K}
    ;
    \\
    a_1,a_2,\dots,a_K \in \R_+
    ,
    \qquad
    b_1,b_2,\dots,b_{K} \in \R_+
    ,
    \qquad
    C, D  \in \R_+
    .
  \end{gathered}
\end{equation}
As before, the variables $\{ a_i, b_j \}$ will have the exponential
time dependence given by~\eqref{eq:residues-time-dependence}.

In addition,
there will be $2(N_j-1)$ internal parameters
associated with each group containing $N_j \ge 2$ peakons,
i.e., with each non-singleton group.
For the $j$th $X$-group from the left,
these internal parameters will be called
\begin{equation}
  \bigl\{ \tau^X_{j,i}, \, \sigma^X_{j,i} \bigr\}_{i=1}^{N^X_j-1}
  ,
\end{equation}
and in the $j$th $Y$-group we use
\begin{equation}
  \bigl\{ \tau^Y_{j,i}, \, \sigma^Y_{j,i} \bigr\}_{i=1}^{N^Y_j-1}
  .
\end{equation}
These parameters have to satisfy certain constraints which arise in the proofs.
By studying the solution formulas in
Sections~\ref{sec:solutions-even} and~\ref{sec:solutions-odd}, one can also
verify directly that these constraints are needed for the peakons
to be ordered correctly (cf. Section~\ref{sec:no-collisions}).

\begin{itemize}
\item The basic constraint is that
  \begin{equation}
    \tau^X_{j,i} > 0
    ,\qquad
    0 < \sigma^X_{j,1} < \sigma^X_{j,2} < \dots < \sigma^X_{j,{N_j^X -1}}
  \end{equation}
  for all~$j$ such that $N_j^X \ge 2$,
  and similarly
  \begin{equation}
    \tau^Y_{j,i} > 0
    ,\qquad
    0 < \sigma^Y_{j,1} < \sigma^Y_{j,2} < \dots < \sigma^Y_{j,{N_j^Y -1}}
    ,
  \end{equation}
  whenever $N_j^Y \ge 2$.
  
\item In case there are two non-singleton groups next to each other,
  then the last~$\sigma$ in the left group must be smaller than the
  first~$\tau$ in the right group:
  \begin{equation}
    \label{eq:constraint-last-sigma-first-tau}
    \sigma^X_{j,{N_j^X -1}} < \tau^Y_{j,1}
    ,
    \qquad \text{or} \qquad
    \sigma^Y_{j,{N_j^Y -1}} < \tau^X_{j+1,1}
    .
  \end{equation}

\item In the even case, if the rightmost group is a singleton ($N_K^Y = 1$)
  and the second rightmost group is not ($N_K^X \ge 2$),
  then the parameter~$D$ must satisfy
  \begin{equation}
    \label{eq:constraint-last-sigma-D-even}
    \sigma^X_{K,{N_K^X -1}} < D
    .
  \end{equation}
  Similarly, in the odd case, if $N_{K+1}^X = 1$ and $N_K^Y \ge 2$, then we require
  \begin{equation}
    \label{eq:constraint-last-sigma-D-odd}
    \sigma^Y_{K,{N_K^Y-1}} < D
    .
  \end{equation}

\item For both the even and the odd case,
  if the leftmost $X$-group is a singleton ($N_1^X = 1$)
  and the leftmost $Y$-group
  is not ($N_1^Y \ge 2$), then
  \begin{equation}
    \label{eq:constraint-C-simpler}
    M < C \, \tau_{1,1}^Y
    .
  \end{equation}
  If the leftmost $X$-group contains $N_1^X \ge 3$ peakons, then
  \begin{equation}
    \label{eq:constraint-C-general}
    \tau^X_{1,1} \, M < C \, \sigma^X_{1,1} \, \tau^X_{1,2}
    ,
  \end{equation}
  where $M$ is the product of all~$\mu_j$
  (i.e., $M = \mu_1 \dotsm \mu_{K-1}$ in the even case
  and $M = \mu_1 \dotsm \mu_{K}$ in the odd case).
  If $N_1^X = 2$, then $C>0$ is the only constraint.
  
\item A special case is the $1+1$ interlacing solution,
  when the above constraints for $C$ and $D$ merge into $1 < C D$;
  cf. Remark~\ref{rem:interlacing-K1} above.
\end{itemize}

When talking about a specific group,
we may simplify the notation by omitting the label $X$ or~$Y$,
and also the group index~$j$, if these are clear from the context.
Thus, in that case we only write $\tau_i$ instead of~$\tau_{j,i}^X$,
for example.
The same thing applies for the sums $T_i$, $S_i$ and~$R_i$
in the following definition.

\begin{definition}
  \label{def:T-S-R}
  For $N \ge 2$ and  $i \in \{1,2,\dots,N-1\}$,
  define the abbreviations
  \begin{equation}
    T_i = \sum_{a=1}^{i} \tau_{a}
    ,\qquad
    S_i = \sum_{b=1}^{i-1} \tau_{b+1} \sigma_{b}
    ,\qquad
    R_{i} = \sigma_i T_i-S_i
    ,
  \end{equation}
  and let $\sigma_0=0$.
  Note that $S_1=0$ since it is an empty sum. Some formulas will also contain $R_0$,
  where by definition we let $R_0=0$.
\end{definition}

\begin{remark}
  The total number of parameters in the solution formulas will
  equal the total number of degrees of freedom in the system, namely
  twice the number of peakons.
  Indeed, the number of spectral parameters is twice the number of
  groups, and each group with $N \ge 2$ peakons
  (i.e., with $N-1$ peakons more than in the interlacing case)
  contributes an
  additional $2(N-1)$ internal parameters.
\end{remark}

We will also use abbreviations corresponding to those defined for
singleton peakons in~\eqref{eq:XYQP}.
\begin{definition}
  \label{def:XYQP-groups}
  Let
  \begin{equation}
    X_{k,i} = \tfrac12 \exp 2 x_{k,i}
    ,\qquad
    Q_{k,i} = 2 m_{j,i} \exp(-x_{k,i})
    ,
  \end{equation}
  and
  \begin{equation}
    Y_{k,i} = \tfrac12 \exp 2 y_{k,i}
    ,\qquad
    P_{k,i} = 2 n_{j,i} \exp(-y_{k,i})
    .
  \end{equation}
\end{definition}

\begin{remark}
  \label{rem:D-is-extra-tau}
  As can be seen in~\eqref{eq:even-Y-rightmost-group-pos},
  for example,
  the parameter~$D$ enters into the solution formulas as if it were an
  additional $\tau$-parameter for the rightmost group,
  i.e., as if that group had parameters
  \begin{equation*}
    \tau_1, \dots, \tau_{N-1}, \tau_N = D
  \end{equation*}
  instead of just $\tau_1,\dots,\tau_{N-1}$.
  In particular, when the rightmost group is a singleton ($N=1$),
  there will be no ``proper'' $\tau$-parameters,
  but the parameter~$D$ will still play a role similar to
  the one played by~$\tau_1$ for $N \ge 2$.
  In this sense,
  the constraints \eqref{eq:constraint-last-sigma-D-even}
  and~\eqref{eq:constraint-last-sigma-D-odd}
  can be viewed as special cases of~\eqref{eq:constraint-last-sigma-first-tau}.

  The constraints on the parameter~$C$ are not as easily interpreted;
  see however the discussions in connection with
  \eqref{eq:char-within-leftmost}
  and~\eqref{eq:leftmost-singleton-rewritten-structure}.
\end{remark}

\section{Examples}
\label{sec:examples}

Before stating the general solution formulas in
Sections~\ref{sec:solutions-even} and~\ref{sec:solutions-odd},
we will write them out in several special cases,
and illustrate the solutions and their asymptotics in figures,
to highlight the new features compared to the $K+K$ interlacing case.
The even case is covered in Section~\ref{sec:examples-groups-even},
and the odd case in Section~\ref{sec:examples-groups-odd}.
But we begin with a couple of examples related to the proofs.

\subsection{Examples of the proof technique}

The general proof will be presented in Section~\ref{sec:proofs-even}
for the even case, and in Section~\ref{sec:proofs-odd} for the odd case.
As a preparation, we here give two examples which illustrate
the technique of our proof.
Example~\ref{ex:proof-technique} shows how to turn a selected peakon
into a zero-amplitude ghostpeakon by reparametrizing the spectral
variables and letting one parameter tend to zero.
Example~\ref{ex:proof-technique2} outlines the overall structure of
the proof, where we start from an interlacing configuration and use
several steps like that in Example~\ref{ex:proof-technique} in order
to obtain the desired non-interlacing configuration.

\begin{example}
  \label{ex:proof-technique}
  Suppose we seek the solution for the configuration
  shown in Figure~\ref{fig:noninterlacing-peakons-K=2} in the introduction,
  where there are $2+2$ groups,
  all singletons except that the second $X$-group contains two peakons:
  \begin{equation*}
    x_1 < y_1 <
    \underbrace{x_{2,1} < x_{2,2}}_{\text{second $X$-group}}
    < y_2
    .
  \end{equation*}
  We will obtain this solution via a limiting procedure,
  where we start from the $3+3$ interlacing case
  \begin{equation*}
    x_1< y_1< x_2 < y_2 < x_3 < y_3
  \end{equation*}
  and make the amplitude~$n_2$ of the peakon at~$y_2$
  tend to zero by manipulating the spectral variables.
  Once this is done, we can simply rename the variables
  for the remaining peakons to get the desired non-interlacing solution.
  
  For $K=3$, the interlacing peakon solutions of the Geng--Xue equation are given by the
  formulas~\eqref{eq:GX-3+3-interlacing-joint}
  in Example~\ref{ex:GX-3+3-interlacing}.
  Let us write $\hat{D}$ instead of~$D$
  in these formulas.
  We aim to kill the peakon at~$y_2$, i.e., to make $n_2=0$.
  For that, we replace the (positive) parameters $\lambda_3$, $\mu_2$, $a_3$, $b_2$, $\hat{D}$
  in the solution formulas with the equivalent set of (positive) parameters
  $\epsilon$, $\tau_1$, $\sigma_1$, $\theta$, $D$ defined by
  \begin{equation}
    \label{eq:replace-spectral-data-3+3}
    \lambda_3 = \frac{\theta}{\tau_1 \epsilon}
    ,\qquad
    \mu_2 = \frac{1}{\epsilon}
    ,\qquad
    a_3 = \frac{\tau_1 (\theta + \tau_1)}{\theta} \epsilon
    ,\qquad
    b_2 = \frac{\sigma_1}{\epsilon}
    ,\qquad
    \hat{D} = D - \sigma_1
    ,
  \end{equation}
  where we are going to let $\epsilon \to 0$ later.
  For $\hat{D}$ to be positive, the constraint $\sigma_1 < D$ must also be imposed.

  Actually, \eqref{eq:replace-spectral-data-3+3} is a slight abuse of notation,
  but we will use it throughout the paper for convenience.
  It is really the constants $a_3(0)$ and~$b_2(0)$ that we are redefining,
  \begin{equation*}
    a_3(0) = \frac{\tau_1 (\theta + \tau_1)}{\theta} \epsilon
    ,\qquad
    b_2(0) = \frac{\sigma_1}{\epsilon}
    ,
  \end{equation*}
  so that, for example, $b_2$ in the computation below actually means
  \begin{equation*}
    b_2(t)
    = b_2(0) \, e^{t/\mu_2}
    = \frac{\sigma_1}{\epsilon} \, e^{\epsilon t}
  \end{equation*}
  rather than just $b_2 = \sigma_1 / \epsilon$.
  The extra factor $e^{\epsilon t}$ has been taken into account when writing
  the $\mathcal{O}(\epsilon)$ terms, but it does not really affect the outcome
  in any way, so simply computing with $b_2 = \sigma_1 / \epsilon$
  would give the same result.
  
  As it turns out, this will give us a ghostpeakon at~$y_2$
  parametrized by the constant~$\theta$,
  and we also obtain the formulas for $x_1$, $y_2$, $x_2$, $x_3$ and~$y_3$,
  which will be independent of~$\theta$;
  see~\eqref{eq:solution-xyxxy-position-old-names}.
  
  Let us first see what becomes of the formula for
  $Y_2 = \frac12 e^{2y_2}$ in~\eqref{eq:GX-3+3-interlacing-positions}:
  \begin{equation}
    \hat{Y}_2 = \frac{\hat{\detJ}_{21}^{00}}{\hat{\detJ}_{10}^{11}}
    .
  \end{equation}
  Here we have written
  $\hat{\detJ}_{rs}^{ij} = \detJ[3,2,r,s,i,j]$
  with a hat,
  to distinguish it from
  $\detJ_{rs}^{ij} = \detJ[2,1,r,s,i,j]$
  without a hat, which will appear in our final formulas.
  Similarly, $\hat{Y}_2$ is the expression that we start with,
  and $Y_2$ will denote the final result after letting $\epsilon \to 0$.
  We need to expand the sums $\hat{\detJ}_{21}^{00}$
  and~$\hat{\detJ}_{10}^{11}$ explicitly in terms of
  $a_1$, $a_2$, $a_3$, $\lambda_1$, $\lambda_2$, $\lambda_3$,
  $b_1$, $b_2$, $\mu_1$, $\mu_2$, $C$ and~$D$.
  Untangling Definition~\ref{def:heineintegral}, we get
  \begin{equation}
    \label{eq:heineintegral0021-explicit}
    \begin{split}
      \hat{\detJ}_{21}^{00} &
      = \frac{(\lambda_1-\lambda_2)^2}{(\lambda_1+\mu_1)(\lambda_2+\mu_1)} a_1 a_2 b_1
      + \frac{(\lambda_1-\lambda_2)^2}{(\lambda_1+\mu_2)(\lambda_2+\mu_2)} a_1 a_2 b_2
      \\ &
      + \frac{(\lambda_1-\lambda_3)^2}{(\lambda_1+\mu_1)(\lambda_3+\mu_1)} a_1 a_3 b_1
      + \frac{(\lambda_2-\lambda_3)^2}{(\lambda_2+\mu_1)(\lambda_3+\mu_1)} a_2 a_3 b_1
      \\ &
      + \frac{(\lambda_1-\lambda_3)^2}{(\lambda_1+\mu_2)(\lambda_3+\mu_2)} a_1 a_3 b_2
      + \frac{(\lambda_2-\lambda_3)^2}{(\lambda_2+\mu_2)(\lambda_3+\mu_2)} a_2 a_3 b_2
    \end{split}
  \end{equation}
  and
  \begin{equation}
    \begin{split}
      \hat{\detJ}_{10}^{11} = a_1 \lambda_1 + a_2 \lambda_2 + a_3 \lambda_3
      .
    \end{split}
  \end{equation}
  Starting with $\hat{\detJ}_{21}^{00}$, we rewrite it a little,
  use the substitution~\eqref{eq:replace-spectral-data-3+3}, and let $\epsilon \to 0$:
  \begin{equation}
    \begin{split}
      \hat{\detJ}_{21}^{00} &
      = \frac{(\lambda_1 - \lambda_2)^2}{(\lambda_1 + \mu_1)(\lambda_2 + \mu_1)} a_1 a_2 b_1
      + \frac{(\lambda_1 - \lambda_2)^2}{(\frac{\lambda_1}{\mu_2} + 1)(\frac{\lambda_2}{\mu_2} + 1) \, \mu_2^2}  a_1 a_2 b_2
      \\ & \quad
      + \frac{\bigl( \frac{\lambda_1}{\lambda_3} - 1 \bigr)^2 \, \lambda_3^2}{(\lambda_1 + \mu_1)(1 + \frac{\mu_1}{\lambda_3}) \, \lambda_3} a_1 a_3 b_1
      + \frac{\bigl(\frac{\lambda_2}{\lambda_3 } - 1 \bigr)^2 \, \lambda_3^2}{(\lambda_2 + \mu_1)(1 + \frac{\mu_1}{\lambda_3}) \, \lambda_3} a_2 a_3 b_1
      \\ & \quad
      + \frac{\bigl(\frac{\lambda_1}{\lambda_3} - 1 \bigr)^2 \lambda_3^2}{(\frac{\lambda_1}{\mu_2} + 1)(\lambda_3 + \mu_2) \, \mu_2} a_1 a_3 b_2
      + \frac{\bigl( \frac{\lambda_2}{\lambda_3} - 1 \bigr)^2 \, \lambda_3^2}{(\frac{\lambda_2}{\mu_2} + 1)(\lambda_3 + \mu_2) \, \mu_2} a_2 a_3 b_2
      \\ &
      = \underbrace{\frac{\bigl( \lambda_1 - \lambda_2 \bigr)^2}{(\lambda_1 + \mu_1)(\lambda_2 + \mu_1)} a_1 a_2 b_1}_{\detJ_{21}^{00}}
      + \underbrace{\left( \frac{a_1 b_1}{\lambda_1 + \mu_1} + \frac{a_2 b_1}{\lambda_2 + \mu_1} \right)}_{\detJ_{11}^{00}} (\theta + \tau_1) \bigl( 1 + \mathcal{O}(\epsilon) \bigr)
      \\ & \quad
      + \underbrace{(a_1 + a_2)}_{\detJ_{10}^{00}} \, (\sigma_1 \, \theta) \bigl( 1 + \mathcal{O}(\epsilon) \bigr)
      \\ &
      \to
      \detJ_{21}^{00} + (\theta + \tau_1) \detJ_{11}^{00} + \sigma_1 \theta \detJ_{10}^{00}
      \,
      ,\qquad
      \text{as $\epsilon \to 0$}
      .
    \end{split}
  \end{equation}
  In $\hat{\detJ}_{10}^{11}$, the parameter $\epsilon$ happens to cancel out completely
  even before we take the limit:
  \begin{equation}
    \begin{split}
      \hat{\detJ}_{10}^{11}&
      = a_1 \lambda_1 + a_2 \lambda_2 + a_3 \lambda_3
      \\ &
      = \underbrace{(a_1 \lambda_1 + a_2 \lambda_2)}_{\detJ_{10}^{11}} + \theta + \tau_1
      .
    \end{split}
  \end{equation}
  This gives the formula for $y_2$,
  which will be the ghostpeakon:
  \begin{equation}
    \begin{split}
      Y_2 &
      = Y_2^{\text{ghost}}
      = \frac{1}{2} \exp 2y_2^{\text{ghost}}
      = \lim_{\epsilon \to 0} \hat{Y}_2
      \\ &
      = \lim_{\epsilon \to 0}  \frac{\hat{\detJ}_{21}^{00}}{\hat{\detJ}_{10}^{11}}
      = \frac{\detJ_{21}^{00} + (\theta + \tau_1) \detJ_{11}^{00} + \sigma_1 \theta \detJ_{10}^{00}}{\detJ_{10}^{11} + \theta + \tau_1}
      ,
    \end{split}
  \end{equation}
  where the parameters $\theta$, $\sigma_1$, $\tau_1$ are positive.
  Similarly, for $\hat{X}_3=\frac12 \exp{2 \hat{x}_3}$ we have
  \begin{equation}
    \begin{split}
      \hat{X}_3 &
      = \hat{\detJ}_{11}^{00}
      \\ &
      = \frac{1}{\lambda_1 + \mu_1} a_1 b_1
      + \frac{1}{\lambda_2 + \mu_1} a_2 b_1
      + \frac{1}{\lambda_3 + \mu_1} a_3 b_1
      \\ & \quad
      + \frac{1}{\lambda_1 + \mu_2} a_1 b_2
      + \frac{1}{\lambda_2 + \mu_2} a_2 b_2
      + \frac{1}{\lambda_3 + \mu_2} a_3 b_2
      \\ &
      = \frac{1}{\lambda_1 + \mu_1} a_1 b_1
      + \frac{1}{\lambda_2 + \mu_1} a_2 b_1
      + \frac{a_3}{\lambda_3} \frac{b_1}{\bigl( 1 + \frac{\mu_1}{\lambda_3} \bigr)}
      \\ & \quad
      + \frac{a_1}{\bigl( \frac{\lambda_1}{\mu_2} + 1 \bigr)} \frac{b_2}{\mu_2}
      + \frac{a_2}{\bigl( \frac{\lambda_2}{\mu_2} + 1 \bigr)} \frac{b_2}{\mu_2}
      + \frac{1}{\lambda_3 + \mu_2} a_3 b_2
      \\ &
      \to
      \underbrace{\frac{1}{\lambda_1 + \mu_1} a_1 b_1 + \frac{1}{\lambda_2 + \mu_1} a_2 b_1}_{\detJ_{11}^{00}}
      + 0
      + \underbrace{(a_1 + a_2)}_{\detJ_{10}^{00}} \sigma_1
      + 0
      \\ &
      = \detJ_{11}^{00} + \sigma_1 \detJ_{10}^{00}
      = X_3
      ,
      \qquad
      \text{as $\epsilon \to 0$}
      .
    \end{split}
  \end{equation}
  In the same way, we find the following expressions for the other variables,
  in the limit $\epsilon\to 0$:
  \begin{equation}
    \label{eq:solution-xyxxy-position-old-names}
    \begin{aligned}
      X_1 &= \frac{\detJ_{21}^{00}}{\detJ_{10}^{11} + C \, \detJ_{11}^{10}}
      ,
      \qquad
      &
      Y_1 &= \frac{\detJ_{21}^{00}}{\detJ_{10}^{11}}
      ,
      \\[1ex]
      X_2 &= \frac{ \detJ_{21}^{00} + \tau_1 \detJ_{11}^{00}}{\detJ_{10}^{11} + \tau_1}
      ,
      &
      Y_2 &= \frac{\detJ_{21}^{00} + (\theta + \tau_1) \detJ_{11}^{00} + \sigma_1 \theta \detJ_{10}^{00}}{\detJ_{10}^{11} + \theta + \tau_1}
      ,
      \\[1ex]
      X_3 &= \detJ_{11}^{00} + \sigma_1 \detJ_{10}^{00}
      ,
      &
      Y_3 &= \detJ_{11}^{00} + D \detJ_{10}^{00}
      .
    \end{aligned}
  \end{equation}
  For the expressions
  $Q_{k} = 2 m_{k} \, e^{-x_{k}}$  and $P_{k} = 2 n_{k} \, e^{-y_{k}}$
  involving the amplitudes, we get
  \begin{equation}
    \label{eq:solution-xyxxy-amplitude-old-names}
    \begin{aligned}
      Q_1 &= \frac{\mu_1 }{\lambda_1 \lambda_2 } \left( \frac{\detJ_{10}^{11}}{ \detJ_{11}^{10}} + C \right)
      ,
      &
      P_1 &= \frac{\detJ_{10}^{11} \detJ_{11}^{10}}{\detJ_{10}^{01} \detJ_{21}^{01}}
      ,
      \\[1ex]
      Q_2 &= \frac{\sigma_1 \detJ_{10}^{01} \left( \detJ_{10}^{11}+ \tau_1 \right)}{\detJ_{11}^{10} \left( \detJ_{11}^{10} + \sigma_1 \left( \detJ_{10}^{10} + \tau_1 \right) \right)}
      ,
      \qquad
      &
      P_2 &= 0
      ,
      \\[1ex]
      Q_3 &= \frac{\detJ_{10}^{01}}{\detJ_{11}^{10} + \sigma_1 \left( \detJ_{10}^{10} + \tau_1 \right)}
      ,
      &
      P_3 &= \frac{1}{\detJ_{10}^{00}}
      .
    \end{aligned}
  \end{equation}
  Note in particular that $P_2=0$, which implies that $n_2=0$, as desired.
  Note also that $Y_2$ (and hence~$y_2$) is the only quantity which depends on~$\theta$.
  Ignoring the ghostpeakon at~$y_2$, the remaining positions are
  \begin{equation}
    x_1 < y_1 < \underbrace{x_2 < x_3} < y_3
    .
  \end{equation}
  Renaming them to
  \begin{equation*}
    x_1 < y_1 < \underbrace{x_{2,1} < x_{2,2}}_{\text{second $X$-group}} < y_2
    ,
  \end{equation*}
  and similarly for the corresponding amplitudes,
  we find that the non-interlacing peakon solution
  \begin{subequations}
    \label{eq:solution-xyxxy-joint}
    \begin{equation}
      \label{eq:peakon-ansatz-xyxxy}
      \begin{split}
        u(x,t) &= m_1(t) \, e^{-\abs{x-x_1(t)}} + m_{2,1}(t) \, e^{-\abs{x-x_{2,1}(t)}} + m_{2,2}(t) \, e^{-\abs{x-x_{2,2}(t)}}
        ,
        \\[1ex]
        v(x,t) &= n_1(t) \, e^{-\abs{x-y_1(t)}} + n_2(t) \, e^{-\abs{x-y_2(t)}}
      \end{split}
    \end{equation}
    is given by the explicit formulas
    \begin{equation}
      \label{eq:solution-xyxxy-position}
      \begin{aligned}
        \tfrac12 e^{2x_1} = X_1
        &
        = \frac{\detJ_{21}^{00}}{\detJ_{10}^{11} + C \, \detJ_{11}^{10}}
        ,\qquad
        &
        \tfrac12 e^{2y_1} = Y_1
        &
        = \frac{\detJ_{21}^{00}}{\detJ_{10}^{11}}
        ,
        \\[1ex]
        \tfrac12 e^{2x_{2,1}} = X_{2,1}
        &
        = \frac{ \detJ_{21}^{00} + \tau_1 \detJ_{11}^{00}}{\detJ_{10}^{11} + \tau_1}
        ,
        \\[1ex]
        \tfrac12 e^{2x_{2,1}} = X_{2,2}
        &
        = \detJ_{11}^{00} + \sigma_1 \detJ_{10}^{00}
        ,
        &
        \tfrac12 e^{2y_2} = Y_2
        &
        = \detJ_{11}^{00} + D \detJ_{10}^{00}
      \end{aligned}
    \end{equation}
    and
    \begin{equation}
      \label{eq:solution-xyxxy-amplitude}
      \begin{aligned}
        2 m_1 \, e^{-x_1} = Q_1
        &
        = \frac{\mu_1 }{\lambda_1 \lambda_2}
        \left( \frac{\detJ_{10}^{11}}{\detJ_{11}^{10}} + C \right)
        ,
        &
        2 n_2 \, e^{-y_2} = P_1
        &
        = \frac{\detJ_{10}^{11} \detJ_{11}^{10}}{\detJ_{10}^{01} \detJ_{21}^{01}}
        ,
        \\[1ex]
        2 m_{2,1} \, e^{-x_{2,1}} = Q_{2,1}
        &= \frac{\sigma_1 \detJ_{10}^{01} \left( \detJ_{10}^{11} + \tau_1 \right)}{\detJ_{11}^{10} \left( \detJ_{11}^{10} + \sigma_1 \left( \detJ_{10}^{10} + \tau_1 \right) \right)}
        ,
        \quad
        \\[1ex]
        2 m_{2,2} \, e^{-x_{2,2}} = Q_{2,2}
        &
        = \frac{\detJ_{10}^{01}}{\detJ_{11}^{10} + \sigma_1 \left( \detJ_{10}^{10} + \tau_1 \right)}
        ,
        &
        2 n_2 \, e^{-y_2} = P_2
        &
        = \frac{1}{\detJ_{10}^{00}}
        .
      \end{aligned}
    \end{equation}
  \end{subequations}
  Here, $\detJ_{rs}^{ij}=\detJ[2,1,r,s,i,j]$,
  so these formulas depend on the spectral parameters
  $\lambda_1$, $\lambda_2$, $\mu_1$,
  $a_1$, $a_2$,~$b_1$,
  as well as on the positive parameters
  $\tau_1 = \tau_{2,1}^X$ and~$\sigma_1 = \sigma_{2,1}^X$
  which are internal to the second $X$-group.
  The substitution $\hat{D} = D - \sigma_1$
  is necessary in order to get a formula for $Y_2$ free from
  the parameter~$\sigma_1$ which we want to keep confined to the $X_2$-group.
  As already pointed out just after~\eqref{eq:replace-spectral-data-3+3},
  this gives rise to the constraint $\sigma_1<D$,
  and we can also see in the formulas for $X_{2,2}$ and~$Y_2$
  in~\eqref{eq:solution-xyxxy-position}
  that this constraint is necessary in order for $x_{2,2} < y_2$ to hold.
  
  Figure~\ref{fig:proof-example-positions} shows the curves
  $x=x_1(t)$, $x=y_1(t)$, $x=x_{2,i}(t)$ and $x=y_2(t)$
  describing the positions of the peakons,
  as given by the explicit solution formulas~\eqref{eq:solution-xyxxy-position}.
  We have taken the spectral data~\eqref{eq:GX-3+3-interlacing-spectral-data}
  from Example~\ref{ex:GX-3+3-interlacing} as our starting point, and then
  applied~\eqref{eq:replace-spectral-data-3+3}
  with $\tau_1 = 10^{10}$ and $\sigma_1 = 10^{5}$,
  so that the parameter values in the figure are
  \begin{equation}
    \label{eq:proof-example-parameters}
    \begin{gathered}
      \lambda_1 = \frac{1}{5}
      ,\quad
      \lambda_2 = 1
      ,\qquad
      \mu_1 = \frac{1}{3}
      ,
      \\
      a_1(0) = 10^{-4}
      ,\quad
      a_2(0) = 10^{1}
      ,\qquad
      b_1(0) = 10^{-6}
      ,
      \\
      C = 10^{20}
      ,\quad
      D = 10^{18} + 10^{5}
      ,
      \\
      \tau_1 = 10^{10}
      ,\quad
      \sigma_1 = 10^{5}
      .
    \end{gathered}
  \end{equation}
  We have also included the curve
  \begin{equation}
    \label{eq:proof-example-ghost}
    x = \xi(t;\theta)
    = \frac12 \ln \frac{2 \bigl( \detJ_{21}^{00} + (\theta + \tau_1) \detJ_{11}^{00} + \sigma_1 \theta \detJ_{10}^{00} \bigr)}{\detJ_{10}^{11} + \theta + \tau_1}
    ,
  \end{equation}
  i.e., the curve obtained from what was called $Y_2$
  in~\eqref{eq:solution-xyxxy-position-old-names},
  before we renamed the variables.
  In other words, it is the position of the ghostpeakon
  whose amplitude became zero in the limit as $\epsilon \to 0$,
  and which is therefore now absent
  from $v(x,t)$ in~\eqref{eq:peakon-ansatz-xyxxy}.
  This is the only quantity which depends on the value of the parameter~$\theta$,
  and in the figure we have used $\theta = 10^{-15}$.
  In fact, as $\theta$ varies in the range $0 < \theta < \infty$,
  the family of curves $x = \xi(t; \theta)$ fills out the region of the $(x,t)$ plane between
  the curves $x = x_{2,1}(t)$ and $x = x_{2,2}(t)$, as illustrated in
  Figure~\ref{fig:proof-example-characteristics};
  note from the formula~\eqref{eq:proof-example-ghost} that
  \begin{equation}
    \lim_{\theta \to 0^+} \xi(t;\theta) = x_{2,1}(t)
    ,\qquad
    \lim_{\theta \to \infty} \xi(t;\theta) = x_{2,2}(t)
    .
  \end{equation}
  These curves are in fact \emph{characteristic curves} associated with this particular
  non-interlacing peakon solution~\eqref{eq:solution-xyxxy-joint},
  i.e., solutions of the ODE $\dot \xi(t) = u(\xi(t),t) \, v(\xi(t),t)$.
  So we obtain the characteristic curves between $x_{2,1}(t)$ and $x_{2,2}(t)$
  as a byproduct of the proof,
  and in fact we can derive formulas for \emph{all} characteristic curves;
  see Section~\ref{sec:characteristic-curves}.
  
  Finally, in Figure~\ref{fig:proof-example-amplitudes} we have plotted the curves
  $s = \ln m_1(t)$, $s = -\ln n_1(t)$, $x = \ln m_{2,i}(t)$ and $s = -\ln n_2(t)$,
  obtained from the solution formulas~\eqref{eq:solution-xyxxy-amplitude}
  together with~\eqref{eq:solution-xyxxy-position},
  for the same parameter values~\eqref{eq:proof-example-parameters}.
  
  These pictures give a first taste of some of the new features appearing
  when the solutions are non-interlacing.
  We will comment upon these features in more detail in the examples that follow below.

  \begin{figure}[H]
    \centering
    \includegraphics[width=13cm]{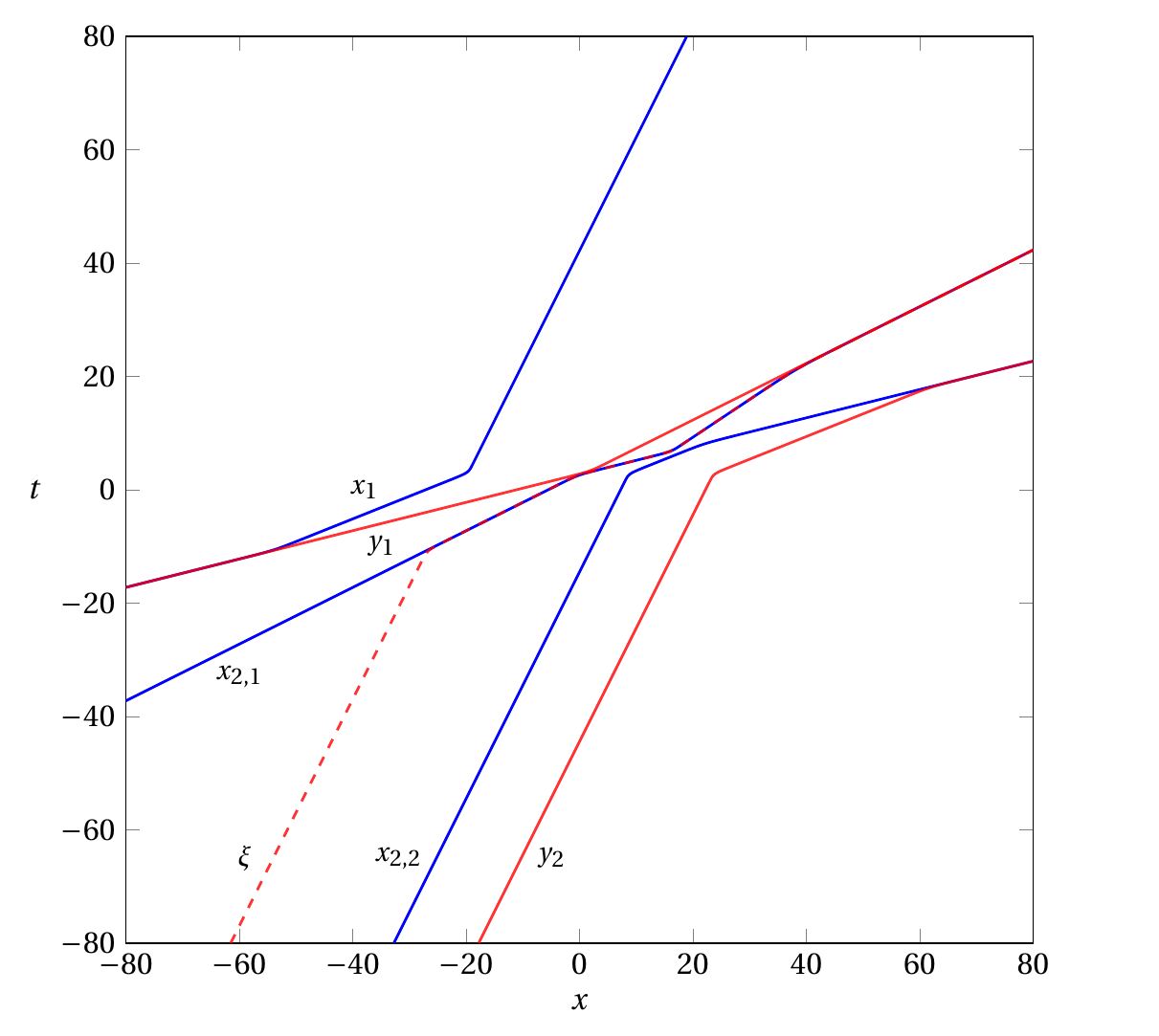}
    \caption{\textbf{Positions in a simple non-interlacing case.}
      Positions of the peakons in the solution~\eqref{eq:solution-xyxxy-joint}
      from Example~\ref{ex:proof-technique},
      with the parameter values~\eqref{eq:proof-example-parameters}.
      The solid blue curves show $X$-peakons and the solid red curves $Y$-peakons.
      The dashed red curve $x = \xi(t)$ is the position of a
      zero-amplitude ``ghostpeakon'',
      the remnant of the peakon that was ``killed'' in order to derive
      this non-interlacing solution with $2+2$ groups from the interlacing $3+3$ solution.
      The ghostpeakon curve must lie between the curves for $x_{2,1}$ and~$x_{2,2}$;
      its exact position depends on the parameter~$\theta$,
      which in this picture has the value $10^{-15}$
      (cf. Figure~\ref{fig:proof-example-characteristics}).
      The strict ordering
      $x_1 < y_1 < x_{2,1} < \xi < x_{2,2} < y_2$
      is preserved for all~$t$
      (see Theorem~\ref{thm:no-collisions}),
      so the curves actually never cross or touch.
      But, for example, the curve $x=y_1(t)$
      approaches the same straight line as the curve $x=x_1(t)$
      as $t \to -\infty$,
      and the same line as the curves $x=x_{2,1}(t)$ and $x = \xi(t)$
      as $t \to +\infty$.
    }
    \label{fig:proof-example-positions}
  \end{figure}

  \begin{figure}[H]
    \centering
    \includegraphics[width=13cm]{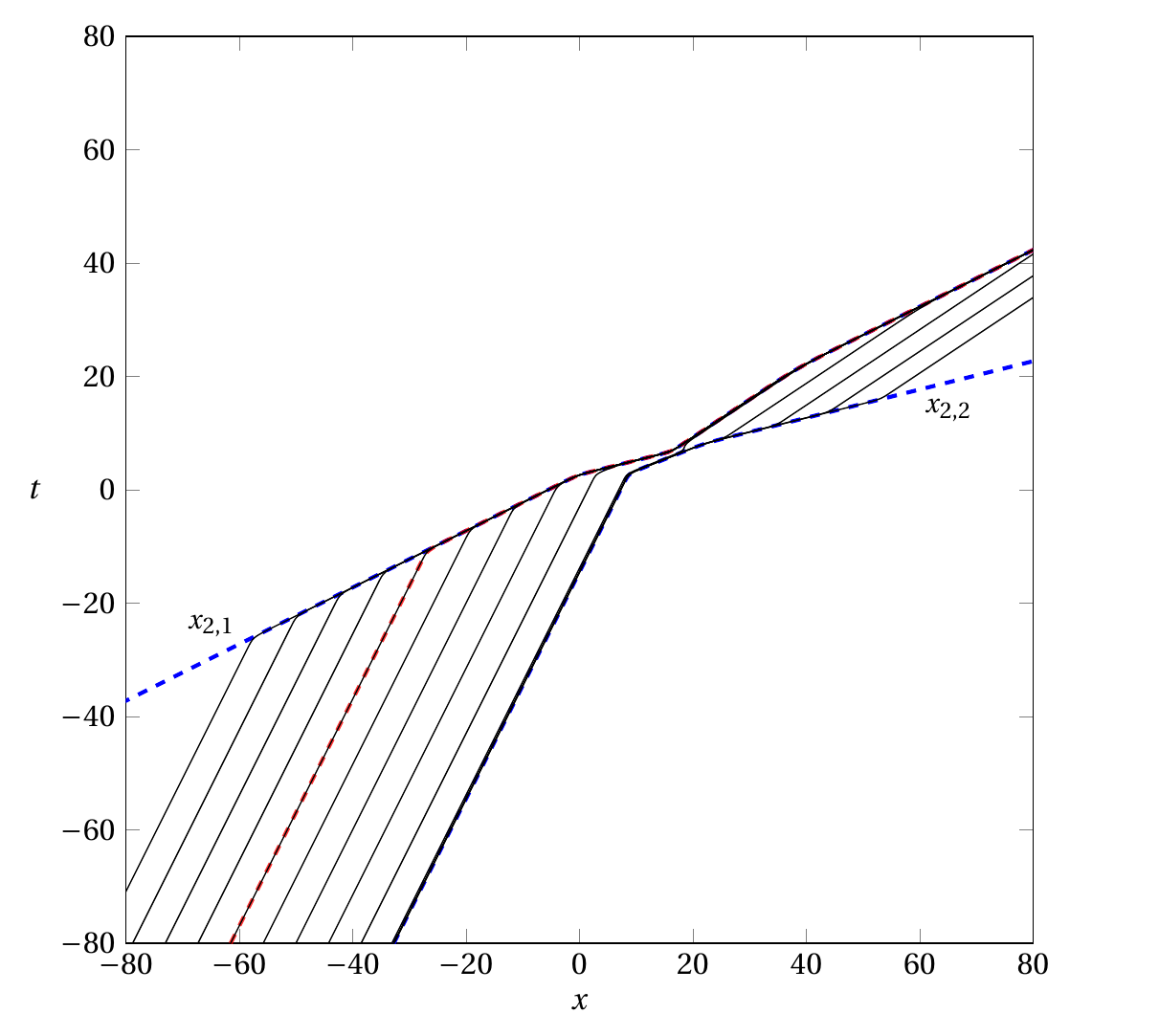}
    \caption{\textbf{Characteristic curves.}
      A selection of characteristic curves $x = \xi(t)$,
      i.e., solutions of the ODE $\dot \xi = u(\xi) \, v(\xi)$,
      for the peakon solution~\eqref{eq:solution-xyxxy-joint}
      in Example~\ref{ex:proof-technique}.
      The dashed blue curves are $x = x_{2,1}(t)$ and $x = x_{2,2}(t)$,
      and the region between them is filled by a
      family of characteristic curves $x = \xi(t; \theta)$
      given by equation~\eqref{eq:proof-example-ghost},
      where $\theta$ varies in the range $0 < \theta < \infty$.
      The formula~\eqref{eq:proof-example-ghost} was obtained as a byproduct of the proof,
      in the form of a ``ghostpeakon''.
      Shown here in solid black are the characteristic curves
      $x = \xi(t; \theta)$
      for $\theta = 10^{-35 + 5k}$, $k = 0,1,\dots,13$,
      so the fifth one from the left ($\theta = 10^{-15}$)
      is the same as the dashed red curve
      in Figure~\ref{fig:proof-example-positions},
      and it is marked here with dashed red as well.
    }
    \label{fig:proof-example-characteristics}
  \end{figure}

  \begin{figure}[H]
    \centering
    \includegraphics[width=13cm]{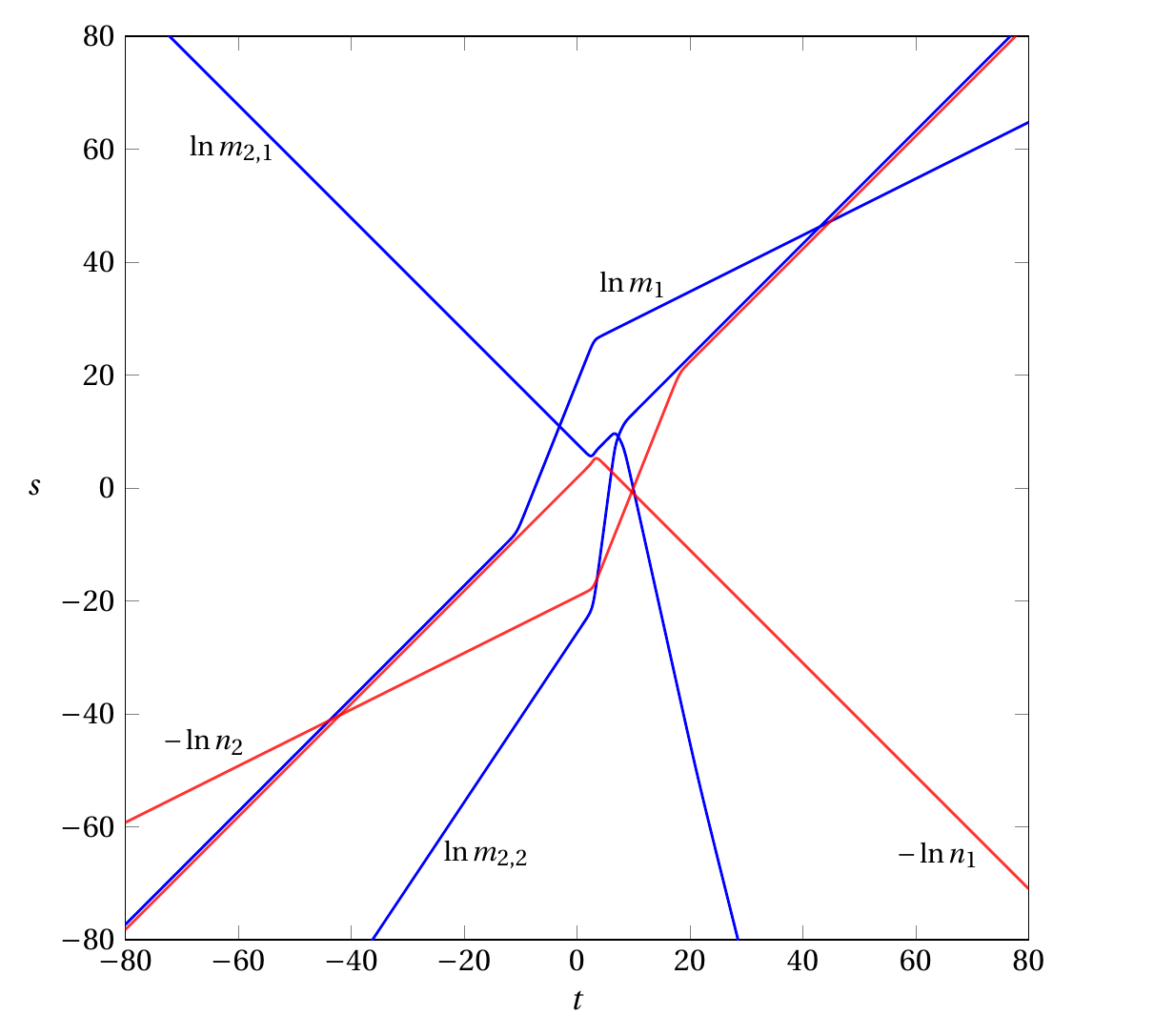}
    \caption{\textbf{Amplitudes in a simple non-interlacing case.}
      Logarithms of the amplitudes
      ($s = \ln m_1(t)$, $s = -\ln n_1(t)$, etc.)
      for the solution~\eqref{eq:solution-xyxxy-joint}
      in example~\ref{ex:proof-technique}
      with the parameters~\eqref{eq:proof-example-parameters},
      i.e., the same solution whose positions were shown in
      Figure~\ref{fig:proof-example-positions}.
    }
    \label{fig:proof-example-amplitudes}
  \end{figure}
\end{example}

\begin{remark}
  The reason that the formulas obtained in the limit $\epsilon \to 0$
  still satisfy the peakon ODEs is the following: for $\epsilon \neq 0$,
  the ODEs are identically satisfied with respect to~$t$ and to all
  the parameters, and the reparametrization is such that all functions
  involved will have a removable singularity at $\epsilon=0$.
\end{remark}

\begin{remark}
  It is perhaps somewhat surprising that the formulas for the singletons
  $x_1$, $y_1$ and~$y_2$ in the non-interlacing solution~\eqref{eq:solution-xyxxy-joint}
  in Example~\ref{ex:proof-technique} are identical to the formulas for $x_1$, $y_1$ and~$y_2$ in
  the $2+2$ interlacing peakon solution with the same spectral parameters
  (which are given by Theorem~\ref{thm:interlacing-solution} with $K=2$,
  and are also written out in detail by Lundmark and
  Szmigielski~\cite[formula~(7.3)]{lundmark-szmigielski:2017:GX-dynamics-interlacing}).
  This is a general phenomenon,
  the reason for which will be explained in Section~\ref{sec:effective}.
\end{remark}

\begin{example}
  \label{ex:proof-technique2}
  Suppose next that we seek the solution formulas for a peakon configuration
  \begin{multline*}
    x_1 <
    \underbrace{y_{1,1} < y_{1,2} < y_{1,3}}_{\text{$Y$-group}} <
    x_2 < y_2 <
    \underbrace{x_{3,1} < x_{3,2} < x_{3,3}}_{\text{$X$-group}} <
    \underbrace{y_{3,1} < y_{3,2} < y_{3,3} < y_{3,4}}_{\text{$Y$-group}}
    ,
  \end{multline*}
  which we can schematically represent with a diagram
  \begin{equation*}
    \X
    \Y \Y \Y
    \X
    \Y
    \X \X \X
    \Y \Y \Y \Y
    ,
  \end{equation*}
  where black dots are $X$-peakons and white dots are $Y$-peakons.
  Clearly there are $3+3$ groups, three of which are singletons.
  Our starting point is the $10+10$ interlacing configuration obtained
  by inserting auxiliary $X$-peakons between adjacent $Y$-peakons,
  and vice versa.
  With braces indicating the original non-singleton groups,
  the diagram then becomes
  \begin{equation*}
    \X
    \underbrace{\Y \X \Y \X \Y}
    \X \Y
    \underbrace{\X \Y \X \Y \X}
    \underbrace{\Y \X \Y \X \Y \X \Y}
    .
  \end{equation*}
  For such a configuration
  $x_1 < y_1 < \dots < x_{10} < y_{10}$,
  the solution formulas are known,
  assuming that all amplitudes $m_k$ and $n_k$
  are nonzero.
  We now ``kill'' the rightmost $X$-peakon,
  i.e., we perform a similar limiting procedure
  as in Example~\ref{ex:proof-technique}
  in order to drive the amplitude $m_{10}$ to zero.
  The details of how to do this will be explained in
  Section~\ref{sec:proofs-even},
  but the result is that we find out the (previously unknown)
  solution formulas for the configuration
  \begin{equation*}
    \X \Y \X \Y \X \Y \X \Y \X \Y \X \Y \X \Y \X \Y \X
    \Y \Z \Y
    .
  \end{equation*}
  Next, we repeat the procedure twice, to kill the next two $X$-peakons
  from the right,
  yielding the formulas for the configuration
  \begin{equation*}
    \X \Y \X \Y \X \Y \X \Y \X \Y \X \Y \X
    \Y \Z \Y \Z \Y \Z \Y
    .
  \end{equation*}
  At this stage, the rightmost $Y$-group with four peakons
  is completed, and the next auxiliary peakon to be killed
  is the $Y$-peakon between the two rightmost remaining $X$-peakons.
  Killing $Y$-peakons is done in a similar way as killing $X$-peakons,
  and we continue in this manner until we are back at the
  original configuration,
  at which point we have the solution formulas that we seek:
  \begin{equation*}
    \begin{aligned}
      & \text{Start} &&
      \X \Y \X \Y \X \Y \X \Y \X \Y \X \Y \X \Y \X \Y \X \Y \X \Y
      ,
      \\
      & \text{Step 1a} &&
      \X \Y \X \Y \X \Y \X \Y \X \Y \X \Y \X \Y \X \Y \X \Y \Z \Y
      ,
      \\
      & \text{Step 1b} &&
      \X \Y \X \Y \X \Y \X \Y \X \Y \X \Y \X \Y \X \Y \Z \Y \Z \Y
      ,
      \\
      & \text{Step 1c} &&
      \X \Y \X \Y \X \Y \X \Y \X \Y \X \Y \X \Y \Z \Y \Z \Y \Z \Y
      ,
      \\
      & \text{Step 2a} &&
      \X \Y \X \Y \X \Y \X \Y \X \Y \X \Z \X \Y \Z \Y \Z \Y \Z \Y
      ,
      \\
      & \text{Step 2b} &&
      \X \Y \X \Y \X \Y \X \Y \X \Z \X \Z \X \Y \Z \Y \Z \Y \Z \Y
      ,
      \\
      & \text{Step 3a} &&
      \X \Y \X \Y \Z \Y \X \Y \X \Z \X \Z \X \Y \Z \Y \Z \Y \Z \Y
      ,
      \\
      & \text{Step 3b (finish)} &&
      \X \Y \Z \Y \Z \Y \X \Y \X \Z \X \Z \X \Y \Z \Y \Z \Y \Z \Y
      .
    \end{aligned}
  \end{equation*}
  It should be clear that what we need to control is how
  the solution formulas change when we kill one peakon,
  i.e., when we perform a typical step in the sequence above.
  We will always be in the situation that we have
  a ``half-finished'' group that we are working on,
  with ``finished'' groups on its right,
  and interlacing singletons on its left.
  When we kill a peakon, we increase the number of peakons in the
  current group by one, since the singleton just to the left of the
  killed peakon is joined to the group.
  This also means that the number of groups decreases by two: a
  singleton of one type ($X$ or~$Y$) disappears, and two groups of the
  other type
  are merged into one.
  
  We will formulate the proof as a kind of induction or descent:
  we assume that the solution for a pre-killing configuration
  with $K+K$ groups
  is given by our general solution formulas (to be formulated in
  Section~\ref{sec:solutions-even}),
  and show that the solution for the post-killing configuration
  is still given by the same general formulas,
  but with $(K-1)+(K-1)$ groups and with $N+1$ peakons
  in the current group instead of~$N$.
  The base case consists of the known solution formulas for the even interlacing case.
\end{example}

\subsection{Examples with an even number of groups}
\label{sec:examples-groups-even}

In this section we present solution formulas, graphics and asymptotics
for several different examples with an even number of groups,
mainly for the $3+3$ case, where we successively add more features
to the
$3+3$ interlacing solution from Example~\ref{ex:GX-3+3-interlacing}.

Examples~\ref{ex:GX-3+3-typicalY2},
\ref{ex:GX-3+3-secondrightmostX3}
and~\ref{ex:GX-3+3-rightmostY3}
illustrate solutions where all groups but one are singletons;
the non-singleton group in these examples (with five peakons in each case)
is $Y_2$, $X_3$ and $Y_3$, respectively.
Example~\ref{ex:GX-3+3-X2Y2} shows a configuration with two adjacent non-singleton groups
(two $X_2$-peakons and five $Y_2$-peakons), the remaining groups being singletons.
In Example~\ref{ex:GX-3+3-allgroups} all six groups are non-singletons,
and in this example we also consider the asymptotic behaviour of the amplitudes in detail;
the previous examples focus on the positions.
Finally, in Example~\ref{ex:GX-1+1-allgroups} we illustrate the somewhat special $1+1$
case with just one $X$-group and one $Y$-group.

\begin{figure}
  \centering
  \includegraphics[width=13cm]{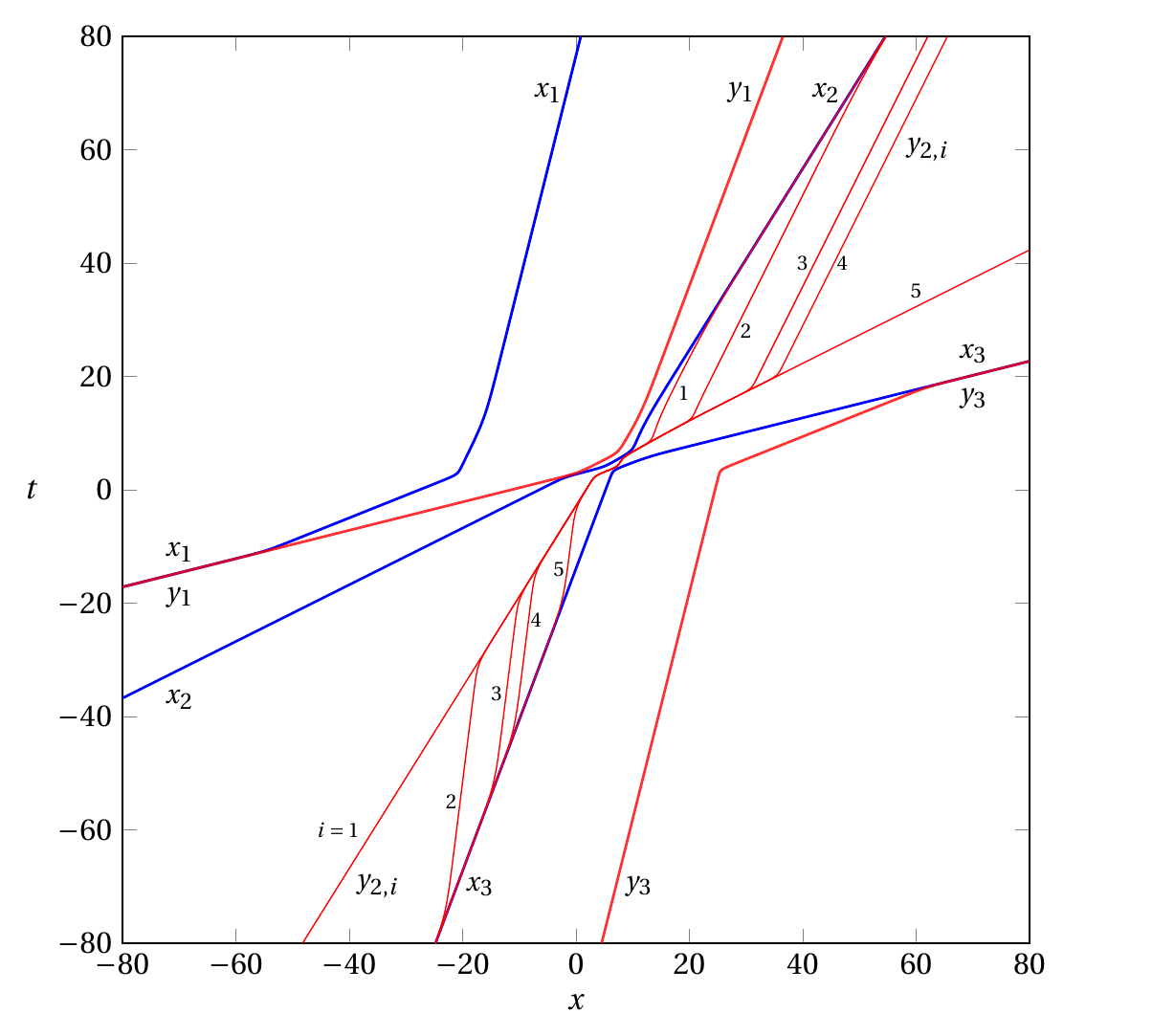}
  \caption{\textbf{Typical group in the middle.}
    Positions of the peakons for the non-interlacing peakon
    solution~\eqref{eq:GX-3+3-typicalY2-solution-joint}
    described in Example~\ref{ex:GX-3+3-typicalY2},
    with the parameter values
    \eqref{eq:GX-3+3-interlacing-spectral-data}
    and~\eqref{eq:GX-3+3-typicalY2-tau-sigma}.
    Compared to the interlacing solution in Figure~\ref{fig:GX-3+3-interlacing-positions-all},
    the only difference is that the singleton $y_2$ is
    replaced by a group of five peakons $y_{2,i}$, $1 \le i \le 5$.
    As $t \to -\infty$, the leftmost curve in the group, $x = y_{2,1}(t)$,
    approaches the same line as the $y_2$ singleton
    in Figure~\ref{fig:GX-3+3-interlacing-positions-all},
    while the other trajectories $x = y_{2,i}(t)$, $2 \le i \le 5$,
    approach the same line as the curve $x = x_3(t)$.
    As $t \to +\infty$, it is instead the rightmost curve in the group,
    $x = y_{2,5}(t)$, that behaves like the $y_2$ singleton,
    while the other peakons $x = y_{2,i}(t)$, $1 \le i \le 4$,
    asymptotically approach the same line as $x = x_2(t)$.
    At $t=80$, the maximal time shown in the picture,
    the peakon at~$x_2$ has not yet caught up with the slower peakons
    at $y_{2,3}$ and~$y_{2,4}$,
    but it will, eventually.
  }
  \label{fig:GX-3+3-typicalY2-positions}
\end{figure}

\begin{example}[A typical $Y$-group]
  \label{ex:GX-3+3-typicalY2}
  Consider the case of $3+3$ groups, where all groups are singletons
  except the second $Y$-group which contains five peakons:
  \begin{equation}
    x_1 < y_1 < x_2 <
    \underbrace{y_{2,1} < y_{2,2} < y_{2,3} < y_{2,4} < y_{2,5}}_{\text{Second $Y$-group}} <
    x_3 < y_3
    .
  \end{equation}
  We have chosen a typical ``middle'' group for our first example,
  since (non-singleton) groups on the far left and right are described by special formulas.
  We state the solution formulas for this case below; these formulas
  are obtained in later sections by the method of ghostpeakons, where
  we start from the $7+7$ interlacing peakon solution,
  kill off $m_3$, $m_4$, $m_5$, $m_6$ by performing appropriate limiting
  procedures with the spectral data, and rename the remaining (nonzero) amplitudes
  \begin{equation*}
    m_1, n_1, m_2,
    n_2, \quad n_3, \quad n_4, \quad n_5, \quad n_6,
    m_7, n_7
  \end{equation*}
  to
  \begin{equation*}
    m_1, n_1, m_2,
    \underbrace{n_{2,1}, n_{2,2}, n_{2,3}, n_{2,4}, n_{2,5}},
    m_3, n_3
    ,
  \end{equation*}
  and similarly for the positions.
  There are $12$ spectral variables $\{a_k, \lambda_k\}^3_{k=1}$, $\{b_k, \mu_k\}^2_{k=1}$,
  $C$, $D$, whose time dependence is given by~\eqref{eq:residues-time-dependence}.
  The solution depends also on $8$ additional constants
  \begin{equation*}
    \{ \tau_{2,i}^Y , \, \sigma_{2,i}^Y \}_{i=1}^4
    ,
  \end{equation*}
  which only appear in the formulas for the variables $y_{2,i}$
  and~$n_{2,i}$ in the second $Y$-group,
  and which will be written as
  \begin{equation*}
    \{ \tau_i, \, \sigma_i \}_{i=1}^4
  \end{equation*}
  for simplicity.
  These parameters are positive, and must also satisfy the constraint
  $\sigma_1 < \sigma_2 < \sigma_3 < \sigma_4$.
  Thus, the solution formulas depend on $20$ parameters in total,
  which is as it should be,
  since there are 20 degrees of freedom in the system
  (position and amplitude for each of the 10 peakons).
  
  With~$\detJ_{ij}^{rs} = \detJ[3, 2, r,s,i,j]$,
  and using the abbreviations $X_k$, $Y_k$, $Q_k$ and $P_k$ defined by~\eqref{eq:XYQP},
  and $Y_{k,i}$ and~$P_{k,i}$ from Definition~\ref{def:XYQP-groups},
  the solution formulas for the positions are
  \begin{subequations}
    \label{eq:GX-3+3-typicalY2-solution-joint}
    \begin{equation}
      \label{eq:GX-3+3-typicalY2-solution-positions}
      \begin{aligned}
        X_1 &= \frac{\detJ_{32}^{00}}{\detJ_{21}^{11} + C \, \detJ_{22}^{10}}
        ,
        \qquad
        Y_1 = \frac{\detJ_{32}^{00}}{\detJ_{21}^{11}}
        ,
        \qquad
        X_2 = \frac{\detJ_{22}^{00}}{\detJ_{11}^{11}}
        ,
        \\[1em]
        Y_{2,1} &= \frac{\detJ_{22}^{00} + \tau_1\detJ_{21}^{00}}{\detJ_{11}^{11} + \tau_1\detJ_{10}^{11}}
        ,
        \\[1ex]
        Y_{2,2} &= \frac{\detJ_{22}^{00} + (\tau_1 + \tau_2) \detJ_{21}^{00} + \tau_2 \sigma_1 \detJ_{11}^{00}}{\detJ_{11}^{11} + (\tau_1 + \tau_2) \detJ_{10}^{11} + \tau_2 \sigma_1}
        ,
        \\[1ex]
        Y_{2,3} &= \frac{\detJ_{22}^{00} + (\tau_1 + \tau_2 + \tau_3) \detJ_{21}^{00} + (\tau_2 \sigma_1 + \tau_3 \sigma_2) \detJ_{11}^{00}}{\detJ_{11}^{11} + (\tau_1 + \tau_2 + \tau_3) \detJ_{10}^{11} + \tau_2 \sigma_1 + \tau_3 \sigma_2}
        ,
        \\[1ex]
        Y_{2,4} &= \frac{\detJ_{22}^{00} + (\tau_1 + \tau_2 + \tau_3 + \tau_4) \detJ_{21}^{00} + (\tau_2 \sigma_1 + \tau_3 \sigma_2 + \tau_4 \sigma_3) \detJ_{11}^{00}}{\detJ_{11}^{11} + (\tau_1 + \tau_2 + \tau_3 + \tau_4)\detJ_{10}^{11} + \tau_2 \sigma_1 + \tau_3 \sigma_2 + \tau_4 \sigma_3}
        ,
        \\[1ex]
        Y_{2,5} &= \frac{{\detJ_{21}^{00}} + \sigma_4 \detJ_{11}^{00}} {{\detJ_{10}^{11}} + \sigma_4}
        ,
        \\[1ex]
        X_3 &= \detJ_{11}^{00}
        ,\qquad
        Y_3 = \detJ_{11}^{00} + D \detJ_{10}^{00}
        ,
      \end{aligned}
    \end{equation}
    and the formulas for the amplitudes are
    \begin{equation}
      \label{eq:GX-3+3-typicalY2-solution-amplitudes}
      \begin{aligned}
        Q_1 &= \frac{\mu_1 \mu_2}{\lambda_1 \lambda_2 \lambda_3} \left( \frac{\detJ_{21}^{11}}{\detJ_{22}^{10}}+ C \right)
        ,
        \qquad
        P_1 = \frac{\detJ_{22}^{10} \detJ_{21}^{11}}{\detJ_{21}^{01} \detJ_{32}^{01}}
        ,
        \qquad
        Q_2 = \frac{\detJ_{21}^{01} \detJ_{11}^{11}}{\detJ_{11}^{10} \detJ_{22}^{10}}
        ,
        \\[1em]
        P_{2,1} &= \sigma_1 \detJ_{11}^{10} \bigl( \detJ_{11}^{11} +  \tau_1 \detJ_{10}^{11} \bigr)
        \\ & \quad
        \times \bigl( \detJ_{21}^{01}(\detJ_{21}^{01} + \sigma_1  \detJ_{11}^{01} +  \tau_1  \sigma_1 \detJ_{10}^{01}) \bigr)^{-1}
        ,
        \\[1ex]
        P_{2,2 } &= (\sigma_2 - \sigma_1) \detJ_{11}^{10} \bigl( \detJ_{11}^{11} + (\tau_1 + \tau_2)  \detJ_{10}^{11} +  \tau_2 \sigma_1 \bigr)
        \\ & \quad
        \times \bigl( \detJ_{21}^{01} + \sigma_2 \detJ_{11}^{01} + \bigl( \sigma_2 (\tau_1 + \tau_2)  - \tau_2 \sigma_1 \bigr) \detJ_{10}^{01} \bigr)^{-1}
        \\ & \quad
        \times \bigl( \detJ_{21}^{01} + \sigma_1 \detJ_{11}^{01} + \tau_1 \sigma_1 \detJ_{10}^{01} \bigr)^{-1}
        ,
        \\[1ex]
        P_{2,3} &= (\sigma_3 - \sigma_2) \detJ_{11}^{10} \bigl( \detJ_{11}^{11} + (\tau_1 + \tau_2 + \tau_3)  \detJ_{10}^{11} + \tau_2 \sigma_1 + \tau_3 \sigma_2 \bigr)
        \\
        & \quad
        \times \bigl( \detJ_{21}^{01} + \sigma_3 \detJ_{11}^{01} + \bigl( \sigma_3 (\tau_1 + \tau_2 + \tau_3) - (\tau_2 \sigma_1 + \tau_3 \sigma_2) \bigr) \detJ_{10}^{01} \bigr)^{-1}
        \\ & \quad
        \times \bigl( \detJ_{21}^{01}+ \sigma_2 \detJ_{11}^{01} + \bigl( \sigma_2 (\tau_1 + \tau_2) - \tau_2 \sigma_1 \bigr) \detJ_{10}^{01} \bigr)^{-1}
        ,
        \\[1ex]
        P_{2,4} &= (\sigma_4 - \sigma_3) \detJ_{11}^{10} \bigl( \detJ_{11}^{11} + (\tau_1 + \tau_2 + \tau_3 + \tau_4) \detJ_{10}^{11} + \tau_2 \sigma_1 + \tau_3 \sigma_2 + \tau_4 \sigma_3 \bigr)
        \\ &\quad
        \times \bigl( \detJ_{21}^{01} + \sigma_4 \detJ_{11}^{01} + \bigl( \sigma_4 (\tau_1 + \tau_2 + \tau_3 + \tau_4 ) - (\tau_2 \sigma_1 + \tau_3 \sigma_2 + \tau_4 \sigma_3) \bigr) \detJ_{10}^{01} \bigr)^{-1}
        \\ &\quad
        \times \bigl( \detJ_{21}^{01} + \sigma_3 \detJ_{11}^{01} + \bigl( \sigma_3 (\tau_1 + \tau_2 + \tau_3) - (\tau_2 \sigma_1 + \tau_3 \sigma_2 ) \bigr) \detJ_{10}^{01} \bigr)^{-1}
        ,
        \\[1ex]
        P_{2,5} &= \frac{\detJ_{11}^{10}\bigl( \detJ_{10}^{11} + \sigma_4  \bigr)}{\detJ_{10}^{01} \bigl( \detJ_{21}^{01} + \sigma_4 \detJ_{11}^{01} + \bigl( \sigma_4 (\tau_1 + \tau_2 + \tau_3 + \tau_4) - (\tau_4 \sigma_3 + \tau_3 \sigma_2 + \tau_2 \sigma_1) \bigr) \detJ_{10}^{01} \bigr)}
        ,
        \\[1em]
        Q_3 &= \frac{\detJ_{10}^{01}}{\detJ_{11}^{10}}
        ,
        \qquad
        P_3 = \frac{1}{\detJ_{10}^{00}}
        .
      \end{aligned}
    \end{equation}
  \end{subequations}
  Figure~\ref{fig:GX-3+3-typicalY2-positions} shows a plot of the peakon trajectories
  \begin{equation*}
    \begin{aligned}
      x &= x_1(t)
      ,\quad
      x = y_1(t)
      ,
      \\[1ex]
      x &= x_2(t)
      ,\quad
      x = y_{2,1}(t)
      ,\quad
      x = y_{2,2}(t)
      ,\quad
      x = y_{2,3}(t)
      ,\quad
      x = y_{2,4}(t)
      ,\quad
      x = y_{2,5}(t)
      ,
      \\[1ex]
      x &= x_3(t)
      ,\quad
      x = y_3(t)
      ,
    \end{aligned}
  \end{equation*}
  with the same values~\eqref{eq:GX-3+3-interlacing-spectral-data}
  for the spectral parameters as in the previous Example~\ref{ex:GX-3+3-interlacing}
  (the $3+3$ interlacing case),
  and with the additional parameters
  $\tau_i = \tau_{2,i}^Y$ and $\sigma_i = \sigma_{2,i}^Y$
  chosen to be
  \begin{equation}
    \label{eq:GX-3+3-typicalY2-tau-sigma}
    \begin{aligned}
      \tau_1 &= 10^5
      ,&
      \sigma_1 &= 10^{-13}
      ,\\
      \tau_2 &= 10^{10}
      ,&
      \sigma_2 &= 10^{-8}
      ,\\
      \tau_3 &= 10^{17}
      ,&
      \sigma_3 &= 10^{-6}
      ,\\
      \tau_4 &= 10^{20}
      ,&
      \sigma_4 &= 10^{-1}
      .
    \end{aligned}
  \end{equation}
  It is visible that the trajectories approach certain straight lines as $t \to \pm \infty$,
  and this will be proved in Section~\ref{sec:asymptotics-even}.
  Moreover, the curves for the singletons $x_1$, $y_1$, $x_2$, $x_3$ and~$y_3$
  are exactly the same as for the corresponding peakons in the $3+3$ interlacing
  peakon solution shown earlier in Figure~\ref{fig:GX-3+3-interlacing-positions-all}.
  This is because the formulas for the singletons here are the same as in the interlacing case,
  and the spectral parameters here have the same values as in our previous example;
  the additional parameters $\{ \tau_i, \sigma_i \}$ appearing here only occur in the formulas for
  the non-singleton $Y_2$-group.
  
  The interesting new thing is what happens inside the group $y_{2,i}$, $1 \le i \le 5$.
  We see that as $t \to +\infty$, the four leftmost peakons in the
  group approach the curve $x=x_2(t)$, while the last peakon $y_{2,5}$
  approaches the curve $x=y_2(t)$ from Figure~\ref{fig:GX-3+3-interlacing-positions-all}
  (the interlacing case).
  In particular, the rightmost peakon in the group separates from the
  other ones in that group as $t \to +\infty$.
  
  In more detail, what happens to each peakon is the following.
  As $t \to +\infty$, the leftmost curve $x=x_1(t)$ approaches the line
  \begin{equation}
    x = \frac{t}{2} \left( \frac{1}{\lambda_3} \right)
    + \frac12 \ln \left( \frac{2 a_3(0) \, \Psi_{[1,3]\{1,2\}}}{C \lambda_1 \lambda_2 \, \Psi_{\{1,2\}\{1,2\}} } \right)
    ,
  \end{equation}
  the curve $x=y_1(t)$ approaches the line
  \begin{equation}
    x = \frac{t}{2} \left( \frac{1}{\lambda_3} + \frac{1}{\mu_2} \right)
    + \frac12 \ln \left( \frac{2 a_3(0) \, b_2(0) \, \Psi_{[1,3]\{1,2\}}}{\lambda_1 \lambda_2 \mu_1 \, \Psi_{\{1,2\}\{1\}}} \right)
    ,
  \end{equation}
  while all the curves $x=x_2(t)$ and $x=y_{2,i}(t)$
  for $1 \le i \le 4$ approach the line
  \begin{equation}
    x = \frac{t}{2} \left( \frac{1}{\lambda_2} + \frac{1}{\mu_2} \right)
    + \frac12 \ln \left( \frac{2 a_2(0) \, b_2(0) \, \Psi_{\{1,2\}\{1,2\}}}{\lambda_1 \mu_1 \, \Psi_{\{1\}\{1\}}} \right)
    ,
  \end{equation}
  the curve $x=y_{2,5}(t)$ approaches
  \begin{equation}
    \label{eq:GX-3+3-asymp-posinf-singleton-y2}
    x = \frac{t}{2} \left( \frac{1}{\lambda_2} + \frac{1}{\mu_1} \right)
    + \frac12 \ln \left( \frac{2 a_2(0) \, b_1(0) \, \Psi_{\{1,2\}\{1\}}}{\lambda_1 } \right)
    ,
  \end{equation}
  and the curves $x=x_3(t)$ and $y_3(t)$ approach
  \begin{equation}
    \label{eq:GX-3+3-asymp-posinf-singleton-x3-y3}
    x = \frac{t}{2} \left( \frac{1}{\lambda_1} + \frac{1}{\mu_1} \right)
    + \frac12 \ln \left( 2 a_1(0) \, b_1(0) \, \Psi_{\{1\}\{1\}} \right)
    .
  \end{equation}

  The asymptotics as $t \to -\infty$ are similar.
  All singletons follow the same curves as the $3+3$ interlacing peakon case.
  For the $Y_2$-group, the first peakon $y_{2,1}$ approaches the same line as the
  peakon~$y_2$ in the $3+3$ interlacing case (see Figure~\ref{fig:GX-3+3-interlacing-positions-all}),
  and the other four peakons approach the line for~$x_3$.

  Here are the details:
  As $t \to -\infty$, the two leftmost curves $x=x_1(t)$ and $x=y_1(t)$ both approach the line
  \begin{equation}
    x = \frac{t}{2} \left( \frac{1}{\lambda_1} + \frac{1}{\mu_1} \right)
    + \frac12 \ln \left( \frac{2 a_1(0) \, b_1(0) \, \Psi_{[1,3]\{1,2\}}}{\lambda_2 \lambda_3 \mu_2 \, \Psi_{\{2,3\}\{2\}}} \right)
    ,
  \end{equation}
  the curve $x=x_2(t)$ approaches the line
  \begin{equation}
    x = \frac{t}{2} \left( \frac{1}{\lambda_2} + \frac{1}{\mu_1} \right)
    + \frac12 \ln \left( \frac{2 a_2(0) \, b_1(0) \, \Psi_{\{2,3\}\{1,2\}}}{\lambda_3  \mu_2 \, \Psi_{\{3\}\{2\}}} \right)
    ,
  \end{equation}
  and the curve $x=y_{2,1}(t)$ approaches the line
  \begin{equation}
    x = \frac{t}{2} \left( \frac{1}{\lambda_2} + \frac{1}{\mu_2} \right)
    + \frac12 \ln \left( \frac{2 a_2(0) \, b_2(0) \, \Psi_{\{2,3\}\{2\}}}{\lambda_3} \right)
    ,
  \end{equation}
  while the curves $x=x_3(t)$ and $x=y_{2,i}(t)$, $2 \le i \le 5$,
  all approach the line
  \begin{equation}
    \label{eq:GX-3+3-asymp-neginf-singleton-x3}
    x = \frac{t}{2} \left( \frac{1}{\lambda_3} + \frac{1}{\mu_2} \right)
    + \frac12 \ln \left( 2 a_3(0) \, b_2(0) \, \Psi_{\{3\}\{2\}} \right)
    ,
  \end{equation}
  and the curve $x=y_3(t)$ approaches the line
  \begin{equation}
    \label{eq:GX-3+3-asymp-neginf-singleton-y3}
    x = \frac{t}{2} \left( \frac{1}{\lambda_3} \right)
    + \frac12 \ln \bigl( 2 D \, a_3(0) \bigr)
    .
  \end{equation}
  As long as the $Y$-group is not the rightmost $Y$-group or the
  leftmost $Y$-group, we get a similar asymptotic behaviour,
  and likewise for $X$-groups ``in the middle'',
  i.e., for all groups except $X_1$, $Y_1$, $X_K$ and~$Y_K$.
  The next two examples illustrate what happens to the groups $X_K$ and~$Y_K$ near the right edge;
  $X_1$ and~$Y_1$ near the left edge are similar.

  We will not discuss the behaviour of the amplitudes in non-singleton groups
  until Example~\ref{ex:GX-3+3-allgroups};
  however, the reader who is eager so see plots of the amplitudes~$n_{2,i}$ in the present example
  may already now look ahead at the red curves in
  Figure~\ref{fig:GX-3+3-allgroups-amplitudes-X2-Y2} in that example.
\end{example}

\begin{figure}
  \centering
  \includegraphics[width=13cm]{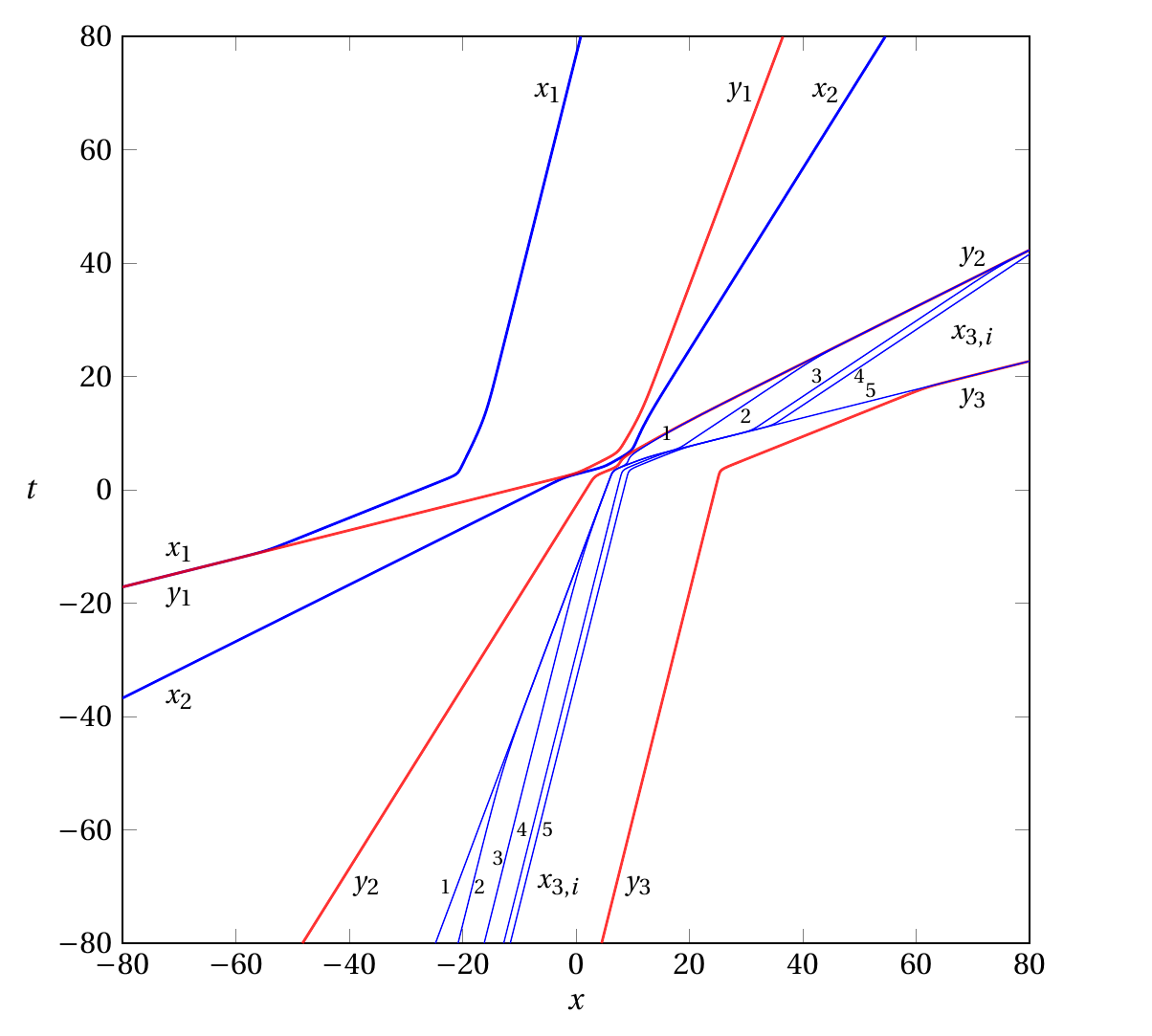}
  \caption{\textbf{Second group from the right.}
    Positions of the peakons for the non-interlacing peakon
    solution~\eqref{eq:GX-3+3-secondrightmostX3-solution-joint} described in
    Example~\ref{ex:GX-3+3-secondrightmostX3},
    with the parameter values
    \eqref{eq:GX-3+3-interlacing-spectral-data}
    and~\eqref{eq:GX-3+3-secondrightmostX3-tau-sigma}.
    Compared to Figure~\ref{fig:GX-3+3-interlacing-positions-all}, the only difference is that
    the singleton~$x_3$ is replaced by a group of five peakons~$x_{3,i}$, $1 \le i \le 5$.
    As $t \to +\infty$, the asymptotics for this group (the rightmost $X$-group,
    second from the right overall)
    are the same as for a ``middle'' group
    (i.e., not $X_1$, $Y_1$, $X_K$ or~$Y_K$),
    like the $Y_2$-group in Example~\ref{ex:GX-3+3-typicalY2} (Figure~\ref{fig:GX-3+3-typicalY2-positions});
    namely, the rightmost peakon in the group, $x = x_{3,5}(t)$, approaches the same line as a singleton
    in the corresponding position would do
    (cf. the curve $x = x_3(t)$ in Figures \ref{fig:GX-3+3-interlacing-positions-all}
    and~\ref{fig:GX-3+3-typicalY2-positions}),
    while the other peakons in the group approach the same line as the
    neighbouring singleton peakon to the left, $x = y_2(t)$.
    Likewise, as $t \to -\infty$, the leftmost peakon in the group, $x = x_{3,1}(t)$,
    behaves like a singleton in the same position would
    (the $x_3$ curve in Figures \ref{fig:GX-3+3-interlacing-positions-all}
    and~\ref{fig:GX-3+3-typicalY2-positions}).
    However, the other peakons in the group are exceptional in that they
    do \emph{not} join the line approached by the singleton peakon $x = y_3(t)$ to the right;
    instead each one approaches its own line \emph{parallel} to the $y_3$ line
    as $t \to -\infty$.}
  \label{fig:GX-3+3-secondrightmostX3-positions}
\end{figure}

\begin{example}[Second group from the right]
  \label{ex:GX-3+3-secondrightmostX3}
  Again, we consider a case with $3+3$ groups where all groups but one are singletons, but
  now it is the $X_3$-group which contains five peakons:
  \begin{equation}
    x_1<y_1<x_2<y_2<\underbrace{x_{3,1}< x_{3,2}<x_{3,3}< x_{3,4}<x_{3,5}}_{\text{Rightmost $X$-group}}<y_3.
  \end{equation}
  The solution formulas for the positions are
  \begin{subequations}
    \label{eq:GX-3+3-secondrightmostX3-solution-joint}
    \begin{equation}
      \label{eq:GX-3+3-secondrightmostX3-solution-positions}
      \begin{aligned}
        X_1 &= \frac{\detJ_{32}^{00}}{\detJ_{21}^{11} + C \detJ_{22}^{10}}
        ,
        \qquad
        Y_1 = \frac{\detJ_{32}^{00}}{\detJ_{21}^{11}}
        ,
        \qquad
        X_2 = \frac{\detJ_{22}^{00}}{\detJ_{11}^{11}}
        ,
        \qquad
        Y_2 = \frac{\detJ_{21}^{00}}{\detJ_{10}^{11}}
        ,
        \\[1ex]
        X_{3,1} &= \frac{\detJ_{21}^{00} + \tau_1 \detJ_{11}^{00}}{\detJ_{10}^{11} + \tau_1}
        ,
        \\[1ex]
        X_{3,2} &= \frac{\detJ_{21}^{00} + (\tau_1 + \tau_2) \detJ_{11}^{00} + \tau_2 \sigma_1 \detJ_{10}^{00} }{\detJ_{10}^{11}+   \tau_1+\tau_2}
        ,
        \\[1ex]
        X_{3,3} &= \frac{\detJ_{21}^{00} + (\tau_1 + \tau_2 + \tau_3) \detJ_{11}^{00} + (\tau_2 \sigma_1 + \tau_3 \sigma_2) \detJ_{10}^{00}}{\detJ_{10}^{11} + \tau_1 + \tau_2 + \tau_3}
        ,
        \\[1ex]
        X_{3,4} &= \frac{\detJ_{21}^{00} + (\tau_1 + \tau_2 + \tau_3 + \tau_4) \detJ_{11}^{00} + (\tau_2 \sigma_1 + \tau_3 \sigma_2 + \tau_4 \sigma_3) \detJ_{10}^{00} }{\detJ_{10}^{11} + \tau_1 + \tau_2 + \tau_3 + \tau_4}
        ,
        \\[1ex]
        X_{3,5} &= \detJ_{11}^{00} + \sigma_4 \detJ_{10}^{00}
        ,
        \\[1ex]
        Y_3 &= \detJ_{11}^{00} + D \detJ_{10}^{00}
        ,
      \end{aligned}
    \end{equation}
    and the amplitudes are obtained from the formulas
    \begin{equation}
      \label{eq:GX-3+3-secondrightmostX3-solution-amplitudes}
      \begin{aligned}
        &
        \sbox0{$Q_2$}\hspace{-\the\wd0}  
        \begin{aligned}
          Q_1 &= \frac{\mu_1 \mu_2}{\lambda_1 \lambda_2 \lambda_3} \left( \frac{\detJ_{21}^{11}}{\detJ_{22}^{10}} + C \right)
          ,
          \qquad
          &
          P_1 &= \frac{\detJ_{22}^{10} \detJ_{21}^{11}}{\detJ_{21}^{01} \detJ_{32}^{01}}
          ,
          \\[1ex]
          Q_2 &= \frac{\detJ_{21}^{01} \detJ_{11}^{11}}{\detJ_{11}^{10} \detJ_{22}^{10}}
          ,
          &
          P_2 &= \frac{\detJ_{11}^{10} \detJ_{10}^{11}}{\detJ_{10}^{01} \detJ_{21}^{01}}
          ,
        \end{aligned}
        \\[1ex]
        Q_{3,1}&= \frac{\sigma_1 \detJ_{10}^{01} \left( \detJ_{10}^{11} + \tau_1 \detJ_{00}^{11} \right) }
        {\detJ_{11}^{10} \left( \detJ_{11}^{10}+ \sigma_1 \detJ_{10}^{10}+\tau_1 \right)}
        ,
        \\[1ex]
        Q_{3,2}&= (\sigma_2 - \sigma_1) \detJ_{10}^{01} \left( \detJ_{10}^{11}+\tau_1+ \tau_2 \right)
        \\[1ex]
        & \times
        \left( \detJ_{11}^{10} + \sigma_2 \detJ_{10}^{10} + \sigma_2 (\tau_1 + \tau_2) - \tau_2 \sigma_1 \right)^{-1}
        \\[1ex]
        & \times
        \left( \detJ_{11}^{10} + \sigma_1 \detJ_{10}^{10} + \sigma_2 \tau_1 \right)^{-1}
        ,
        \\[1ex]
        Q_{3,3} &= (\sigma_3 - \sigma_2) \detJ_{10}^{01} \left( \detJ_{10}^{11} + \tau_1 + \tau_2 + \tau_3 \right)
        \\[1ex]
        & \times
        \left( \detJ_{11}^{10}+ \sigma_3 \detJ_{10}^{10} + \sigma_3 (\tau_1 + \tau_2 + \tau_3) - (\tau_2 \sigma_1 + \tau_2 \sigma_3) \right)^{-1}
        \\[1ex]
        & \times
        \left( \detJ_{11}^{10}+ \sigma_2 \detJ_{10}^{10} + \sigma_2 \left( \tau_1 +\tau_2 \right) - \tau_2 \sigma_1 \right)^{-1}
        ,
        \\[1ex]
        Q_{3,4} &= (\sigma_4 - \sigma_3) \detJ_{10}^{01} \left( \detJ_{10}^{11} + \tau_1 + \tau_2 + \tau_3 + \tau_4 \right)
        \\[1ex]
        & \times
        \left( \detJ_{11}^{10} + \sigma_4 \detJ_{10}^{10} + \sigma_4 (\tau_1 + \tau_2 + \tau_3 + \tau_4) - (\tau_2 \sigma_1 + \tau_2 \sigma_3 + \tau_4 \sigma_3) \right)^{-1}
        \\[1ex]
        & \times
        \left( \detJ_{11}^{10}+ \sigma_3 \detJ_{10}^{10} + \sigma_3 (\tau_1 + \tau_2 + \tau_3) - (\tau_2 \sigma_1 + \tau_2 \sigma_3) \right)^{-1}
        ,
        \\[1ex]
        Q_{3,5} &= \frac{\detJ_{10}^{01}}{\detJ_{11}^{10} + \sigma_4 \detJ_{10}^{10} + \sigma_4 (\tau_1 + \tau_2 + \tau_3 + \tau_4) - (\tau_2 \sigma_1 + \tau_2 \sigma_3 + \tau_4 \sigma_3)}
        ,
        \\[1ex]
        P_3 &= \frac{1}{\detJ_{10}^{00}}
        ,
      \end{aligned}
    \end{equation}
  \end{subequations}
  where the group parameters
  $\tau_i = \tau_{3,i}^X$ and $\sigma_i = \sigma_{3,i}^X$
  must be positive and satisfy
  \begin{equation*}
    \sigma_1 < \sigma_2 < \sigma_3 < \sigma_4 < D
    .
  \end{equation*}
  From the formulas for $X_{3,5}$ and~$Y_3$
  in~\eqref{eq:GX-3+3-secondrightmostX3-solution-positions}
  one can see that the constraint $\sigma_4 < D$ is necessary in order to have
  $x_{3,5} < y_3$.
  One can think of this constraint as coming from the general requirement
  that the last $\sigma$ for each non-singleton group must be less than
  the first~$\tau$ for the next group, whenever that next group is not a singleton.
  In this case, the next group (the $Y_3$-group) \emph{is} in fact a singleton,
  but it is the \emph{rightmost} group, to which special rules apply:
  if it is a singleton,
  the spectral parameter $D$ will step in and play the role of~$\tau_{3,1}^Y$;
  see Remark~\ref{rem:D-is-extra-tau}.
  
  The formulas~\eqref{eq:GX-3+3-secondrightmostX3-solution-positions}
  are obtained by letting $j=1$ in the general
  solution formula~\eqref{eq:even-X-typical-group-pos}
  for the positions in the $X_{K+1-j}$-group, where $1 \le j \le K-1$
  (i.e., any $X$-group except~$X_1$, the leftmost one).
  As $t \to -\infty$,
  the dominant term in the denominator of~\eqref{eq:even-X-typical-group-pos} is
  usually~$\detJ_{j-1,j-2}^{11}$;
  however, this term is absent in the case $j=1$,
  since $\detJ_{0,-1}^{11}=0$.
  This will make the $t \to -\infty$ asymptotics for the positions in the
  rightmost $X$-group different
  than for the other $X$-groups given by the same
  formula~\eqref{eq:even-X-typical-group-pos}.
  
  The positions
  \begin{equation*}
    \begin{aligned}
      x &= x_1(t)
      ,&
      x &= y_1(t)
      ,
      \\[1ex]
      x &= x_2(t)
      ,&
      x &= y_2(t)
      ,
      \\[1ex]
      x &= x_{3,1}(t)
      ,\quad
      x = x_{3,2}(t)
      ,\quad
      x = x_{3,3}(t)
      ,\quad
      x = x_{3,4}(t)
      ,\quad
      x = x_{3,5}(t)
      ,&
      x &= y_3(t)
    \end{aligned}
  \end{equation*}
  are plotted in Figure~\ref{fig:GX-3+3-secondrightmostX3-positions},
  with the same spectral parameters~\eqref{eq:GX-3+3-interlacing-spectral-data}
  as in Example~\ref{ex:GX-3+3-interlacing},
  and with the group parameters
  $\tau_i = \tau_{3,i}^X$ and $\sigma_i = \sigma_{3,i}^X$
  equal to
  \begin{equation}
    \label{eq:GX-3+3-secondrightmostX3-tau-sigma}
    \begin{aligned}
      \tau_1 &= 10^5
      ,&
      \sigma_1 &= 10^{-4}
      ,\\
      \tau_2 &= 10^{11}
      ,&
      \sigma_2 &= 1
      ,\\
      \tau_3 &= 10^{18}
      ,&
      \sigma_3 &= 10^{3}
      ,\\
      \tau_4 &= 10^{20}
      ,&
      \sigma_4 &= 10^{4}
      .
    \end{aligned}
  \end{equation}
  As $t \to +\infty$,
  all peakons in the group except $x_{3,5}$ approach the same
  line~\eqref{eq:GX-3+3-asymp-posinf-singleton-y2} as~$y_2$,
  \begin{equation}
    x = \frac{t}{2} \left( \frac{1}{\lambda_2} + \frac{1}{\mu_1} \right)
    + \frac12 \ln \left( \frac{2 a_2(0) \, b_1(0) \, \Psi_{[1,2]\{1\}}}{\lambda_1} \right)
    ,
  \end{equation}
  while $x_{3,5}$, the rightmost peakon in the group,
  approaches the same line~\eqref{eq:GX-3+3-asymp-posinf-singleton-x3-y3}
  as a singleton~$x_3$
  (cf. Figures \ref{fig:GX-3+3-interlacing-positions-all},
  \ref{fig:GX-3+3-typicalY2-positions}
  and~\ref{fig:GX-3+3-rightmostY3-positions}), namely
  \begin{equation}
    x = \frac{t}{2} \left( \frac{1}{\lambda_1} + \frac{1}{\mu_1} \right)
    + \frac12 \ln \left( 2 a_1(0) \, b_1(0) \, \Psi_{\{1\}\{1\}} \right)
    .
  \end{equation}
  As $t \to -\infty$, the peakons approach separate lines;
  the leftmost curve $x=x_{3,1}(t)$ approaches the same
  line~\eqref{eq:GX-3+3-asymp-neginf-singleton-x3}
  as a singleton~$x_3$, namely
  \begin{equation}
    x = \frac{t}{2} \left( \frac{1}{\lambda_3} + \frac{1}{\mu_2} \right)
    + \frac12 \ln \left( 2 a_3(0) \, b_2(0) \, \Psi_{\{3\}\{2\}} \right)
    ,
  \end{equation}
  while the other curves $x=x_{3,i}(t)$ approach the lines
  \begin{equation}
    \begin{cases}
      \displaystyle
      x = \frac{t}{2} \left( \frac{1}{\lambda_3} \right)
      + \frac12 \ln \left( \frac{2 S_i \, a_3(0)}{T_i} \right)
      ,
      &
      2 \le i \le 4
      ,
      \\[2ex]
      \displaystyle
      x = \frac{t}{2} \left( \frac{1}{\lambda_3} \right)
      + \frac12 \ln \bigl( 2 \sigma_4 \, a_3(0) \bigr)
      ,
      &
      i = 5
      ,
    \end{cases}
  \end{equation}
  which are all parallel to the $y_3$ singleton line~\eqref{eq:GX-3+3-asymp-neginf-singleton-y3}.

  For plots of the amplitudes~$m_{3,i}$, see the blue curves in
  Figure~\ref{fig:GX-3+3-allgroups-amplitudes-Y1-X3} in
  Example~\ref{ex:GX-3+3-allgroups}.
\end{example}

\begin{figure}
  \centering
  \includegraphics[width=13cm]{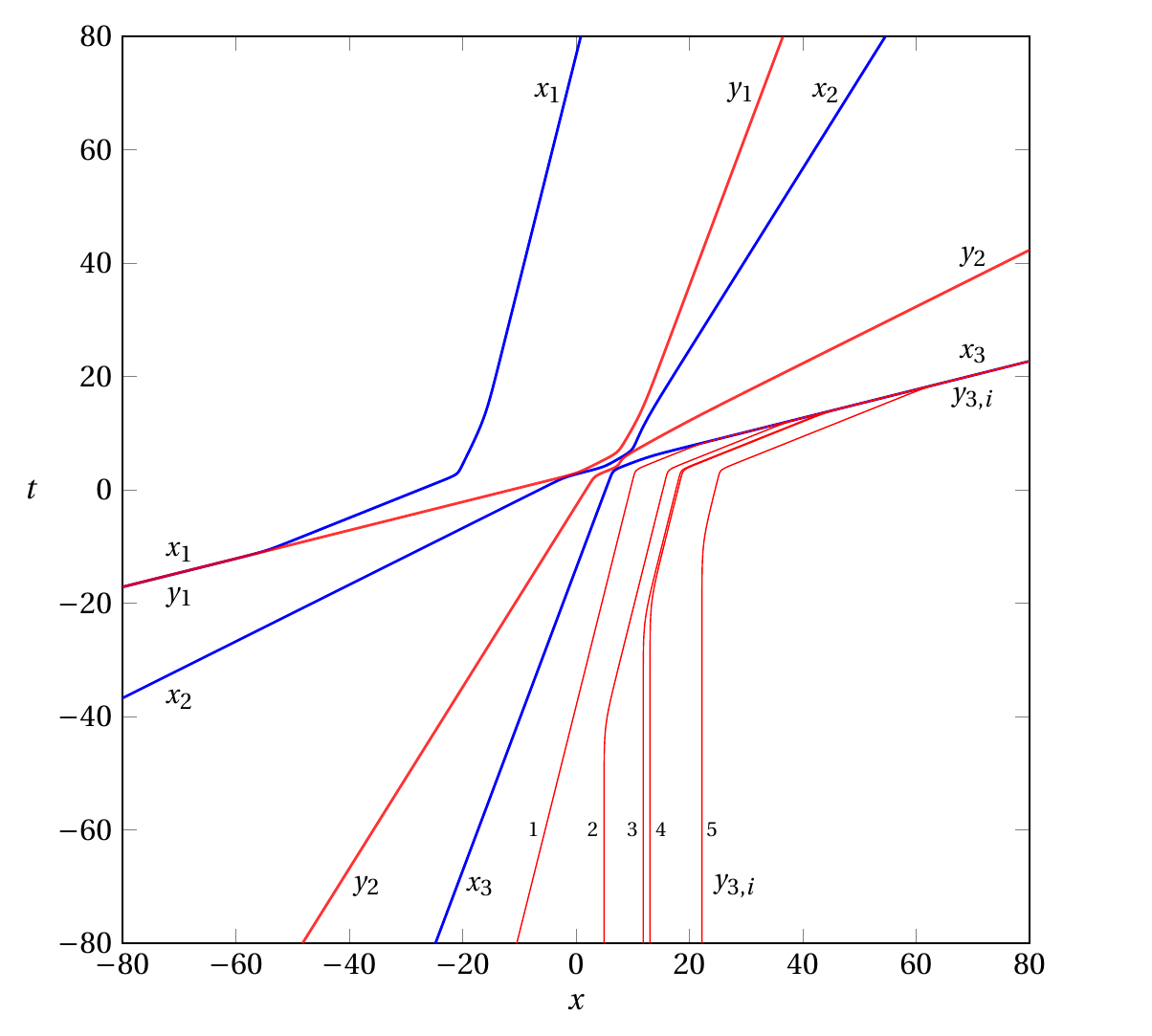}
  \caption{\textbf{Rightmost group.}
    Positions of the peakons for the non-interlacing peakon
    solution~\eqref{eq:GX-3+3-rightmostY3-solution-joint}
    described in Example~\ref{ex:GX-3+3-rightmostY3},
    with the parameter values
    \eqref{eq:GX-3+3-interlacing-spectral-data}
    and~\eqref{eq:GX-3+3-rightmostY3-tau-sigma}.
    Compared to the interlacing solution in Figure~\ref{fig:GX-3+3-interlacing-positions-all},
    the only difference is that the rightmost singleton~$y_3$ is
    replaced by a group of five peakons $y_{3,i}$, $1 \le i\le 5$.
    As $t \to +\infty$, all of the peakons in this rightmost group
    approach the same line as the curve $x = x_3(t)$.
    As $t \to -\infty$, each one approaches a separate line;
    the leftmost peakon~$y_{3,1}$ runs parallel to the line approached by the singleton~$y_3$
    in Figures
    \ref{fig:GX-3+3-interlacing-positions-all}
    and~\ref{fig:GX-3+3-typicalY2-positions},
    while the other $y_{3,i}$ tend to constant values.}
  \label{fig:GX-3+3-rightmostY3-positions}
\end{figure}

\begin{example}[Rightmost group]
  \label{ex:GX-3+3-rightmostY3}
  Now consider the following configuration of $3+3$ groups,
  where the $Y_3$-group consists of five peakons, while all other groups are singletons:
  \begin{equation}
    x_1 < y_1 < x_2 < y_2 < x_3 <
    \underbrace{ y_{3,1} < y_{3,2} < y_{3,3} < y_{3,4} < y_{3,5}}_{\text{Rightmost $Y$-group}}
    .
  \end{equation}
  The solution formulas for the positions are
  \begin{subequations}
    \label{eq:GX-3+3-rightmostY3-solution-joint}
    \begin{equation}
      \label{eq:GX-3+3-rightmostY3-solution-positions}
      \begin{aligned}
        X_1 &= \frac{\detJ_{32}^{00}}{\detJ_{21}^{11} + C \detJ_{22}^{10}}
        ,
        \qquad
        Y_1 = \frac{\detJ_{32}^{00}}{\detJ_{21}^{11}}
        ,
        \qquad
        X_2 = \frac{\detJ_{22}^{00}}{\detJ_{11}^{11}}
        ,
        \qquad
        Y_2 = \frac{\detJ_{21}^{00}}{\detJ_{10}^{11}}
        ,
        \\[1ex]
        X_3 &= \detJ_{11}^{00}
        ,
        \\[1ex]
        Y_{3,1} &= \detJ_{11}^{00} + \tau_1 \detJ_{10}^{00}
        ,
        \\[1ex]
        Y_{3,2} &= \detJ_{11}^{00} + (\tau_1 + \tau_2) \detJ_{10}^{00} +  \tau_2 \sigma_1
        ,
        \\[1ex]
        Y_{3,3} &= \detJ_{11}^{00} + (\tau_1 + \tau_2 + \tau_3) \detJ_{10}^{00} + \tau_2 \sigma_1+ \tau_3 \sigma_2
        ,
        \\[1ex]
        Y_{3,4} &= \detJ_{11}^{00} + (\tau_1 + \tau_2 + \tau_3 + \tau_4) \detJ_{10}^{00} + \tau_2 \sigma_1 + \tau_3 \sigma_2 + \tau_4 \sigma_3
        ,
        \\[1ex]
        Y_{3,5} &= \detJ_{11}^{00} + (\tau_1 + \tau_2 + \tau_3 + \tau_4 + D) \detJ_{10}^{00} + \tau_2 \sigma_1+ \tau_3 \sigma_2 + \tau_4 \sigma_3 + D \, \sigma_4
        ,
      \end{aligned}
    \end{equation}
    and the amplitudes are given by
    \begin{equation}
      \begin{aligned}
        \label{eq:GX-3+3-rightmostY3-solution-amplitudes}
        Q_1 &= \frac{\mu_1 \mu_2}{ \lambda_1 \lambda_2 \lambda_3 } \left( \frac{\detJ_{21}^{11}}{  \detJ_{22}^{10}} + C \right)
        ,
        \qquad
        &
        P_1 &= \frac{\detJ_{22}^{10} \detJ_{21}^{11}}{\detJ_{21}^{01} \detJ_{32}^{01}}
        ,
        \\[1ex]
        Q_2 &= \frac{\detJ_{21}^{01} \detJ_{11}^{11}}{\detJ_{11}^{10} \detJ_{22}^{10}}
        ,
        &
        P_2 &= \frac{\detJ_{11}^{10} \detJ_{10}^{11}}{\detJ_{10}^{01} \detJ_{21}^{01}}
        ,
        \\[1ex]
        Q_3 &= \frac{\detJ_{10}^{01}}{\detJ_{11}^{10}}
        ,
        \\[1ex]
        P_{3,1} &= \frac{\sigma_1}{ \left( \detJ_{10}^{01} + \sigma_1 \right) \detJ_{10}^{01}}
        ,
        \\[1ex]
        P_{3,2} &= \frac{\sigma_2 - \sigma_1}{\left( \detJ_{10}^{01} + \sigma_2 \right) \left( \detJ_{10}^{01} + \sigma_1 \right)}
        ,
        \\[1ex]
        P_{3,3} &= \frac{\sigma_3 - \sigma_2}{\left( \detJ_{10}^{01} + \sigma_3 \right) \left( \detJ_{10}^{01} + \sigma_2 \right)}
        ,
        \\[1ex]
        P_{3,4} &= \frac{\sigma_4 - \sigma_3}{\left( \detJ_{10}^{01} + \sigma_4 \right) \left( \detJ_{10}^{01} + \sigma_3 \right)}
        ,
        \\[1ex]
        P_{3,5} &= \frac{1}{\detJ_{10}^{01} + \sigma_4}
        .
      \end{aligned}
    \end{equation}
  \end{subequations}
  Here the parameters
  $\tau_i = \tau_{3,i}^Y$ and $\sigma_i = \sigma_{3,i}^Y$
  are all positive and satisfy
  $\sigma_1< \sigma_2<\sigma_3<\sigma_4$.

  Figure~\ref{fig:GX-3+3-rightmostY3-positions} shows a plot of the peakon trajectories
  \begin{equation*}
    \begin{aligned}
      x &= x_1(t)
      ,\quad
      x = y_1(t)
      ,
      \\[1ex]
      x &= x_2(t)
      ,\quad
      x = y_2(t)
      ,
      \\[1ex]
      x &= x_3(t)
      ,\quad
      x = y_{3,1}(t)
      ,\quad
      x = y_{3,2}(t)
      ,\quad
      x = y_{3,3}(t)
      ,\quad
      x = y_{3,4}(t)
      ,\quad
      x = y_{3,5}(t)
      ,
    \end{aligned}
  \end{equation*}
  with the same spectral parameters~\eqref{eq:GX-3+3-interlacing-spectral-data}
  as in Example~\ref{ex:GX-3+3-interlacing},
  and with the group parameters
  $\tau_i = \tau_{3,i}^Y$ and $\sigma_i = \sigma_{3,i}^Y$
  equal to
  \begin{equation}
    \label{eq:GX-3+3-rightmostY3-tau-sigma}
    \begin{aligned}
      \tau_1 &= 10^5
      ,&
      \sigma_1 &= 10^{-6}
      ,\\
      \tau_2 &= 10^{10}
      ,&
      \sigma_2 &= 10^{-2}
      ,\\
      \tau_3 &= 10^{12}
      ,&
      \sigma_3 &= 10^{-1}
      ,\\
      \tau_4 &= 10^{12}
      ,&
      \sigma_4 &= 10^{1}
      .
    \end{aligned}
  \end{equation}
  We see that as $t \to +\infty$,
  all peakons in the $Y_3$-group asymptotically behave the same;
  the curves $x=y_{3,i}(t)$, $1 \le i \le 5$, all approach the line
  \begin{equation}
    x =
    \frac{t}{2} \left( \frac{1}{\lambda_1} + \frac{1}{\mu_1} \right)
    + \frac12 \ln \left( 2 a_1(0) \, b_1(0) \, \Psi_{\{1\}\{1\}} \right)
    ,
  \end{equation}
  which is the same line~\eqref{eq:GX-3+3-asymp-posinf-singleton-x3-y3}
  that is approached by the curve $x = x_3(t)$,
  and also by the singleton $x = y_3(t)$ in Examples
  \ref{ex:GX-3+3-interlacing} and~\ref{ex:GX-3+3-typicalY2}.

  As $t \to -\infty$, on the other hand,
  each peakon in the group approaches a separate line,
  where the formulas for the asymptotics depend on~$i$;
  the leftmost curve $x=y_{3,1}(t)$ approaches the line
  \begin{equation}
    x =
    \frac{t}{2} \left( \frac{1}{\lambda_3} \right)
    + \frac12 \ln \bigl( 2 \tau_1 \, a_3(0) \bigr)
    ,
  \end{equation}
  which is parallel to the line~\eqref{eq:GX-3+3-asymp-neginf-singleton-y3}
  approached by a singleton~$y_3$,
  while the other $y_{3,i}(t)$ approach the constant values (vertical lines)
  \begin{equation}
    \begin{cases}
      x = \frac12 \ln (2 S_i)
      , &
      2 \le i \le 4
      , \\
      x = \frac12 \ln \bigl( 2 (S_4 + D \, \sigma_4) \bigr)
      , &
      i = 5,
    \end{cases}
  \end{equation}
  with $S_i = S_{3,i}^Y$ as in Definition~\ref{def:T-S-R}.

  For plots of the amplitudes~$n_{3,i}$, see the red curves in
  Figure~\ref{fig:GX-3+3-allgroups-amplitudes-X1-Y3} in
  Example~\ref{ex:GX-3+3-allgroups}.
\end{example}

\begin{figure}
  \centering
  \includegraphics[width=13cm]{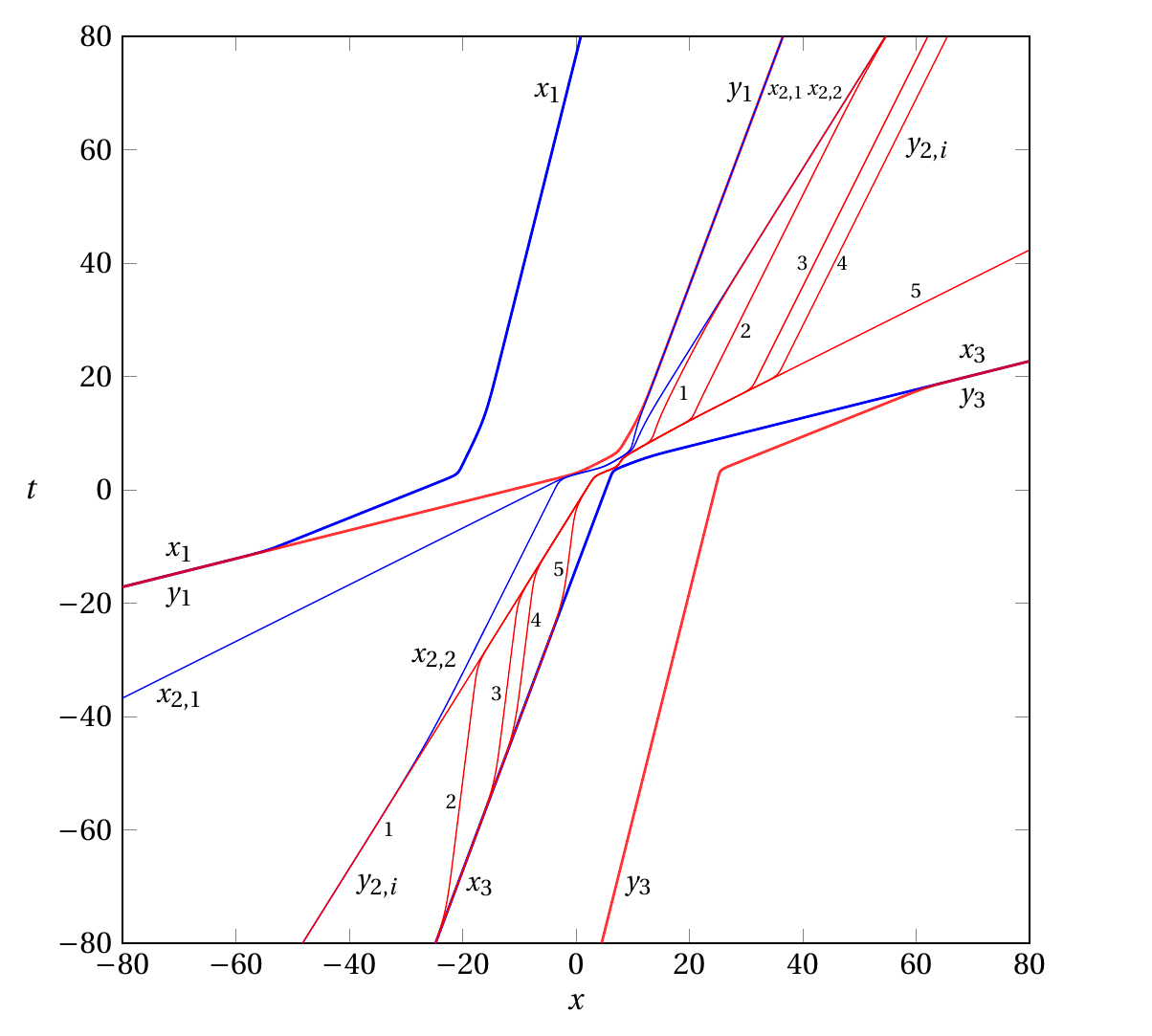}
  \caption{\textbf{Two adjacent non-singleton groups.}
    Positions of the peakons for the non-interlacing
    solution~\eqref{eq:GX-3+3-X2Y2-solution-joint}
    described in Example~\ref{ex:GX-3+3-X2Y2},
    with the parameter values
    \eqref{eq:GX-3+3-interlacing-spectral-data},
    \eqref{eq:GX-3+3-typicalY2-tau-sigma},
    and~\eqref{eq:GX-3+3-X2Y2-tau-sigma-X2}.
    Compared to Figure~\ref{fig:GX-3+3-typicalY2-positions},
    the singleton~$x_2$ has been
    replaced by a group consisting of two peakons, $x_{2,1}$ and $x_{2,2}$.
    The $X_2$-group and the $Y_2$-group both exhibit the asymptotic behaviour
    of typical ``middle'' groups (not $X_1$, $Y_1$, $X_K$ or~$Y_K$):
    as $t \to \pm \infty$, one of the peakons in the group follows the path that a
    singleton in the corresponding position would take,
    while the others approach a neighbouring peakon of the other type.}
  \label{fig:GX-3+3-X2Y2-positions}
\end{figure}

\begin{example}[Two adjacent non-singleton groups]
  \label{ex:GX-3+3-X2Y2}
  We consider next the following configuration of $3+3$ groups, where the second $X$-group
  contains two peakons, the second $Y$-group contains five peakons, and
  the remaining groups are singletons:
  \begin{equation}
    x_1<y_1<\underbrace{x_{2,1}<x_{2,2}}_{\text{Second $X$-group}}<\underbrace{y_{2,1}<y_{2,2}<y_{2,3}<y_{2,4}<y_{2,5}}_{\text{Second $Y$-group}}<x_3<y_3
    .
  \end{equation}
  The solution formulas for the positions in this case are
  \begin{subequations}
    \label{eq:GX-3+3-X2Y2-solution-joint}
    \begin{equation}
      \label{eq:GX-3+3-X2Y2-solution-positions}
      \begin{aligned}
        X_1 &= \frac{\detJ_{32}^{00}}{\detJ_{21}^{11} + C \detJ_{22}^{10}}
        ,
        \qquad
        Y_1 = \frac{\detJ_{32}^{00}}{\detJ_{21}^{11}}
        ,
        \\[1ex]
        X_{2,1} &= \frac{ \detJ_{32}^{00} + \tau_1^X \detJ_{22}^{00} }{\detJ_{21}^{11} + \tau_1^X \detJ_{11}^{11}}
        ,
        \\[1ex]
        X_{2,2}&=\frac{ \detJ_{22}^{00} + \sigma_1^X \detJ_{21}^{00} }{\detJ_{11}^{11} + \sigma_1^X \detJ_{10}^{11}}
        ,
        \\[1em]
        Y_{2,1} &= \frac{\detJ_{22}^{00} + \tau_1^Y \detJ_{21}^{00}}{\detJ_{11}^{11} + \tau_1^Y\detJ_{10}^{11}}
        ,
        \\[1ex]
        Y_{2,2} &= \frac{\detJ_{22}^{00} + (\tau_1^Y + \tau_2^Y) \detJ_{21}^{00} + (\tau_2^Y \sigma_1^Y) \detJ_{11}^{00}}{\detJ_{11}^{11} + (\tau_1^Y + \tau_2^Y) \detJ_{10}^{11} + \tau_2^Y \sigma_1^Y}
        ,
        \\[1ex]
        Y_{2,3} &= \frac{\detJ_{22}^{00} + (\tau_1^Y + \tau_2^Y + \tau_3^Y) \detJ_{21}^{00} + (\tau_2^Y \sigma_1^Y + \tau_3^Y \sigma_2^Y) \detJ_{11}^{00}}{\detJ_{11}^{11} + (\tau_1^Y + \tau_2^Y + \tau_3^Y) \detJ_{10}^{11} + \tau_2^Y \sigma_1^Y + \tau_3^Y \sigma_2^Y}
        ,
        \\[1ex]
        Y_{2,4} &= \frac{\detJ_{22}^{00} + (\tau_1^Y + \tau_2^Y + \tau_3^Y + \tau_4^Y) \detJ_{21}^{00} + (\tau_2^Y \sigma_1^Y  + \tau_3^Y \sigma_2^Y + \tau_4^Y \sigma_3^Y) \detJ_{11}^{00}}{\detJ_{11}^{11} + (\tau_1^Y + \tau_2^Y + \tau_3^Y + \tau_4^Y) \detJ_{10}^{11} + \tau_2^Y \sigma_1^Y + \tau_3^Y \sigma_2^Y + \tau_4^Y \sigma_3^Y}
        ,
        \\[1ex]
        Y_{2,5} &= \frac{{\detJ_{21}^{00}} + \sigma_4^Y \detJ_{11}^{00}} {{\detJ_{10}^{11}} + \sigma_4^Y}
        ,
        \\[1ex]
        X_3 &= \detJ_{11}^{00}
        ,
        \qquad
        Y_3= \detJ_{11}^{00} + D \, \detJ_{10}^{00}
        ,
      \end{aligned}
    \end{equation}
    and the amplitudes are given by
    \begin{equation}
      \label{eq:GX-3+3-X2Y2-solution-amplitudes}
      \begin{aligned}
        Q_1 &= \frac{\mu_1 \mu_2}{\lambda_1 \lambda_2 \lambda_3 } \left( \frac{\detJ_{21}^{11}}{\detJ_{22}^{10}} + C \right)
        ,
        \qquad
        P_1 = \frac{\detJ_{21}^{11} \, \detJ_{22}^{10}}{\detJ_{21}^{01} \, \detJ_{32}^{01}}
        ,
        \\[1ex]
        Q_{2,1} &= \frac{ \sigma_1^X \, \detJ_{21}^{01} \, \left( \detJ_{21}^{11} + \tau_1^X \, \detJ_{11}^{11} \right)}{\detJ_{22}^{10} \, \left( \detJ_{22}^{10} + \sigma_1^X \left( \detJ_{21}^{10} + \tau_1^X \detJ_{11}^{10} \right) \right)}
        ,
        \\[1ex]
        Q_{2,2} &=\frac{\detJ_{21}^{01} \, \left( \detJ_{11}^{11} + \sigma_1^X \detJ_{10}^{11} \right)}{\detJ_{11}^{10} \, \left( \detJ_{22}^{10} + \sigma_1^X \left( \detJ_{21}^{10} + \tau_1^X \detJ_{11}^{10} \right) \right)}
        ,
        \\[1em]
        P_{2,1} &= \sigma_1^Y \detJ_{11}^{10} \bigl( \detJ_{11}^{11} + \tau_1^Y \detJ_{10}^{11} \bigr)
        \\ & \quad
        \times \bigl( \detJ_{21}^{01}(\detJ_{21}^{01} + \sigma_1^Y \detJ_{11}^{01} + \tau_1^Y \sigma_1^Y \detJ_{10}^{01}) \bigr)^{-1},
        \\[1ex]
        P_{2,2} &= (\sigma_2^Y - \sigma_1^Y) \detJ_{11}^{10} \bigl( \detJ_{11}^{11} + (\tau_1^Y + \tau_2^Y) \detJ_{10}^{11} + \tau_2^Y \sigma_1^Y \bigr)
        \\ & \quad
        \times \bigl( \detJ_{21}^{01} + \sigma_2^Y \detJ_{11}^{01} + \bigl( \sigma_2^Y (\tau_1^Y + \tau_2^Y) - \tau_2^Y \sigma_1^Y \bigr) \detJ_{10}^{01} \bigr)^{-1}
        \\ & \quad
        \times \bigl( \detJ_{21}^{01} + \sigma_1^Y \detJ_{11}^{01} +  \tau_1^Y  \sigma_1^Y \detJ_{10}^{01} \bigr)^{-1}
        ,
        \\[1ex]
        P_{2,3} &= (\sigma_3^Y - \sigma_2^Y) \detJ_{11}^{10} \bigl( \detJ_{11}^{11} + (\tau_1^Y + \tau_2^Y + \tau_3^Y) \detJ_{10}^{11} + \tau_2^Y \sigma_1^Y + \tau_3^Y \sigma_2^Y \bigr)
        \\ & \quad
        \times \bigl( \detJ_{21}^{01} + \sigma_3^Y \detJ_{11}^{01} + \bigl( \sigma_3^Y (\tau_1^Y + \tau_2^Y + \tau_3^Y) - (\tau_2^Y \sigma_1^Y + \tau_3^Y \sigma_2^Y) \bigr) \detJ_{10}^{01} \bigr)^{-1}
        \\ & \quad
        \times \bigl( \detJ_{21}^{01} + \sigma_2^Y \detJ_{11}^{01} + \bigl( \sigma_2^Y (\tau_1^Y + \tau_2^Y) - \tau_2^Y \sigma_1^Y \bigr) \detJ_{10}^{01} \bigr)^{-1}
        ,
        \\[1ex]
        P_{2,4} &= (\sigma_4^Y - \sigma_3^Y) \detJ_{11}^{10} \bigl( \detJ_{11}^{11} + (\tau_1^Y + \tau_2^Y + \tau_3^Y + \tau_4^Y)  \detJ_{10}^{11} + \tau_2^Y \sigma_1^Y + \tau_3^Y \sigma_2^Y + \tau_4^Y \sigma_3^Y \bigr)
        \\ & \quad
        \times \bigl( \detJ_{21}^{01} + \sigma_4^Y \detJ_{11}^{01} + \bigl( \sigma_4^Y (\tau_1^Y + \tau_2^Y + \tau_3^Y + \tau_4^Y) - (\tau_2^Y \sigma_1^Y + \tau_3^Y \sigma_2^Y + \tau_4^Y \sigma_3^Y) \bigr) \detJ_{10}^{01} \bigr)^{-1}
        \\ & \quad
        \times \bigl( \detJ_{21}^{01} + \sigma_3^Y \detJ_{11}^{01} + \bigl( \sigma_3^Y (\tau_1^Y + \tau_2^Y + \tau_3^Y) - (\tau_2^Y \sigma_1^Y + \tau_3^Y \sigma_2^Y) \bigr) \detJ_{10}^{01} \bigr)^{-1}
        ,
        \\[1ex]
        P_{2,5} &= \frac{\detJ_{11}^{10}\bigl( \detJ_{10}^{11} + \sigma_4^Y  \bigr)}{\detJ_{10}^{01} \bigl( \detJ_{21}^{01} + \sigma_4^Y \detJ_{11}^{01} + \bigl( \sigma_4^Y (\tau_1^Y + \tau_2^Y + \tau_3^Y + \tau_4^Y) - (\tau_2^Y \sigma_1^Y + \tau_3^Y \sigma_2^Y + \tau_4^Y \sigma_3^Y) \bigr) \detJ_{10}^{01} \bigr)}
        ,
        \\[1ex]
        Q_3 &=\frac{\detJ_{10}^{01}}{\detJ_{11}^{10}}
        ,
        \qquad
        P_3 = \frac{1}{\detJ_{10}^{00}}
        .
      \end{aligned}
    \end{equation}
  \end{subequations}
  These equations are the same as~\eqref{eq:GX-3+3-typicalY2-solution-joint}
  in Example~\ref{ex:GX-3+3-typicalY2},
  except that the formulas for $X_2$ and~$Q_2$ are replaced by
  the formulas for $X_{2,i}$ and~$Q_{2,i}$ above,
  and that we have written superscripts $X$ and~$Y$
  on the parameters $\tau$ and $\sigma$,
  in order to see which group they belong to.
  There are two such parameters for the $X_2$-group
  ($\tau_1^X = \tau_{2,1}^X$ and $\sigma_1^X = \sigma_{2,1}^X$)
  and eight for the $Y_2$-group
  ($\tau_i^Y = \tau_{2,i}^Y$ and $\sigma_i^Y = \sigma_{2,i}^Y$ for $1 \le i \le 4$).
  These parameters must be positive and satisfy
  $\sigma_1^Y < \sigma_2^Y < \sigma_3^Y < \sigma_4^Y$,
  and also $\sigma_1^X < \tau_1^Y$ since we have two non-singleton groups next to each other;
  cf.~\eqref{eq:constraint-last-sigma-first-tau}.
  From the expressions for $X_{2,2}$ and~$Y_{2,1}$ in~\eqref{eq:GX-3+3-X2Y2-solution-positions},
  it can be seen that this last constraint is necessary in order to have $x_{2,2} < y_{2,1}$.

  For the plot of the positions in Figure~\ref{fig:GX-3+3-X2Y2-positions}
  we have taken the usual spectral parameters~\eqref{eq:GX-3+3-interlacing-spectral-data},
  and the same group parameters~\eqref{eq:GX-3+3-typicalY2-tau-sigma} for the $Y_2$-group
  as in Example~\ref{ex:GX-3+3-typicalY2}.
  The parameters for the $X_2$-group are
  \begin{equation}
    \label{eq:GX-3+3-X2Y2-tau-sigma-X2}
    \begin{aligned}
      \tau_1^X = \tau_{2,1}^X &= 10^5
      ,&
      \sigma_1^X = \sigma_{2,1}^X &= 10^{-3}
      .
    \end{aligned}
  \end{equation}
\end{example}

\begin{example}[All groups non-singletons]
  \label{ex:GX-3+3-allgroups}
  Consider next the following configuration of $3+3$ groups,
  represented schematically as in Example~\ref{ex:proof-technique2}:
  \begin{equation*}
    \underbrace{\X \X \X}_{X_1} \underbrace{\Y \Y \Y \Y}_{Y_1}
    \underbrace{\X \X}_{X_2} \underbrace{\Y \Y \Y \Y \Y}_{Y_2}
    \underbrace{\X \X \X \X \X}_{X_3} \underbrace{\Y \Y \Y \Y \Y}_{Y_3}
    .
  \end{equation*}
  Here all the groups are non-singletons,
  and the solution formulas for the positions and amplitudes are obtained from
  the general results in Section~\ref{sec:solutions-even} by taking $K=3$.
  We will not write out the formulas here,
  but note that the formulas for the four rightmost groups ($X_2$, $Y_2$, $X_3$ and~$Y_3$)
  have already been given in
  Examples
  \ref{ex:GX-3+3-X2Y2},
  \ref{ex:GX-3+3-typicalY2},
  \ref{ex:GX-3+3-secondrightmostX3}
  and~\ref{ex:GX-3+3-rightmostY3}
  above.
  For the plots we will use the same
  spectral parameters~\eqref{eq:GX-3+3-interlacing-spectral-data}
  and group parameters
  \eqref{eq:GX-3+3-X2Y2-tau-sigma-X2},
  \eqref{eq:GX-3+3-typicalY2-tau-sigma},
  \eqref{eq:GX-3+3-secondrightmostX3-tau-sigma}
  and~\eqref{eq:GX-3+3-rightmostY3-tau-sigma}
  as in those examples, and in addition for the two leftmost groups we take
  \begin{equation}
    \label{eq:GX-3+3-allgroups-tau-sigma-X1}
    \begin{aligned}
      \tau_{1,1}^X &= 10^{-10}
      ,&
      \sigma_{1,1}^X &= 10^{-15}
      ,\\
      \tau_{1,2}^X &= 10^{1}
      ,&
      \sigma_{1,2}^X &= 10^{-3}
    \end{aligned}
  \end{equation}
  and
  \begin{equation}
    \label{eq:GX-3+3-allgroups-tau-sigma-Y1}
    \begin{aligned}
      \tau_{1,1}^Y &= 10^2
      ,&
      \sigma_{1,1}^Y &= 10^{-15}
      ,\\
      \tau_{1,2}^Y &= 10^{4}
      ,&
      \sigma_{1,2}^Y &= 10^{-10}
      ,\\
      \tau_{1,3}^Y &= 10^{8}
      ,&
      \sigma_{1,3}^Y &= 10^{-3}
      .
    \end{aligned}
  \end{equation}
  Note that the last~$\sigma$ for each group is less than the first~$\tau$
  for the next group; cf.~\eqref{eq:constraint-last-sigma-first-tau}.

  A plot of the positions of all 24 peakons at the same time is shown in
  Figure~\ref{fig:GX-3+3-allgroups-positions}.
  The paths of the four rightmost groups are the same as in the previous examples.
  The $X_1$-group (leftmost) behaves asymptotically like the $Y_3$-group (rightmost),
  but with left and right interchanged and with $t \to +\infty$ and $t \to -\infty$ interchanged.
  Similarly for the $Y_1$-group (second leftmost) and the $X_3$-group (second rightmost).

  Now it is finally time to illustrate the asymptotics of the amplitudes for non-interlacing groups,
  as given by Theorem~\ref{thm:asymptotics-amplitudes-even}.
  In order to avoid information overload, we have made three separate plots:
  the amplitudes for the central groups $X_2$ and~$Y_2$ are shown in
  Figure~\ref{fig:GX-3+3-allgroups-amplitudes-X2-Y2},
  the second outermost groups $Y_1$ and~$X_3$ in
  Figure~\ref{fig:GX-3+3-allgroups-amplitudes-Y1-X3},
  and the outermost groups $X_1$ and~$Y_3$ in
  Figure~\ref{fig:GX-3+3-allgroups-amplitudes-X1-Y3}.
  As in Figure~\ref{fig:GX-3+3-interlacing-amplitudes-all} for the interlacing case,
  we plot the logarithms $s = \ln m_{k,i}(t)$ and $s = -\ln n_{k,i}(t)$.
  The curves for the corresponding singleton amplitudes
  (Figure~\ref{fig:GX-3+3-interlacing-amplitudes-all})
  are included as dashed curves in the background.
  
  As can be seen in the pictures, one amplitude in each group will asymptotically follow the same line as
  the singleton (or a line parallel to it, in the case of the outermost groups),
  whereas all the other amplitudes follow lines with a different slope.
  These new slopes, which do not occur in the interlacing case, are (in order):
  \begin{equation}
    \label{eq:asymptotic-group-slopes-example}
    \begin{aligned}
      \frac12 \left( \frac{3}{\lambda_1} - \frac{1}{\mu_1} \right) &= 6
      ,\\
      \frac12 \left( \frac{1}{\lambda_2} - \frac{3}{\mu_1} \right) &= -4
      ,\\
      \frac12 \left( \frac{3}{\lambda_2} - \frac{1}{\mu_2} \right) &= \frac{11}{8}
      ,\\
      \frac12 \left( \frac{1}{\lambda_3} - \frac{3}{\mu_2} \right) &= -\frac{1}{8}
      ,\\
      \frac12 \frac{3}{\lambda_3} &= \frac{3}{4}
      ,\\
      &
      \phantom{=}
      \,\,\,\,
      0
      .
    \end{aligned}
  \end{equation}
  More precisely: as $t \to -\infty$,
  \begin{itemize}
  \item all but the first $s = \ln m_{1,i}(t)$ approach parallel lines $s = 6t + \text{constant}$,
  \item all but the first $s = -\ln n_{1,i}(t)$ approach parallel lines $s = -4t + \text{constant}$,
  \item all but the first $s = \ln m_{2,i}(t)$ approach parallel lines $s = \tfrac{11}{8} t + \text{constant}$,
  \item all but the first $s = -\ln n_{2,i}(t)$ approach parallel lines $s = - \tfrac{1}{8} t + \text{constant}$,
  \item all but the first $s = \ln m_{3,i}(t)$ approach parallel lines $s = \tfrac{3}{4} t + \text{constant}$,
  \item all but the first $s = -\ln n_{3,i}(t)$ approach parallel lines $s = 0t + \text{constant}$,
  \end{itemize}
  and as $t \to +\infty$,
  \begin{itemize}
  \item all but the last $s = \ln m_{1,i}(t)$ approach parallel lines $s = 0t + \text{constant}$,
  \item all but the last $s = -\ln n_{1,i}(t)$ approach parallel lines $s = \tfrac{3}{4} t + \text{constant}$,
  \item all but the last $s = \ln m_{2,i}(t)$ approach parallel lines $s = - \tfrac{1}{8} t + \text{constant}$,
  \item all but the last $s = -\ln n_{2,i}(t)$ approach parallel lines $s = \tfrac{11}{8} t + \text{constant}$,
  \item all but the last $s = \ln m_{3,i}(t)$ approach parallel lines $s = -4t + \text{constant}$,
  \item all but the last $s = -\ln n_{3,i}(t)$ approach parallel lines $s = 6t + \text{constant}$.
  \end{itemize}
\end{example}

\begin{figure}[H]
  \centering
  \includegraphics[width=13cm]{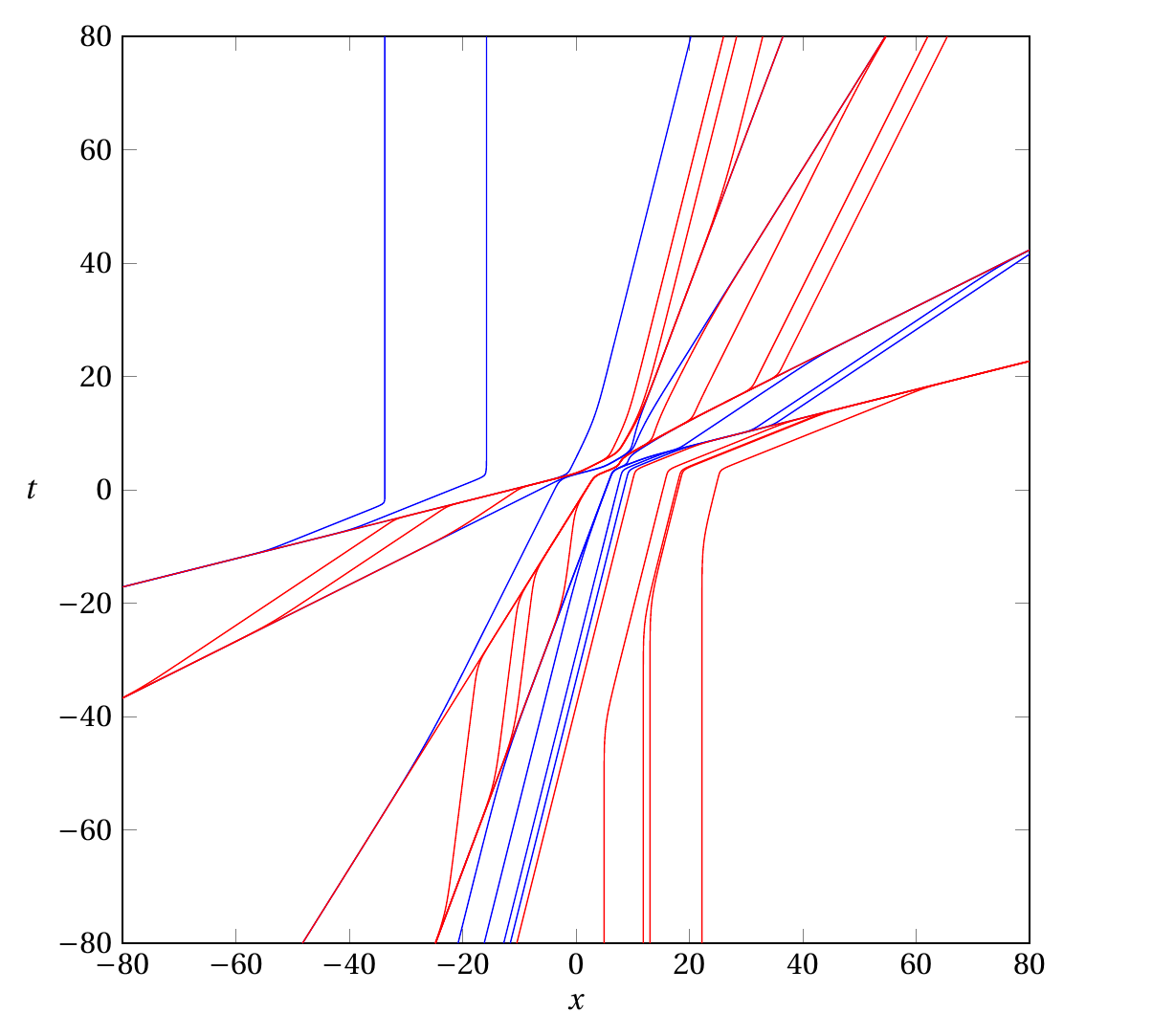}
  \caption{\textbf{Positions in a case with $3+3$ non-singleton groups.}
    Positions of the peakons, $x=x_{k,i}(t)$ and $x=y_{k,i}(t)$,
    for the solution described in Example~\ref{ex:GX-3+3-allgroups}.
    There are $3+3$ groups, all non-singletons,
    with $3+4+2+5+5+5$ peakons in total.
    The curves for the three rightmost groups ($Y_2$, $X_3$ and~$Y_3$,
    with five peakons each)
    are the same as for the corresponding non-singleton groups in
    Figures
    \ref{fig:GX-3+3-typicalY2-positions},
    \ref{fig:GX-3+3-secondrightmostX3-positions}
    and~\ref{fig:GX-3+3-rightmostY3-positions}
    above,
    and also the $X_2$ group with two peakons is like in
    Figure~\ref{fig:GX-3+3-X2Y2-positions}.
    The asymptotic behaviour of the leftmost group ($X_1$) as $t \to \pm\infty$
    is analogous to that of the rightmost group ($Y_3$) as $t \to \mp\infty$,
    and similarly for the other ``mirror-image'' pairs $(Y_1,X_3)$ and $(X_2,Y_2)$.
  }
  \label{fig:GX-3+3-allgroups-positions}
\end{figure}

\begin{figure}[H]
  \centering
  \includegraphics[width=13cm]{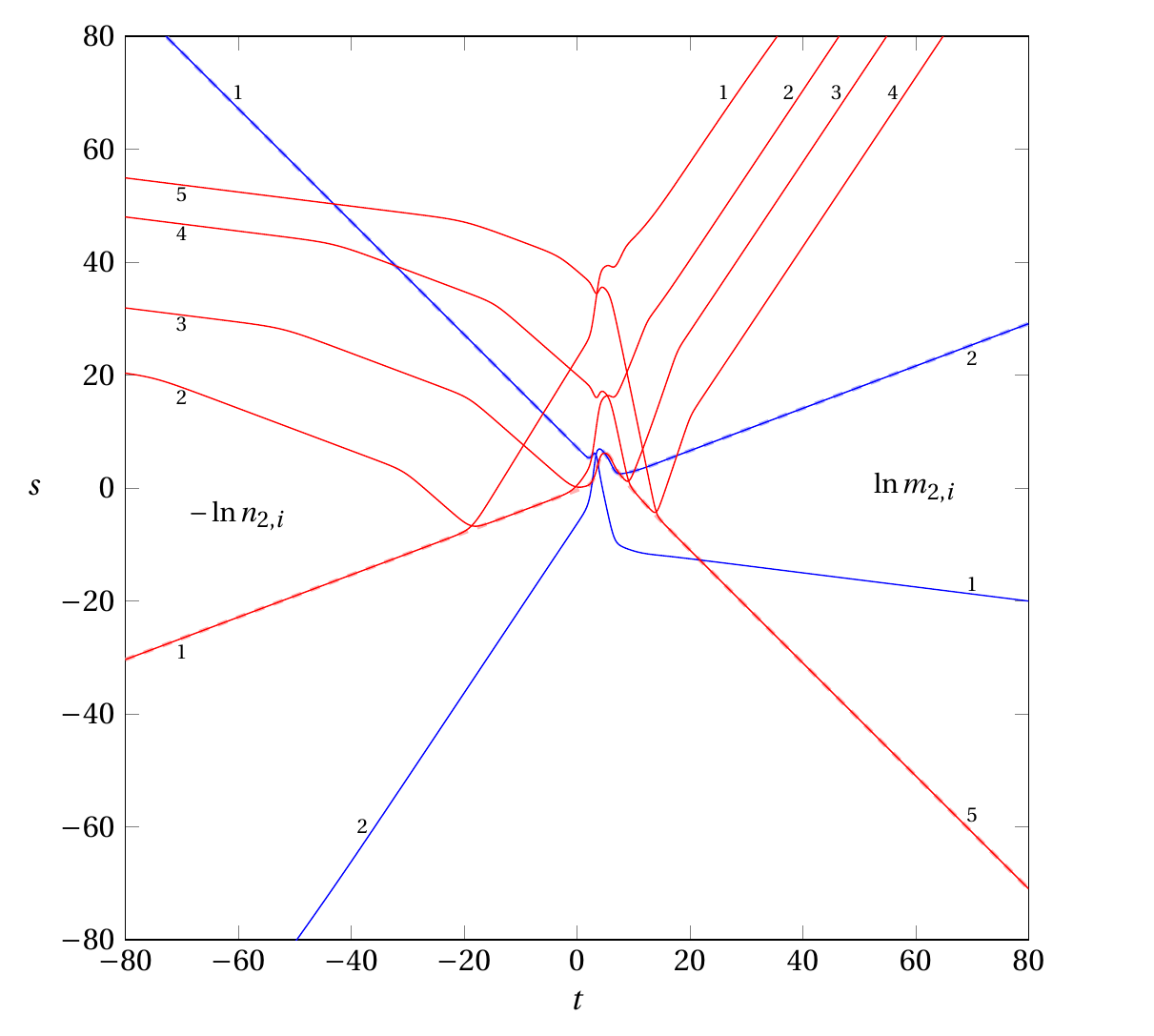}
  \caption{\textbf{Amplitudes for the middle groups ($X_2$ and~$Y_2$).}
    Amplitudes of the peakons for the solution described in Example~\ref{ex:GX-3+3-allgroups}.
    The solid blue curves are $s = \ln m_{2,i}(t)$ and the solid red curves are $s = -\ln n_{2,i}(t)$.
    The dashed curves in the background are the corresponding singleton curves
    $s = \ln m_2(t)$ and $s = -\ln n_2(t)$
    from Figure~\ref{fig:GX-3+3-interlacing-amplitudes-all}.
    As $t \to \pm \infty$, one peakon in each group follows the singleton curve asymptotically,
    while the others approach parallel lines with other slopes, not seen in the interlacing case,
    and given by~\eqref{eq:asymptotic-group-slopes-example}.
  }
  \label{fig:GX-3+3-allgroups-amplitudes-X2-Y2}
\end{figure}

\begin{figure}[H]
  \centering
  \includegraphics[width=13cm]{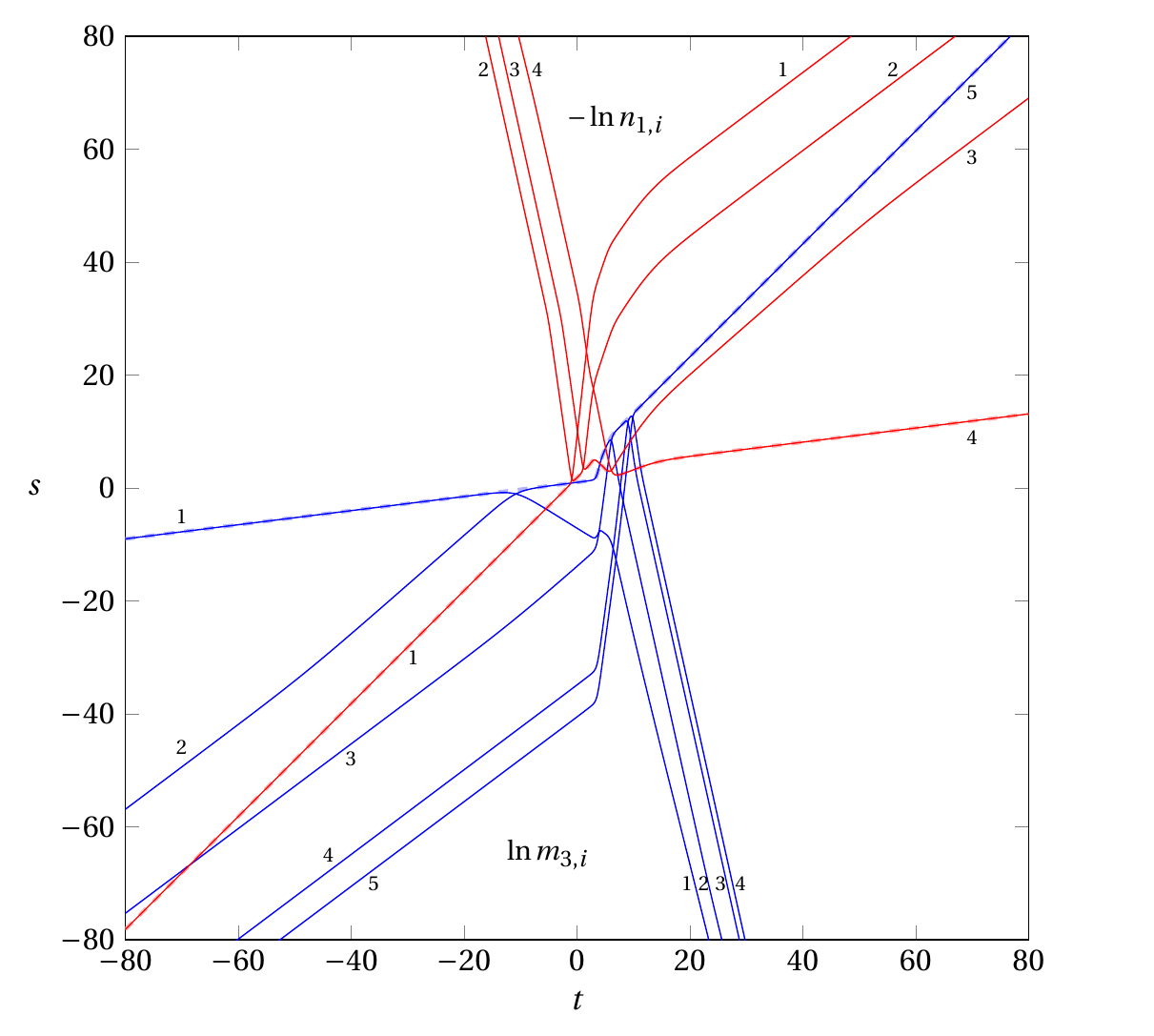}
  \caption{\textbf{Amplitudes for the second leftmost and second rightmost groups ($Y_1$ and~$X_3$).}
    The same as Figure~\ref{fig:GX-3+3-allgroups-amplitudes-X2-Y2},
    but for
    the second leftmost group ($s = -\ln n_{1,i}(t)$, red)
    and
    the second rightmost group ($s = \ln m_{3,i}(t)$, blue).
    The asymptotics for the amplitudes in these groups are exactly analogous to
    the middle groups in
    Figure~\ref{fig:GX-3+3-allgroups-amplitudes-X2-Y2}.
    This is in contrast to the positions, which behave somewhat exceptionally as $t \to -\infty$
    for the second rightmost group ($X_3$, cf. Figures
    \ref{fig:GX-3+3-secondrightmostX3-positions}
    and~\ref{fig:GX-3+3-allgroups-positions}),
    and as $t \to +\infty$ for the second leftmost
    group~($Y_1$, cf. Figure~\ref{fig:GX-3+3-allgroups-positions}).
  }
  \label{fig:GX-3+3-allgroups-amplitudes-Y1-X3}
\end{figure}

\begin{figure}[H]
  \centering
  \includegraphics[width=13cm]{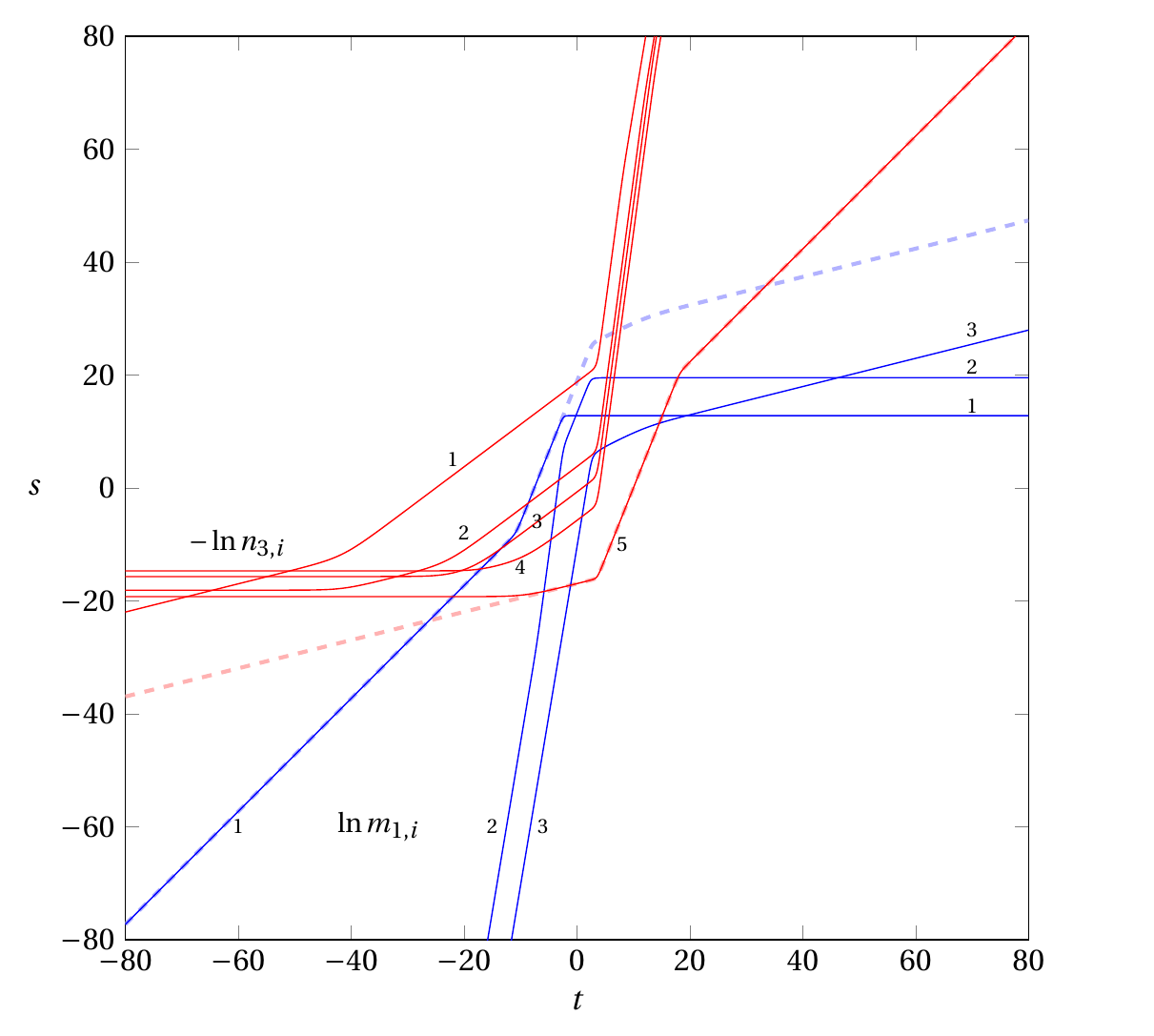}
  \caption{\textbf{Amplitudes for the leftmost and rightmost groups ($X_1$ and~$Y_3$).}
    The same as Figures \ref{fig:GX-3+3-allgroups-amplitudes-X2-Y2}
    and~\ref{fig:GX-3+3-allgroups-amplitudes-Y1-X3},
    but for the leftmost group ($s = \ln m_{1,i}(t)$, blue)
    and the rightmost group ($s = -\ln n_{3,i}(t)$, red).
    The amplitudes in these outermost groups have somewhat exceptional asymptotics,
    since the one peakon in each group that deviates from the others
    only approaches the corresponding singleton curve in one time direction;
    in the other time direction it instead approaches a line \emph{parallel} to the asymptote
    of the singleton curve (in the picture,
    the blue curve $s = \ln m_{1,3}(t)$ becomes parallel to the dashed blue curve as $t \to +\infty$,
    and the red curve $s = -\ln n_{3,1}(t)$ becomes parallel to the dashed red curve as $t \to -\infty$).
  }
  \label{fig:GX-3+3-allgroups-amplitudes-X1-Y3}
\end{figure}

\begin{example}[Two groups only]
  \label{ex:GX-1+1-allgroups}
  If there are just two groups ($X_1$ and~$Y_1$),
  we have $A=1$ and $B=0$ in the determinants $\detJ_{ij}^{rs}$,
  which means that there is only one remaining eigenvalue $\lambda_1$ and no $\mu_j$ at all.
  With four peakons in each group, for example,
  the solution formulas for the $X$-group are
  \begin{equation}
    \begin{aligned}
      X_{1,1} &= \frac{\tau_1 \sigma_1 a_1}{\tau_1 + C \, \sigma_1 (\lambda_1 a_1 + \tau_1)}
      ,
      \\[1ex]
      X_{1,2} &= \frac{\tau_2 \sigma_1 a_1}{\lambda_1 a_1 + \tau_1 + \tau_2}
      ,
      \\[1ex]
      X_{1,3} &= \frac{(\tau_2 \sigma_1 + \tau_3 \sigma_2) a_1}{\lambda_1 a_1 + \tau_1 + \tau_2 + \tau_3}
      ,
      \\[1ex]
      X_{1,4} &= \sigma_3 a_1
    \end{aligned}
  \end{equation}
  and
  \begin{equation}
    \begin{aligned}
      Q_{1,1} &= \frac{1}{\lambda_1} \left( \frac{\tau_1}{\sigma_1 \lambda_1 a_1 + \sigma_1 \tau_1} + C \right)
      ,
      \\[1ex]
      Q_{1,2} &= \frac{(\sigma_2-\sigma_1) a_1 (\lambda_1 a_1 + \tau_1+\tau_2)}{\sigma_1 (\lambda_1 a_1 + \tau_1) \left(   \sigma_2 \lambda_1 a_1+ \sigma_2 (\tau_1 + \tau_2) - \tau_2 \sigma_1 \right)}
      ,
      \\[1ex]
      Q_{1,3} &= (\sigma_3-\sigma_2) a_1 (\lambda_1 a_1 + \tau_1 + \tau_2 + \tau_3)
      \\ & \quad
      \times \left( \sigma_3 \lambda_1 a_1 + \sigma_3 (\tau_1+\tau_2+\tau_3) -(\tau_2 \sigma_1 + \tau_3 \sigma_2) \right)^{-1}
      \\
      &
      \quad \times \left( \sigma_2 \lambda_1 a_1+ \sigma_2 (\tau_1 + \tau_2) - \tau_2 \sigma_1 \right)^{-1}
      ,
      \\[1ex]
      Q_{1,4} &= \frac{a_1}{\sigma_3 \lambda_1 a_1 + \sigma_3 (\tau_1 + \tau_2 + \tau_3) - (\tau_2 \sigma_1 + \tau_3 \sigma_2)}
      ,
    \end{aligned}
  \end{equation}
  where we have written $\tau_k$ and~$\sigma_k$ instead of $\tau_{1,k}^X$ and~$\sigma_{1,k}^X$.
  As usual, $\tau_k > 0$ and $0 < \sigma_1 < \sigma_2 < \sigma_3$.
  The parameter~$C$ is not only positive, but must actually satisfy the stronger
  constraint~\eqref{eq:constraint-C-general}
  (where $M=1$ is the empty product),
  namely
  \begin{equation}
    \tau_1 < C \, \sigma_1 \, \tau_2
    ,
  \end{equation}
  to ensure that
  \begin{equation*}
    X_{1,1} =
    \frac{\sigma_1 a_1}{1 + \frac{C \, \sigma_1}{\tau_1} \, (\lambda_1 a_1 + \tau_1)}
    <
    \frac{\sigma_1 a_1}{1 + \frac{1}{\tau_2} \, (\lambda_1 a_1 + \tau_1)}
    = X_{1,2}
    .
  \end{equation*}
  The formulas for the $Y$-group read as follows,
  where now $\tau_k$ and~$\sigma_k$ stand for $\tau_{1,k}^Y$ and~$\sigma_{1,k}^Y$
  (again with $\tau_k > 0$ and $0 < \sigma_1 < \sigma_2 < \sigma_3$):
  \begin{equation}
    \begin{aligned}
      Y_{1,1} &= \tau_1 a_1
      ,
      \\[1ex]
      Y_{1,2} &= (\tau_1 + \tau_2) a_1 + \tau_2 \sigma_1
      ,
      \\[1ex]
      Y_{1,3} &= (\tau_1 + \tau_2 + \tau_3) a_1 + \tau_2 \sigma_1 + \tau_3 \sigma_2
      ,
      \\[1ex]
      Y_{1,4} &=  (\tau_1+\tau_2+\tau_3+D) a_1 + \tau_2 \sigma_1 + \tau_3 \sigma_2 + D \sigma_3
    \end{aligned}
  \end{equation}
  and
  \begin{equation}
    \begin{aligned}
      P_{1,1} &= \frac{\sigma_1}{(a_1+\sigma_1)a_1}
      ,
      \\[1ex]
      P_{1,2} &= \frac{\sigma_2 - \sigma_1}{(a_1 + \sigma_2)(a_1 + \sigma_1)}
      ,
      \\[1ex]
      P_{1,3} &=\frac{\sigma_3 - \sigma_2}{(a_1 + \sigma_3)(a_1 + \sigma_2)}
      ,
      \\[1ex]
      P_{1,4} &= \frac{1}{a_1+\sigma_3}
      .
    \end{aligned}
  \end{equation}
  The constraint~\eqref{eq:constraint-last-sigma-first-tau} says that
  $\sigma_{1,3}^X < \tau_{1,1}^Y$,
  which ensures that
  \begin{equation*}
    X_{1,4} = \sigma_{1,3}^X \, a_1 < \tau_{1,1}^Y \, a_1 = Y_{1,1}
    ,
  \end{equation*}
  and $D$ must be positive.

  The positions are plotted in Figure~\ref{fig:GX-1+1-allgroups-positions}
  and the amplitudes in Figure~\ref{fig:GX-1+1-allgroups-amplitudes},
  together with the corresponding singleton curves (dashed),
  which are given by
  \begin{equation}
    \label{eq:GX-1+1-interlacing-solution-abbrev}
    X_1 = \frac{a_1}{C}
    ,\qquad
    Y_1 = D a_1
    ,\qquad
    Q_1 = \frac{C}{\lambda_1}
    ,\qquad
    P_1 = \frac{1}{a_1}
    .
  \end{equation}
  More explicitly, \eqref{eq:GX-1+1-interlacing-solution-abbrev}
  means that the $1+1$ interlacing solution is simply
  \begin{equation}
    \label{eq:GX-1+1-interlacing-solution}
    x_1(t) = x_1(0) + ct
    ,\quad
    y_1(t) = y_1(0) + ct
    ,\quad
    m_1(t) = m_1(0) \, e^{ct}
    ,\quad
    n_1(t) = n_1(0) \, e^{-ct}
    ,
  \end{equation}
  where $c = m_1(0) \, n_1(0) \, e^{x_1(0)-y_1(0)} = \frac{1}{2\lambda_1} > 0$.
  This can also easily be derived directly from the governing ODEs~\eqref{eq:GX-peakon-ode},
  which in this case are
  \begin{equation*}
    \dot x_1 = \dot x_2 = \frac{\dot m_1}{m_1} = - \frac{\dot n_1}{n_1} = m_1 n_2 e^{x_1-x_2}
    ,
  \end{equation*}
  where $m_1 n_2 e^{x_1-x_2}$ is a constant of motion.
  
  The spectral data used in the plots are
  \begin{equation}
    \label{eq:GX-1+1-allgroups-parameters}
    \lambda_1 = 1
    ,\qquad
    C = 10^{-6}
    ,\qquad
    D = 10^{26}
    ,
  \end{equation}
  while the group parameters are
  \begin{equation}
    \label{eq:GX-1+1-allgroups-tau-sigma-X1}
    \begin{aligned}
      \tau_{1,1}^X &= 10^{-10}
      ,&
      \sigma_{1,1}^X &= 10^{-8}
      ,\\
      \tau_{1,2}^X &= 10^{10}
      ,&
      \sigma_{1,2}^X &= 10^{-3}
      ,\\
      \tau_{1,3}^X &= 10^{10}
      ,&
      \sigma_{1,3}^X &= 10^{-1}
    \end{aligned}
  \end{equation}
  and
  \begin{equation}
    \label{eq:GX-1+1-allgroups-tau-sigma-Y1}
    \begin{aligned}
      \tau_{1,1}^Y &= 10^2
      ,&
      \sigma_{1,1}^Y &= 10^{-15}
      ,\\
      \tau_{1,2}^Y &= 10^{4}
      ,&
      \sigma_{1,2}^Y &= 10^{-10}
      ,\\
      \tau_{1,3}^Y &= 10^{25}
      ,&
      \sigma_{1,3}^Y &= 10^{-3}
      .
    \end{aligned}
  \end{equation}
  Note that these numbers do \emph{not} satisfy the constraint~\eqref{eq:constraint-C-simpler},
  $1 < C \, \tau_{1,1}^Y$,
  which is a requirement for the \emph{singleton}~$X_1$ to form a valid solution
  together with this $Y_1$-group:
  \begin{equation*}
    X_1 = \tfrac{1}{C} \, a_1 \not< \tau_{1,1}^Y \, a_1 = Y_{1,1}
    .
  \end{equation*}
  Indeed, in Figure~\ref{fig:GX-1+1-allgroups-positions}
  the dashed blue line does not lie to the left of
  all the red curves.
  Changing $\tau_{1,1}^Y$ to a value greater than $1/C$ would remedy this,
  causing all the red curves to move to the right of the dashed blue line.
  The constraint~\eqref{eq:constraint-last-sigma-D-even},
  $\sigma_{1,3}^X < D$,
  is however satisfied,
  which means that the singleton~$Y_1$ would form a valid solution
  together with the $X_1$-group:
  \begin{equation*}
    X_{1,4} = \sigma_{1,3}^X \, a_1 < D a_1 = Y_1
    .
  \end{equation*}
  As can be seen in Figure~\ref{fig:GX-1+1-allgroups-positions},
  the dashed red line indeed lies to the right of all the blue curves.
\end{example}

\begin{figure}[H]
  \centering
  \includegraphics[width=13cm]{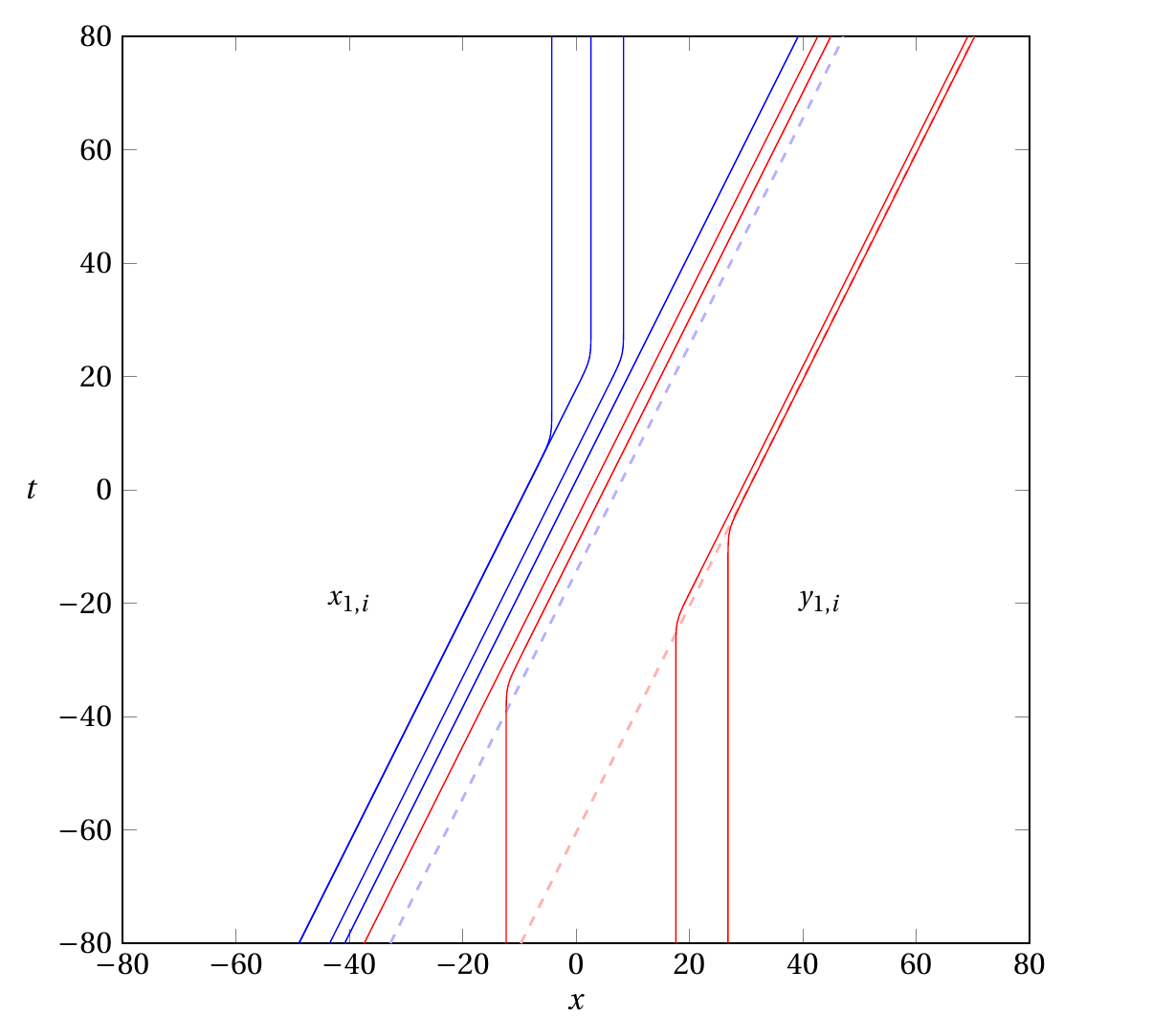}
  \caption{\textbf{Positions in the case of $1+1$ groups.}
    Positions of the peakons, $x = x_{1,i}(t)$ (blue) and $x = y_{1,i}(t)$ (red)
    for the solution described in Example~\ref{ex:GX-1+1-allgroups},
    where the is one $X$-group and one $Y$-group with four peakons each.
    The dashed curves in the background are the corresponding singleton curves
    $x = x_1(t)$ and $x = y_1(t)$ with the same spectral data.
    These singletons curves, as well as the curves for $x_{1,4}$ and $y_{1,1}$,
    are straight lines
    $x = \frac{t}{2\lambda_1} + \text{constant}$
    (where $\lambda_1=1$ in this example).
    The other peakons asymptotically run parallel to these lines in one time direction,
    and tend to constant values in the other time direction.
  }
  \label{fig:GX-1+1-allgroups-positions}
\end{figure}

\begin{figure}[H]
  \centering
  \includegraphics[width=13cm]{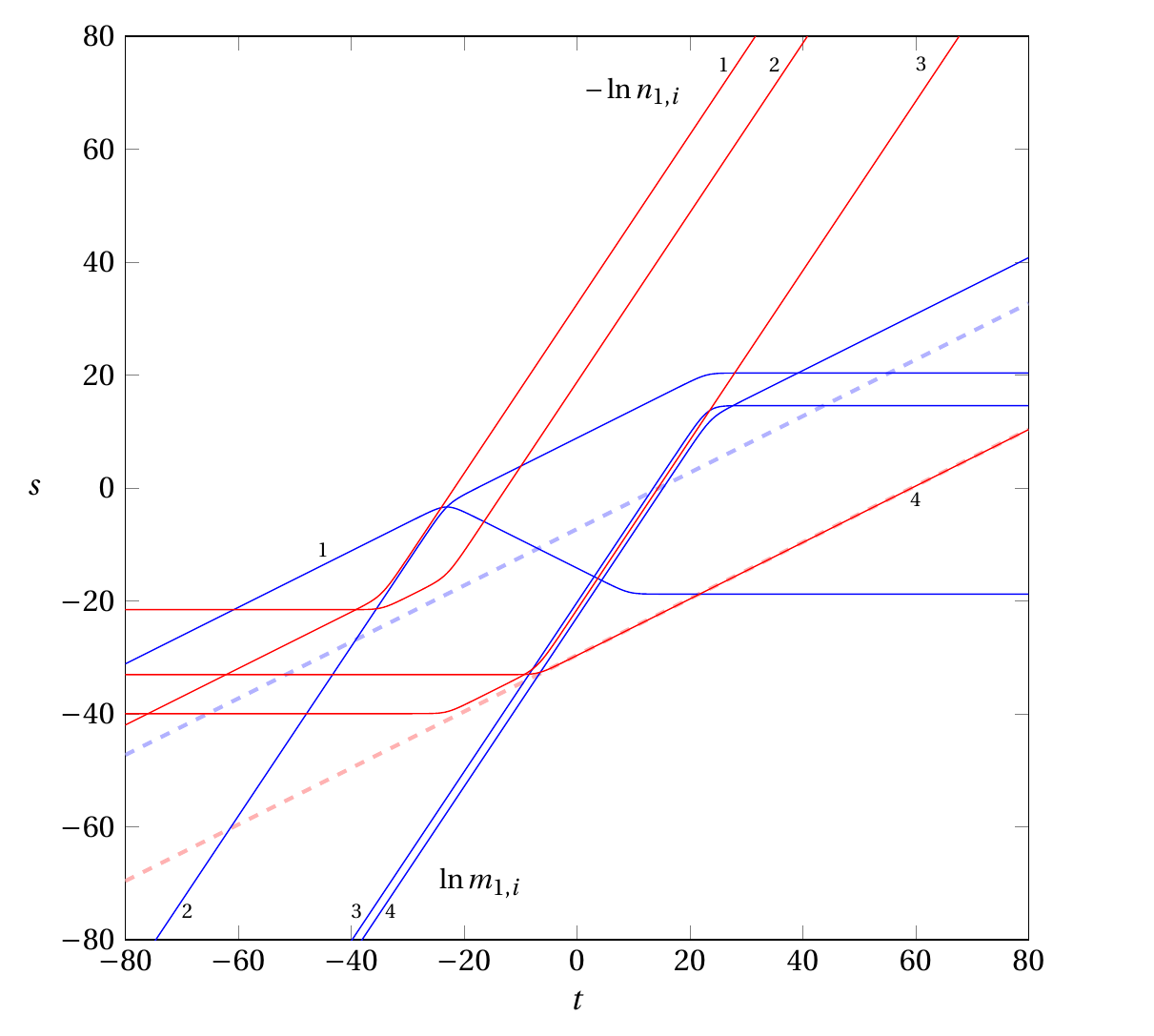}
  \caption{\textbf{Amplitudes in the case of $1+1$ groups.}
    Amplitudes of the peakons for the same solution as in
    Figure~\ref{fig:GX-1+1-allgroups-positions} (Example~\ref{ex:GX-1+1-allgroups}).
    The solid blue curves are $s = \ln m_{1,i}(t)$ and the solid red curves are $s = -\ln n_{1,i}(t)$.
    The dashed curves in the background are the corresponding singleton curves
    $s = \ln m_1(t)$ and $s = -\ln n_1(t)$
    with the same spectral data;
    these are straight lines with slope $\frac{1}{2\lambda_1}$.
    For the amplitudes in the non-singleton groups,
    as $t \to +\infty$ or as $t \to -\infty$,
    one curve in each group will asymptotically have this slope as well,
    whereas the others approach lines with slope $0$ or $\frac{3}{2\lambda_1}$.
  }
  \label{fig:GX-1+1-allgroups-amplitudes}
\end{figure}

\subsection{Examples with an odd number of groups}
\label{sec:examples-groups-odd}

In the odd case with $(K+1)+K$ groups,
the solution formulas (given in Section~\ref{sec:solutions-odd})
are slightly different; there is an offset in one of the lower indices
of the determinants~$\detJ_{ij}^{rs}$ compared to the even case,
and these determinants are computed with spectral data containing $K+K$ eigenvalues
and residues instead of $K+(K-1)$.
The asymptotic behaviour of the solutions in the odd case is similar to the even case
in many respects, but also displays some peculiar differences, as we shall see.

\begin{example}[The $2+1$ interlacing case]
  \label{ex:GX-2+1-interlacing}

  The simplest example of a solution with an odd number of groups
  is the $2+1$ interlacing configuration
  \begin{equation*}
    x_1 < y_1 < x_2
    ,
  \end{equation*}
  for which the governing ODEs~\eqref{eq:GX-peakon-ode} are
  \begin{equation}
    \label{eq:GX-2+1-interlacing-ode}
    \begin{aligned}
      \dot x_1 &= \bigl( m_1 + m_2 e^{x_1-x_2} \bigr) \, n_1 e^{x_1-y_1}
      ,\\
      \dot y_1 &= \bigl( m_1 e^{x_1-y_1} + m_2 e^{y_1-x_2} \bigr) \, n_1
      ,\\
      \dot x_2 &= \bigl( m_1 e^{x_1-x_2} + m_2 \bigr) \, n_1 e^{y_1-x_2}
      ,\\
      \dot m_1 / m_1 &= \bigl( m_1 - m_2 e^{x_1-x_2} \bigr) \, n_1 e^{x_1-y_1}
      ,\\
      \dot n_1 / n_1 &= \bigl( -m_1 e^{x_1-y_1} + m_2 e^{y_1-x_2} \bigr) \, n_1
      ,\\
      \dot m_2 / m_2 &= \bigl( m_1 e^{x_1-x_2} - m_2 \bigr) \, n_1 e^{y_1-x_2}
      ,
    \end{aligned}
  \end{equation}
  and the solution formulas,
  with $\detJ_{ij}^{rs} =\detJ[1,1,r,s,i,j]$,
  are
  \begin{subequations}
    \label{eq:GX-2+1-interlacing-joint}
    \begin{equation}
      \label{eq:GX-2+1-interlacing-positions}
      \begin{aligned}
        X_1 = \tfrac12 e^{2x_1}
        &= \frac{\detJ_{11}^{00}}{\detJ_{00}^{11}+C\detJ_{01}^{10}}
        = \frac{\frac{a_1 b_1}{\lambda_1 + \mu_1}}{1 + C b_1}
        ,\\
        Y_1 = \tfrac12 e^{2y_1}
        &= \frac{\detJ_{11}^{00}}{\detJ_{00}^{11}}
        = \frac{a_1 b_1}{\lambda_1 + \mu_1}
        ,\\
        X_2 = \tfrac12 e^{2x_2}
        &= \detJ_{11}^{00} + D \detJ_{01}^{00}
        = \frac{a_1 b_1}{\lambda_1 + \mu_1} + D b_1
      \end{aligned}
    \end{equation}
    and
    \begin{equation}
      \label{eq:GX-2+1-interlacing-amplitudes}
      \begin{aligned}
        Q_1 = 2 m_1 e^{-x_1}
        &= \frac{\mu_1}{\lambda_1} \, \left( \frac{\detJ_{00}^{11}}{\detJ_{01}^{10}} + C \right)
        = \frac{\mu_1}{\lambda_1} \, \left( \frac{1}{b_1} + C \right)
        ,\\
        P_1 = 2 n_1 e^{-y_1}
        &= \frac{\detJ_{01}^{10}}{\detJ_{11}^{01}}
        = \frac{b_1}{\frac{\mu_1 a_1 b_1}{\lambda_1 + \mu_1}}
        = \frac{\lambda_1 + \mu_1}{\mu_1 a_1}
        ,\\
        Q_2 = 2 m_2 e^{-x_2}
        &= \frac{1}{\detJ_{01}^{10}}
        = \frac{1}{b_1}
        ,
      \end{aligned}
    \end{equation}
  \end{subequations}
  or, more explicitly,
  \begin{subequations}
    \label{eq:GX-2+1-interlacing-explicit-joint}
    \begin{equation}
      \label{eq:GX-2+1-interlacing-explicit-positions}
      \begin{aligned}
        x_1(t)
        &=
        \frac{t}{2} \left( \frac{1}{\lambda_1} + \frac{1}{\mu_1} \right)
        + \frac12 \ln \left( \frac{1}{1 + C b_1(0) \, e^{t/\mu_1}} \right)
        + \frac12 \ln \left( \frac{2 a_1(0) \, b_1(0)}{\lambda_1 + \mu_1}  \right)
        ,\\
        y_1(t)
        &=
        \frac{t}{2} \left( \frac{1}{\lambda_1} + \frac{1}{\mu_1} \right)
        + \frac12 \ln \left( \frac{2 a_1(0) \, b_1(0)}{\lambda_1 + \mu_1} \right)
        ,\\
        x_2(t)
        &=
        \frac{t}{2} \left( \frac{1}{\lambda_1} + \frac{1}{\mu_1} \right)
        + \frac12 \ln \left( \frac{2 a_1(0) \, b_1(0)}{\lambda_1 + \mu_1} + 2D b_1(0) \, e^{-t/\lambda_1} \right)
      \end{aligned}
    \end{equation}
    and
    \begin{equation}
      \label{eq:GX-2+1-interlacing-explicit-amplitudes}
      \begin{aligned}
        \ln m_1(t)
        &=
        x_1(t)
        + \ln \left( \frac{1}{b_1(0) \, e^{t/\mu_1}} + C \right)
        + \ln \left( \frac{\mu_1}{2 \lambda_1} \right)
        \\ &=
        \frac{t}{2} \left( \frac{1}{\lambda_1} - \frac{1}{\mu_1} \right)
        + \frac12 \ln \left( 1 + C b_1(0) \, e^{t/\mu_1} \right)
        \\ & \quad
        + \frac12 \ln \left( \frac{2 a_1(0)}{(\lambda_1 + \mu_1) \, b_1(0)} \right)
        + \ln \left( \frac{\mu_1}{2 \lambda_1} \right)
        ,\\
        -\ln n_1(t)
        &=
        - y_1(t)
        + \frac{t}{\lambda_1}
        + \ln \left( \frac{2\mu_1 \, a_1(0)}{\lambda_1 + \mu_1} \right)
        \\ &=
        \frac{t}{2} \left( \frac{1}{\lambda_1} - \frac{1}{\mu_1} \right)
        + \frac12 \ln \left( \frac{2 a_1(0)}{ (\lambda_1 + \mu_1) b_1(0)} \right)
        + \ln \left( \mu_1 \right)
        ,\\[1ex]
        \ln m_2(t)
        &=
        x_2(t)
        - \frac{t}{\mu_1}
        - \ln \bigl( 2 b_1(0) \bigr)
        \\ &=
        \frac{t}{2} \left( \frac{1}{\lambda_1} - \frac{1}{\mu_1} \right)
        + \frac12 \ln \left( 1 + \frac{D (\lambda_1 + \mu_1)}{a_1(0)}  \, e^{-t/\lambda_1} \right)
        \\ & \quad
        + \frac12 \ln \left( \frac{2 a_1(0)}{ (\lambda_1 + \mu_1) b_1(0)} \right)
        - \ln 2
        .
      \end{aligned}
    \end{equation}
  \end{subequations}
  This imples that the positions behave as follows:
  \begin{itemize}
  \item the curve
    $x = y_1(t) = \frac12 \left( \frac{1}{\lambda_1} + \frac{1}{\mu_1} \right) t + \text{constant}$
    is a straight line,
  \item $x = x_1(t)$ approaches the line $x = y_1(t)$ as $t \to -\infty$,
  \item $x = x_1(t)$ approaches a line $x = \frac{1}{2\lambda_1} \, t + \text{constant}$ as $t \to +\infty$,
  \item $x = x_2(t)$ approaches a line $x = \frac{1}{2 \mu_1} \, t + \text{constant}$ as $t \to -\infty$,
  \item $x = x_2(t)$ approaches the line $x = y_1(t)$ as $t \to +\infty$.
  \end{itemize}
  This is illustrated in Figure~\ref{fig:GX-2+1-interlacing-positions-all},
  with
  \begin{equation}
    \label{eq:GX-2+1-interlacing-parameters}
    \lambda_1 = 1
    ,\quad
    \mu_1 = 3
    ,\quad
    a_1(0) = b_1(0) = 1
    ,\quad
    C = 10^6
    ,\quad
    D = 10^{20}
    .
  \end{equation}
  Note that the incoming velocity
  $(2 \mu_1)^{-1}$
  of the rightmost peakon
  is \emph{different} from the outgoing velocity
  $(2 \lambda_1)^{-1}$
  of the leftmost peakon!
  In particular, it does not seem meaningful to compare those two curves and talk about a ``phase shift''.
  This is in contrast to the even case, where the incoming velocities always appear as outgoing velocities
  in the opposite order.
  
  For the amplitudes
  (Figure~\ref{fig:GX-2+1-interlacing-amplitudes-all})
  we similarly find:
  \begin{itemize}
  \item the curve
    $s = -\ln n_1(t) = \frac12 \left( \frac{1}{\lambda_1} - \frac{1}{\mu_1} \right) t + \text{constant}$
    is a straight line,
  \item $s = \ln m_1(t)$ approaches a line parallel to $x = -\ln n_1(t)$ as $t \to -\infty$,
  \item $s = \ln m_1(t)$ approaches a line $x = \frac{1}{2 \lambda_1} \, t + \text{constant}$ as $ \to +\infty$,
  \item $s = \ln m_2(t)$ approaches a line $x = - \frac{1}{2 \mu_1} \, t + \text{constant}$ as $ \to -\infty$,
  \item $s = \ln m_2(t)$ approaches a line parallel to $x = -\ln n_1(t)$ as $t \to +\infty$.
  \end{itemize}
\end{example}

\begin{figure}[H]
  \centering
  \includegraphics[width=13cm]{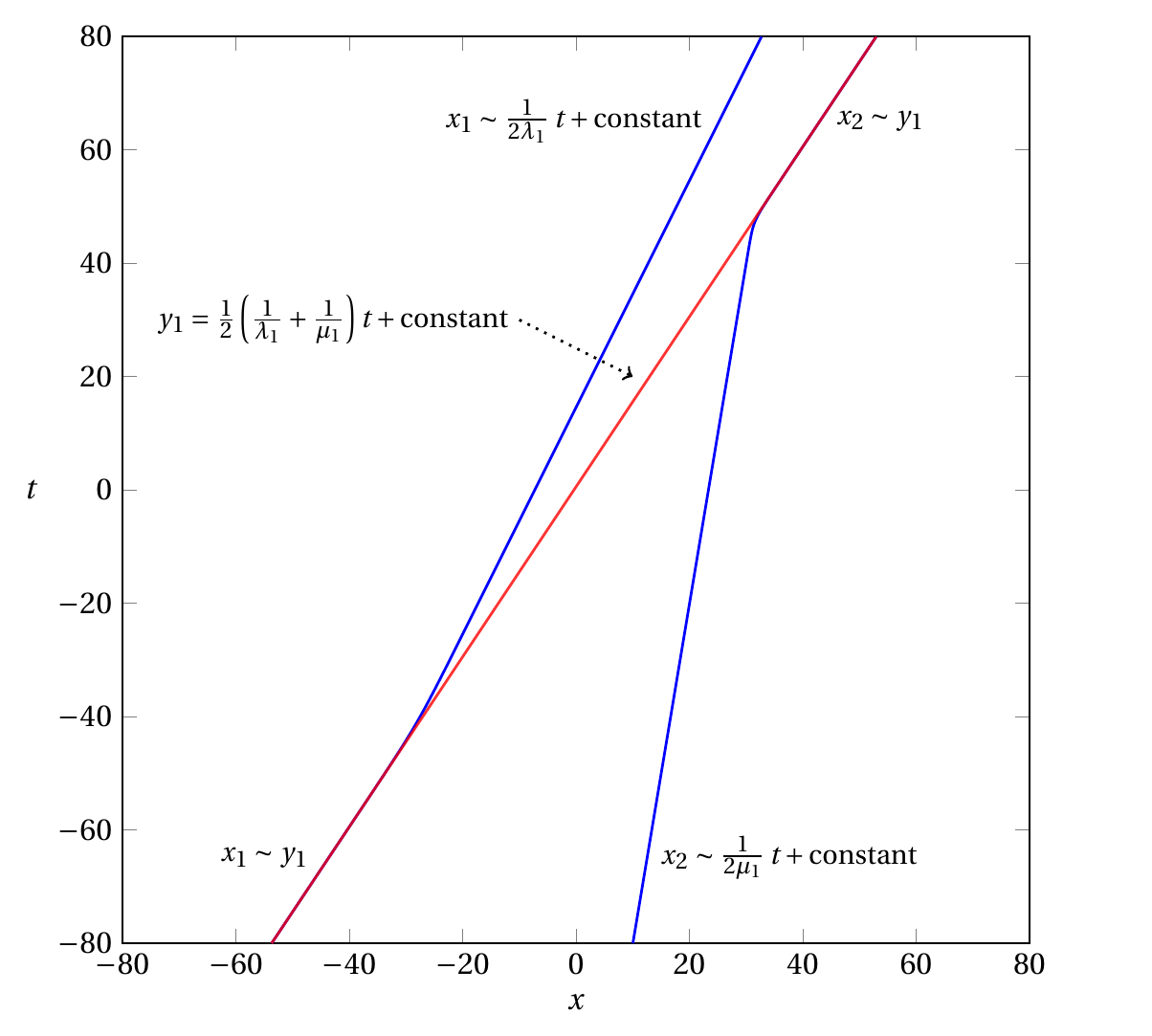}
  \caption{\textbf{Positions for a $2+1$ interlacing solution.}
    Positions of the peakons in a solution with the $2+1$ interlacing configuration
    $x_1 < y_1 < x_2$,
    as described in Example~\ref{ex:GX-2+1-interlacing}.
    The parameters are given in~\eqref{eq:GX-2+1-interlacing-parameters};
    in particular, the eigenvalues are $\lambda_1 = 1$ and $\mu_1 = 3$.
    The middle peakon at~$y_1$ (red curve) travels with constant velocity
    $\frac12 \bigl( \frac{1}{\lambda_1} + \frac{1}{\mu_1} \bigr)$,
    and is approached by the $x_1$-peakon as $t \to -\infty$
    and by the $x_2$-peakon as $t \to +\infty$.
    More remarkably, the asymptotic velocity of the $x_2$-peakon as $t \to -\infty$,
    namely $\frac{1}{2 \mu_1}$,
    is \emph{not} the same as that of the $x_1$-peakon as $t \to +\infty$,
    which is~$\frac{1}{2 \lambda_1}$.
  }
  \label{fig:GX-2+1-interlacing-positions-all}
\end{figure}

\begin{figure}[H]
  \centering
  \includegraphics[width=13cm]{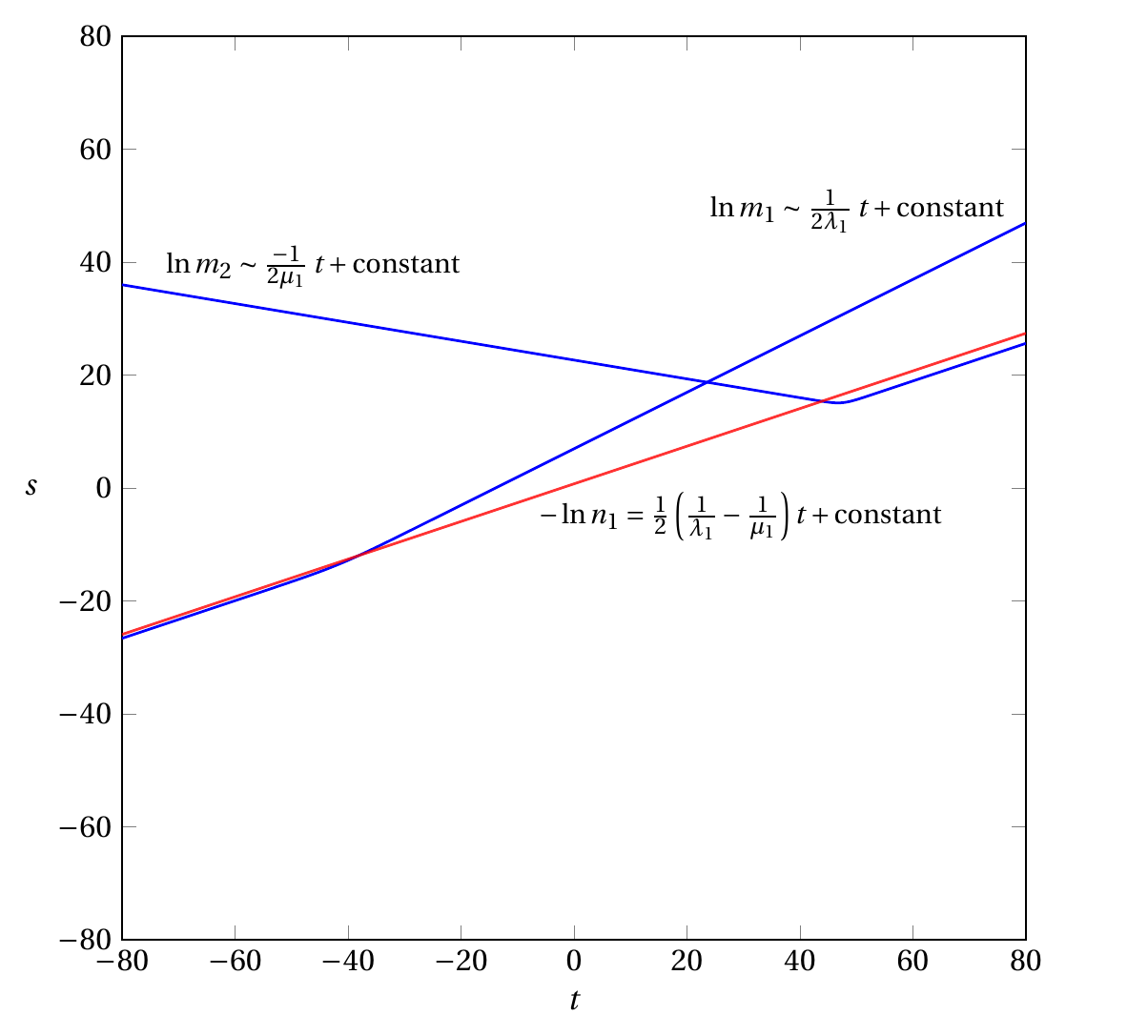}
  \caption{\textbf{Amplitudes for a $2+1$ interlacing solution.}
    Amplitudes of the peakons for the same solution as in
    Figure~\ref{fig:GX-2+1-interlacing-positions-all}
    (Example~\ref{ex:GX-2+1-interlacing}).
    The curve $s = -\ln n_1(t)$ is a straight line (red)
    with slope $\frac12 \bigl( \frac{1}{\lambda_1} - \frac{1}{\mu_1} \bigr)$.
    As $t \to -\infty$, the curve $s = \ln m_1(t)$ approaches a line parallel to that line,
    and similarly for $s = \ln m_2(t)$ as $t \to +\infty$.
    However, in the other time direction they approach lines with \emph{different} slopes,
    $\frac{1}{2 \lambda_1}$
    and~$\frac{-1}{2 \mu_1}$.
  }
  \label{fig:GX-2+1-interlacing-amplitudes-all}
\end{figure}

\begin{example}[The $4+3$ interlacing case]
  \label{ex:GX-4+3-interlacing}

  The formulas for the $4+3$ interlacing solution are,
  with $\detJ_{ij}^{rs}=\detJ[3,3,r,s,i,j]$:
  \begin{subequations}
    \label{eq:GX-4+3-interlacing-joint}
    \begin{equation}
      \label{eq:GX-4+3-interlacing-positions}
      \begin{aligned}
        X_1 = \tfrac12 e^{2x_1}
        &= \frac{\detJ_{33}^{00}}{\detJ_{22}^{11} + C \detJ_{23}^{10}}
        ,
        \qquad
        &
        Y_1 = \tfrac12 e^{2y_1}
        &= \frac{\detJ_{33}^{00}}{\detJ_{22}^{11}}
        ,
        \\[1ex]
        X_2 = \tfrac12 e^{2x_2}
        &= \frac{\detJ_{23}^{00}}{\detJ_{12}^{11}}
        ,
        &
        Y_2 = \tfrac12 e^{2y_2}
        &= \frac{\detJ_{22}^{00}}{\detJ_{11}^{11}}
        ,
        \\[1ex]
        X_3 = \tfrac12 e^{2x_3}
        &= \frac{\detJ_{12}^{00}}{\detJ_{01}^{11}}
        ,
        &
        Y_3 = \tfrac12 e^{2y_3}
        &= \detJ_{11}^{00}
        ,
        \\[1ex]
        X_4 = \tfrac12 e^{2x_4}
        &= \detJ_{11}^{00} + D \detJ_{01}^{00}
        &&
      \end{aligned}
    \end{equation}
    and
    \begin{equation}
      \label{eq:GX-4+3-interlacing-amplitudes}
      \begin{aligned}
        Q_1 = 2 m_1 e^{-x_1}
        &=  \frac{\mu_1 \mu_2 \mu_3}{\lambda_1 \lambda_2 \lambda_3 }
        \, \left( \frac{\detJ_{22}^{11}}{\detJ_{23}^{10}} + C \right)
        ,
        \qquad
        &
        P_1 = 2 n_1 e^{-y_1}
        &= \frac{\detJ_{22}^{11} \detJ_{23}^{10}}{\detJ_{22}^{01} \detJ_{33}^{01}}
        ,
        \\[1ex]
        Q_2 = 2 m_2 e^{-x_2}
        &= \frac{\detJ_{12}^{11} \detJ_{22}^{01}}{\detJ_{12}^{10} \detJ_{23}^{10}}
        ,
        &
        P_2 = 2 n_2 e^{-y_2}
        &= \frac{\detJ_{11}^{11} \detJ_{12}^{10}}{\detJ_{11}^{01} \detJ_{22}^{01}}
        ,
        \\[1ex]
        Q_3 = 2 m_3 e^{-x_3}
        &= \frac{\detJ_{01}^{11} \detJ_{11}^{01}}{\detJ_{01}^{10} \detJ_{12}^{10}}
        ,
        &
        P_3 = 2 n_3 e^{-y_3}
        &= \frac{\detJ_{01}^{10}}{\detJ_{11}^{01}}
        ,
        \\[1ex]
        Q_4 = 2 m_4 e^{-x_4}
        &= \frac{1}{\detJ_{01}^{10}}
        .
        &&
      \end{aligned}
    \end{equation}
  \end{subequations}
  (Compare with the $3+3$ interlacing solution~\eqref{eq:GX-3+3-interlacing-joint}
  in Example~\ref{ex:GX-3+3-interlacing}.)
  
  In Figure~\ref{fig:GX-4+3-interlacing-positions-all}
  we have plotted the positions
  obtained from~\eqref{eq:GX-4+3-interlacing-positions}
  with the spectral data
  \begin{equation}
    \label{eq:GX-4+3-interlacing-spectral-data}
    \begin{gathered}
      \lambda_1 = \frac{1}{5}
      ,\quad
      \lambda_2 = 1
      ,\quad
      \lambda_3 = 2
      ,\qquad
      \mu_1 = \frac{1}{3}
      ,\quad
      \mu_2 = 4
      ,\quad
      \mu_3 = 8
      ,
      \\
      a_1(0) = 10^{-4}
      ,\quad
      a_2(0) = 10^{1}
      ,\quad
      a_3(0) = 10^{3}
      ,\qquad
      b_1(0) = 10^{-6}
      ,\quad
      b_2(0) = 10^{2}
      ,\quad
      b_3(0) = 1
      ,
      \\
      C = 10^{20}
      ,\quad
      D = 10^{18}
      ,
    \end{gathered}
  \end{equation}
  i.e., the same values as for the $3+3$ interlacing
  solution in Example~\ref{ex:GX-3+3-interlacing}, plus the new parameters
  $\mu_3$ and~$a_3(0)$.
  Here we see again the same phenomenon as in
  Example~\ref{ex:GX-2+1-interlacing},
  namely that in the odd case, the asymptotic velocities for the blue curves $x = x_k(t)$
  as $t \to -\infty$
  are not the same as when $t \to +\infty$
  (although they \emph{are} for the red curves $x = y_k(t)$).
  In fact, by Theorem~\ref{thm:asymptotics-singletons-odd}
  the incoming velocites ($t \to -\infty$) are, from left to right,
  \begin{equation}
    \label{eq:asymptotic-velocities-odd-example-neginf}
    \begin{aligned}
      \text{$x_1$ and $y_1$} &:&
      \frac12 \left( \frac{1}{\lambda_1} + \frac{1}{\mu_1} \right) &= 4
      ,\\
      x_2 &:&
      \frac12 \left( \frac{1}{\lambda_2} + \frac{1}{\mu_1} \right) &= 2
      ,\\
      y_2 &:&
      \frac12 \left( \frac{1}{\lambda_2} + \frac{1}{\mu_2} \right) &= \frac{5}{8}
      ,\\
      x_3 &:&
      \frac12 \left( \frac{1}{\lambda_3} + \frac{1}{\mu_2} \right) &= \frac{3}{8}
      ,\\
      y_3 &:&
      \frac12 \left( \frac{1}{\lambda_3} + \frac{1}{\mu_3} \right) &= \frac{5}{16}
      ,\\
      x_4 &:&
      \frac12 \frac{1}{\mu_3} &= \frac{1}{16}
      .
    \end{aligned}
  \end{equation}
  (where the five leftmost peakons agree with~\eqref{eq:asymptotic-velocities-example} from the $3+3$ case),
  whereas the outgoing velocites ($t \to +\infty$) are, from right to left,
  \begin{equation}
    \label{eq:asymptotic-velocities-odd-example-posinf}
    \begin{aligned}
      \text{$x_4$ and $y_3$} &:&
      \frac12 \left( \frac{1}{\lambda_1} + \frac{1}{\mu_1} \right) &= 4
      ,\\
      x_3 &:&
      \frac12 \left( \frac{1}{\lambda_1} + \frac{1}{\mu_2} \right) &= \frac{21}{8}
      ,\\
      y_2 &:&
      \frac12 \left( \frac{1}{\lambda_2} + \frac{1}{\mu_2} \right) &= \frac{5}{8}
      ,\\
      x_2 &:&
      \frac12 \left( \frac{1}{\lambda_2} + \frac{1}{\mu_3} \right) &= \frac{9}{16}
      ,\\
      y_1 &:&
      \frac12 \left( \frac{1}{\lambda_3} + \frac{1}{\mu_3} \right) &= \frac{5}{16}
      ,\\
      x_1 &:&
      \frac12 \frac{1}{\lambda_3} &= \frac{1}{4}
      .
    \end{aligned}
  \end{equation}
  Figure~\ref{fig:GX-4+3-interlacing-amplitudes-all} shows the amplitudes for the same
  interlacing $4+3$ solution.
  As usual, the curves $s = \ln m_k(t)$ and $s = -\ln n_k(t)$ asymptotically
  approach straight lines.
  As $t \to -\infty$, the asymptotic slopes are
  (in order):
  \begin{equation}
    \label{eq:asymptotic-slopes-odd-example-neginf}
    \begin{aligned}
      \text{$\ln m_1$ and $-\ln n_1$} &:&
      \frac12 \left( \frac{1}{\lambda_1} - \frac{1}{\mu_1} \right) &= 1
      ,\\
      \ln m_2 &:&
      \frac12 \left( \frac{1}{\lambda_2} - \frac{1}{\mu_1} \right) &= -1
      ,\\
      -\ln n_2 &:&
      \frac12 \left( \frac{1}{\lambda_2} - \frac{1}{\mu_2} \right) &= \frac{3}{8}
      ,\\
      \ln m_3 &:&
      \frac12 \left( \frac{1}{\lambda_3} - \frac{1}{\mu_2} \right) &= \frac{1}{8}
      ,\\
      -\ln n_3 &:&
      \frac12 \left( \frac{1}{\lambda_3} - \frac{1}{\mu_3} \right) &= \frac{3}{16}
      ,\\
      \ln m_4 &:&
      - \frac{1}{2 \mu_3} &= -\frac{1}{8}
      ,
    \end{aligned}
  \end{equation}
  and as $t \to +\infty$, they are
  (in reverse order):
  \begin{equation}
    \label{eq:asymptotic-slopes-odd-example-posinf}
    \begin{aligned}
      \text{$\ln m_4$ and $-\ln n_3$} &:&
      \frac12 \left( \frac{1}{\lambda_1} - \frac{1}{\mu_1} \right) &= 1
      ,\\
      \ln m_3 &:&
      \frac12 \left( \frac{1}{\lambda_1} - \frac{1}{\mu_2} \right) &= \frac{19}{8}
      ,\\
      -\ln n_2 &:&
      \frac12 \left( \frac{1}{\lambda_2} - \frac{1}{\mu_2} \right) &= \frac{3}{8}
      ,\\
      \ln m_2 &:&
      \frac12 \left( \frac{1}{\lambda_2} - \frac{1}{\mu_3} \right) &= \frac{7}{16}
      ,\\
      -\ln n_1 &:&
      \frac12 \left( \frac{1}{\lambda_3} - \frac{1}{\mu_3} \right) &= \frac{3}{16}
      ,\\
      \ln m_1 &:&
      \frac{1}{2 \lambda_3} &= \frac{1}{4}
      .
    \end{aligned}
  \end{equation}
\end{example}

\begin{figure}[H]
  \centering
  \includegraphics[width=13cm]{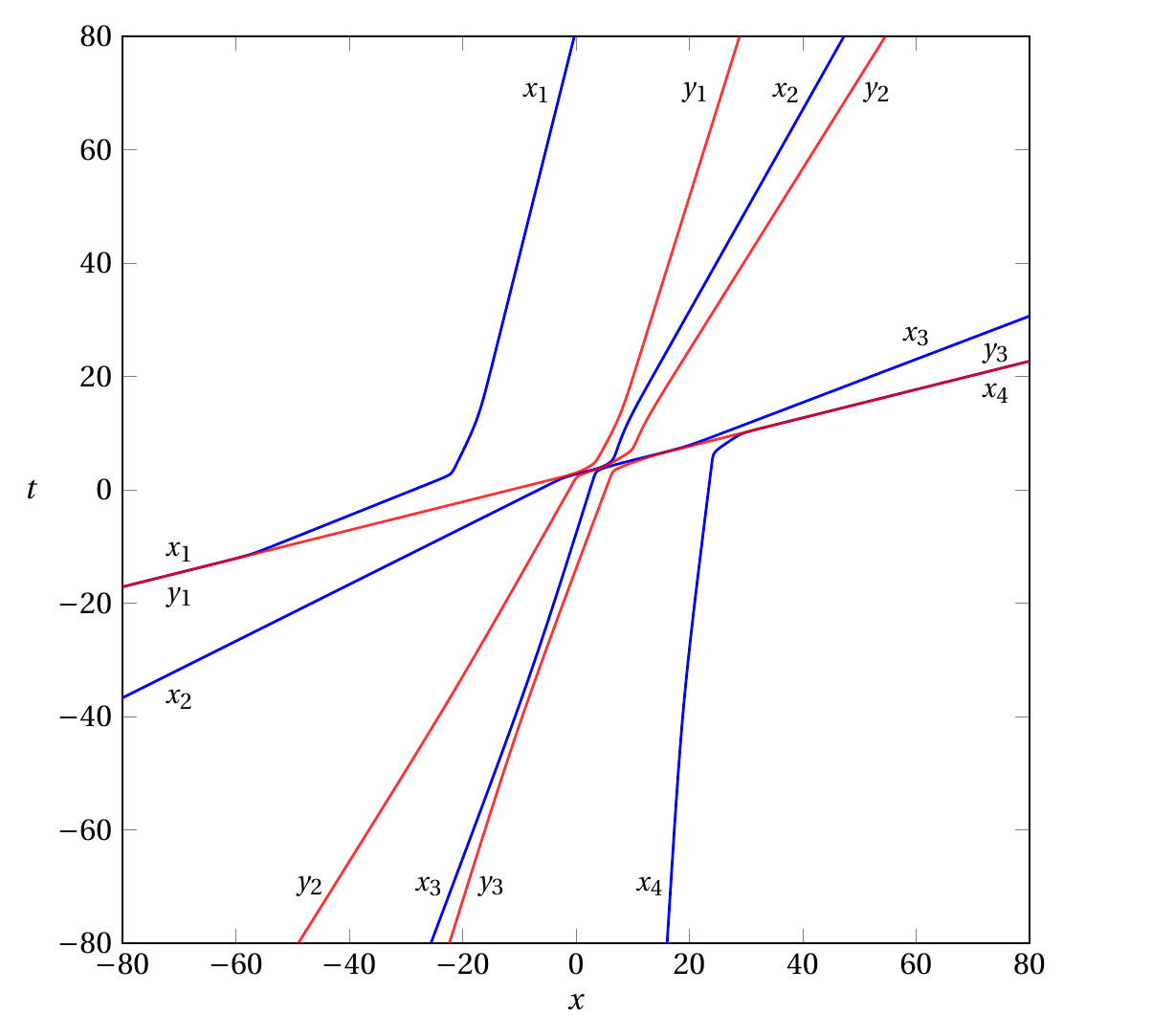}
  \caption{\textbf{Positions in the $4+3$ interlacing case.}
    Positions of the peakons in the $4+3$ interlacing
    solution~\eqref{eq:GX-4+3-interlacing-joint}
    in Example~\ref{ex:GX-4+3-interlacing},
    with the parameter values~\eqref{eq:GX-4+3-interlacing-spectral-data}.
    For the red curves $x = y_k(t)$, the outgoing velocities are the
    same as the incoming velocities, namely $4$, $\frac{5}{8}$ and~$\frac{5}{16}$.
    But the blue curves $x = x_k(t)$ have incoming velocities
    $4$, $2$, $\frac{3}{8}$, $\frac{1}{16}$ that (except for the fastest one)
    are different from the outgoing velocities
    $4$, $\frac{21}{8}$, $\frac{9}{16}$ and~$\frac{1}{4}$;
    see~\eqref{eq:asymptotic-velocities-odd-example-neginf} and~\eqref{eq:asymptotic-velocities-odd-example-posinf}.
  }
  \label{fig:GX-4+3-interlacing-positions-all}
\end{figure}

\begin{figure}[H]
  \centering
  \includegraphics[width=13cm]{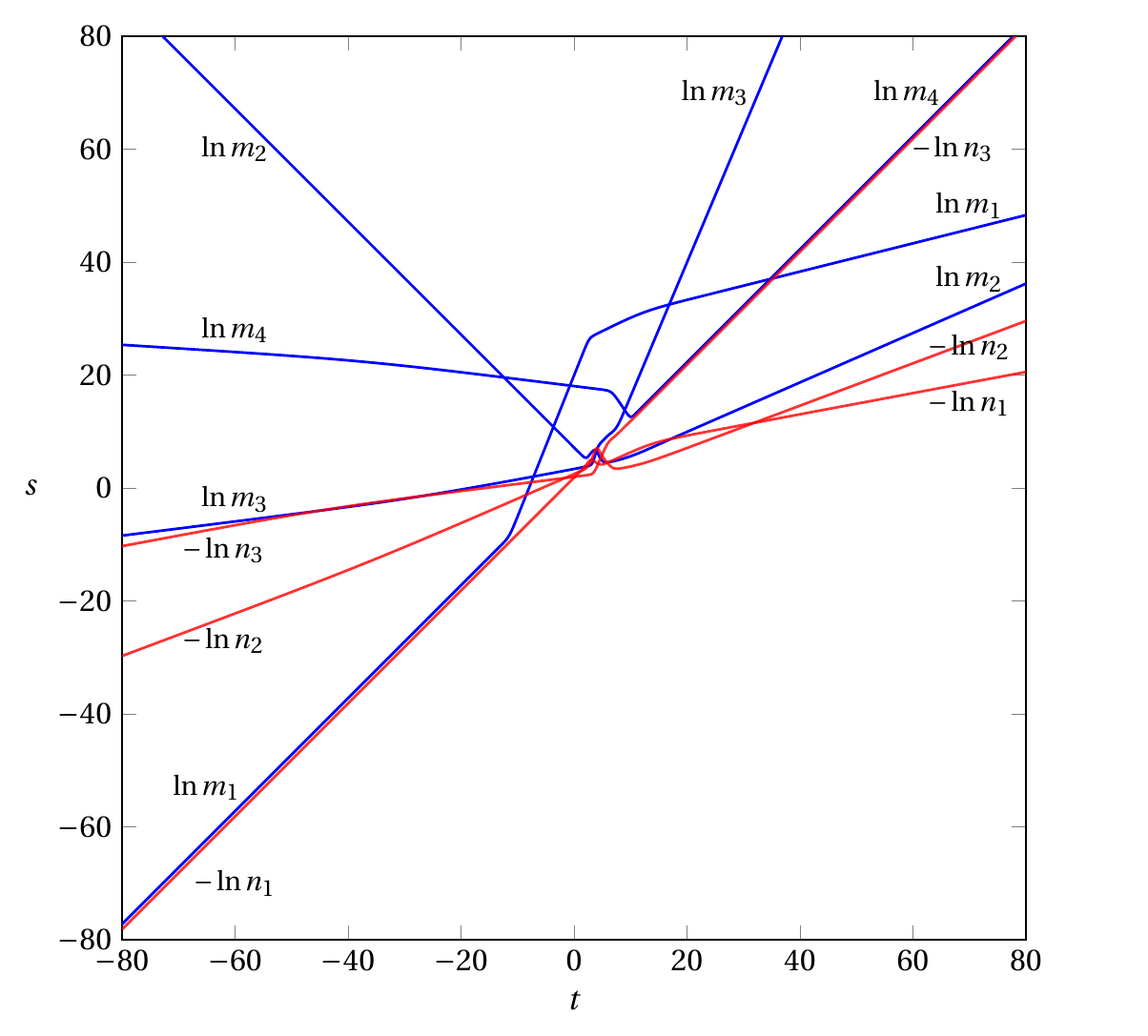}
  \caption{\textbf{Amplitudes in the $4+3$ interlacing case.}
    Amplitudes of the peakons for the same solution as in
    Figure~\ref{fig:GX-4+3-interlacing-positions-all}
    (Example~\ref{ex:GX-4+3-interlacing}).
    For the red curves $s = -\ln n_k(t)$, the asymptotic slopes as $t \to -\infty$
    are equal to those as $t \to +\infty$, namely $1$, $\frac{3}{8}$ and~$\frac{3}{16}$.
    But for the blue curves $s = \ln m_k(t)$, they are
    $1$, $-1$, $\frac{1}{8}$ and~$-\frac{1}{8}$ as $t \to -\infty$
    and
    $1$, $\frac{19}{8}$, $\frac{7}{16}$ and~$\frac{1}{4}$ as $t \to +\infty$;
    see~\eqref{eq:asymptotic-slopes-odd-example-neginf} and~\eqref{eq:asymptotic-slopes-odd-example-posinf}.
  }
  \label{fig:GX-4+3-interlacing-amplitudes-all}
\end{figure}

\begin{example}[Non-singleton groups in the odd case]
  \label{ex:GX-4+3-allgroups}

  When it comes to groups with more than one peakon,
  the odd case is very similar to the even case in terms of formulas.
  For example, if we have $4+3$ groups where the $Y_2$-group contains five peakons,
  then the singleton solution formulas from Example~\ref{ex:GX-4+3-interlacing},
  \begin{equation*}
    Y_2 = \frac{ \detJ_{22}^{00}}{\detJ_{11}^{11}}
    ,\qquad
    P_2 = \frac{\detJ_{11}^{11} \detJ_{12}^{10}}{\detJ_{11}^{01} \detJ_{22}^{01}}
    ,
  \end{equation*}
  are replaced by
  \begin{subequations}
    \label{eq:GX-4+3-typicalY2-solution-joint}
    \begin{equation}
      \label{eq:GX-4+3-typicalY2-solution-positions}
      \begin{aligned}
        Y_{2,1} &= \frac{\detJ_{23}^{00} + \tau_1 \detJ_{22}^{00}}{\detJ_{12}^{11} + \tau_1 \detJ_{11}^{11}}
        ,
        \\[1ex]
        Y_{2,2} &= \frac{\detJ_{23}^{00} + (\tau_1 + \tau_2) \detJ_{22}^{00} + (\tau_2 \sigma_1) \detJ_{12}^{00}}{\detJ_{12}^{11} + (\tau_1 + \tau_2) \detJ_{11}^{11} + \tau_2 \sigma_1 \detJ_{01}^{11}}
        ,
        \\[1ex]
        Y_{2,3} &= \frac{\detJ_{23}^{00} + (\tau_1 + \tau_2 + \tau_3) \detJ_{22}^{00} + (\tau_2 \sigma_1 + \tau_3 \sigma_2) \detJ_{12}^{00}}{\detJ_{12}^{11} + (\tau_1 + \tau_2 + \tau_3) \detJ_{11}^{11} + (\tau_2 \sigma_1 + \tau_3 \sigma_2) \detJ_{01}^{11}}
        ,
        \\[1ex]
        Y_{2,4} &= \frac{\detJ_{23}^{00} + (\tau_1 + \tau_2 + \tau_3 + \tau_4) \detJ_{22}^{00} + (\tau_2 \sigma_1 + \tau_3 \sigma_2 + \tau_4 \sigma_3) \detJ_{12}^{00}}{\detJ_{12}^{11} + (\tau_1 + \tau_2 + \tau_3 + \tau_4) \detJ_{11}^{11} + (\tau_2 \sigma_1 +\tau_3 \sigma_2 + \tau_4 \sigma_3) \detJ_{01}^{11}}
        ,
        \\[1ex]
        Y_{2,5} &= \frac{{\detJ_{22}^{00}} + \sigma_4 \detJ_{12}^{00}} {{\detJ_{11}^{11}} + \sigma_4 \detJ_{01}^{11}}
      \end{aligned}
    \end{equation}
    and
    \begin{equation}
      \label{eq:GX-4+3-typicalY2-solution-amplitudes}
      \begin{aligned}
        P_{2,1} &= \sigma_1 \detJ_{12}^{10} \bigl( \detJ_{12}^{11} +  \tau_1  \detJ_{11}^{11}  \bigr)
        \\ & \quad
        \times \bigl( \detJ_{22}^{01} (\detJ_{22}^{01} + \sigma_1 \detJ_{12}^{01} +  \tau_1 \sigma_1 \detJ_{11}^{01}) \bigr)^{-1}
        ,
        \\[1ex]
        P_{2,2} &= (\sigma_2 - \sigma_1) \detJ_{12}^{10}  \bigl( \detJ_{12}^{11} + (\tau_1 + \tau_2) \detJ_{11}^{11} +  \tau_2 \sigma_1 \detJ_{01}^{11} \bigr)
        \\ & \quad
        \times \bigl( \detJ_{22}^{01} + \sigma_2 \detJ_{12}^{01} + \bigl( \sigma_2 (\tau_1 + \tau_2) - \tau_2 \sigma_1 \bigr) \detJ_{11}^{01} \bigr)^{-1}
        \\ & \quad
        \times \bigl( \detJ_{22}^{01} + \sigma_1 \detJ_{12}^{01} +  \tau_1 \sigma_1 \detJ_{11}^{01} \bigr)^{-1}
        ,
        \\[1ex]
        P_{2,3} &= (\sigma_3 - \sigma_2) \detJ_{12}^{10} \bigl( \detJ_{12}^{11} + (\tau_1 + \tau_2 + \tau_3)  \detJ_{11}^{11} + (\tau_2 \sigma_1 + \tau_3 \sigma_2) \detJ_{01}^{11} \bigr)
        \\ & \quad
        \times \bigl( \detJ_{22}^{01} + \sigma_3 \detJ_{12}^{01} + \bigl( \sigma_3 (\tau_1 + \tau_2 + \tau_3) - (\tau_2 \sigma_1 + \tau_3 \sigma_2) \bigr) \detJ_{11}^{01} \bigr)^{-1}
        \\ & \quad
        \times \bigl( \detJ_{22}^{01} + \sigma_2 \detJ_{12}^{01} + \bigl( \sigma_2 (\tau_1 + \tau_2) - \tau_2 \sigma_1 \bigr) \detJ_{11}^{01} \bigr)^{-1}
        ,
        \\[1ex]
        P_{2,4} &= (\sigma_4 - \sigma_3) \detJ_{12}^{10} \bigl( \detJ_{12}^{11} + (\tau_1 + \tau_2 + \tau_3 + \tau_4) \detJ_{11}^{11} + (\tau_2 \sigma_1 + \tau_3 \sigma_2 + \tau_4 \sigma_3) \detJ_{01}^{11} \bigr)
        \\ & \quad
        \times \bigl( \detJ_{22}^{01} + \sigma_4 \detJ_{12}^{01} + \bigl( \sigma_4 (\tau_1 + \tau_2 + \tau_3 + \tau_4) - (\tau_2 \sigma_1 + \tau_3 \sigma_2 + \tau_4 \sigma_3) \bigr) \detJ_{11}^{01} \bigr)^{-1}
        \\ & \quad
        \times \bigl( \detJ_{22}^{01} + \sigma_3 \detJ_{12}^{01} + \bigl( \sigma_3 (\tau_1 + \tau_2 + \tau_3) - (\tau_2 \sigma_1 + \tau_3 \sigma_2) \bigr) \detJ_{11}^{01} \bigr)^{-1}
        ,
        \\[1ex]
        P_{2,5} &= \frac{\detJ_{12}^{10} \bigl( \detJ_{11}^{11} + \sigma_4 \detJ_{01}^{11} \bigr)}{\detJ_{11}^{01} \bigl( \detJ_{22}^{01} + \sigma_4 \detJ_{12}^{01} + \bigl( \sigma_4 (\tau_1 + \tau_2 + \tau_3 + \tau_4) - (\tau_2 \sigma_1 + \tau_3 \sigma_2 + \tau_4 \sigma_3) \bigr) \detJ_{11}^{01} \bigr)}
        ,
      \end{aligned}
    \end{equation}
  \end{subequations}
  where $\tau_i = \tau^Y_{2,i}$
  and $\sigma_i = \sigma^Y_{2,i}$.
  (Compare with~\eqref{eq:GX-3+3-typicalY2-solution-joint}
  in Example~\ref{ex:GX-3+3-typicalY2}.)

  This means that as $t \to -\infty$, when terms $\detJ_{ij}^{rs}$
  with small $i$ and~$j$ are dominant,
  the $y_{2,1}$-peakon will behave like the singleton~$y_2$,
  while the other four peakons in the $Y_2$-group
  will approach the neighbouring $x_3$-singleton to the right,
  given by $X_3 = \detJ_{12}^{00}/\detJ_{01}^{11}$.
  Likewise, as $t \to +\infty$, when terms $\detJ_{ij}^{rs}$
  with large $i$ and~$j$ are dominant,
  the $y_{2,5}$-peakon will behave like the singleton~$y_2$,
  while the other four will approach the neighbouring $x_2$-singleton to the left,
  given by $X_2 = \detJ_{23}^{00}/\detJ_{12}^{11}$.
  So this aspect of the asymptotics is just like for the even case,
  although this of course also means that the differences in singleton behaviour
  between the even and odd cases will carry over and affect the groups as well.

  As an illustration of all this,
  Figure~\ref{fig:GX-4+3-allgroups-positions} shows the positions $x = x_{k,i}(t)$
  and $x = y_{k_i}(t)$ for a case with $4+3$ groups,
  where we have used the same parameters as
  for the case with $3+3$ groups in Example~\ref{ex:GX-3+3-allgroups},
  with the spectral data $\mu_3 = 8$ and $b_3(0) = 1$
  added as in the $4+3$ interlacing case in Example~\ref{ex:GX-4+3-interlacing},
  together with the group parameters
  \begin{equation}
    \label{eq:GX-4+3-X4-parameters}
    \tau_{4,1}^X = \sigma_{4,1}^X = 10^5
  \end{equation}
  for an $X_4$-group with two peakons.
  
  Regarding the amplitudes,
  as $x \to -\infty$ or $x \to \infty$
  one peakon in each group (the first one or the last one, respectively)
  will follow the corresponding singleton
  (or a parallel line, for the outermost groups $X_1$ and~$X_{K+1}$),
  while the other curves $s = \ln m_{k,i}$ and $s = -\ln n_{k,i}$ follow parallel lines with
  other slopes;
  cf. Example~\ref{ex:GX-3+3-allgroups} for the even case.
  The slopes that occur for these other lines are, as $t \to -\infty$ (in order,
  and with $2 \le i \le N$, where $N$ is the number of peakons in the group in question):
  \begin{equation}
    \label{eq:asymptotic-group-slopes-odd-example-neginf}
    \begin{aligned}
      \ln m_{1,i} &:&
      \frac12 \left( \frac{3}{\lambda_1} - \frac{1}{\mu_1} \right) &= 6
      ,\\
      -\ln n_{1,i} &:&
      \frac12 \left( \frac{1}{\lambda_2} - \frac{3}{\mu_1} \right) &= -4
      ,\\
      \ln m_{2,i} &:&
      \frac12 \left( \frac{3}{\lambda_2} - \frac{1}{\mu_2} \right) &= \frac{11}{8}
      ,\\
      -\ln n_{2,i} &:&
      \frac12 \left( \frac{1}{\lambda_3} - \frac{3}{\mu_2} \right) &= -\frac{1}{8}
      ,\\
      \ln m_{3,i} &:&
      \frac12 \left( \frac{3}{\lambda_3} - \frac{1}{\mu_3} \right) &= \frac{11}{16}
      ,\\
      -\ln n_{3,i} &:&
      -\frac12 \frac{3}{\mu_3} &= -\frac{3}{16}
      ,\\
      \ln m_{4,i} &:&
      &
      \phantom{=}
      \,\,\,\,
      0
      .
    \end{aligned}
  \end{equation}
  And as $t \to +\infty$, they are (in reverse order,
  and with $1 \le i \le N-1$):
  \begin{equation}
    \label{eq:asymptotic-group-slopes-odd-example-posinf}
    \begin{aligned}
      \ln m_{4,i} &:&
      \frac12 \left( \frac{1}{\lambda_1} - \frac{3}{\mu_1} \right) &= -2
      ,\\
      -\ln n_{3,i} &:&
      \frac12 \left( \frac{3}{\lambda_1} - \frac{1}{\mu_2} \right) &= \frac{59}{8}
      ,\\
      \ln m_{3,i} &:&
      \frac12 \left( \frac{1}{\lambda_2} - \frac{3}{\mu_2} \right) &= \frac{1}{8}
      ,\\
      -\ln n_{2,i} &:&
      \frac12 \left( \frac{1}{\lambda_2} - \frac{3}{\mu_3} \right) &= \frac{23}{16}
      ,\\
      \ln m_{2,i} &:&
      \frac12 \left( \frac{1}{\lambda_3} - \frac{3}{\mu_3} \right) &= \frac{1}{16}
      ,\\
      -\ln n_{1,i} &:&
      \frac12 \frac{3}{\lambda_3} &= \frac{3}{4}
      ,\\
      \ln m_{1,i} &:&
      &
      \phantom{=}
      \,\,\,\,
      0
      .
    \end{aligned}
  \end{equation}
  So for both $X$-groups and $Y$-groups (not just $X$-groups),
  these amplitude curves have different asymptotic
  slopes in the two time directions.
  For illustrations,
  showing the amplitudes of the $j$th group from the left
  together with the $j$th group from the right,
  see
  Figures~\ref{fig:GX-4+3-allgroups-amplitudes-Y2},
  \ref{fig:GX-4+3-allgroups-amplitudes-X2-X3},
  \ref{fig:GX-4+3-allgroups-amplitudes-Y1-Y3}
  and~\ref{fig:GX-4+3-allgroups-amplitudes-X1-X4},
  with $Y_2$, $\{ X_2, X_3 \}$, $\{ Y_1, Y_3 \}$ and $\{ X_1, X_4 \}$,
  respectively.
\end{example}

\begin{remark}
  As should be clear from the examples in this section,
  the solutions in the odd case display a kind of symmetry breaking
  which is not present in the even case.
  On the other hand, since there are equally many eigenvalues of each kind
  in the odd case,
  there is the possibility of obtaining a ``symmetric partner solution''
  to a given solution by swapping $\lambda_i \leftrightarrow \mu_i$
  for each~$i \in [1,K]$.
\end{remark}

\begin{figure}[H]
  \centering
  \includegraphics[width=13cm]{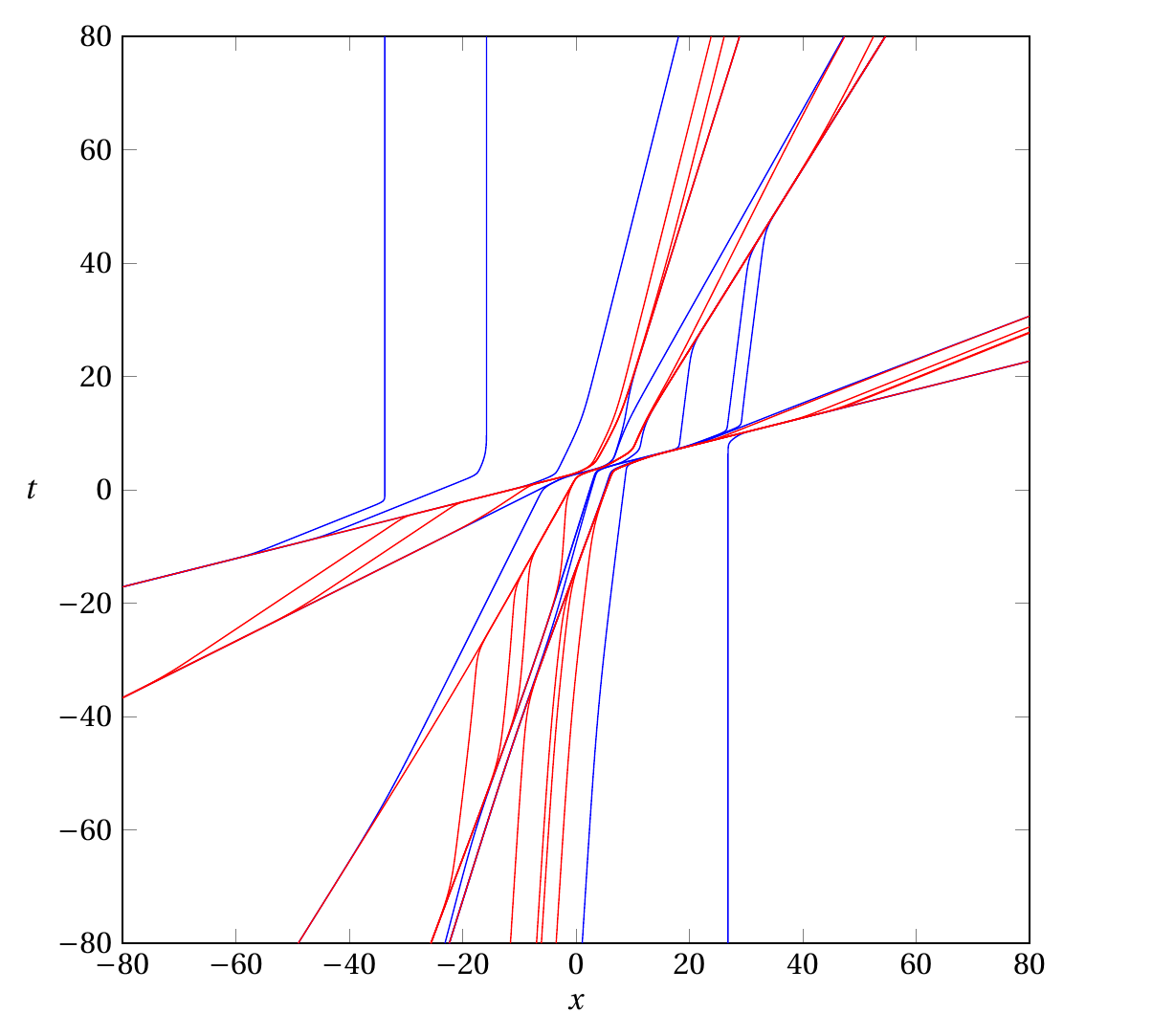}
  \caption{\textbf{Positions in a case with $4+3$ non-singleton groups.}
    Positions of the peakons, $x=x_{k,i}(t)$ and $x=y_{k,i}(t)$,
    for the solution described in Example~\ref{ex:GX-4+3-allgroups}.
    There are $4+3$ groups, all non-singletons,
    with $3+4+2+5+5+5+2$ peakons in total.
    This odd case is similar to the even case (Example~\ref{ex:GX-3+3-allgroups})
    in the way that one peakon in each group
    asymptotically approaches the corresponding singleton curve
    (Figure~\ref{fig:GX-4+3-interlacing-positions-all}, Example~\ref{ex:GX-4+3-interlacing}),
    while the other peakons in the group instead approach a neighbouring singleton curve,
    with the same exceptions near the edges as in the even case.
    But since these asymptotic lines are inherited from the odd interlacing case,
    the odd case with groups differs from the even case in the same way that
    the odd interlacing case differs from the even interlacing case
    (cf. Example~\ref{ex:GX-4+3-interlacing}).
  }
  \label{fig:GX-4+3-allgroups-positions}
\end{figure}

\begin{figure}[H]
  \centering
  \includegraphics[width=13cm]{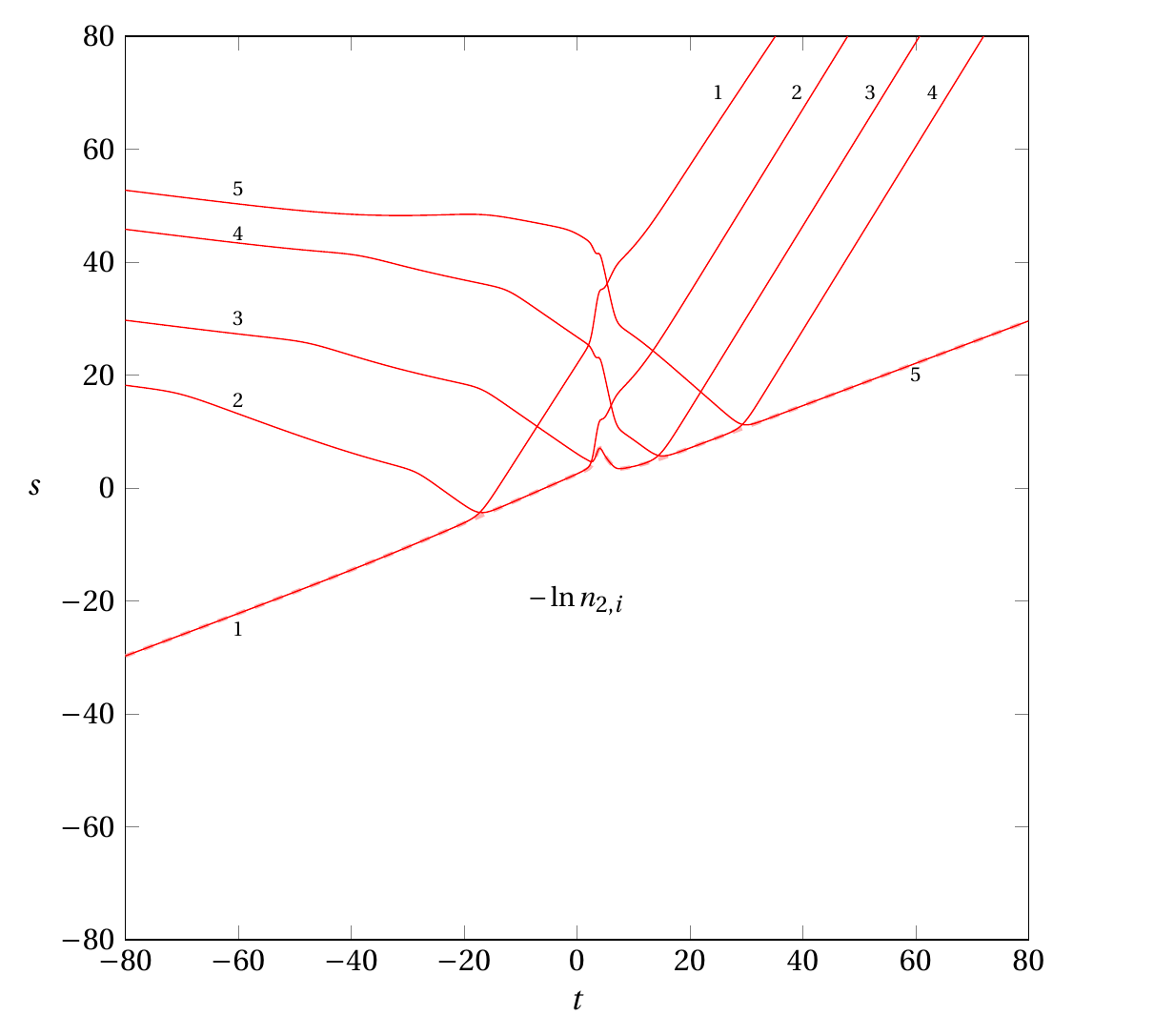}
  \caption{\textbf{Amplitudes for the middle group ($Y_2$).}
    Amplitudes of the five peakons in the $Y_2$-group for the solution with $4+3$ groups
    described in Example~\ref{ex:GX-4+3-allgroups},
    whose positions where shown in Figure~\ref{fig:GX-4+3-allgroups-positions}.
    The solid red curves are $s = -\ln n_{2,i}(t)$,
    and the dashed curve in the background is the corresponding singleton curve
    $s = -\ln n_2(t)$
    from Figure~\ref{fig:GX-4+3-interlacing-amplitudes-all}.
    As $t \to \pm \infty$, one peakon in the group follows the singleton curve asymptotically,
    while the others approach parallel lines with other slopes
    given by
    \eqref{eq:asymptotic-group-slopes-odd-example-neginf}
    and~\eqref{eq:asymptotic-group-slopes-odd-example-posinf},
    namely $-\frac{1}{8}$ as $t \to -\infty$
    and $\frac{23}{16}$ as $t \to +\infty$.
  }
  \label{fig:GX-4+3-allgroups-amplitudes-Y2}
\end{figure}

\begin{figure}[H]
  \centering
  \includegraphics[width=13cm]{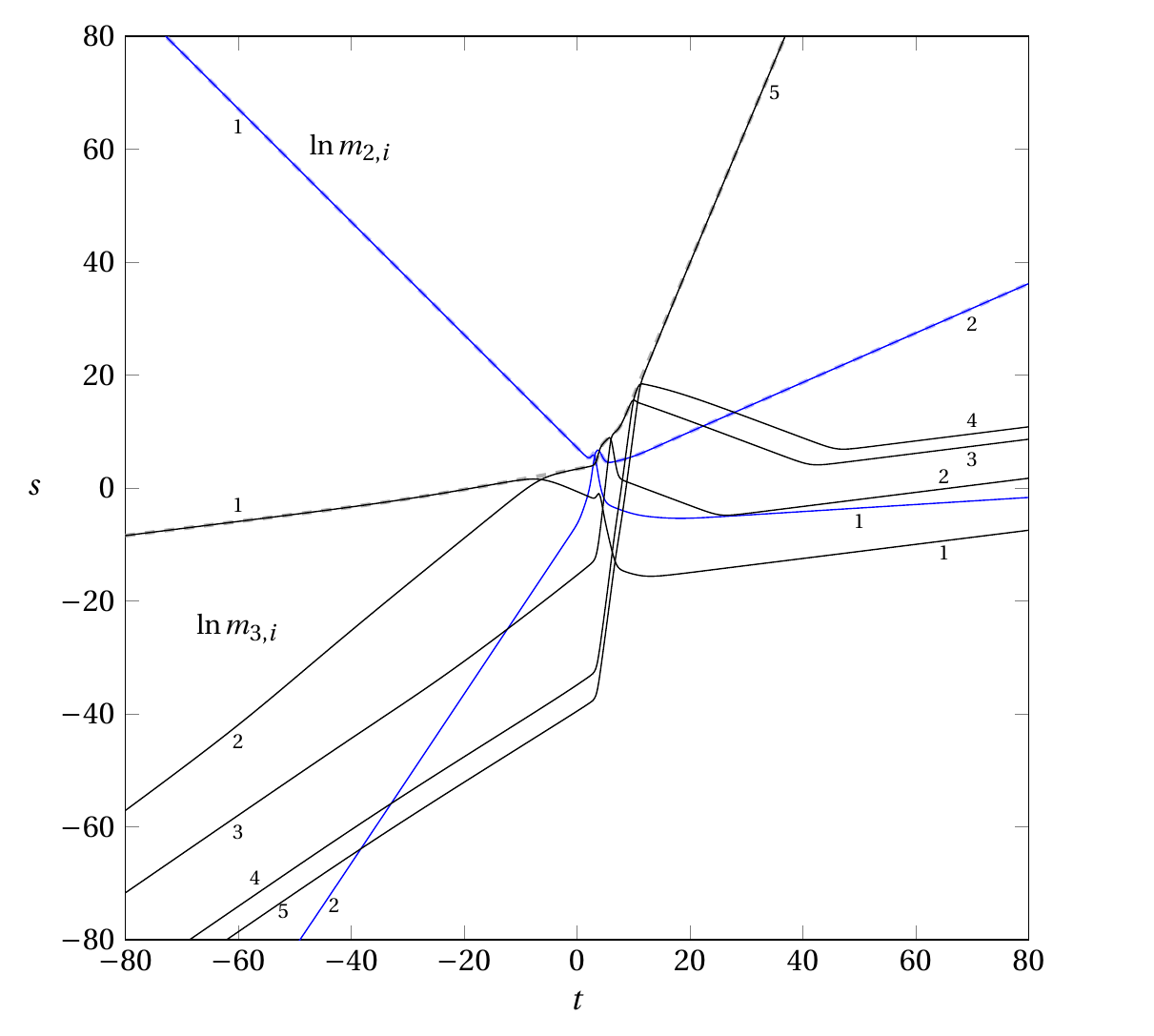}
  \caption{\textbf{Amplitudes for the middle $X$-groups ($X_2$ and~$X_3$).}
    The same as Figure~\ref{fig:GX-4+3-allgroups-amplitudes-Y2}
    but for the typical $X$-groups in the middle: $s = \ln m_{2,i}(t)$ (blue)
    and $s = \ln m_{3,i}(t)$ (black).
    The dashed lines are the corresponding singleton curve
    $s = \ln m_2(t)$ and $s = \ln m_3(t)$
    from Figure~\ref{fig:GX-4+3-interlacing-amplitudes-all}.
    The asymptotic slopes of the curves which do not approach the singleton curves are given by
    \eqref{eq:asymptotic-group-slopes-odd-example-neginf}
    and~\eqref{eq:asymptotic-group-slopes-odd-example-posinf},
    namely $\frac{11}{8}$ (blue) and $\frac{11}{16}$ (black) as $t \to -\infty$,
    and $\frac{1}{16}$ (blue) and $\frac{1}{8}$ (black) as $t \to +\infty$.
  }
  \label{fig:GX-4+3-allgroups-amplitudes-X2-X3}
\end{figure}

\begin{figure}[H]
  \centering
  \includegraphics[width=13cm]{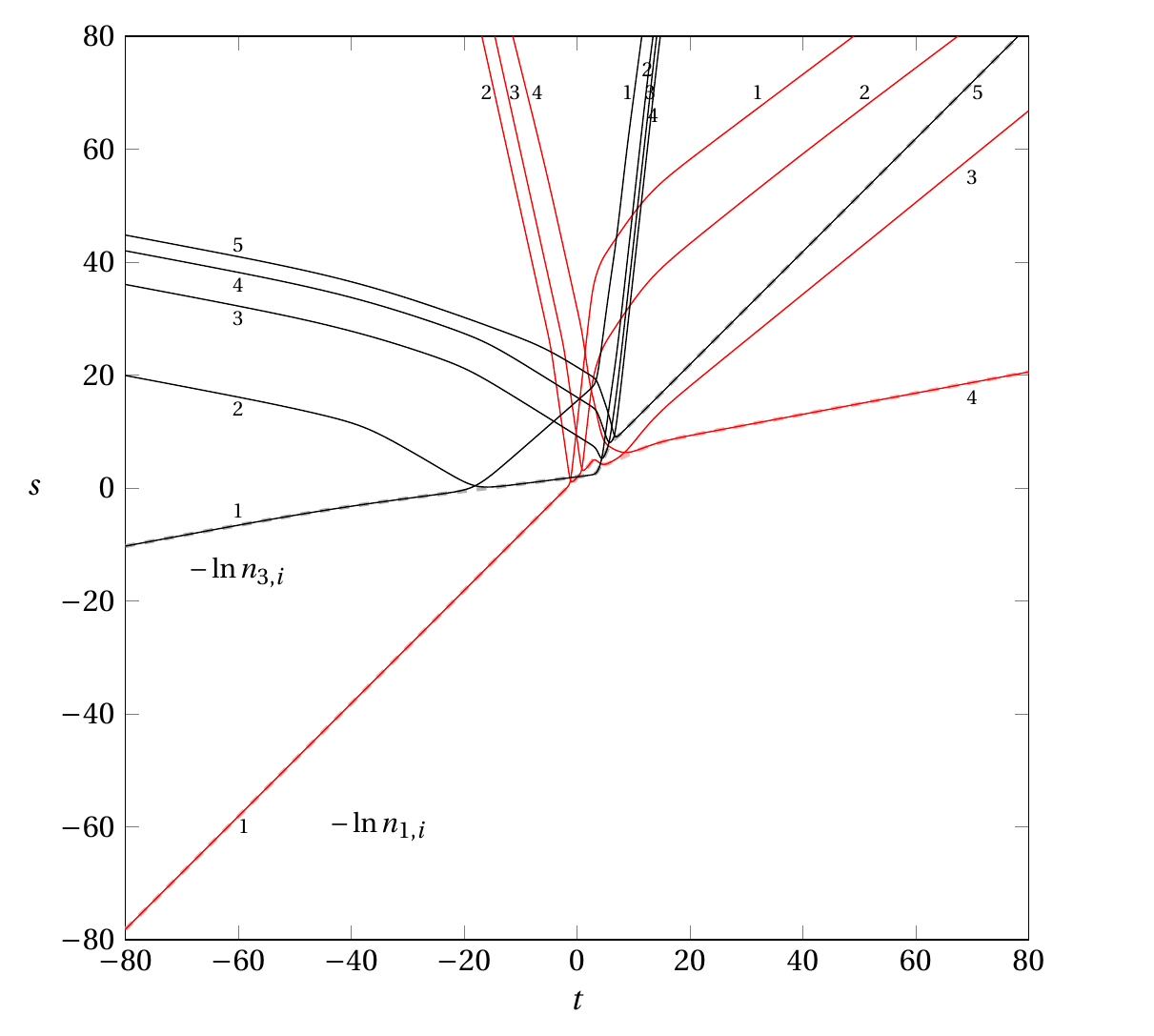}
  \caption{\textbf{Amplitudes for the second leftmost and second rightmost groups ($Y_1$ and~$Y_3$).}
    The same as Figures~\ref{fig:GX-4+3-allgroups-amplitudes-Y2}
    and~\ref{fig:GX-4+3-allgroups-amplitudes-X2-X3},
    but for the second leftmost group ($s = -\ln n_{1,i}(t)$, red)
    and the second rightmost group ($s = -\ln n_{3,i}(t)$, black).
    The asymptotic slopes of the curves which do not approach the (dashed) singleton curves
    are $-4$ (red) and $-\frac{3}{16}$ (black) as $t \to -\infty$,
    and $\frac{3}{4}$ (red) and $\frac{59}{8}$ (black) as $t \to +\infty$.
  }
  \label{fig:GX-4+3-allgroups-amplitudes-Y1-Y3}
\end{figure}

\begin{figure}[H]
  \centering
  \includegraphics[width=13cm]{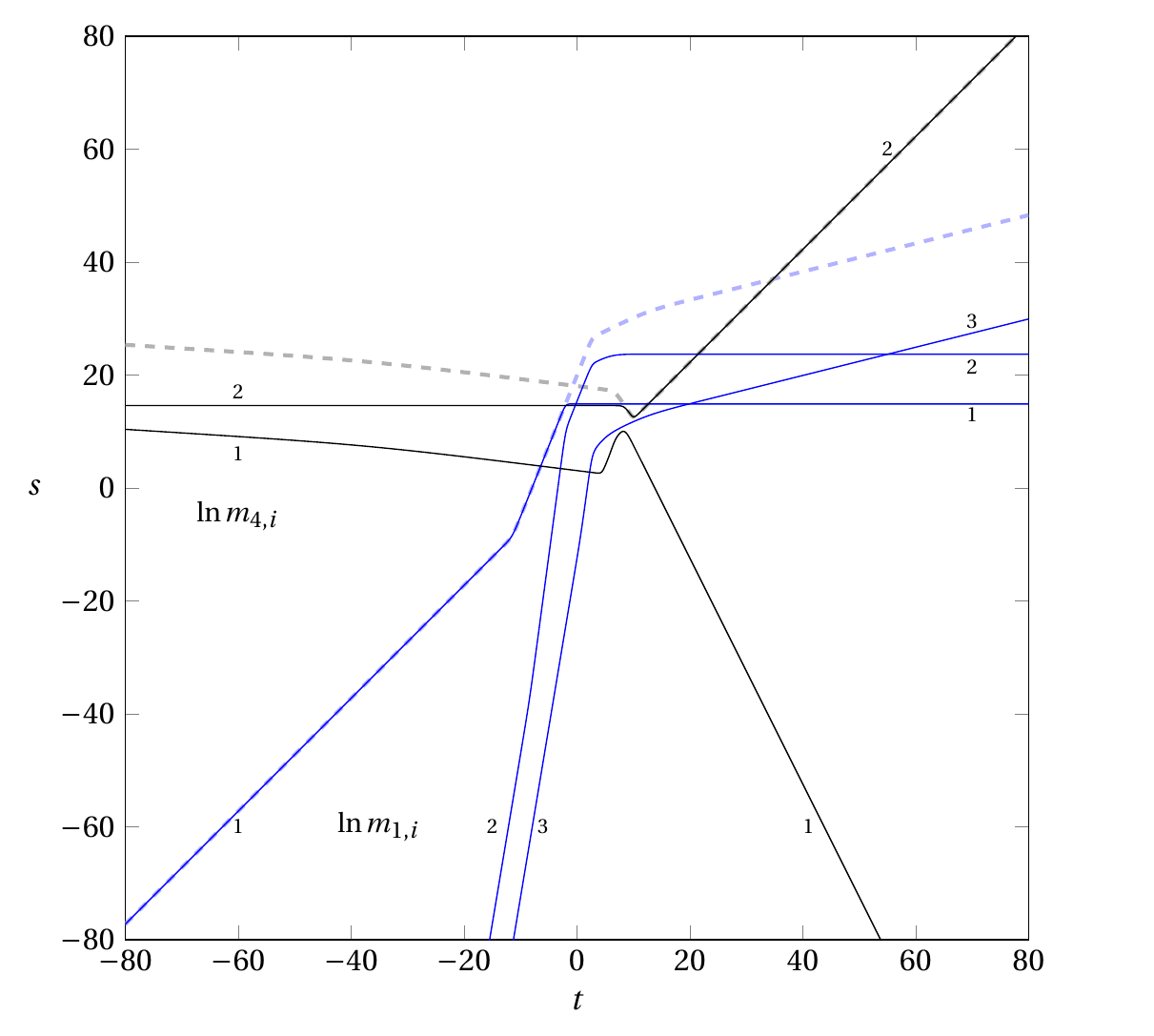}
  \caption{\textbf{Amplitudes for the leftmost and rightmost groups ($X_1$ and~$X_4$).}
    The same as Figures~\ref{fig:GX-4+3-allgroups-amplitudes-Y2},
    \ref{fig:GX-4+3-allgroups-amplitudes-X2-X3}
    and~\ref{fig:GX-4+3-allgroups-amplitudes-Y1-Y3},
    but for the leftmost group ($s = \ln m_{1,i}(t)$, blue)
    and the rightmost group ($s = \ln m_{4,i}(t)$, black).
    As for the outermost groups in the even case,
    one peakon in each group approaches the (dashed) singleton curve
    in one time direction, but only a line \emph{parallel} to the asymptote
    of the singleton curve in the other time direction.
    The asymptotic slopes of the remaining curves
    are $6$ (blue) and $0$ (black) as $t \to -\infty$,
    and $0$ (blue) and $-2$ (black) as $t \to +\infty$.
  }
  \label{fig:GX-4+3-allgroups-amplitudes-X1-X4}
\end{figure}

\section{Solution formulas for the even case  ($2K$ groups)}
\label{sec:solutions-even}

After all these examples, it is now time to turn to the general results.
In this section, we will state all the formulas for
the peakon solutions of the Geng--Xue equation~\eqref{eq:GX} in the even case,
where we have $K$ groups of each kind,
and we assume without loss of generality that the first group is an $X$-group (so that the last group is a $Y$-group).
The proofs will be given in Section~\ref{sec:proofs-even},
and the asymptotics of the solutions as $t \to \pm\infty$
will be studied in Section~\ref{sec:asymptotics-even}.
After that, we will do the same things for the odd case
in Sections~\ref{sec:solutions-odd}, \ref{sec:proofs-odd} and~\ref{sec:asymptotics-odd}, respectively.

To write down the solution formulas for an arbitrary even peakon
configuration, just use the $K+K$ interlacing formulas in
Sections~\ref{sec:solutions-even-X-singleton} and~\ref{sec:solutions-even-Y-singleton} for any
group consisting of a single peakon.
For $X$-groups with more than one peakon,
the formulas are in Section~\ref{sec:solutions-even-X-typical-group},
except that the leftmost group is special, with formulas given in Section~\ref{sec:solutions-even-X-leftmost-group}.
Similarly, for $Y$-groups with more than one peakon, the typical group
is described in Section~\ref{sec:solutions-even-Y-typical-group}, and the exceptional
rightmost group in Section~\ref{sec:solutions-even-Y-rightmost-group}.

We remind the reader about the notation defined in
Section~\ref{sec:more-notation}. For simplicity, when considering a
particular group, for example $X$-group number~$j$,
we will write just $\tau_i$ and $\sigma_i$ instead of $\tau_{j,i}^X$ and $\sigma_{j,i}^X$,
and similarly for the corresponding sums $T_i$, $S_i$ and~$R_i$.
We will also omit the superscript $X$ and $Y$ on $N_j^X$ and $N_j^Y$
when it is clear from the context whether we are talking about an $X$-group or a $Y$-group.
Also remember that all the parameters satisfy the constraints in Section~\ref{sec:more-notation};
we will not state this explicitly in connection with the solution formulas,
but it is understood throughout the paper.

In Sections \ref{sec:solutions-even}, \ref{sec:proofs-even}
and~\ref{sec:asymptotics-even} which deal with the even case,
all the determinants $\detJ_{ij}^{rs}$ defined by~\eqref{eq:heine-integral-as-sum}
will have $A=K$ and $B=K-1$, i.e.,
\begin{equation}
  \detJ_{ij}^{rs} = \detJ[K,K-1,r,s,i,j]
  ,
\end{equation}
and we will make heavy use of the abbreviation
\begin{equation}
  j'= K+1-j
  .
\end{equation}
We will also use the abbreviations
$X_{j,i} = \tfrac12 \exp 2 x_{j,i}$,
etc., from Definition~\ref{def:XYQP-groups}.

\begin{remark}
  Whenever the conditions for some formulas cannot be satisfied,
  such as the inequality $2 \le j \le K$ in case $K=1$,
  or $2 \le i \le N-1$ in case $N=2$,
  then those formulas are simply disregarded as not relevant in that situation.
\end{remark}

\subsection{Solutions for $X$-groups}
\label{sec:solutions-even-X}

The solution formulas for peakons in an $X$-group depend on whether
the group is a singleton with just one peakon, or a true group
consisting of $N_j^X \ge 2$ peakons. Moreover, the solution formulas
for the leftmost $X$-group are different from those for the other
$X$-groups. Also note that the rightmost peakon $X_{j,N_j^X}$ in
each group is given by a separate formula; see for
example~\eqref{eq:even-X-typical-group-pos}.

\subsubsection{$X$-singletons}
\label{sec:solutions-even-X-singleton}

For those groups which consist of a single peakon, the solution
formulas coincide with the formulas derived for the interlacing
peakons by Lundmark and
Szmigielski~\cite{lundmark-szmigielski:2017:GX-dynamics-interlacing},
already given in Theorem~\ref{thm:interlacing-solution} but repeated
here for convenience:
\begin{equation}
  \label{eq:even-X-typical-singleton}
  X_{j'} = \frac{\detJ_{jj}^{00}}{\detJ_{j-1,j-1}^{11}}
  ,\qquad
  Q_{j'} = \frac{\detJ_{j-1,j-1}^{11} \detJ_{j,j-1}^{01}}{\detJ_{jj}^{10} \detJ_{j-1,j-1}^{10}}
  ,
\end{equation}
where $j'=K+1-j$ and $1 \le j \le K-1$,
together with the leftmost $X$-peakon,
\begin{equation}
  \label{eq:even-X-leftmost-singleton}
  X_1 = \frac{\detJ_{K,K-1}^{00}}{\detJ_{K-1,K-2}^{11} + C \, \detJ_{K-1,K-1}^{10}}
  ,\qquad
  Q_1 =
  \frac{M}{L}
  \left( \frac{\detJ_{K-1,K-2}^{11}}{\detJ_{K-1,K-1}^{10}} + C \right)
  ,
\end{equation}
where $L = \prod_{i=1}^K \lambda_i$ and $M = \prod_{j=1}^{K-1} \mu_j$.

\subsubsection{All $X$-groups except the leftmost one (typical $X$-groups)}
\label{sec:solutions-even-X-typical-group}

Next, we give the formulas for $X$-group number $j'=K+1-j$, i.e., the
$j$th $X$-group from the right, in case it contains $N_{j'}=N_{j'}^X
\ge 2$ peakons. Here $1 \le j \le K-1$; the leftmost group $j=K$ is
treated separately below. The solutions for the positions are
\begin{equation}
  \label{eq:even-X-typical-group-pos}
  \begin{aligned}
    X_{j',i} &
    = \frac{\detJ_{j+1,j}^{00} + T_i \detJ_{jj}^{00} + S_i \detJ_{j,j-1}^{00}}{\detJ_{j,j-1}^{11} + T_i \detJ_{j-1,j-1}^{11} + S_i \detJ_{j-1,j-2}^{11}}
    ,\\[1ex]
    X_{j',N_{j'}} &
    = \frac{\detJ_{jj}^{00} + \sigma_{N_{j'}-1} \detJ_{j,j-1}^{00}}{\detJ_{j-1,j-1}^{11} + \sigma_{N_{j'}-1} \detJ_{j-1,j-2}^{11}}
    ,
  \end{aligned}
\end{equation}
and for the amplitudes
\begin{equation}
  \label{eq:even-X-typical-group-amp}
  \begin{split}
    Q_{j',i} &
    = (\sigma_i - \sigma_{i-1}) \, 
    \detJ_{j,j-1}^{01}
    \left( \detJ_{j,j-1}^{11} + T_i \detJ_{j-1,j-1}^{11} + S_i \detJ_{j-1,j-2}^{11} \right)
    \\
    &\quad
    \times
    \left( \detJ_{jj}^{10} + \sigma_i \detJ_{j,j-1}^{10} + R_{i} \detJ_{j-1,j-1}^{10} \right)^{-1}
    \\
    &\quad
    \times
    \left( \detJ_{jj}^{10} + \sigma_{i-1} \detJ_{j,j-1}^{10} + R_{i-1} \detJ_{j-1,j-1}^{10} \right)^{-1}
    ,
    \\[1ex]
    Q_{j',N_{j'}} &
    = \frac{\detJ_{j,j-1}^{01} \left( \detJ_{j-1,j-1}^{11} + \sigma_{N_{j'}-1} \detJ_{j-1,j-2}^{11} \right)}{\detJ_{j-1,j-1}^{10} \left( \detJ_{jj}^{10} + \sigma_{N_{j'}-1} \detJ_{j,j-1}^{10} + R_{N_{j'}-1} \detJ_{j-1,j-1}^{10} \right)}
    ,
  \end{split}
\end{equation}
where $1 \le i \le N_{j'}-1$, and the sums $T_i$, $S_i$, $R_i$ are as in Definition~\ref{def:T-S-R}.

\subsubsection{The leftmost $X$-group}
\label{sec:solutions-even-X-leftmost-group}

If the leftmost $X$-group contains more than one peakon,
then the positions are given by
\begin{equation}
  \label{eq:even-X-leftmost-group-pos}
  \begin{split}
    X_{1,1} &
    = \frac{\detJ_{K,K-1}^{00}}{\detJ_{K-1,K-2}^{11} + \frac{1}{\sigma_1} \detJ_{K-1,K-1}^{11} + C \left( \detJ_{K-1,K-1}^{10} + \frac{1}{\tau_1} \detJ_{K,K-1}^{10} \right)}
    ,
    \\[1ex]
    X_{1,i} &
    = \frac{S_i \detJ_{K,K-1}^{00}}{\detJ_{K,K-1}^{11} + T_i \detJ_{K-1,K-1}^{11} + S_i \detJ_{K-1,K-2}^{11}}
    ,
    \\[1ex]
    X_{1,N_1} &
    = \frac{\sigma_{N_1-1} \detJ_{K,K-1}^{00}}{\detJ_{K-1,K-1}^{11} + \sigma_{N_1-1} \detJ_{K-1,K-2}^{11}}
    ,
  \end{split}
\end{equation}
and the amplitudes by
\begin{equation}
  \label{eq:even-X-leftmost-group-amp}
  \begin{split}
    Q_{1,1} &
    =
    \frac{M}{L}
    \left( \frac{\detJ_{K-1,K-2}^{11} + \frac{1}{\sigma_1} \detJ_{K-1,K-1}^{11}}{\detJ_{K-1,K-1}^{10} + \frac{1}{\tau_1} \detJ_{K,K-1}^{10}} + C \right)
    ,
    \\[1ex]
    Q_{1,i} &
    = (\sigma_i - \sigma_{i-1}) \, \detJ_{K,K-1}^{01} \left( \detJ_{K,K-1}^{11} + T_i \detJ_{K-1,K-1}^{11} + S_i \detJ_{K-1,K-2}^{11} \right)
    \\ & \quad
    \times \left( \sigma_i \detJ_{K,K-1}^{10} + R_{i} \detJ_{K-1,K-1}^{10} \right)^{-1}
    \\ & \quad
    \times \left( \sigma_{i-1} \detJ_{K,K-1}^{10} + R_{i-1} \detJ_{K-1,K-1}^{10} \right)^{-1}
    ,
    \\[1ex]
    Q_{1,N_1} &
    = \frac{\detJ_{K,K-1}^{01} \left( \detJ_{K-1,K-1}^{11} + \sigma_{N_1-1} \detJ_{K-1,K-2}^{11} \right)}{\detJ_{K-1,K-1}^{10} \left( \sigma_{N_1-1} \detJ_{K,K-1}^{10} + R_{N_1-1} \detJ_{K-1,K-1}^{10} \right)}
    ,
  \end{split}
\end{equation}
where $2 \le i \le N_1-1$.
Note that in this group, both the rightmost and leftmost peakons are given by separate formulas.

\subsection{Solutions  for $Y$-groups}
\label{sec:solutions-even-Y}

Like for $X$-groups, the formulas for a $Y$-group depend on whether
the group contains a single peakon or more than one peakon. Moreover,
the rightmost group has special formulas.

\subsubsection{$Y$-singletons}
\label{sec:solutions-even-Y-singleton}

In case the group is a singleton, the solution formulas
coincide with the formulas for interlacing peakons.
With $j'=K+1-j$ as usual, the formulas	for $2 \le j \le K$ are
\begin{equation}
  \label{eq:even-Y-typical-singleton}
  Y_{j'} =
  \frac{\detJ_{j,j-1}^{00}}{\detJ_{j-1,j-2}^{11}}
  ,\qquad
  P_{j'} =
  \frac{\detJ_{j-1,j-1}^{10} \detJ_{j-1,j-2}^{11}}{\detJ_{j,j-1}^{01} \detJ_{j-1,j-2}^{01}}
  ,
\end{equation}
while the formulas for the rightmost peakon are
\begin{equation}
  \label{eq:even-Y-rightmost-singleton}
  Y_{K} = \detJ_{11}^{00} + D \detJ_{10}^{00}
  ,\qquad
  P_{K} = \frac{1}{\detJ_{10}^{00}}
  .
\end{equation}

\subsubsection{All $Y$-groups except the rightmost one (typical $Y$-groups)}
\label{sec:solutions-even-Y-typical-group}

If a $Y$-group contains $N_{j'} = N_{j'}^Y \ge 2$ peakons, where
$2 \le j \le K$, then the formulas for the positions are
\begin{equation}
  \label{eq:even-Y-typical-group-pos}
  \begin{aligned}
    Y_{j',i} &= \frac{\detJ_{jj}^{00} + T_{i} \detJ_{j,j-1}^{00} + S_i \detJ_{j-1,j-1}^{00}}{\detJ_{j-1,j-1}^{11} + T_{i} \detJ_{j-1,j-2}^{11} + S_i \detJ_{j-2,j-2}^{11}}
    ,
    \\[1ex]
    Y_{j',N_{j'}} &=  \frac{\detJ_{j,j-1}^{00} + \sigma_{N_{j'}-1} \detJ_{j-1,j-1}^{00}}{\detJ_{j-1,j-2}^{11} + \sigma_{N_{j'}-1} \detJ_{j-2,j-2}^{11}}
    ,
  \end{aligned}
\end{equation}
and the formulas for the amplitudes are
\begin{equation}
  \label{eq:even-Y-typical-group-amp}
  \begin{split}
    P_{j',i} &
    = (\sigma_i - \sigma_{i-1}) \, \detJ_{j-1,j-1}^{10}
    \\ & \quad
    \times \left( \detJ_{j-1,j-1}^{11} + T_i \detJ_{j-1,j-2}^{11} + S_i \detJ_{j-2,j-2}^{11} \right)
    \\ & \quad
    \times \left( \detJ_{j,j-1}^{01} + \sigma_i \detJ_{j-1,j-1}^{01} + R_{i} \detJ_{j-1,j-2}^{01} \right)^{-1}
    \\ & \quad
    \times \left( \detJ_{j,j-1}^{01} + \sigma_{i-1} \detJ_{j-1,j-1}^{01} + R_{i-1} \detJ_{j-1,j-2}^{01} \right)^{-1},
    \\[1ex]
    P_{j',N_{j'}} &
    =\frac{\detJ_{j-1,j-1}^{10} \left( \detJ_{j-1,j-2}^{11} + \sigma_{N_{j'}-1} \detJ_{j-2,j-2}^{11} \right)}{\detJ_{j-1,j-2}^{01} \left( \detJ_{j,j-1}^{01} + \sigma_{N_{j'}-1} \detJ_{j-1,j-1}^{01} + R_{N_{j'}-1} \detJ_{j-1,j-2}^{01} \right)}
    ,
  \end{split}
\end{equation}
for $1 \le i \le N_{j'}-1$.

\subsubsection{The rightmost $Y$-group}
\label{sec:solutions-even-Y-rightmost-group}

If the rightmost $Y$-group contains $N_K \ge 2$ peakons, the formulas
for the positions are
\begin{equation}
  \label{eq:even-Y-rightmost-group-pos}
  \begin{aligned}
    Y_{K,i} &
    = \detJ_{11}^{00} +  T_{i} \detJ_{10}^{00} + S_i
    ,
    \\[1ex]
    Y_{K,N_K} &
    = \detJ_{11}^{00} + \left( T_{N_K-1} + D \right) \detJ_{10}^{00}
    + \left( S_{N_K-1} + D \, \sigma_{N_K-1} \right)
    ,
  \end{aligned}
\end{equation}
while the amplitudes are given by
\begin{equation}
  \label{eq:even-Y-rightmost-group-amp}
  \begin{aligned}
    P_{K,i} &
    = \frac{\sigma_i - \sigma_{i-1}}{\left( \detJ_{10}^{01} + \sigma_i \right) \left( \detJ_{10}^{01} + \sigma_{i-1} \right)}
    ,
    \\[1ex]
    P_{K,N_K} &
    = \frac{1}{\detJ_{10}^{01} + \sigma_{N_K-1}}
    ,
  \end{aligned}
\end{equation}
for $1 \le i \le N_K-1$.

\section{Proofs  for the even  case}
\label{sec:proofs-even}

In this section we will prove the solution formulas stated in Section~\ref{sec:solutions-even}.
First we will explain the method of the proof, which was already
illustrated in Examples~\ref{ex:proof-technique}
and~\ref{ex:proof-technique2} for some special cases.

\subsection{Method of the proof: ``killing'' peakons}
\label{sec:killing}

To obtain the solution formulas for a specific peakon configuration
with $K_1+K_1$ groups, our starting point is the $K_2+K_2$ interlacing
peakon configuration obtained by adding extra $Y$-peakons between
adjacent $X$-peakons, and vice versa.
(In both these configurations, we are as usual assuming that we begin
with an $X$-group on the left and finish with a $Y$-group on the right.)
Thus, we are inserting $N-1$ new peakons into each group of $N$
peakons, so that the total number of peakons in the interlacing
$K_2+K_2$ configuration will be
\begin{equation*}
  K_2+K_2 = \sum_{j=1}^{K_1} (2N_j^X-1) + \sum_{j=1}^{K_1} (2N_j^Y-1)
  .
\end{equation*}
We know the solution formulas \eqref{eq:interlacing-solution-positions}
and~\eqref{eq:interlacing-solution-amplitudes}
for this interlacing configuration, with the determinants
$\detJ_{ij}^{rs}= \detJ[K_2,K_2-1,r,s,i,j]$.
Next, we successively kill off the extra peakons that we inserted,
one by one, from the right, until we reach the desired configuration again.
By \emph{killing} a peakon, we mean that we turn it into a ``ghostpeakon''
whose amplitude $m_{j,i}(t)$ or $n_{j,i}(t)$ is identically zero,
by making a substitution of the form~\eqref{eq:substitution-kill-peakon-even} below,
and letting $\epsilon \to 0$.
Such a peakon does not influence the dynamics of the other peakons, so
by simply disregarding all ghostpeakons at the end of the process,
whatever remains of the solution formulas will be the solution that we
seek.

In each step during the procedure, if we forget about the peakons we
have already killed, we have groups (with one or more peakons each) to
the right of the peakon we are about to kill, but only interlacing
peakons (singletons) on its left.

As explained in Example~\ref{ex:proof-technique2}, it suffices to show
that our general solution formulas ``survive'' in each such step,
i.e., if the solution before killing a peakon is given by the solution
formulas for $K+K$ groups, with $N$ peakons in the group just to the
right of the peakon being killed, then the solution after the killing
is given by the solution formulas for $(K-1)+(K-1)$ groups, with $N+1$
peakons in the newly formed group.

\begin{remark}
  \label{rem:hats}
  We will use \textbf{symbols with hats} for quantities related to
  the original configuration with $K+K$ groups before
  killing the peakon,
  while \textbf{symbols without hats} refer to the configuration
  with $(K-1)+(K-1)$ groups after the killing.
  For example, $\hat{y}_{K,i}(t)$ is the position of the $i$th peakon
  in the rightmost group in the original configuration.
  After killing another peakon to the left of it,
  this group will be $Y$-group number $K-1$,
  so the new notation becomes
  \begin{equation*}
    y_{K-1,i}(t) = \lim_{\epsilon \to 0} \hat{y}_{K,i}(t)
    .
  \end{equation*}
  Some care is needed if using the notation~$j'$ for ``the $j$th object from the right'',
  since this will be number $K+1-j$ from the left before killing, but number $(K-1)+1-j$ afterwards.
  Thus, for example,
  \begin{equation}
    \label{eq:j-prime-warning}
    \hat{X}_{j'} = \hat{X}_{K+1-j}
    \qquad\text{but}\qquad
    X_{j'} = X_{K-j}
    .
  \end{equation}
\end{remark}

\begin{remark}
  The inspiration for this method comes from our previous
  paper~\cite{lundmark-shuaib:2018p:ghostpeakons}, where we used
  ghostpeakons to compute the characteristic curves for peakon
  solutions of the Camassa--Holm, Degasperis--Procesi and Novikov
  equations.
\end{remark}

\subsection{Preparations}

In the proof, we will need to know what becomes of the determinants
\begin{equation*}
  \hat{\detJ}_{ij}^{rs} = \detJ[K,K-1,r,s,i,j]
\end{equation*}
(cf. Definition~\ref{def:heineintegral} and Remark~\ref{rem:hats})
when we make a substitution of the form
\begin{equation}
  \label{eq:substitution-kill-peakon-even}
  \begin{aligned}
    \lambda_K &= \frac{\alpha}{\epsilon}
    ,
    &
    \mu_{K-1} &= \frac{\beta}{\epsilon}
    ,
    \\
    a_K &= Z \, \epsilon^{k_1}
    ,
    &
    b_{K-1} &= W \, \epsilon^{k_2}
    ,
  \end{aligned}
\end{equation}
and in particular what happens as we let $\epsilon \to 0$ afterwards.
This is the content of Lemma~\ref{lem:heineintegral-epsilon} below,
which is our primary technical tool.

\begin{remark}
  When we apply Lemma~\ref{lem:heineintegral-epsilon} later,
  the numbers $k_1$ and~$k_2$
  in the substitution~\eqref{eq:substitution-kill-peakon-even}
  will just be some integers,
  while $\alpha$, $\beta$, $Z$ and~$W$
  will be expressions depending on three new parameters $\tau_1$, $\sigma_1$ and~$\theta$,
  which together with the parameter~$\epsilon$
  are replacing the four original parameters
  $\lambda_K$, $\mu_{K-1}$, $a_K$ and~$b_{K-1}$.
\end{remark}

\begin{remark}
  As we already explained in Example~\ref{ex:proof-technique},
  the substitutions are actually
  $a_K(0) = Z \, \epsilon^{k_1}$
  and
  $b_{K-1}(0) = W \, \epsilon^{k_2}$,
  but for simplicity of notation we will
  write just $a_K$ and~$b_{K-1}$.
  The quantities $a_K(t)$ and $a_K(0)$
  only differ by a factor $e^{t \epsilon/\alpha}$
  which tends to~$1$ as $\epsilon \to 0$,
  and similarly for $b_{K-1}(t)$ and~$b_{K-1}(0)$.
\end{remark}

\begin{lemma}
  \label{lem:heineintegral-epsilon}
  Let
  \begin{equation}
    \hat{\detJ}_{ij}^{rs}=\detJ[K,K-1,r,s,i,j]
    \qquad\text{and}\qquad
    \detJ_{ij}^{rs}=\detJ[K-1,K-2,r,s,i,j]
    .
  \end{equation}
  If we reparametrize the spectral data according to~\eqref{eq:substitution-kill-peakon-even},
  then, as $\epsilon \to 0$, the determinant $\hat{\detJ}_{ij}^{rs}$ satisfies
  \begin{equation}
    \begin{split}
      \label{eq:heineintegral-epsilon}
      \hat{\detJ}_{ij}^{rs}
      &
      = \detJ_{ij}^{ rs}
      \\ &
      + \detJ_{i-1,j}^{ rs} \, \bigl( Z \alpha^{q_1} \bigr) \, \epsilon^{k_1-q_1} \, \bigl( 1 + \mathcal{O}(\epsilon) \bigr)
      \\ &
      + \detJ_{i,j-1}^{ rs} \, \bigl( W \beta^{q_2} \bigr) \, \epsilon^{k_2-q_2} \, \bigl( 1 + \mathcal{O}(\epsilon) \bigr)
      \\ &
      + \detJ_{i-1,j-1}^{rs} \, \bigl( Z \alpha^{q_1} \bigr) \, \bigl( W \beta^{q_2} \bigr) \,
      \left( \frac{\alpha \, \beta}{\alpha +\beta} \right) \, \epsilon^{(k_1-q_1) + (k_2- q_2) - 1} \,
      \bigl( 1 + \mathcal{O}(\epsilon) \bigr)
      ,
    \end{split}
  \end{equation}
  where
  \begin{equation}
    q_1 = 2(i-1) - j + r
    ,\qquad
    q_2 = 2(j-1) - i + s
    .
  \end{equation}
\end{lemma}

\begin{proof}
  In the definition of $\hat{\detJ}_{ij}^{rs}$ from
  equation~\eqref{eq:heine-integral-as-sum}, we split the sums into four cases,
  according to whether or not $K \in I$ or $K-1 \in J$:
  \begin{equation}
    \begin{split}
      \hat{\detJ}_{ij}^{rs} &=
      \sum_{I \in \binom{[K]}{i}} \sum_{J \in \binom{[K-1]}{j}}
      \Psi_{IJ} \, \lambda_I^r a_I \, \mu_J^s \, b_J
      \\
      &= \sum_{I' \in \binom{[K-1]}{i}} \sum_{J' \in \binom{[K-2]}{j}}
      \Psi_{I' J'} \, \lambda_{I'}^r \, a_{I'} \, \mu_{J'}^s \, b_{J'}
      \\ & \quad
      + \sum_{{I'} \in \binom{[K-1]}{i-1}} \sum_{J' \in \binom{[K-2]}{j}}
      \Psi_{I' J'} \, \lambda_{I'}^r a_{I'} \, \mu_{J'}^s b_{J'}
      \left( \lambda_K^r \, a_K \, \right)
      \left( \frac{\lambda_K^{2(i-1)}}{\lambda_K^j} \right)
      \left( \frac{\prod^{i-1}_{a=1} \bigl( \lambda_K^{-1} \lambda _a - 1 \bigr)^2}
        {\prod^{j}_{b=1} \bigl( \frac{\mu_b}{\lambda_K} + 1 \bigr)} \right)
      \\ & \quad
      + \sum_{I' \in \binom{[K-1]}{i}} \sum_{J' \in \binom{[K-2]}{j-1}}
      \Psi_{I' J'} \, \lambda_{I'}^r a_{I'} \, \mu_{J'}^s b_{J'} \,
      \left( \mu_{K-1}^s \, b_{K-1} \right)
      \left( \frac{\mu_{K-1}^{2(j-1)}}{\mu_{K-1}^i} \right)
      \left( \frac{\prod^{j-1}_{b=1} \bigl( \mu_{K-1}^{-1} \mu_b - 1 \bigr)^2}
        {\prod^{i}_{a=1} \bigl( \frac{\lambda_a}{\mu_{K-1}} + 1 \bigr)} \right)
      \\ & \quad
      + \sum_{I' \in \binom{[K-1]}{i-1}} \sum_{J' \in \binom{[K-2]}{j-1}}
      \Psi_{I' J'} \, \lambda_{I'}^r a_{I'} \, \mu_{J'}^s b_{J'} \,
      \left( \lambda_K^r \, a_K \right) \left( \mu_{K-1}^s \,  b_{K-1} \right)
      \left( \frac{\lambda_K^{2(i-1)}}{\lambda_K^{j-1}} \right)
      \left( \frac{\mu_{K-1}^{2(j-1)}}{\mu_{K-1}^{i-1}} \right)
      \\ & \qquad \qquad
      \times \frac{1}{\lambda_K + \mu_{K-1}}
      \left( \frac{\prod_{a=1}^{i-1} \bigl( \lambda_K^{-1} \lambda _a - 1 \bigr)^2}
        {\prod_{b=1}^{j-1} \bigl( \frac{\mu_b}{\lambda_K} + 1 \bigr)} \right)
      \left( \frac{\prod_{b=1}^{j-1} \bigl( \mu_{K-1}^{-1} \mu_b - 1 \bigr)^2}
        {\prod_{a=1}^{i-1} \bigl( \frac{\lambda_a}{\mu_{K-1}} + 1 \bigr)} \right)
      .
    \end{split}
  \end{equation}
  Using the substitution~\eqref{eq:substitution-kill-peakon-even} in this equation,
  and factoring out the lowest power of~$\epsilon$ in each sum,
  we get~\eqref{eq:heineintegral-epsilon}.
\end{proof}

\begin{remark}
  \label{rem:dominant-contribution}
  According to Lemma~\ref{lem:heineintegral-epsilon},
  the dominant contribution as $\epsilon \to 0$ will come from the term (or terms) with
  the lowest power of~$\epsilon$:
  \begin{equation}
    \min\bigl( 0, k_1-q_1, k_2-q_2, (k_1-q_1)+(k_2-q_2)-1 \bigr)
    .
  \end{equation}
\end{remark}

\begin{remark}
  \label{rem:q1-q2-invariant}
  The numbers $q_1$ and~$q_2$ in Lemma~\ref{lem:heineintegral-epsilon}
  do not change if the upper indices are both increased by one
  and simultaneously the lower indices are both decreased by one.
  That is, the matrices
  \begin{equation*}
    \hat{\detJ}_{ij}^{rs}
    \qquad\text{and}\qquad
    \hat{\detJ}_{i-1,j-1}^{r+1,s+1}
  \end{equation*}
  have the same $q_1$ and~$q_2$.
  This will be used frequently in the proof to save work,
  since in many of the solution formulas a term $\hat{\detJ}_{ij}^{rs}$
  in the numerator is accompanied by a term $\hat{\detJ}_{i-1,j-1}^{r+1,s+1}$
  at the corresponding position in the denominator.
\end{remark}

\subsection{How to kill a $Y$-singleton}
\label{sec:kill-Y-even}

Assume that we are at a specific stage of the procedure described
in Section~\ref{sec:killing},
where we currently have $K+K$ groups, where $K \ge 2$,
and we want to kill the singleton $\hat{Y}_{K-j}$-peakon which
constitutes $Y$-group number~$j+1$ from the right.
The $\hat{Y}_K$-peakon on the far right ($j=0$) is only ever killed in order
to enter the odd case treated in Section~\ref{sec:proofs-odd},
so here in the even case we assume
\begin{equation}
  \label{eq:assumption-j-while-killing-Y}
  1 \le j \le K-1
  .
\end{equation}
The effect of killing this $Y$-peakon is that the $\hat{X}_{K+1-j}$-group,
consisting of
\begin{equation*}
  N = \hat{N}_{K+1-j}^X
\end{equation*}
peakons,
and having
(if $N \ge 2$)
the internal parameters
\begin{equation}
  \hat{\tau}_i = \hat{\tau}_{K+1-j,i}^X
  ,\qquad
  \hat{\sigma}_i = \hat{\sigma}_{K+1-j,i}^X
  ,\qquad
  1 \le i \le N-1
  ,
\end{equation}
will be joined with the singleton $\hat{X}_{K-j}$ to form a new group,
the $X_{K-j}$-group, containing $N_{K-j}^X = N+1$ peakons,
and having
(for any $N \ge 1$)
the internal parameters
\begin{equation}
  \tau_i=\tau_{K-j,i}^X
  ,\qquad
  \sigma_i=\sigma_{K-j,i}^X
  ,\qquad
  1 \le i \le N
  .
\end{equation}
For example, using the schematic notation from Example~\ref{ex:proof-technique2},
the configuration may be
\begin{equation*}
  \X \Y \X \Y \X \Y \X \Y
  \underbrace{\X \Y \overbrace{\X \X \X \X}^N}
  \Y \Y \Y \Y
  \X \X \X
  \Y
  \X \X \X
  \Y
  \X
  \Y \Y
  \X \X \X \X
  \Y
  \X \X
  \Y \Y \Y
  ,
\end{equation*}
which after the killing becomes
\begin{equation*}
  \X \Y \X \Y \X \Y \X \Y
  \underbrace{\X \Z \X \X \X \X}_{N+1}
  \Y \Y \Y \Y
  \X \X \X
  \Y
  \X \X \X
  \Y
  \X
  \Y \Y
  \X \X \X \X
  \Y \X \X
  \Y \Y \Y
  .
\end{equation*}
To achieve this,
we use a substitution of the form~\eqref{eq:substitution-kill-peakon-even} with
\begin{equation}
  \label{eq:substitution-kill-Y-even}
  \begin{gathered}
    k_1 = j
    ,\qquad
    k_2 = j-2
    ,
    \\[1ex]
    \alpha = \frac{\theta}{\tau_1}
    ,\qquad
    \beta = 1
    ,\qquad
    Z = \frac{(\alpha + 1) \, \tau_1}{\alpha^{j}}
    = (\tau_1 + \theta) (\tau_1 / \theta)^j
    ,\qquad
    W = \sigma_1
    ,
  \end{gathered}
\end{equation}
and if $N \ge 2$ we also redefine the internal parameters,
\begin{equation}
  \label{eq:redefine-parameters-sigma-tau}
  \begin{aligned}
    \hat{\tau}_1 &= \tau_2 - \theta
    ,
    \\
    \hat{\tau}_i &= \tau_{i+1}
    ,&& 2 \le i \le N-1
    ,
    \\
    \hat{\sigma}_i &= \sigma_{i+1} - \sigma_1
    ,&& 1 \le i \le N-1
    .
  \end{aligned}
\end{equation}
We apply these substitutions in the solution formulas for $K+K$
groups that describe the current configuration.
The effect is to replace the $2N+2$ parameters
\begin{equation*}
  \lambda_K
  ,\quad
  \mu_{K-1}
  ,\quad
  a_K
  ,\quad
  b_{K-1}
  ,\quad
  \{ \hat{\sigma}_i, \hat{\tau}_i \}_1^{N-1}
\end{equation*}
with the equivalent set of $2N+2$ parameters
\begin{equation*}
  \epsilon
  ,\quad
  \theta
  ,\quad
  \{ \sigma_i, \tau_i \}_1^{N}
\end{equation*}
It is clear that the old parameters are positive
and the $\hat{\sigma}_i$ are increasing
if and only if the new parameters are positive and
the $\sigma_i$ are increasing.
(In particular, when $N=2$, the requirement $\hat{\sigma}_1 > 0$
gives rise to $\sigma_2 > \sigma_1$ to begin with.)

We then let ${\epsilon \to 0}$.
As we will show, this turns the $\hat{Y}_{K-j}$-peakon into a ghostpeakon,
which we discard.
After relabeling the other peakons,
we obtain the new configuration with $(K-1)+(K-1)$ groups,
whose $X_{K-j}$-group is made up of
\begin{equation}
  \begin{aligned}
    X_{K-j,1} &= \lim_{\epsilon \to 0} \hat{X}_{K-j}
    ,\\
    X_{K-j,i} &= \lim_{\epsilon \to 0} \hat{X}_{K+1-j,i-1}
    ,\qquad 2 \le i \le N+1
    .
  \end{aligned}
\end{equation}

In fact there are some further substitutions that must also be made in certain cases.
They will be explained at the relevant places in the proof,
but are summarized here, for convenience.
It is understood that whenever a parameter is not affected by these rules,
it is left unchanged
(just remove the hat, and relabel if necessary to take into account that the number
of groups has changed:
$\hat{D} = D$, $\hat{\tau}_{a,i}^X = \tau_{a,i}^X$ for $a < K-j$
but $\hat{\tau}_{a+1,i}^X = \tau_{a,i}^X$ for $a > K-j$, etc.).
\begin{itemize}
\item If $\hat{N}_{K+1-j}^Y \ge 2$, let
  \begin{equation}
    \hat{\tau}_{K+1-j,1}^Y = \tau_{K-j,1}^Y - \sigma_{K-j,1}^X
    .
  \end{equation}
  See~\eqref{eq:funny-redefinition} for the case $2 \le j \le K-1$
  and~\eqref{eq:funny-redefinition-rightmost} for the case $j = 1$.

\item If $j=1$ and $\hat{N}_{K}^Y = 1$, let
  \begin{equation}
    \hat{D} = D - \sigma_{K-1,1}^X
    .
  \end{equation}
  See~\eqref{eq:redefine-D}.

\item If $1 \le j \le K-2$, let
  \begin{equation}
    \hat{C} = \frac{\theta C}{\tau_1}
    .
  \end{equation}
  See~\eqref{eq:redefine-C}.
  
\item If $j = K-1$, let
  \begin{equation}
    \hat{C} = \frac{\theta C}{\tau_1} - \frac{M}{\sigma_1}
    ,\qquad
    \text{where}
    \quad
    M = \prod_{j=1}^{K-2} \mu_j
    .
  \end{equation}
  See~\eqref{eq:kill-Y1-redefine-C}.
  The interpretation when $K=2$ is that $M=1$ (empty product).

\end{itemize}
These substitutions (together with similar ones when killing an $X$-peakon,
as described in Section~\ref{sec:kill-X-even})
are what gives rise to the constraints listed in
Section~\ref{sec:more-notation},
as will be explained in the sections below.
They have been designed in such a way that the
``ghost parameter'' $\theta$ will only appear in the formula for the
position of the ghostpeakon,
\begin{equation*}
  y_{\text{ghost}}(t) = \lim_{\epsilon \to 0} \hat{y}_{K-j}(t)
  ,
\end{equation*}
and such that all internal parameters will be confined to their own groups
rather than ``leaking'' into the formulas for other groups.

We are assuming as an induction hypothesis that the original $K+K$ configuration
is described by our solution formulas,
and what we need to show is that the formulas obtained by making the
substitutions above in these formulas,
and then letting $\epsilon \to 0$, agree with the
claimed solution formulas for the new $(K-1)+(K-1)$ configuration.

Recall from~\eqref{eq:interlacing-solution-positions}
or~\eqref{eq:even-Y-typical-singleton}
the formula for the position of $Y$-peakon number~$p'=K+1-p$, where $p=j+1 \in [2,K]$,
which is the one that we aim to kill:
\begin{equation}
  \hat{Y}_{K-j}
  = \hat{Y}_{(j+1)'}
  = \hat{Y}_{p'}
  = \frac{\hat{\detJ}_{p,p-1}^{00}}{\hat{\detJ}_{p-1,p-2}^{11}}
  = \frac{\hat{\detJ}_{j+1,j}^{00}}{\hat{\detJ}_{j,j-1}^{11}}
  .
\end{equation}
Under the substitutions above, we now expand
the numerator and the denominator
according to Lemma~\ref{lem:heineintegral-epsilon}.
To determine the dominant terms as $\epsilon\to 0$,
we compute the numbers listed in Remark~\ref{rem:dominant-contribution}.
For the numerator~$\hat{\detJ}_{j+1,j}^{00}$ we find that
\begin{equation}
  \begin{aligned}
    q_1 &= 2 \bigl( (j+1)-1 \bigr) - j + 0 = j
    ,\\
    q_2 &= 2(j-1) - (j+1 ) + 0 =  j-3
    ,
  \end{aligned}
\end{equation}
and since $k_1=j$ and $k_2=j-2$, this implies that
\begin{equation}
  \begin{aligned}
    k_1-q_1 &= 0
    ,\\
    k_2-q_2 &= 1
    ,\\
    (k_1-q_1) + (k_2-q_2) - 1 &= 0
    .
  \end{aligned}
\end{equation}
According to Remark~\ref{rem:q1-q2-invariant}
we obtain the same numbers for the denominator~$\hat{\detJ}_{j,j-1}^{11}$.
Thus, for both the numerator and the denominator,
the first, second and fourth term in~\eqref{eq:heineintegral-epsilon}
will give contributions of order $\epsilon^0$,
while the third term is of order $\epsilon^1$ and can be included
in the remainder~$\mathcal{O}(\epsilon)$:
\begin{equation}
  \begin{split}
    \hat{Y}_{K-j}
    = \frac{\hat{\detJ}_{j+1,j}^{00}}{\hat{\detJ}_{j,j-1}^{11}}
    &
    = \frac{\detJ_{j+1,j}^{00}
      + Z \alpha^{q_1} \detJ_{jj}^{00}
      + \dfrac{Z W \alpha^{q_1+1} \beta^{q_2+1}}{\alpha + \beta} \, \detJ_{j,j-1}^{00}
      + \mathcal{O}(\epsilon)}
    {\detJ_{j,j-1}^{11}
      + Z \alpha^{q_1} \detJ_{j-1,j-1}^{11}
      + \dfrac{Z W \alpha^{q_1+1} \beta^{q_2+1}}{\alpha + \beta} \, \detJ_{j-1,j-2}^{11}
      + \mathcal{O} (\epsilon)}
    \\ &
    = \frac{\detJ_{j+1,j}^{00}
      + (\tau_1 + \theta) \detJ_{jj}^{00}
      + \theta \sigma_1 \detJ_{j,j-1}^{00}
      + \mathcal{O}(\epsilon)}
    {\detJ_{j,j-1}^{11}
      + (\tau_1 + \theta) \detJ_{j-1,j-1}^{11}
      + \theta \sigma_1 \detJ_{j-1,j-2}^{11}
      + \mathcal{O} (\epsilon)}
    .
  \end{split}
\end{equation}
Letting $\epsilon \to 0$, we get the formula for $y_{\text{ghost}}(t)$\,:
\begin{equation}
  \label{eq:proof-Y-ghost}
  Y_{\text{ghost}}
  = \frac12 \exp \bigl( 2 y_{\text{ghost}} \bigr)
  = \lim_{\epsilon \to 0} \hat{Y}_{K-j}
  =
  \frac{\detJ_{j+1,j}^{00}
    + (\tau_1 + \theta) \detJ_{jj}^{00}
    + \theta \sigma_1 \detJ_{j,j-1}^{00}
  }{\detJ_{j,j-1}^{11}
    + (\tau_1 + \theta) \detJ_{j-1,j-1}^{11}
    + \theta \sigma_1 \detJ_{j-1,j-2}^{11}}
  ,
\end{equation}
where $\theta$ is the positive  ghost parameter.
This formula will actually now be discarded, but we record it here anyway for
use in Section~\ref{sec:characteristic-curves} about characteristic curves.
And it was also necessary to check that the limit was \emph{finite},
as we will see in a moment.

Next, we verify that the amplitude of that peakon becomes zero, under the same procedure,
so that we really obtain a ghostpeakon as claimed.
Writing $p=j+1$ as above, the formula for the amplitude
in~\eqref{eq:interlacing-solution-amplitudes}
or~\eqref{eq:even-Y-typical-singleton} is
\begin{equation}
  \hat{P}_{K-j}
  = \hat{P}_{p'}
  = \frac{\hat{\detJ}_{p-1,p-1}^{10} \hat{\detJ}_{p-1,p-2}^{11}}{\hat{\detJ}_{p,p-1}^{01} \hat{\detJ}_{p-1,p-2}^{01}}
  = \frac{\hat{\detJ}_{jj}^{10} \hat{\detJ}_{j,j-1}^{11}}{\hat{\detJ}_{j+1,j}^{01} \hat{\detJ}_{j,j-1}^{01}}
  .
\end{equation}
We need to compare the powers of $\epsilon$ after expanding each
$\hat{\detJ}_{ij}^{rs}$ in this formula with similar
computations as for the positions above. We find that the smallest
power of $\epsilon$ is $-1$, which appears (only) in the last term in
Lemma~\ref{lem:heineintegral-epsilon} for
$\hat{\detJ}_{j+1,j}^{01}$,
since for this factor we have $q_1 = j$ and $q_2 = j-2$,
so that
\begin{equation}
  \begin{aligned}
    k_1-q_1 &= 0
    ,\\
    k_2-q_2 &= 0
    ,\\
    (k_1-q_1) + (k_2-q_2) - 1 &= -1
    .
  \end{aligned}
\end{equation}
We thus have $\epsilon^{-1}$ in the denominator,
so after multiplying numerator and denominator by~$\epsilon$
we find that
\begin{equation*}
  \hat{P}_{K-j}
  =
  \mathcal{O}(\epsilon)
  \to 0
  ,\qquad
  \text{as $\epsilon \to 0$}
  .
\end{equation*}
Now recall that $\hat{P}_{K-j} = 2 \hat{n}_{K-j} \exp(- \hat{y}_{K-j})$,
and that the limit $\hat{y}_{K-j} \to y_{\text{ghost}}$ is finite, as we saw above.
This implies that the amplitude $\hat{n}_{K-j}$ indeed becomes zero in the limit
as $\epsilon \to 0$.

\subsubsection{What happens to $X$-groups}
\label{sec:what-happens-X-groups}

We have now seen that making the right substitution and letting
$\epsilon \to 0$ will turn the selected $Y$-singleton
(namely $\hat{Y}_{(j+1)'} = \hat{Y}_{K-j}$)
into a ghostpeakon.
Now let us see what happens to all non-singleton groups
to the right of the killed $Y$-peakon when we perform the same
operations.
We begin with $X$-groups here, and consider $Y$-groups
in Section~\ref{sec:what-happens-Y-groups}.
For singletons (to the left or to the right), see
Section~\ref{sec:what-happens-X-singletons} and~\ref{sec:what-happens-Y-singletons}.

Consider first the  solution formulas for the $j$th $X$-group from the right
(i.e., the $\hat{X}_{j'}$-group where $j'=K+1-j$, with $1 \le j \le K-1$),
\emph{the one next to  the  killed $Y$-singleton},
in case this group consists of
$N = \hat{N}_{j'}^X \ge 2$ peakons before the killing:
\begin{equation*}
  \X \Y \X \Y \X \Y \X \Y
  \underbrace{\X \Y \overbrace{\X \X \X \X}^{N}}
  \Y \Y \Y \Y
  \X \X \X
  \Y
  \X \X \X
  \Y
  \X
  \Y \Y
  \X \X \X \X
  \Y
  \X \X
  \Y \Y \Y
  .
\end{equation*}
(The underbrace is the same as before, indicating the location of the
new group formed when the $Y$-peakon is killed, while the overbrace
highlights the group whose solution formulas we are currently
studying.)
We would like to show that the form of the surviving formulas
are preserved by the procedure, but with the new parameters $\{ \sigma_i,\tau_i\}_{i=1}^{N}$
instead of the original ones
$\{ \hat{\sigma}_i, \hat{\tau}_i \}_{i=1}^{N-1}$,
and also that the new formula for the old singleton $X$-group number
$j+1$ from the right becomes the first peakon in the new $j$th $X$-group from the right:
\begin{equation}
  \label{eq:proof-X-group-verify-pos}
  \begin{aligned}
    \lim_{\epsilon \to 0} \hat{X}_{(j+1)'}
    &=
    X_{j',1}
    ,\\
    \lim_{\epsilon \to 0} \hat{X}_{j',i}
    &=
    X_{j',i+1}
    ,\qquad
    1 \le i \le N
    .
  \end{aligned}
\end{equation}
To avoid any misunderstanding,
we remind the reader of the warning~\eqref{eq:j-prime-warning}:
after the killing, the notation~$j'$ means $K-j$.
So if we count from the left, the formulas look as follows:
\begin{equation}
  \label{eq:proof-X-group-verify-pos-fromleft}
  \begin{aligned}
    \lim_{\epsilon \to 0} \hat{X}_{K-j}
    &=
    X_{K-j,1}
    ,\\
    \lim_{\epsilon \to 0} \hat{X}_{K+1-j,i}
    &=
    X_{K-j,i+1}
    ,\qquad
    1 \le i \le N
    .
  \end{aligned}
\end{equation}
The first of these equations will be proved in Section~\ref{sec:what-happens-X-singletons}
where we investigate what happens to $X$-singletons;
see~\eqref{eq:first-peakon-new-X-group} for the case
$1 \le j \le K-2$ and \eqref{eq:first-peakon-new-leftmost-X-group} for the case $j = K-1$.
Here we will show the second one.
Recall the solution formulas~\eqref{eq:even-X-typical-group-pos},
for $1 \le j \le K-1$:
\begin{equation}
  \begin{split}
    \label{eq:proof-Xij}
    \hat{X}_{j',i} &= \frac{\hat{\detJ}_{j+1,j}^{00} + \hat{T}_i \hat{\detJ}_{jj}^{00} + \hat{S}_i \hat{\detJ}_{j,j-1}^{00}}{\hat{\detJ}_{j,j-1}^{11} + \hat{T}_i \hat{\detJ}_{j-1,j-1}^{11} + \hat{S}_i \hat{\detJ}_{j-1,j-2}^{11}}
    ,
    \qquad
    1 \le i \le N-1
    ,
    \\[1em]
    \hat{X}_{j',N} &= \frac{\hat{\detJ}_{jj}^{00} + \hat{\sigma}_{N-1} \hat{\detJ}_{j,j-1}^{00}}{\hat{\detJ}_{j-1,j-1}^{11} + \hat{\sigma}_{N-1} \hat{\detJ}_{j-1,j-2}^{11}}
    .
  \end{split}
\end{equation}
We are going to use Lemma~\ref{lem:heineintegral-epsilon} to
replace each $\hat{\detJ}_{ij}^{rs}$ in equation~\eqref{eq:proof-Xij},
and then use Remark~\ref{rem:dominant-contribution} to calculate the dominant
contribution of~$\epsilon$.

Let us start with the first equation in~\eqref{eq:proof-Xij},
and consider the terms $\hat{\detJ}_{j+1,j}^{00}$ and $\hat{\detJ}_{j,j-1}^{11}$,
which by Remark~\ref{rem:q1-q2-invariant} contain the same powers of~$\epsilon$.
We compute
\begin{equation}
  \begin{aligned}
    q_1 &= 2((j+1)-1)-j+0 = j
    ,\\ 
    q_2 &= 2(j-1)-(j+1)+0 = j-3
    .
  \end{aligned}
\end{equation}
Since $k_1=j$ and $k_2=j-2$, this implies that
\begin{equation}
  \begin{aligned}
    k_1-q_1 &= 0
    ,\\
    k_2-q_2 &= 1
    ,\\
    (k_1-q_1)+( k_2-q_2)-1 &= 0
    ,
  \end{aligned}
\end{equation}
so the dominant contribution with the smallest exponent is $\epsilon^0$,
coming from the first, second and fourth term in~\eqref{eq:heineintegral-epsilon}.
Similarly, for $\hat{\detJ}_{jj}^{00}$ and $\hat{\detJ}_{j-1,j-1}^{11}$,
the power $\epsilon^0$ dominates,
appearing in the first and third term in~\eqref{eq:heineintegral-epsilon},
since $q_1 = q_2 = j-2$ so that
\begin{equation}
  \begin{aligned}
    k_1-q_1 &= 2
    ,\\
    k_2-q_2 &= 0
    ,\\
    (k_1-q_1)+( k_2-q_2)-1 &= 1
    .
  \end{aligned}
\end{equation}
And for $\hat{\detJ}_{j,j-1}^{00}$ and $\hat{\detJ}_{j-1,j-2}^{11}$,
the lowest-order contribution is also~$\epsilon^0$,
but only for the first term in~\eqref{eq:heineintegral-epsilon},
since $q_1 = j-1$ and $q_2 = j-4$ so that
the other exponents become greater than zero:
\begin{equation}
  \begin{aligned}
    k_1-q_1 &= 1
    ,\\
    k_2-q_2 &= 2
    ,\\
    (k_1-q_1)+( k_2-q_2)-1 &= 2
    .
  \end{aligned}
\end{equation}
Substituting first in the numerator, we have (for $1 \le i \le N-1$)
\begin{equation}
  \begin{split}
    &
    \hat{\detJ}_{j+1,j}^{00}
    + \hat{T}_i \hat{\detJ}_{jj}^{00} + \hat{S}_i \hat{\detJ}_{j,j-1}^{00}
    \\[1ex]
    &
    =
    \detJ_{j+1,j}^{00}
    + \detJ_{jj}^{00} Z  \alpha^{j} \bigl( 1 + \mathcal{O}(\epsilon) \bigr)
    + \detJ_{j,j-1}^{00} Z W \alpha^{j+1} \beta^{(j-3)+1} (\alpha + \beta)^{-1} \bigl( 1 + \mathcal{O}(\epsilon) \bigr)
    \\ & \quad
    + \hat{T}_i \left( \detJ_{jj}^{00} + \detJ_{j,j-1}^{00} W \beta^{j-2} \bigl( 1 + \mathcal{O}(\epsilon) \bigr) \right)
    \\ & \quad
     + \hat{S}_i \detJ_{j,j-1}^{00} \, \bigl( 1 + \mathcal{O}(\epsilon) \bigr)
    \\[1ex]
    &
    =
    \detJ_{j+1,j}^{00} + \detJ_{jj}^{00} \, (\theta + \tau_1)
    + \detJ_{j,j-1}^{00} \, \sigma_1 \theta
    \\ & \quad
    + \hat{T}_i \left( \detJ_{jj}^{00} + \detJ_{j,j-1}^{00} \, \sigma_1 \right)
    + \hat{S}_i \detJ_{j,j-1}^{00}
    + \mathcal{O}(\epsilon)
    \\[1ex]
    &
    \to
    \detJ_{j+1,j}^{00}
    + \bigl( \theta + \tau_1 + \hat{T}_i \bigr)  \detJ_{jj}^{00}
    + \bigl( \sigma_1 \theta + \sigma_1 \hat{T}_i + \hat{S}_i \bigr) \detJ_{j,j-1}^{00}
    \\[1ex]
    &
    =
    \detJ_{j+1,j}^{00} + T_{i+1} \detJ_{jj}^{00} + S_{i+1} \detJ_{j,j-1}^{00}
    , \qquad
    \text{as $\epsilon \to 0$}
    ,
  \end{split}
\end{equation}
where we get the last equation by redefining the parameters as
in~\eqref{eq:redefine-parameters-sigma-tau}, namely
\begin{equation*}
  \begin{aligned}
    \hat{\tau}_1 &= \tau_2 - \theta
    ,
    \\
    \hat{\tau}_i &= \tau_{i+1}
    ,&&
    2 \le i \le N-1
    ,
    \\
    \hat{\sigma}_i &= \sigma_{i+1} - \sigma_1
    ,&&
    1 \le i \le N-1
    .
  \end{aligned}
\end{equation*}
Indeed, this gives (for $1 \le i \le N-1$, with $\sum_2^i = 0$ if $i=1$)
\begin{equation}
  \label{eq:redefine-T-hat}
  \theta + \tau_1 + \hat{T}_i
  = \theta + \tau_1 + \left( \hat{\tau}_1 + \sum_{a=2}^{i} \hat{\tau}_a \right)
  = \theta + \tau_1 + \left( \tau_2 - \theta + \sum_{a=2}^{i} \tau_{a+1} \right)
  = \sum_{a=1}^{i+1} \tau_{a}
  = T_{i+1}
\end{equation}
and (with $\sum_1^{i-1} = 0$ if $i=1$)
\begin{equation}
  \label{eq:redefine-S-hat}
  \begin{split}
    \sigma_1 \theta + \sigma_1 \hat{T}_i + \hat{S}_i
    &
    = \sigma_1 \bigl( \theta + \hat{T}_i \bigr) + \sum_{a=1}^{i-1} \hat{\sigma}_a \hat{\tau}_{a+1}
    \\ &
    = \sigma_1 (T_{i+1} - \tau_1) + \sum_{a=1}^{i-1} (\sigma_{a+1} - \sigma_1) \tau_{a+2}
    \\ &
    = \sigma_1 \sum_{a=1}^{i} \tau_{a+1} + \sum_{a=2}^{i} (\sigma_a - \sigma_1) \tau_{a+1}
    = \sum_{a=1}^{i} \sigma_a \tau_{a+1}
    = S_{i+1}
    .
  \end{split}
\end{equation}
The denominator is similar to the numerator (Remark~\ref{rem:q1-q2-invariant}),
so in the limit $\epsilon \to 0$ we obtain,
in agreement with the solution formulas~\eqref{eq:even-X-typical-group-pos},
\begin{equation}
  X_{j',i+1}
  = \lim_{\epsilon \to 0} \hat{X}_{j',i}
  = \frac{\detJ_{j+1,j}^{00} + T_{i+1} \detJ_{jj}^{00} + S_{i+1} \detJ_{j,j-1}^{00}}{\detJ_{j,j-1}^{11} + T_{i+1} \detJ_{j-1,j-1}^{11} + S_{i+1} \detJ_{j-1,j-2}^{11}}
  ,\qquad
  1 \le i \le N-1
  .
\end{equation}
Similarly, for the rightmost peakon in the group,
the position is given by the second equation in~\eqref{eq:proof-Xij} above,
and using $\hat{\sigma}_{N-1} = \sigma_{N}-\sigma_1$
we find, as desired, that
\begin{equation}
  \begin{split}
    X_{j',N+1}
    = \lim_{\epsilon \to 0} \hat{X}_{j',N}
    &
    = \lim_{\epsilon \to 0}
    \frac{\hat{\detJ}_{jj}^{00} + \hat{\sigma}_{N-1} \hat{\detJ}_{j,j-1}^{00}}{\hat{\detJ}_{j-1,j-1}^{11} + \hat{\sigma}_{N-1} \hat{\detJ}_{j-1,j-2}^{11}}
    \\ &
    = \lim_{\epsilon \to 0} \frac{\left( \detJ_{jj}^{00} + \sigma_1 \detJ_{j,j-1}^{00} \right) + (\sigma_{N} - \sigma_1) \detJ_{j,j-1}^{00} + \mathcal{O}(\epsilon)}{\left( \detJ_{j-1,j-1}^{11} + \sigma_1 \detJ_{j-1,j-2}^{11} \right) + (\sigma_{N} - \sigma_1) \detJ_{j-1,j-2}^{11} + \mathcal{O}(\epsilon)}
    ,
    \\[1em]
    &
    = \frac{\detJ_{jj}^{00} + \sigma_{N} \detJ_{j,j-1}^{00}}{\detJ_{j-1,j-1}^{11} + \sigma_{N} \detJ_{j-1,j-2}^{11}}
    .
  \end{split}
\end{equation}
This concludes the proof that the second equation in~\eqref{eq:proof-X-group-verify-pos} holds.

One also needs to verify that the formulas
for the amplitudes in this group behave in the corresponding way:
\begin{equation}
  \label{eq:proof-X-group-verify-amp}
  \begin{aligned}
    \lim_{\epsilon \to 0} \hat{Q}_{(j+1)'}
    &=
    Q_{j',1}
    ,\\
    \lim_{\epsilon \to 0} \hat{Q}_{j',i}
    &=
    Q_{j',i+1}
    ,\qquad
    1 \le i \le N
    .
  \end{aligned}
\end{equation}
The investigation of the singleton's amplitude $\hat{Q}_{(j+1)'}$ belongs
to Section~\ref{sec:what-happens-X-singletons};
as for $\hat{X}_{(j+1)'}$, one must treat the cases
$1 \le j \le K-2$ and $j = K-1$ separately, since $\hat{Q}_1$ is given by a special formula.
From~\eqref{eq:even-X-typical-group-amp}, the formula for $\hat{Q}_{j',i}$ when $1 \le i \le N-1$ is
\begin{equation*}
  \begin{split}
    \hat{Q}_{j',i} &
    = (\hat{\sigma}_i - \hat{\sigma}_{i-1}) \, 
    \hat{\detJ}_{j,j-1}^{01}
    \left( \hat{\detJ}_{j,j-1}^{11} + \hat{T}_i \hat{\detJ}_{j-1,j-1}^{11} + \hat{S}_i \hat{\detJ}_{j-1,j-2}^{11} \right)
    \\
    &\quad
    \times
    \left( \hat{\detJ}_{jj}^{10} + \hat{\sigma}_i \hat{\detJ}_{j,j-1}^{10} + \hat{R}_i \hat{\detJ}_{j-1,j-1}^{10} \right)^{-1}
    \\
    &\quad
    \times
    \left( \hat{\detJ}_{jj}^{10} + \hat{\sigma}_{i-1} \hat{\detJ}_{j,j-1}^{10} + \hat{R}_{i-1} \hat{\detJ}_{j-1,j-1}^{10} \right)^{-1}
    ,
  \end{split}
\end{equation*}
where the last bracket on the first line has already been studied above.
Just by definition we have
\begin{equation}
  \hat{\sigma}_i - \hat{\sigma}_{i-1}
  = (\sigma_{i+1} - \sigma_1) - (\sigma_i - \sigma_1)
  = \sigma_{i+1} - \sigma_i
  ,
\end{equation}
and for the other factors we compute using Lemma~\ref{lem:heineintegral-epsilon} that
$\hat{\detJ}_{j,j-1}^{01} = \detJ_{j,j-1}^{01} + \mathcal{O}(\epsilon)$
and
\begin{equation}
  \begin{split}
    &
    \hat{\detJ}_{jj}^{10} + \hat{\sigma}_i \hat{\detJ}_{j,j-1}^{10} + \hat{R}_i \hat{\detJ}_{j-1,j-1}^{10}
    \\ &
    =
    \Bigl( \detJ_{jj}^{10} + \sigma_1 \detJ_{j,j-1}^{10} + \sigma_1 \tau_1 \detJ_{j-1,j-1}^{10} \Bigr)
    \\ & \quad
    + \hat{\sigma}_i \Bigl( \detJ_{j,j-1}^{10} + (\tau_1 + \theta) \detJ_{j-1,j-1}^{10} \Bigr)
    + \hat{R}_i \detJ_{j-1,j-1}^{10}
    + \mathcal{O}(\epsilon)
    \\ &
    =
    \detJ_{jj}^{10}
    + \underbrace{(\sigma_1 + \hat{\sigma}_i)}_{= \sigma_{i+1}} \detJ_{j,j-1}^{10}
    + \underbrace{\left( \sigma_1 \tau_1 +   \hat{\sigma}_i (\tau_1 + \theta) + \hat{R}_i \right)}_{= R_{i+1}} \detJ_{j-1,j-1}^{10}
    + \mathcal{O}(\epsilon)
    ,
  \end{split}
\end{equation}
where the identity yielding $R_{i+1}$ is easily verified
as in \eqref{eq:redefine-T-hat} and~\eqref{eq:redefine-S-hat} above.
The conclusion is that in the limit as $\epsilon \to 0$, all expressions in the formula
for~$\hat{Q}_{j',i}$ reduce to their counterparts without hats, except that $i$ is replaced
by $i+1$ everywhere, which proves that
$\lim_{\epsilon \to 0} \hat{Q}_{j',i} = Q_{j',i+1}$
for $1 \le i \le N-1$.
The case~$i=N$ is very similiar; we omit the details.

\bigskip

Next, let us see
\emph{what happens to the non-singleton $X$-groups further right of the killed $Y$-peakon},
indicated by overbraces in the diagram:
\begin{equation*}
  \X \Y \X \Y \X \Y \X \Y
  \underbrace{\X \Y \X \X \X \X}
  \Y \Y \Y \Y
  \overbrace{\X \X \X}
  \Y
  \overbrace{\X \X \X}
  \Y
  \X
  \Y \Y
  \overbrace{\X \X \X \X}
  \Y
  \overbrace{\X \X}
  \Y \Y \Y
  .
\end{equation*}
There are no such groups if $j=1$, so assume $2 \le j \le K-1$.
Writing $N = \hat{N}_{p'}^X$,
the positions~\eqref{eq:even-X-typical-group-pos}
for the $p$th $X$-group from the right, where $1 \le p < j$, are
\begin{equation}
  \label{eq:proof-even-X-typical-group-pos}
  \begin{split}
    \hat{X}_{p',i} &= \frac{\hat{\detJ}_{p+1,p}^{00} + \hat{T}_{i} \hat{\detJ}_{pp}^{00} + \hat{S}_{i} \hat{\detJ}_{p,p-1}^{00}}{\hat{\detJ}_{p,p-1}^{11} + \hat{T}_{i} \hat{\detJ}_{p-1,p-1}^{11} + \hat{S}_{i} \hat{\detJ}_{p-1,p-2}^{11}}
    ,\qquad
    1 \le i \le N-1
    ,
    \\[1ex]
    \hat{X}_{p',N} &= \frac{\hat{\detJ}_{pp}^{00} + \hat{\sigma}_{N-1} \hat{\detJ}_{p,p-1}^{00}}{\hat{\detJ}_{p-1,p-1}^{11} + \hat{\sigma}_{N-1} \hat{\detJ}_{p-1,p-2}^{11}}
    .
  \end{split}
\end{equation}
Here $\hat{\tau}_i=\hat{\tau}_{p',i}^X$, etc.,
and in this case we do not need to redefine the internal parameters,
but simply remove the hats: $\hat{\tau}_i=\tau_i$, etc.

We find the leading contributions for $\hat{\detJ}_{p+1,p}^{00}$
and $\hat{\detJ}_{p,p-1}^{11}$ using Lemma~\ref{lem:heineintegral-epsilon}.
A simple computation gives $q_1 = p$ and $q_2 = p-3$,
and since $k_1=j$ and $k_2=j-2$ as always, we get
\begin{equation}
  \begin{aligned}
    k_1-q_1 &= j-p
    ,\\
    k_2-q_2 &= j-p+1
    ,\\
    (k_1-q_1) + (k_2-q_2) - 1 &= 2(j-p)
    ,
  \end{aligned}
\end{equation}
so the lowest power is $\epsilon^0$ (since $j-p>0$),
and we get $\detJ_{p+1,p}^{00} + \mathcal{O}(\epsilon)$ and
$\detJ_{p,p-1}^{11} + \mathcal{O}(\epsilon)$ instead of the ones with hats.
Similarly, $\hat{\detJ}_{ij}^{rs} = \detJ_{ij}^{rs} + \mathcal{O}(\epsilon)$
for all terms appearing in the formulas~\eqref{eq:proof-even-X-typical-group-pos},
so in the limit as $\epsilon \to 0$ we get the same formulas but without hats,
as desired.
The formulas for the amplitudes are proved just as easily.

\subsubsection{What happens to $Y$-groups}
\label{sec:what-happens-Y-groups}

The same argument as at the end of the previous section
(the dominant power is $\epsilon^0$, coming from the first term in Lemma~\ref{lem:heineintegral-epsilon})
also proves that the
formulas for all non-singleton $Y$-groups to the right of the killed peakon
will be preserved by the procedure,
except that a separate argument is needed for the
$Y$-group nearest to the killed peakon,
namely the $j$th $Y$-group from the right,
if it is a non-singleton:
\begin{equation*}
  \X \Y \X \Y \X \Y \X \Y \underbrace{\X \Y \X \X \X \X} \overbrace{\Y \Y \Y \Y}^{N} \X \X \X \Y \X \X \X \Y \X \Y \Y \X \X \X \X \Y \X \X \Y \Y \Y
  .
\end{equation*}
If this is the rightmost group ($j=1$),
it is given by special formulas, so we consider two cases:
\begin{itemize}
\item
  If $2 \le j \le K-1$, the formulas~\eqref{eq:even-Y-typical-group-pos}
  for the positions in this group (not the rightmost one) are, with $N=\hat{N}_{j'}^Y \ge 2$,
  \begin{equation}
    \label{eq:remove-my-hats-please}
    \begin{split}
      \hat{Y}_{j',i} &=
      \frac{\hat{\detJ}_{jj}^{00} + \hat{T}_{j',i}^Y \hat{\detJ}_{j,j-1}^{00} + \hat{S}_{j',i}^Y \hat{\detJ}_{j-1,j-1}^{00}}{\hat{\detJ}_{j-1,j-1}^{11} + \hat{T}_{j',i}^Y \hat{\detJ}_{j-1,j-2}^{11} + \hat{S}_{j',i}^Y \hat{\detJ}_{j-2,j-2}^{11}}
      ,\qquad
      1 \le i \le N-1
      ,
      \\[1ex]
      \hat{Y}_{j',N} &=
      \frac{\hat{\detJ}_{j,j-1}^{00} + \hat{\sigma}_{j',N-1}^Y \hat{\detJ}_{j-1,j-1}^{00}}{\hat{\detJ}_{j-1,j-2}^{11} + \hat{\sigma}_{j',N-1}^Y \hat{\detJ}_{j-2,j-2}^{11}}
      .
    \end{split}
  \end{equation}
  With the help of Lemma~\ref{lem:heineintegral-epsilon}
  and Remark~\ref{rem:dominant-contribution} as usual,
  we find that the dominant power is~$\epsilon^0$,
  and that all terms $\hat{\detJ}_{ij}^{rs}$ in these formulas simply tend to
  $\detJ_{ij}^{rs}$ as $\epsilon \to 0$, except the two terms
  $\hat{\detJ}_{jj}^{00}$ and~$\hat{\detJ}_{j-1,j-1}^{11}$
  which contain an additional contribution of the order~$\epsilon^0$:
  \begin{equation}
    \begin{split}
      Y_{j',i} &
      = \lim_{\epsilon \to 0} \hat{Y}_{j',i}
      = \frac{\left( \detJ_{jj}^{00} + \sigma_{j',1}^X \detJ_{j,j-1}^{00} \right) + \hat T_{j',i}^Y \detJ_{j,j-1}^{00} + \hat S_{j',i}^Y \detJ_{j-1,j-1}^{00}}{\left( \detJ_{j-1,j-1}^{11} + \sigma_{j',1}^X \detJ_{j-1,j-2}^{11} \right) + \hat T_{j',i}^Y \detJ_{j-1,j-2}^{11} + \hat S_{j',i}^Y \detJ_{j-2,j-2}^{11}}
      ,
      \\[1ex]
      Y_{j',N} &
      = \lim_{\epsilon \to 0} \hat{Y}_{j',N}
      = \frac{\detJ_{j,j-1}^{00} + \hat{\sigma}_{j',N-1}^Y \detJ_{j-1,j-1}^{00}}{\detJ_{j-1,j-2}^{11} + \hat{\sigma}_{j',N-1}^Y \detJ_{j-2,j-2}^{11}}
      .
    \end{split}
  \end{equation}
  In order to get the formulas that we want, namely
  \eqref{eq:remove-my-hats-please} with the hats removed,
  we therefore need to redefine~$\hat{\tau}_{j',1}^Y$ as
  \begin{equation}
    \label{eq:funny-redefinition}
    \hat{\tau}_{j',1}^Y = \tau_{j',1}^Y - \sigma_{j',1}^X
    ,
  \end{equation}
  while keeping all the other internal parameters for this $Y$-group unchanged
  (just remove the hats).
  Then
  \begin{equation}
    \hat T_{j',i}^Y + \sigma_{j',1}^X = T_{j',i}^Y
    ,\qquad
    \text{for $1 \le i \le N-1$}
    ,
  \end{equation}
  so that we get the correct coeffient~$T_{j',i}^Y$
  in front of $\detJ_{j,j-1}^{00}$ and~$\detJ_{j-1,j-2}^{11}$.
  The amplitudes are studied similarly.
  
  Before the killing, we have by induction hypothesis that
  if $\hat{N}_{j'}^X \ge 2$, then the first constraint
  in~\eqref{eq:constraint-last-sigma-first-tau}
  holds,
  ``the last~$\sigma$ in the $X$-group is less than the first~$\tau$
  in the adjacent $Y$-group'':
  \begin{equation}
    \hat{\sigma}_{j', \hat{N}_{j'}^X-1}^X
    <
    \hat{\tau}_{j',1}^Y
    .
  \end{equation}
  With the redefinitions
  \eqref{eq:redefine-parameters-sigma-tau} and~\eqref{eq:funny-redefinition}
  that we have made, this becomes
  \begin{equation}
    \sigma_{j',N_{j'}^X}^X
    - \sigma_{j',1}^X
    <
    \tau_{j',1}^Y
    - \sigma_{j',1}^X
    ,
  \end{equation}
  so
  \begin{equation}
    \sigma_{j',N_{j'}^X}^X
    <
    \tau_{j',1}^Y
    ,
  \end{equation}
  meaning that~\eqref{eq:constraint-last-sigma-first-tau} holds also after the killing.
  If instead $\hat{N}_{j'}^X = 1$, then before the killing we only have the constraint
  $0 < \hat{\tau}_{j',1}^Y$,
  which after the redefinition becomes
  $0 < \tau_{j',1}^Y - \sigma_{j',1}^X$,
  i.e.,
  \begin{equation}
    \sigma_{j',1}^X
    <
    \tau_{j',1}^Y
    .
  \end{equation}
  In this case, the $X_{j'}$-group created by the killing has only two members,
  so $\sigma_{j',1}^X$ is its last (and only) $\sigma$-parameter,
  and we see that if the neighbouring $Y$-group is a non-singleton,
  then its first~$\tau$ must be greater than this newly introduced~$\sigma$.
  This is the mechanism which causes the
  constraint~\eqref{eq:constraint-last-sigma-first-tau}
  to arise in the first place.
  
\item 
  If $j=1$, then we are looking at the rightmost $Y$-group, for which the positions are given
  by the special formulas~\eqref{eq:even-Y-rightmost-group-pos}:
  \begin{equation}
    \begin{aligned}
      \hat{Y}_{K,i} &
      = \hat{\detJ}_{11}^{00} +  \hat{T}_{K,i}^Y \hat{\detJ}_{10}^{00} + \hat{S}_{K,i}^Y
      ,\qquad
      1 \le i \le N-1
      ,
      \\[1ex]
      \hat{Y}_{K,N} &
      = \hat{\detJ}_{11}^{00} + \left( \hat{T}_{K,N-1}^Y + \hat{D} \right) \hat{\detJ}_{10}^{00}
      + \left( \hat{S}_{K,N-1}^Y + \hat{D} \, \hat{\sigma}_{K,N-1}^Y \right)
      ,
    \end{aligned}
  \end{equation}
  where $N = \hat{N}_K^Y \ge 2$.
  As above, we obtain
  \begin{equation}
    \begin{aligned}
      Y_{K-1,i} &
      = \lim_{\epsilon \to 0} \hat{Y}_{K,i}
      = \left( \detJ_{11}^{00} + \sigma_{K-1,1}^X \detJ_{10}^{00} \right)
      + \hat T_{K,i}^Y \detJ_{10}^{00}
      + \hat S_{K,i}^Y \detJ_{00}^{00}
      ,\qquad
      1 \le i \le N-1
      ,
      \\[1ex]
      Y_{K-1,N} &
      = \lim_{\epsilon \to 0} \hat{Y}_{K,N}
      = \left( \detJ_{11}^{00} + \sigma_{K-1,1}^X \detJ_{10}^{00} \right)
      + \left( \hat{T}_{K,N-1}^Y + \hat{D} \right) \detJ_{10}^{00}
      + \left( \hat{S}_{K,N-1}^Y + \hat{D} \, \hat{\sigma}_{K,N-1}^Y \right)
      ,
    \end{aligned}
  \end{equation}
  and we see that the desired formulas are obtained with the same redefinition
  of~$\hat{\tau}_1$ as above,
  with all other internal parameters in the rightmost $Y$-group kept unchanged:
  \begin{equation}
    \label{eq:funny-redefinition-rightmost}
    \begin{aligned}
      \hat{\tau}_{K,1}^Y &= \tau_{K-1,1}^Y - \sigma_{K-1,1}^X
      ,\\
      \hat{\tau}_{K,i}^Y &= \tau_{K-1,i}^Y
      ,\qquad
      2 \le i \le N-1
      ,\\
      \hat{\sigma}_{K,i}^Y &= \sigma_{K-1,i}^Y
      ,\qquad
      1 \le i \le N-1
      ,\\
      \hat{D} &= D
      .
    \end{aligned}
  \end{equation}
  Similarly for the amplitudes.
\end{itemize}

\subsubsection{What happens to $X$-singletons}
\label{sec:what-happens-X-singletons}

In this section we will show what becomes of all $X$-singletons
when we kill the $\hat{Y}_{(j+1)'}$-peakon ($1 \le j \le K-1$) as described above.
For the most parts, we will only write the proofs for the positions;
the arguments for the amplitudes are similar.

We consider first $X$-group number $p$ from the right,
where $1 \le p \le K-1$, assuming it is a singleton,
in which case is given by the formulas~\eqref{eq:even-X-typical-singleton}:
\begin{equation}
  \label{eq:proof-X-singleton}
  \hat{X}_{p'} = \frac{\hat{\detJ}_{pp}^{00}}{\hat{\detJ}_{p-1,p-1}^{11}}
  ,
  \qquad
  \hat{Q}_{p'} = \frac{\hat{\detJ}_{p-1,p-1}^{11} \hat{\detJ}_{p,p-1}^{01}}{\hat{\detJ}_{pp}^{10} \hat{\detJ}_{p-1,p-1}^{10}}
  ,
\end{equation}
where $p' = K+1-p$.
The leftmost $X$-peakon ($p=K$) is given by a separate formula,
and will be investigated separately below.
Thus we are now looking at the $X$-singletons indicated by the arrows in this picture
(we use arrows instead of overbraces here for reasons of space):
\begin{equation*}
  \X \Y
  \arrowX \Y \arrowX \Y \arrowX \Y
  \underbrace{\arrowX \Y \X \X \X \X}
  \Y \Y \Y \Y
  \X \X \X
  \Y
  \X \X \X
  \Y
  \arrowX
  \Y \Y
  \X \X \X \X
  \Y
  \X \X
  \Y \Y \Y
  .
\end{equation*}
With the same substitution~\eqref{eq:substitution-kill-Y-even} as before,
and using Lemma~\ref{lem:heineintegral-epsilon},
the numerator of the positions in~\eqref{eq:proof-X-singleton} becomes
\begin{equation}
  \begin{split}
    \hat{\detJ}^{ 00}_{pp} &= \detJ_{pp}^{ 00}\epsilon^0
    \\ & \quad
    + \detJ_{p-1,p}^{ 00} Z \alpha^{q_1} \epsilon^{k_1-q_1}\bigl( 1 + \mathcal{O}(\epsilon) \bigr)
    \\ & \quad
    + \detJ_{p,p-1}^{ 00} W \beta^{q_2} \epsilon^{k_2-q_2}\bigl( 1 + \mathcal{O}(\epsilon) \bigr)
    \\ & \quad
    + \detJ_{p-1,p-1}^{00} \frac{Z W \alpha^{q_1+1} \beta^{q_2+1} }{\alpha+\beta} \epsilon^{(k_1-q_2) + (k_2-q_2) - 1} \bigl( 1 + \mathcal{O}(\epsilon) \bigr)
    ,
  \end{split}
\end{equation}
and the denominator becomes
\begin{equation}
  \begin{split}
    \hat{\detJ}^{ 11}_{p-1,p-1} &= \detJ_{p-1,p-1}^{ 11} \epsilon^0
    \\ & \quad
    + \detJ_{p-2,p-1}^{ 11} Z \alpha^{q_1} \epsilon^{k_1-q_1} \bigl( 1 + \mathcal{O}(\epsilon) \bigr)
    \\ & \quad
    + \detJ_{p-1,p-2}^{ 11} W \beta^{q_2} \epsilon^{k_2-q_2} \bigl( 1 + \mathcal{O}(\epsilon) \bigr)
    \\ & \quad
    + \detJ_{p-2,p-2}^{11} \frac{Z W \alpha^{q_1+1} \beta^{q_2+1} }{\alpha+\beta} \epsilon^{(k_1-q_2) + (k_2-q_2) - 1} \bigl( 1 + \mathcal{O}(\epsilon) \bigr)
    .
  \end{split}
\end{equation}
In~\eqref{eq:substitution-kill-Y-even} we have $k_1=j$ and $k_2=j-2$,
and we compute $q_1 = q_2 = p-2$ for both the numerator and the denominator, so that
\begin{equation}
  \begin{aligned}
    k_1-q_1 &= j-p+2
    ,\\
    k_2-q_2 &= j-p
    ,\\
    (k_1-q_1) + (k_2-q_2) - 1 &= 2(j-p)+1
    .
  \end{aligned}
\end{equation}
Thus, the powers of~$\epsilon$ which appear are
\begin{equation}
  \epsilon^0
  ,\qquad
  \epsilon^{j-p+2}
  ,\qquad
  \epsilon^{j-p}
  ,\qquad
  \epsilon^{2(j-p)+1}
  ,
\end{equation}
and we have several cases to study, depending on the value of $j-p$:
\begin{itemize}
\item If $1 \le p < j \le K-1$, then the $X$-singleton lies to the right of
  the killed $Y$-peakon
  and also to the right of its adjacent $X$-group:
  \begin{equation*}
    \X \Y \X \Y \X \Y \X \Y
    \underbrace{\X \Y \X \X \X \X}
    \Y \Y \Y \Y
    \X \X \X
    \Y
    \X \X \X
    \Y
    \arrowX
    \Y \Y
    \X \X \X \X
    \Y
    \X \X
    \Y \Y \Y
    .
  \end{equation*}
  Then $j-p > 0$ and the smallest exponent is~$0$, so
  \begin{equation}
    \hat{X}_{p'}
    = \frac{\hat{\detJ}_{pp}^{00}}{\hat{\detJ}_{p-1,p-1}^{11}}
    = \frac{\detJ_{pp}^{00} \bigl( 1 + \mathcal{O}(\epsilon) \bigr)}{\detJ_{p-1,p-1}^{11} \bigl( 1 + \mathcal{O}(\epsilon) \bigr)}
    \to \frac{\detJ_{pp}^{00}}{\detJ_{p-1,p-1}^{11}}
    ,\qquad
    \text{as $\epsilon \to 0$}
    .
  \end{equation}
  Thus, the formulas for the positions of all such $X$-singletons are preserved.
  
\item If $1 \le p = j \le K-1$, then the $X$-peakon that we are looking at
  is a singleton lying immediately to the right of the killed $Y$-peakon,
  and it will become the last (i.e., second) peakon in
  the new $j$th $X$-group from the right.
  This does not occur in the configuration that we have been using as
  an example, but in a configuration like this instead:
  \begin{equation*}
    \X \Y \X \Y \X \Y \X \Y \underbrace{\X \Y \arrowX} \Y \Y \Y \Y \X \X \X \Y \X \X \X \Y \X \Y \Y \X \X \X \X \Y \X \X \Y \Y \Y
    .
  \end{equation*}
  Then $j-p = 0$, and there are two contributions of order~$\epsilon^0$:
  \begin{equation}
    \begin{split}
      \hat{X}_{p'}= \hat{X}_{j'}
      = \frac{\hat{\detJ}_{jj}^{00}}{\hat{\detJ}_{j-1,j-1}^{11}}
      &
      = \frac{\detJ_{jj}^{00}+ \sigma_1 \detJ_{j,j-1}^{00} + \mathcal{O}(\epsilon)}{\detJ_{j-1,j-1}^{11}+ \sigma_1 \detJ_{j-1,j-2}^{00} + \mathcal{O}(\epsilon)}
      \\ &
      \to \frac{ \detJ_{jj}^{00}+ \sigma_1 \detJ_{j,j-1}^{00}}{\detJ_{j-1,j-1}^{11} + \sigma_1 \detJ_{j-1,j-2}^{00}}
      ,\qquad
      \text{as $\epsilon \to 0$}
      .
    \end{split}
  \end{equation}
  After killing the $Y$-peakon, we relabel the variables to include this
  singleton as the second (and last) member in the new $j$th $X$-group from the
  right,
  \begin{equation}
    X_{j',2}
    = \lim_{\epsilon \to 0} \hat{X}_{j'}
    = \frac{ \detJ_{jj}^{00}+ \sigma_1 \detJ_{j,j-1}^{00}}{\detJ_{j-1,j-1}^{11} + \sigma_1 \detJ_{j-1,j-2}^{00}}
    ,
  \end{equation}
  which agrees with the formula for $X_{j',N_{j'}}$ in~\eqref{eq:even-X-typical-group-pos}
  when $N_{j'} = N_{j'}^X = 2$.

\item If $p = j+1$, where $1 \le j \le K-2$, then we are considering the $X$-peakon
  lying immediately to the left of the killed $Y$-peakon;
  it will become the first peakon in the new
  $j$th $X$-group from the right:
  \begin{equation*}
    \X \Y \X \Y \X \Y \X \Y
    \underbrace{\arrowX \Y \X \X \X \X}
    \Y \Y \Y \Y
    \X \X \X
    \Y
    \X \X \X
    \Y
    \X
    \Y \Y
    \X \X \X \X
    \Y
    \X \X
    \Y \Y \Y
    .
  \end{equation*}
  We are assuming $j \le K-2$ to keep this from being the leftmost peakon,
  which is studied separately below.
  
  Since $j-p = -1$, the smallest exponents are both $j-p = -1$
  and $2(j-p)+1 = -1$, and we find
  \begin{equation}
    \begin{split}
      \hat{X}_{(j+1)'}
      =
      \hat{X}_{p'}
      &
      = \frac{\hat{\detJ}_{pp}^{00}}{\hat{\detJ}_{p-1,p-1}^{11}}
      \\[1ex]
      &
      = \frac{\left( \detJ_{p,p-1}^{00} W \beta^{q_2} + \detJ_{p-1,p-1}^{00} \dfrac{Z W \alpha^{q_1+1} \beta^{q_2+1}}{\alpha + \beta} \right) \epsilon^{-1} \bigl( 1 + \mathcal{O}(\epsilon) \bigr)}{\left( \detJ_{p-1,p-2}^{11} W \beta^{q_2} + \detJ_{p-2,p-2}^{11} \dfrac{Z W \alpha^{q_1+1} \beta^{q_2+1}}{\alpha + \beta} \right) \epsilon^{-1} \bigl( 1 + \mathcal{O}(\epsilon) \bigr)}
      \\[1ex]
      &
      = \frac{\left( \detJ_{p,p-1}^{ 00} \sigma_1 + \detJ_{p-1,p-1}^{00} \tau_1 \sigma_1 \right) \epsilon^{-1} (1+ \mathcal{O}(\epsilon))}{\left( \detJ_{p-1,p-2}^{11} \sigma_1 + \detJ_{p-2,p-2}^{11} \tau_1 \sigma_1 \right) \epsilon^{-1} \bigl( 1 + \mathcal{O}(\epsilon) \bigr)}
      \\[1ex]
      &
      \to
      \frac{\detJ_{j+1,j}^{00} + \tau_1 \detJ_{jj}^{00}  }{\detJ_{j,j-1}^{11} + \tau_1 \detJ_{j-1,j-1}^{11}}
      \qquad \text{as $\epsilon \to 0$}
      .
    \end{split}
  \end{equation}
  After the killing, we relabel the variables to include this singleton
  as the first member in the $j$th $X$-group from the right:
  \begin{equation}
    \label{eq:first-peakon-new-X-group}
    X_{j',1}
    = \lim_{\epsilon \to 0} \hat{X}_{(j+1)'}
    = \frac{\detJ_{j+1,j}^{00} + \tau_1 \detJ_{jj}^{00}  }{\detJ_{j,j-1}^{11} + \tau_1 \detJ_{j-1,j-1}^{11}}
    .
  \end{equation}
  This agrees with the desired expression,
  namely the case $i=1$ in~\eqref{eq:even-X-typical-group-pos},
  repeated here for convenience
  (remember that $T_1=\tau_1$ and $S_1=0$):
  \begin{equation*}
    X_{j',i}
    = \frac{\detJ_{j+1,j}^{00} + T_{i} \detJ_{jj}^{00} + S_{i} \detJ_{j,j-1}^{00}}{\detJ_{j,j-1}^{11} + T_{i} \detJ_{j-1,j-1}^{11} + S_{i} \detJ_{j-1,j-2}^{11}}
    .
  \end{equation*}
  
\item If $2 \le j+1 < p < K$, then the $X$-singleton in question
  lies to the left of the killed peakon (but not on the very left):
  \begin{equation*}
    \X \Y
    \arrowX \Y \arrowX \Y \arrowX \Y
    \underbrace{\X \Y \X \X \X \X}
    \Y \Y \Y \Y
    \X \X \X
    \Y
    \X \X \X
    \Y
    \X
    \Y \Y
    \X \X \X \X
    \Y
    \X \X
    \Y \Y \Y
    .
  \end{equation*}
  Then $j-p < -1$, with the smallest power of $\epsilon$ being $2(j-p)+1$,
  and we get the following formula for the $X$-singletons to the left of the killed peakon:
  \begin{equation}
    \begin{split}
      \hat{X}_{p'}
      = \frac{\hat{\detJ}_{pp}^{00}}{\hat{\detJ}_{p-1,p-1}^{11}}
      &
      = \frac{\detJ_{p-1,p-1}^{00} \epsilon^{2(j-p)+1} \tau_1 \sigma_1 \bigl( \frac{\theta_1}{\tau_1} \bigr)^{p-1-j} \bigl( 1 + \mathcal{O}(\epsilon) \bigr)}{\detJ_{p-2,p-2}^{11} \epsilon^{2(j-p)+1} \tau_1 \sigma_1 \bigl( \frac{\theta_1}{\tau_1} \bigr)^{p-1-j} \bigl( 1 + \mathcal{O}(\epsilon) \bigr)}
      \\[1ex]
      &
      \to \frac{\detJ_{p-1,p-1}^{00}}{\detJ_{p-2,p-2}^{11}}
      \qquad \text{as $\epsilon \to 0$}
      .
    \end{split}
  \end{equation}
  So these singletons are given by the same formulas as before,
  except that $p$ has been replaced by $p-1$.
  But this is exactly what we want, since what before the killing was the $p$th
  $X$-group from the right is
  after the killing $X$-group number $p-1$ from the right
  (since $p > j+1$):
  \begin{equation}
    X_{(p-1)'}
    = \lim_{\epsilon \to 0} \hat{X}_{p'}
    = \frac{\detJ_{p-1,p-1}^{00}}{\detJ_{p-2,p-2}^{11}}
    .
  \end{equation}
  
\end{itemize}

Next we consider what happens to the leftmost peakon,
which is given by the special formulas~\eqref{eq:even-X-leftmost-singleton}:
\begin{equation}
  \hat{X}_1 = \frac{\hat{\detJ}_{K,K-1}^{00}}{\hat{\detJ}_{K-1,K-2}^{11} + \hat{C} \, \hat{\detJ}_{K-1,K-1}^{10}}
  ,\qquad
  \hat{Q}_1 =
  \frac{\hat{M}}{\hat{L}}
  \left( \frac{\hat{\detJ}_{K-1,K-2}^{11}}{\hat{\detJ}_{K-1,K-1}^{10}} + \hat{C} \right)
  .
\end{equation}
Here we have put hats on the eigenvalue products before the killing,
\begin{equation}
  \hat{L} = \prod_{i=1}^K \lambda_i
  ,\qquad
  \hat{M} = \prod_{j=1}^{K-1} \mu_j
  ,
\end{equation}
since the corresponding products after the killing will be written without hats,
\begin{equation}
  L = \prod_{i=1}^{K-1} \lambda_i
  ,\qquad
  M = \prod_{j=1}^{K-2} \mu_j
  .
\end{equation}

\begin{itemize}
\item 
  Suppose first that $j < K-1$, so that the $\hat{Y}_{(j+1)'}$-peakon
  that we are killing is not adjacent to
  the $\hat{X}_1$-peakon that we are going to investigate:
  \begin{equation*}
    \arrowX \Y
    \X \Y \X \Y \X \Y
    \underbrace{\X \Y \X \X \X \X}
    \Y \Y \Y \Y
    \X \X \X
    \Y
    \X \X \X
    \Y
    \X
    \Y \Y
    \X \X \X \X
    \Y
    \X \X
    \Y \Y \Y
    .
  \end{equation*}
  This means that $X_1$ after the killing will be a singleton, too.

  For $\hat{\detJ}_{K,K-1}^{00}$ and $\hat{\detJ}_{K-1,K-2}^{11}$ we compute
  $q_1 = K-1$ and $q_2 = K-4$, so
  \begin{equation}
    \label{eq:proof-leftmost-epsilon-powers-1}
    \begin{aligned}
      k_1-q_1 &= j-K+1
      ,\\
      k_2-q_2 &= j-K+2
      ,\\
      (k_1-q_1) + (k_2-q_2) - 1 &= 2(j-K+1)
      ,
    \end{aligned}
  \end{equation}
  which (because of the assumption $2 \le j \le K-2$)
  means that the dominant contribution 
  will come from the fourth term in Lemma~\ref{lem:heineintegral-epsilon},
  of order~$\epsilon^{2(j-K+1)}$.
  For $\hat{\detJ}_{K-1,K-1}^{10}$ we find
  $q_1 = K-2$ and $q_2 = K-3$, so
  \begin{equation}
    \label{eq:proof-leftmost-epsilon-powers-2}
    \begin{aligned}
      k_1-q_1 &= j-K+2
      ,\\
      k_2-q_2 &= j-K+1
      ,\\
      (k_1-q_1) + (k_2-q_2) - 1 &= 2(j-K+1)
      ,
    \end{aligned}
  \end{equation}
  which implies that the dominant power is~$\epsilon^{2(j-K+1)}$ here as well.
  However, in the coefficient
  \begin{equation*}
    \bigl( Z \alpha^{q_1} \bigr) \, \bigl( W \beta^{q_2} \bigr) \,
    \left( \frac{\alpha \, \beta}{\alpha +\beta} \right)
  \end{equation*}
  appearing in the the dominant term,
  the factor $\alpha^{q_1}$
  will be different, since $q_1 = K-1$ in the first case and $q_1 = K-2$
  in the second case.
  (But since $\beta=1$, the factor $\beta^{q_2}$ makes no difference.)
  We can compensate for this discrepancy by redefining
  \begin{equation}
    \label{eq:redefine-C}
    \hat{C} = \alpha C = (\theta / \tau_1) \, C
    ,
  \end{equation}
  which gives
  \begin{equation}
    X_1 = \lim_{\epsilon \to 0} \hat{X}_1
    = \frac{\detJ_{K-1,K-2}^{00}}{\detJ_{K-2,K-3}^{11} + \hat{C} \, \alpha^{-1} \, \detJ_{K-2,K-2}^{10}}
    = \frac{\detJ_{K-1,K-2}^{00}}{\detJ_{K-2,K-3}^{11} + C \, \detJ_{K-2,K-2}^{10}}
    ,
  \end{equation}
  as desired.
  The amplitude formula works out as it should too, since
  $\mu_{K-1} = 1/\epsilon$
  and
  $\lambda_K = \alpha/\epsilon$,
  so that $\hat{M}/\hat{L} = M/(\alpha L)$:
  \begin{equation}
    Q_1 = \lim_{\epsilon \to 0} \hat{Q}_1
    =
    \frac{M}{\alpha L}
    \left( \frac{\detJ_{K-2,K-3}^{11}}{\alpha^{-1} \detJ_{K-2,K-2}^{10}} + \hat{C} \right)
    =
    \frac{M}{L}
    \left( \frac{\detJ_{K-2,K-3}^{11}}{\detJ_{K-2,K-2}^{10}} + C \right)
    .
  \end{equation}
  Clearly $C > 0$ if and only if $\hat{C} = \alpha C > 0$,
  so we get no new constraints from this redefinition.
  
\item
  Finally, if $p = j+1$ where $j = K-1$, then we are looking at what happens to the
  $\hat{X}_1$-peakon in a configuration like this,
  where are killing the $\hat{Y}_1$-peakon:
  \begin{equation*}
    \underbrace{\arrowX \Y \X \X \X \X}
    \Y \Y \Y \Y
    \X \X \X
    \Y
    \X \X \X
    \Y
    \X
    \Y \Y
    \X \X \X \X
    \Y
    \X \X
    \Y \Y \Y
    .
  \end{equation*}
  We get the same numbers as in
  \eqref{eq:proof-leftmost-epsilon-powers-1}
  and~\eqref{eq:proof-leftmost-epsilon-powers-2} above,
  but now with $j = K-1$, so that $j-K+1=0$.
  Therefore there are three dominant contributions of order~$\epsilon^0$ in each term,
  some of which are however zero because the lower right index is outside of the range $[0, K-2]$:
  \begin{equation}
    \begin{aligned}
      \hat{\detJ}_{K,K-1}^{00} &
      =
      \underbrace{\detJ_{K,K-1}^{00}}_{=0}
      + \bigl( Z \alpha^{K-1} \bigr) \underbrace{\detJ_{K-1,K-1}^{00}}_{=0}
      + \bigl( Z \alpha^{K-1} \bigr) \, W \, \frac{\alpha}{\alpha + 1} \, \detJ_{K-1,K-2}^{00}
      + \mathcal{O}(\epsilon)
      \\ &
      = 
      \sigma_1 \theta \, \detJ_{K-1,K-2}^{00}
      + \mathcal{O}(\epsilon)
      ,
      \\[1ex]
      \hat{\detJ}_{K-1,K-2}^{11} &
      =
      \detJ_{K-1,K-2}^{11}
      + \bigl( Z \alpha^{K-1} \bigr) \detJ_{K-2,K-2}^{11}
      + \bigl( Z \alpha^{K-1} \bigr) \, W \, \frac{\alpha}{\alpha + 1} \, \detJ_{K-2,K-3}^{11}
      + \mathcal{O}(\epsilon)
      \\ &
      =
      \detJ_{K-1,K-2}^{11}
      + (\tau_1 + \theta) \detJ_{K-2,K-2}^{11}
      + \sigma_1 \theta \, \detJ_{K-2,K-3}^{11}
      + \mathcal{O}(\epsilon)
      \\ &
      =
      M \detJ_{K-1,K-2}^{10}
      + (\tau_1 + \theta) \, M \detJ_{K-2,K-2}^{10}
      + \sigma_1 \theta \, \detJ_{K-2,K-3}^{11}
      + \mathcal{O}(\epsilon)
      ,
      \\[1ex]
      \hat{\detJ}_{K-1,K-1}^{10} &
      =
      \underbrace{\detJ_{K-1,K-1}^{10}}_{=0}
      + W \detJ_{K-1,K-2}^{10}
      + \bigl( Z \alpha^{K-2} \bigr) \, W \, \frac{\alpha}{\alpha + 1} \, \detJ_{K-2,K-2}^{10}
      + \mathcal{O}(\epsilon)
      \\ &
      = 
      \sigma_1 \detJ_{K-1,K-2}^{10}
      + \sigma_1 \tau_1 \, \detJ_{K-2,K-2}^{10}
      + \mathcal{O}(\epsilon)
      ,
    \end{aligned}
  \end{equation}
  where we have used in the middle equation that
  \begin{equation}
    \detJ_{i,K-2}^{rs}
    = \sum_{I \in \binom{[K-1]}{i}}
    \Psi_{I[K-2]} \, \lambda_I^r \, a_I \, \mu_{[K-2]}^s b_{[K-2]}
    =  M^s \, \detJ_{i,K-2}^{r0}
    ,
  \end{equation}
  with $M = \mu_{[K-2]} = \prod_{j=1}^{K-2} \mu_j$.
  This gives
  \begin{equation}
    \begin{split}
      X_{1,1}
      &
      = \lim_{\epsilon \to 0} \hat{X}_1
      =
      \frac{\sigma_1 \theta \, \detJ_{K-1,K-2}^{00}}
      {\left[
          \begin{aligned}
            & M \detJ_{K-1,K-2}^{10}
            + (\tau_1 + \theta) \, M \detJ_{K-2,K-2}^{10}
            + \sigma_1 \theta \, \detJ_{K-2,K-3}^{11}
            \\ &
            + \hat{C} \sigma_1
            \bigl( \detJ_{K-1,K-2}^{10} + \tau_1 \, \detJ_{K-2,K-2}^{10} \bigr)
          \end{aligned}
        \right]}
      \\ &
      =
      \frac{\detJ_{K-1,K-2}^{00}}
      {\detJ_{K-2,K-3}^{11}
        + \left( \dfrac{M}{\sigma_1} + \dfrac{\tau_1 (M + \hat{C} \sigma_1)}{\sigma_1 \theta} \right) \detJ_{K-2,K-2}^{10}
        + \dfrac{M + \hat{C} \sigma_1}{\sigma_1 \theta} \, \detJ_{K-1,K-2}^{10}}
      .
    \end{split}
  \end{equation}
  If we now let $C = \tau_1 (M + \hat{C} \sigma_1)/(\sigma_1 \theta)$,
  i.e., if we make the redefinition
  \begin{equation}
    \label{eq:kill-Y1-redefine-C}
    \hat{C} = \frac{\theta C}{\tau_1} - \frac{M}{\sigma_1}
    ,
  \end{equation}
  then this formula agrees with the desired result,
  namely~\eqref{eq:even-X-leftmost-group-pos}
  with $K-1$ instead of~$K$:
  \begin{equation}
    \label{eq:first-peakon-new-leftmost-X-group}
    \begin{aligned}
      X_{1,1}
      = \lim_{\epsilon \to 0} \hat{X}_1
      &
      = \frac{\detJ_{K-1,K-2}^{00}}{\detJ_{K-2,K-3}^{11} + \frac{M}{\sigma_1} \detJ_{K-2,K-2}^{10} + C \left( \detJ_{K-2,K-2}^{10} + \frac{1}{\tau_1} \detJ_{K-1,K-2}^{10} \right)}
      \\ &
      = \frac{\detJ_{K-1,K-2}^{00}}{\detJ_{K-2,K-3}^{11} + \frac{1}{\sigma_1} \detJ_{K-2,K-2}^{11} + C \left( \detJ_{K-2,K-2}^{10} + \frac{1}{\tau_1} \detJ_{K-1,K-2}^{10} \right)}
    .
    \end{aligned}
  \end{equation}
  Similarly, for the amplitude we obtain~\eqref{eq:even-X-leftmost-group-amp}
  with $K-1$ instead of~$K$:
  \begin{equation}
    \begin{split}
      Q_{1,1}
      &
      = \lim_{\epsilon \to 0} \hat{Q}_1
      = \lim_{\epsilon \to 0}
      \frac{\hat{M}}{\hat{L}}
      \left( \frac{\hat{\detJ}_{K-1,K-2}^{11}}{\hat{\detJ}_{K-1,K-1}^{10}} + \hat{C} \right)
      \\ &
      = \frac{M}{L \, \theta / \tau_1}
      \left(
        \frac{M \detJ_{K-1,K-2}^{10}
          + (\tau_1 + \theta) \, M \detJ_{K-2,K-2}^{10}
          + \sigma_1 \theta \, \detJ_{K-2,K-3}^{11}
        }{\sigma_1 \bigl( \detJ_{K-1,K-2}^{10} + \tau_1 \, \detJ_{K-2,K-2}^{10} \bigr)}
        + \hat{C}
      \right)
      \\ &
      = \frac{M \, \tau_1}{L \, \theta}
      \left(
        \frac{\theta \, \bigl( M \detJ_{K-2,K-2}^{10} + \sigma_1 \, \detJ_{K-2,K-3}^{11} \bigr)
        }{\sigma_1 \bigl( \detJ_{K-1,K-2}^{10} + \tau_1 \, \detJ_{K-2,K-2}^{10} \bigr)}
        + \frac{M}{\sigma_1}
        + \hat{C}
      \right)
      \\ &
      =
      \frac{M}{L}
      \left( \frac{ \tau_1 \bigl( \detJ_{K-2,K-2}^{11} + \sigma_1 \detJ_{K-2,K-3}^{11} \bigr)}{\sigma_1 \bigl( \detJ_{K-1,K-2}^{10} + \tau_1 \detJ_{K-2,K-2}^{10} \bigr)} + C \right)
      \\ &
      =
      \frac{M}{L}
      \left( \frac{\detJ_{K-2,K-3}^{11} + \frac{1}{\sigma_1} \detJ_{K-2,K-2}^{11}}{\detJ_{K-2,K-2}^{10} + \frac{1}{\tau_1} \detJ_{K-1,K-2}^{10}} + C \right)
      .
    \end{split}
  \end{equation}
  The condition $\hat{C} > 0$ clearly implies that we must require $C>0$
  in the redefinition~\eqref{eq:kill-Y1-redefine-C}, to begin with.
  However, we also need the ghost parameter~$\theta$ to satisfy
  $\frac{\tau_1 M}{\sigma_1 C} < \theta$.
  If $\hat{N}_2^X \ge 2$ (i.e., $N_1^X \ge 3$),
  then we also have the condition $\theta < \tau_2$
  coming from the redefinition $\hat{\tau}_1 = \tau_2 - \theta$
  in~\eqref{eq:redefine-parameters-sigma-tau},
  and the combination of these two inequalities for~$\theta$ leads to the stronger
  constraint~\eqref{eq:constraint-C-general} for the parameter~$C$:
  \begin{equation}
    \frac{\tau_1 M}{\sigma_1 C} < \tau_2
    .
  \end{equation}
  But if $\hat{N}_2^X = 1$ (i.e., $N_1^X = 2$),
  then there is no $\tau_2$, so then the constraint is just $C > 0$.
\end{itemize}

\subsubsection{What happens to $Y$-singletons}
\label{sec:what-happens-Y-singletons}

The computations for singleton $Y$-groups are similar, and we omit the details,
except for the rightmost $Y$-peakon which is slightly exceptional when the second rightmost
$Y$-peakon is killed:
if $j=1$ and $\hat{N}_{K}^Y = 1$,
perform the additional substitution
\begin{equation}
  \label{eq:redefine-D}
  \hat{D} = D - \sigma_{K-1,1}^X
  ,
\end{equation}
to get
\begin{equation}
  Y_{K-1}
  = \lim_{\epsilon \to 0} \hat{Y}_{K}
  = \lim_{\epsilon \to 0} \left( \hat{\detJ}^{00}_{11} + \hat{D} \hat{\detJ}^{00}_{10} \right)
  = \left( \detJ_{11}^{00} + \sigma_{K-1,1}^X \detJ_{10}^{00} \right)
  + \hat{D} \detJ_{10}^{00}
  = \detJ_{11}^{00} + D \detJ_{10}^{00}
  .
\end{equation}
Let $N = \hat{N}_{K}^X $,
so that the rightmost $X$-group contains $N+1$ peakons after the killing.
If $N=1$, then the requirement $\hat{D} > 0$ implies that $D > \sigma_{K-1,1}^X$,
whereas if $N \ge 2$,
we find inductively that the condition
$\hat{D} > \hat{\sigma}_{K,N-1}^X$,
together with the redefinition
$\hat{\sigma}_{K,N-1}^X = \sigma_{K-1,N}^X - \sigma_{K-1,1}^X$
from~\eqref{eq:redefine-parameters-sigma-tau},
implies that
$D > \sigma_{K-1,N}^X$.
This is the origin of the constraint~\eqref{eq:constraint-last-sigma-D-even}.

We may remark that Example~\ref{ex:proof-technique} contained an illustration of
this substitution in the case~$N=1$.

The computations above reflect the fact
mentioned in Remark~\ref{rem:D-is-extra-tau},
that the parameter~$D$
appears in the solution formulas as if it were an
``extra $\tau$-parameter'' for the rightmost $Y$-group.
If $\hat{N}_K^Y \ge 2$, then we would redefine
$\hat{\tau}_{K,1}^Y = \tau_{K-1,1}^Y - \sigma_{K-1,1}^X$
as described in~\eqref{eq:funny-redefinition-rightmost}.
But if $\hat{N}_K^Y = 1$, then there is no $\hat{\tau}_{K,1}^Y$ to redefine,
and instead we perform the corresponding redefinition on $\hat{D}$,
which steps in to play the role of that group's first (and only)
$\tau$-parameter.

\subsection{How to kill an $X$-singleton}
\label{sec:kill-X-even}

We have shown in detail above how to kill any $Y$-singleton
(except the rightmost one which is treated in Section~\ref{sec:proofs-odd})
and how to investigate what then becomes of the solution formulas for the other peakons.
We will now outline how to deal with the $X$-singletons.
The verifications are similar to those in Section~\ref{sec:kill-Y-even},
and we only indicate some exceptional cases.

\paragraph{Killing a typical $X$-peakon.}

Since we are always assuming that the configuration starts with an $X$-peakon,
we never need to kill the leftmost $X$-peakon.
Therefore we shall consider the killing of a singleton constituting
the $j$th $X$-group from the right,
i.e., the $\hat{X}_{j'}$-peakon
where $j' = K+1-j$ and $1 \le j \le K-1$.
This will join the singleton $\hat{Y}_{(j+1)'}$-peakon to the
$\hat{Y}_{j'}$-group containing $N = \hat{N}_{j'}^Y$ peakons,
which after renumbering will be the $Y_{K-j}$-group
containing $N+1$ peakons,
i.e.,
the $j$th $Y$-group from the right in a configuration with
$(K-1)+(K-1)$ groups.
The rightmost and the second leftmost $X$-peakons are special,
so we assume first that $2 \le j \le K-2$.

To achieve this, we choose the parameters
in~\eqref{eq:substitution-kill-peakon-even} as follows:
\begin{equation}
  \label{eq:substitution-kill-X-even}
  \begin{gathered}
    k_1 = j-1
    ,\qquad
    k_2 = j-2
    ,
    \\[1ex]
    \alpha = 1
    ,\qquad
    \beta = \frac{\theta}{\tau_1}
    ,\qquad
    Z= \sigma_1
    ,\qquad
    W = \frac{(\beta + 1) \, \tau_1}{\beta^{j-2}}
    = (\tau_1 + \theta) (\tau_1 / \theta)^{j-2}
    .
  \end{gathered}
\end{equation}
If $N \ge 2$, we also redefine the internal parameters of the $Y$-group being enlarged:
\begin{equation}
  \label{eq:redefine-parameters-sigma-tau-2}
  \begin{aligned}
    \hat{\tau}_1 &= \tau_2 - \theta
    \\
    \hat{\tau}_i &= \tau_{i+1}
    ,&& 2 \le i \le N-1
    ,
    \\
    \hat{\sigma}_i &= \sigma_{i+1} - \sigma_1
    ,&& 1 \le i \le N-1
    ,
  \end{aligned}
\end{equation}
where
\begin{equation*}
  \hat{\tau}_i = \hat{\tau}_{K+1-j,i}^Y
  ,\qquad
  \hat{\sigma}_i = \hat{\sigma}_{K+1-j,i}^Y
  ,\qquad
  \tau_i = \tau_{K-j,i}^Y
  ,\qquad
  \sigma_i = \sigma_{K-j,i}^Y
  .
\end{equation*}
Then we substitute in the solution formulas and let $\epsilon \to 0$.

If $j \ge 2$ and the $\hat{X}_{(j-1)'}$-group
(the $X$-group to the right of the $X$-singleton that we are killing)
contains two or more peakons, then like in~\eqref{eq:funny-redefinition} above
we must also redefine its first $\tau$-parameter as
\begin{equation}
  \label{eq:funny-redefinition-2}
  \hat{\tau}_{K+1-(j-1),1}^X = \tau_{(K-1)+1-(j-1),1}^X - \sigma_{(K-1)+1-j,1}^Y
\end{equation}
to prevent the $\sigma_1$-parameter from the intermediate $Y$-group
from contaminating the formulas.
In the same way as for~\eqref{eq:funny-redefinition},
this is the origin of the second constraint
in~\eqref{eq:constraint-last-sigma-first-tau},
``the last~$\sigma$ in a $Y$-group is less than the first~$\tau$
in the next $X$-group''.

\paragraph{Killing the second leftmost $X$-peakon.}

To kill the $\hat{X}_2$-peakon, follow the same steps as for a typical $X$-peakon above,
with $j=K-1$.
This leads to
\begin{equation}
  \begin{split}
    X_1 = \lim_{\epsilon \to 0} \hat{X}_1
    &
    = \lim_{\epsilon \to 0} \frac{\hat{\detJ}^{00}_{K,K-1}}{\hat{\detJ}^{11}_{K-1,K-2} + \hat{C} \, \hat{\detJ}^{10}_{K-1,K-1}}
    \\ &
    = \frac{\sigma_1 \overbrace{\detJ_{K-1,K-1}^{00}}^{=0} + \sigma_1 \tau_1 \detJ_{K-1,K-2}^{00}}{\sigma_1 \underbrace{\detJ_{K-2,K-2}^{11}}_{= M \detJ_{K-2,K-2}^{10} } + \sigma_1 \tau_1 \detJ_{K-2,K-3}^{11} + \hat{C} \, \sigma_1 \theta \detJ_{K-2,K-2}^{10}}
    \\ &
    = \frac{\detJ_{K-1,K-2}^{00}}{\detJ_{K-2,K-3}^{11} + \frac{\hat{C} \, \theta + M}{\tau_1} \detJ_{K-2,K-2}^{10}}
    ,
  \end{split}
\end{equation}
where $M = \mu_1 \dotsm \mu_{K-2}$,
so that we must also let $C = (\hat{C} \, \theta + M)/\tau_1$,
or equivalently
\begin{equation}
  \hat{C} = (C \, \tau_1 - M) / \theta
  ,
\end{equation}
in order to get the correct formula for~$X_1$
(and also for $Q_1$ in a similar way),
free from $\tau_1 = \tau_{1,1}^Y$ and~$\theta$.
This redefinition is the reason for the constraint~\eqref{eq:constraint-C-simpler},
namely $M < C \, \tau_1$.

\paragraph{Killing the rightmost $X$-peakon.}

To kill the $\hat{X}_K$-peakon, follow the same steps as above, with $j=1$,
but if the rightmost group is a singleton ($\hat{N}_K^Y = 1$),
also set
\begin{equation}
  \hat{D} = D - \theta
  ,
\end{equation}
so that $\theta$ will not appear in the formulas for
the rightmost peakon after the killing:
\begin{equation}
  \begin{split}
    Y_{K-1,2} = \lim_{\epsilon \to 0} \hat{Y}_{K}
    &
    = \lim_{\epsilon \to 0} \bigl( \hat{\detJ}^{00}_{11} + \hat{D} \, \hat{\detJ}^{00}_{10} \bigr)
    \\ &
    = \detJ_{11}^{00}
    + (\tau_1 + \theta) \detJ_{10}^{00}
    + \sigma_1 \theta
    + \hat{D} \bigl( \detJ_{10}^{00} + \sigma_1 \bigr)
    \\ &
    = \detJ_{11}^{00}
    + (\tau_1 + \theta + \hat{D}) \detJ_{10}^{00}
    + \sigma_1 (\theta + \hat{D})
    \\ &
    = \detJ_{11}^{00}
    + (\tau_1 + D) \detJ_{10}^{00}
    + \sigma_1 D
    .
  \end{split}
\end{equation}
This redefinition does not lead to any new constraint, only $D > 0$ as usual.
(But $\theta$ must satisfy $0 < \theta < D$, which is relevant when looking at
the characteristic curves; cf.~\eqref{eq:char-within-rightmost}.)

For $\hat{N}_K^Y \ge 2$, the redefinition $\hat{\tau}_{K,1}^Y = \tau_{K-1,2}^Y - \theta$
in~\eqref{eq:redefine-parameters-sigma-tau-2}
already does the job of absorbing the~$\theta$
in the formulas for the rightmost $Y$-group,
so in that case we just let $\hat{D} = D$.
Like in Section~\ref{sec:what-happens-Y-singletons},
we see here how the parameter~$D$ acts as an additional
``last $\tau$-parameter'' for the rightmost $Y$-group:
if $\hat{N}_K^Y = 1$, then there is no $\hat{\tau}_{K,1}^Y$ to redefine,
but we perform the corresponding redefinition on $\hat{D}$ instead.

\begin{remark}
  For $K=2$ the rightmost and the second leftmost $X$-peakons coincide,
  so then if $\hat{N}_2^Y=1$ we must perform both of the special substitutions described above.
  This only occurs when killing~$\hat{X}_2$ in the $2+2$ interlacing configuration:
  \begin{equation*}
    \X \Y \arrowX \Y
  \end{equation*}
\end{remark}

\section{Asymptotics for the even  case}
\label{sec:asymptotics-even}

In this section, we study the limiting behavior as $t \to \pm \infty$
of peakon solutions of the Geng--Xue equation, in the even case ($K+K$ groups).
The explicit formulas describing these solutions were given in Section~\ref{sec:solutions-even},
in terms of exponential functions with different growth rates,
and determining the asymptotics is simply a matter of identifying
the dominant exponential terms in these formulas
as $t \to \pm \infty$.

We will see that asymptotically each peakon trajectory approaches a straight line,
and also the logarithms of the amplitudes are asymptotically straight lines,
as was illustrated in the examples in Section~\ref{sec:examples-groups-even}.

Since singletons are given by the same formulas as in the interlacing case,
the asymptotics will agree with the interlacing case
\cite[Theorems 9.3 and~9.8]{lundmark-szmigielski:2017:GX-dynamics-interlacing};
we first recall these results here for completeness,
in Theorem~\ref{thm:asymptotics-singletons-even}.
Our new result concern the asymptotics for the groups with $N \ge 2$,
given in Theorem~\ref{thm:asymptotics-positions-even} for the positions
and Theorem~\ref{thm:asymptotics-amplitudes-even} for the amplitudes.

As usual, we will assume that the eigenvalues are numbered in increasing order:
\begin{equation}
  \label{eq:eigenvalues-ordered-even}
  0 < \lambda_1 < \lambda_2 < \dots < \lambda_K
  , \qquad
  0 < \mu_1 < \mu_2 < \dots < \mu_{K-1}
  .
\end{equation}
We remind the reader that all notation has been defined in
Sections~\ref{sec:interlacing-review} and~\ref{sec:more-notation}.
In particular, the products of the eigenvalues are denoted by
\begin{equation*}
  L = \lambda_{[1,K]} = \lambda_1 \dotsm \lambda_K
  ,\qquad
  M = \mu_{[1,K-1]} = \mu_1 \dotsm \mu_{K-1}
  .
\end{equation*}

\begin{theorem}
  \label{thm:asymptotics-singletons-even}

  Any singletons in a solution with $K+K$ groups,
  and in particular all peakons in the $K+K$ interlacing solution,
  where $K \ge 1$ and the eigenvalues are ordered as in~\eqref{eq:eigenvalues-ordered-even},
  satisfy the following asymptotic formulas.

  \begin{itemize}
  \item  As $t \to +\infty$, if $K \ge 2$:

    All $X$-singletons except the leftmost one
    ($j'=K+1-j$ with $1 \le j \le K-1$):
    \begin{equation}
      \begin{split}
        x_{j'}(t) &=
        \frac{t}{2} \left( \frac{1}{\lambda_{j}} + \frac{1}{\mu_j} \right)
        + \frac12 \ln \left( \frac{2 a_{j}(0) \, b_j(0) \, \Psi_{[1,j][1,j]}}{\lambda_{[1,j-1]} \, \mu_{[1,j-1]} \, \Psi_{[1,j-1][1,j-1]}} \right)
        + o(1)
        ,
        \\[1em]
	\ln m_{j'}(t) &=
        \frac{t}{2} \left( \frac{1}{\lambda_{j}} - \frac{1}{\mu_j} \right)
        + \frac12 \ln \left( \frac{2 a_{j}(0) \, b_j(0) \, \Psi_{[1,j][1,j]}}{\lambda_{[1,j-1]} \, \mu_{[1,j-1]} \, \Psi_{[1,j-1][1,j-1]}} \right)
        \\
        &\quad
        +\ln \left( \frac{\mu^2_{[1,j-1]} \Psi_{[1,j][1,j-1]}}{2 b_j(0) \, \lambda_{[1,j]} \, \Psi_{[1,j][1,j]}} \right)
        + o(1)
        .
      \end{split}
    \end{equation}
    The leftmost $X$-singleton:
    \begin{equation}
      \begin{split}
        x_1(t) &=
        \frac{t}{2} \left( \frac{1}{\lambda_{K}} \right)
        + \frac12 \ln \left( \frac{2 a_K(0) \, \Psi_{[1,K][1,K-1]}}{C \, \lambda_{[1,K-1]} \, \Psi_{[1,K-1][1,K-1]}} \right)
        + o(1)
        ,
        \\[1em]
	\ln m_1(t) &=
        \frac{t}{2} \left( \frac{1}{\lambda_{K}} \right)
        + \frac12 \ln \left( \frac{2 a_K(0) \, \Psi_{[1,K][1,K-1]}}{C \, \lambda_{[1,K-1]} \, \Psi_{[1,K-1][1,K-1]}} \right)
	+ \ln \left( \frac{M C}{2L} \right)
        + o(1)
        .
      \end{split}
    \end{equation}
    All $Y$-singletons except the rightmost one ($j'=K+1-j$ with $2 \le j \le K$):
    \begin{equation}
      \begin{split}
        y_{j'}(t) &=
        \frac{t}{2} \left( \frac{1}{\mu_{j-1}} + \frac{1}{\lambda_j} \right)
        + \frac12 \ln \left( \frac{2 a_j(0) \, b_{j-1}(0) \, \Psi_{[1,j][1,j-1]}}{\lambda_{[1,j-1]} \, \mu_{[1,j-2]} \, \Psi_{[1,j-1][1,j-2]}} \right)
        + o(1)
        ,
        \\[1ex]
        \ln n_{j'}(t) &=
        \frac{t}{2} \left( \frac{1}{\mu_{j-1}} - \frac{1}{\lambda_j} \right)
        + \frac12 \ln \left( \frac{2 a_j(0) \, b_{j-1}(0) \, \Psi_{[1,j][1,j-1]}}{\lambda_{[1,j-1]} \, \mu_{[1,j-2]} \, \Psi_{[1,j-1][1,j-2]}} \right)
        \\ & \quad
        + \ln \left( \frac{\lambda_{[1,j-1]}^2 \, \Psi_{[1,j-1][1,j-1]}}{2 a_j(0) \, \mu_{[1,j-1]} \, \Psi_{[1,j][1,j-1]}} \right)
        + o(1)
        .
      \end{split}
    \end{equation}
    The rightmost $Y$-singleton:
    \begin{equation}
      \begin{split}
        y_K(t) &=
        \frac{t}{2} \left( \frac{1}{\lambda_1} + \frac{1}{\mu_1} \right)
        + \frac12 \ln \bigl( 2 a_1(0) \, b_1(0) \, \Psi_{\{1\}\{1\}} \bigr)
        + o(1)
        ,
        \\[1ex]
        \ln n_K(t) &=
        \frac{t}{2} \left( \frac{1}{\mu_1} - \frac{1}{\lambda_1} \right)
        + \frac12 \ln \bigl( 2 a_1(0) \, b_1(0) \, \Psi_{\{1\}\{1\}} \bigr)
        + \ln \left( \frac{1}{2 a_1(0) \, \Psi_{\{1\}\emptyset}} \right)
        + o(1)
        .
      \end{split}
    \end{equation}
    
  \item As $t \to -\infty$, if $K \ge 2$:

    All $X$-singletons except the leftmost one ($2 \le j \le K$):
    \begin{equation}
      \begin{split}
        x_j(t) &=
        \frac{t}{2} \left( \frac{1}{\lambda_{j}} + \frac{1}{\mu_{j-1}} \right)
        + \frac12 \ln \left( \frac{2 a_j(0) \, b_{j-1}(0) \, \Psi_{[j,K][j-1,K-1]}}{\lambda_{[j+1,K]} \, \mu_{[j,K-1]} \, \Psi_{[j+1,K][j,K-1]}} \right)
        + o(1)
        ,
     	\\[1ex]
        \ln m_j(t) &=
        \frac{t}{2} \left( \frac{1}{\lambda_{j}} - \frac{1}{\mu_{j-1}} \right)
        + \frac12 \ln \left( \frac{2 a_j(0) \, b_{j-1}(0) \, \Psi_{[j,K][j-1,K-1]}}{\lambda_{[j+1,K]} \, \mu_{[j,K-1]} \, \Psi_{[j+1,K][j,K-1]}} \right)
    	\\ & \quad
    	+ \ln \left( \frac{\mu_{[j,K-1]}^2 \Psi_{[j,K][j,K-1]}}{2 b_{j-1}(0) \, \lambda_{[j,K]} \, \Psi_{[j,K][j-1,K-1]}} \right)
    	+ o(1)
        .
      \end{split}
    \end{equation}
    The leftmost $X$-singleton:
    \begin{equation}
      \begin{split}
        x_1(t) &=
        \frac{t}{2} \left( \frac{1}{\lambda_1} + \frac{1}{\mu_1} \right)
        + \frac12 \ln \left( \frac{2 a_1(0) \, b_1(0) \, \Psi_{[1,K][1,K-1]}}{\lambda_{[2,K]} \, \mu_{[2,K-1]} \, \Psi_{[2,K][2,K-1]}} \right)
        + o(1)
        ,
        \\[1ex]
	\ln m_1(t) &=
        \frac{t}{2} \left( \frac{1}{\lambda_1} - \frac{1}{\mu_1} \right)
        + \frac12 \ln \left( \frac{2 a_1(0) \, b_1(0) \,  \Psi_{[1,K][1,K-1]}}{\lambda_{[2,K]} \, \mu_{[2,K-1]} \, \Psi_{[2,K][2,K-1]}} \right)
        \\ & \quad
        + \ln \left( \frac{\mu_{[2,K-1]} \, M \, \Psi_{[2,K][2,K-1]}}{2 b_1(0) \, L \, \Psi_{[2,K][1,K-1]}} \right)
        + o(1)
        .
      \end{split}
    \end{equation}
    All $Y$-singletons except the rightmost one:
    \begin{equation}
      \begin{split}
     	y_j(t) &=
        \frac{t}{2} \left( \frac{1}{\lambda_{j}} + \frac{1}{\mu_{j}} \right)
        + \frac12 \ln \left( \frac{2 a_j(0) \, b_{j}(0) \, \Psi_{[j,K][j,K-1]}}{\lambda_{[j+1,K]} \, \mu_{[j+1,K-1]} \, \Psi_{[j+1,K][j+1,K-1]}} \right)
        + o(1)
        ,
     	\\[1ex]
        \ln n_j(t) &=
        \frac{t}{2} \left( \frac{1}{\mu_j} - \frac{1}{\lambda_{j}} \right)
        + \frac12 \ln \left( \frac{2 a_j(0) \, b_{j}(0) \, \Psi_{[j,K][j,K-1]}}{\lambda_{[j+1,K]} \, \mu_{[j+1,K-1]} \, \Psi_{[j+1,K][j+1,K-1]}} \right)
     	\\ & \quad
     	+ \ln \left( \frac{\lambda_{[j+1,K]}^2 \Psi_{[j+1,K][j,K-1]}}{2 a_{j}(0) \, \mu_{[j,K-1]} \, \Psi_{[j,K][j,K-1]}} \right)
     	+ o(1)
        .
      \end{split}
    \end{equation}
    The rightmost $Y$-singleton:
    \begin{equation}
      \begin{split}
        y_K(t) &= \frac{t}{2} \left( \frac{1}{\lambda_{K}} \right)
        + \frac12  \ln \bigl( 2 D \, a_K(0) \bigr)
        + o(1)
        ,
        \\[1ex]
	\ln n_K(t) &= \frac{t}{2} \left( \frac{-1}{\lambda_{K}} \right)
        + \frac12  \ln \bigl( 2 D \, a_K(0) \bigr)
        + \ln \left( \frac{1}{2 a_K(0)} \right)
        + o(1)
        .
      \end{split}
    \end{equation}

  \item The special case $K=1$,
    where the formulas are exact for all~$t$, not just asymptotically:

    The only $X$-singleton:
    \begin{equation}
      \begin{split}
        x_1(t) &=
        \frac{t}{2} \left( \frac{1}{\lambda_1} \right)
        + \frac12 \ln \left( \frac{2 a_1(0)}{C} \right)
        ,
        \\
       	\ln m_1(t) &=
        \frac{t}{2} \left( \frac{1}{\lambda_1} \right)
        + \frac12 \ln \left( \frac{2 a_1(0)}{C} \right)
        + \ln \left( \frac{C}{\lambda_1} \right)
       	.
      \end{split}
    \end{equation}
    The only $Y$-singleton:
    \begin{equation}
      \begin{split}
        y_1 &=
        \frac{t}{2} \left( \frac{1}{\lambda_1} \right)
        + \frac12 \ln \bigl( 2 D \, a_1(0) \bigr)
        ,
        \\
	\ln n_1 &=
        \frac{t}{2} \left( \frac{-1}{\lambda_1} \right)
        + \frac12 \ln \bigl( 2 D \, a_1(0) \bigr)
        + \ln \left( \frac{1}{2 a_1(0)} \right)
        .
      \end{split}
    \end{equation}
  \end{itemize}
\end{theorem}

\begin{remark}
  In the formulas above, and also in the theorems which follow below,
  one can expand the definition of $\Psi_{IJ}$ and cancel all common
  factors from the ratios involving $\Delta_I^2$, $\twin \Delta_J^2$
  and~$\Gamma_{IJ}$ (see
  \cite{lundmark-szmigielski:2017:GX-dynamics-interlacing}, at the end
  of the proof of Theorem 9.3), for example
  \begin{equation*}
    \begin{split}
      &
      \frac{\Psi_{[1,j][1,j-1]}}{\lambda_{[1,j-1]} \, \mu_{[1,j-2]} \, \Psi_{[1,j-1][1,j-2]}}
      \\[1ex]
      &=
      \left( \prod_{r=1}^{j-1} \frac{(\lambda_r - \lambda_j)^2}{\lambda_r (\lambda_r + \mu_{j-1})} \right)
      \left( \prod_{s=1}^{j-2} \frac{(\mu_s - \mu_{j-1})^2}{(\lambda_j + \mu_s) \mu_s} \right)
      \frac{1}{\lambda_j + \mu_{j-1}}
      .
    \end{split}
  \end{equation*}
  However, in order to save space we have not done that.
\end{remark}

\begin{remark}
  We are using the notation $[1,j]$ instead of the shorter form $[j]$
  for the integer interval $[1,2,\dots,j]$, in order to make the
  formulas for $t \to +\infty$ more similar to the formulas for $t \to -\infty$
  where intervals of the form $[j,K]$ and $[j,K-1]$ appear.
\end{remark}

Now we present the asymptotics for the non-singleton groups.
We begin with the positions $x_{j,i}$ and~$y_{j,i}$.
Remark~\ref{rem:how-to-read-asymptotics-pos} after the theorem explains the most interesting features;
see also the examples in Section~\ref{sec:examples-groups-even}.

\begin{theorem}
  \label{thm:asymptotics-positions-even}
  
  The positions for non-singleton groups in the even case,
  with the eigenvalues ordered as in~\eqref{eq:eigenvalues-ordered-even},
  satisfy the following asymptotic formulas.

  \begin{itemize}
  \item As $t \to +\infty$:

    All $X$-groups except the leftmost one ($j'=K+1-j$ with $1 \le j \le K-1$), if $K \ge 2$:
    \begin{equation}
      \begin{aligned}
        x_{j',i}(t) &=
        \frac{t}{2} \left( \frac{1}{\lambda_{j+1}} + \frac{1}{\mu_j} \right)
        + \frac12 \ln \left( \frac{2 a_{j+1}(0) \, b_j(0) \, \Psi_{[1,j+1][1,j]}}{\lambda_{[1,j]} \, \mu_{[1,j-1]} \, \Psi_{[1,j][1,j-1]}} \right)
        + o(1)
        ,
        \\
        & \quad \text{for $1 \le i \le N_{j'}-1$}
        ,
        \\[1ex]
        x_{j',N_{j'}}(t) &= \frac{t}{2} \left( \frac{1}{\lambda_j} + \frac{1}{\mu_{j}} \right) + \frac12 \ln \left( \frac{ 2 a_j(0) \, b_{j}(0) \Psi_{[1,j][1,j]}}{\Psi_{[1,j-1][1,j-1]} \lambda_{[1,j-1]} \, \mu_{[1,j-1]}} \right)
        + o(1)
        .
      \end{aligned}
    \end{equation}
    The leftmost $X$-group:
    \begin{equation}
      \begin{aligned}
        x_{1,1}(t) &= \frac12 \ln \left( \frac{2 \tau_1}{C L} \right)
        + o(1)
        ,
        \\[1ex]
        x_{1,i}(t) &= \frac12 \ln \left( \frac{2 S_i}{L M} \right)
        + o(1)
        ,\quad \text{for $2 \le i \le N_1-1$}
        ,
        \\[1ex]
        x_{1,N_1}(t) &=\frac{t}{2} \left( \frac{1}{\lambda_K} \right) + \frac12 \ln \left( \frac{2 \sigma_{N_1-1} \, a_K(0) \, \Psi_{[1,K][1,K-1]}}{\lambda_{[1,K-1]} \, M \, \Psi_{[1,K-1][1,K-1]}} \right)
        + o(1)
        .
      \end{aligned}
    \end{equation}
    All $Y$-groups except the leftmost and rightmost ones
    ($j'=K+1-j$ with $2 \le j \le K-1$), if $K \ge 3$:
    \begin{equation}
      \label{eq:asy-even-posinf-Y-typical-group}
      \begin{aligned}
        y_{j',i}(t) &= \frac{t}{2} \left( \frac{1}{\lambda_j} + \frac{1}{\mu_j} \right) + \frac12 \ln \left( \frac{2 a_j(0) \, b_j(0) \, \Psi_{[1,j][1,j]}}{\lambda_{[1,j-1]} \, \mu_{[1,j-1]} \, \Psi_{[1,j-1][1,j-1]}} \right)
        + o(1)
        ,
        \\
        & \quad \text{for $1 \le i \le N_{j'}-1$}
        ,
        \\[1ex]
        y_{j',N_{j'}}(t) &= \frac{t}{2}  \left( \frac{1}{\lambda_j}+ \frac{1}{\mu_{j-1}} \right)+ \frac12 \ln \left( \frac{ 2 a_j(0) \, b_{j-1}(0) \, \Psi_{[1,j][1,j-1]}}{\lambda_{[1,j-1]} \, \mu_{[1,j-2]} \, \Psi_{[1,j-1][1,j-2]}} \right)
        + o(1)
        .
      \end{aligned}
    \end{equation}
    The leftmost $Y$-group, if $K \ge 2$:
    \begin{equation}
      \begin{aligned}
        y_{1,i}(t) &=
        \frac{t}{2} \left( \frac{1}{\lambda_{K}} \right)
        + \frac12 \ln \left( \frac{2 T_i a_{K}(0) \, \Psi_{[1,K][1,K-1]}}{\lambda_{[1,K-1]} \, M \, \Psi_{[1,K-1][1,K-1]}} \right)
        + o(1)
        ,
        \\ &
        \quad \text{for $1 \le i \le N_1-1$}
        ,
        \\[1ex]
        y_{1,N_1}(t) &=
        \frac{t}{2} \left( \frac{1}{\lambda_K} + \frac{1}{\mu_{K-1}} \right)
        + \frac12 \ln \left( \frac{2 a_K(0) \, b_{K-1}(0) \, \Psi_{[1,K][1,K-1]}}{\lambda_{[1,K-1]} \, \mu_{[1,K-2]} \, \Psi_{[1,K-1][1,K-2]}} \right)
        + o(1)
        .
      \end{aligned}
    \end{equation}
    The rightmost $Y$-group, if $K \ge 2$:
    \begin{equation}
      \begin{aligned}
        y_{K,i}(t) &=
        \frac{t}{2} \left( \frac{1}{\lambda_1} + \frac{1}{\mu_1} \right)
        + \frac12 \ln \bigl( 2 a_1(0) \, b_1(0) \, \Psi_{\{1\}\{1\}} \bigr)
        + o(1)
        ,\quad \text{for $1 \le i \le N_{K}$}
        .
      \end{aligned}
    \end{equation}
    The only $Y$-group, if $K=1$:
    \begin{equation}
      \begin{aligned}
        y_{1,i}(t) &=
        \frac{t}{2} \left( \frac{1}{\lambda_1} \right)
        + \frac12 \ln \bigl( 2 T_i \, a_1(0) \bigr)
        + o(1)
        ,\quad \text{for $1 \le i \le N_1-1$}
        ,
        \\
        y_{1,N_1}(t) &=
        \frac{t}{2} \left( \frac{1}{\lambda_1} \right)
        + \frac12 \ln \bigl( 2 (T_{N_1-1} + D) \,  a_1(0) \bigr) + o(1)
        ,
      \end{aligned}
    \end{equation}
    where the formula
    $y_{1,1}(t) = \frac{t}{2\lambda_1} + \frac12 \ln \bigl( 2 \tau_1 a_1(0) \bigr)$
    is exact and does not need the $o(1)$ term.

  \item As $t \to -\infty$:

    All $X$-groups except the leftmost and rightmost ones
    ($2 \le j \le K-1$), if $K \ge 3$:
    \begin{equation}
      \begin{aligned}
        x_{j,1}(t) &=
        \frac{t}{2} \left( \frac{1}{\lambda_{j}} + \frac{1}{\mu_{j-1}} \right)
        + \frac12 \ln \left( \frac{2 a_{j}(0) \, b_{j-1}(0) \, \Psi_{[j,K][j-1,K-1]}}{\lambda_{[j+1,K]} \, \mu_{[j,K-1]} \, \Psi_{[j+1,K][j,K-1]}} \right)
        + o(1)
        ,
        \\[1ex]
        x_{j,i}(t) &=
        \frac{t}{2} \left( \frac{1}{\lambda_{j}} + \frac{1}{\mu_{j}} \right)
        + \frac12 \ln \left( \frac{2 a_{j}(0) \, b_{j}(0) \, \Psi_{[j,K][j,K-1]}}{\lambda_{[j+1,K]} \, \mu_{[j+1,K-1]} \, \Psi_{[j+1,K][j+1,K-1]}} \right)
        + o(1)
        ,
        \\
        & \quad\text{for $2 \le i \le N_{j}$}
        .
      \end{aligned}
    \end{equation}
    The leftmost $X$-group, if $K \ge 2$:
    \begin{equation}
      \begin{aligned}
        x_{1,i}(t) &= \frac{t}{2} \left( \frac{1}{\lambda_1} + \frac{1}{\mu_1} \right) + \frac12 \ln \left( \frac{2 a_1(0) \, b_1(0) \, \Psi_{[1,K][1,K-1]}}{\lambda_{[2,K]} \, \mu_{[2,K-1]} \, \Psi_{[2,K][2,K-1]}} \right)
        + o(1)
        ,
        \\
        &
        \quad\text{for $1 \le i \le N_1$}
	.
      \end{aligned}
    \end{equation}
    The rightmost $X$-group, if $K \ge 2$:
    \begin{equation}
      \begin{aligned}
        x_{K,1}(t) &= \frac{t}{2} \left( \frac{1}{\lambda_{K}} + \frac{1}{\mu_{K-1}} \right)
        + \frac12 \ln \bigl( 2 a_K(0) \, b_{K-1}(0) \, \Psi_{\{K\}\{K-1\}} \bigr)
        + o(1)
        ,
        \\[1ex]
        x_{K,i}(t) &= \frac{t}{2} \left( \frac{1}{\lambda_{K}} \right)
        + \frac12 \ln \left( \frac{2 S_i \, a_K(0)}{T_i} \right)
        + o(1)
        ,
        \quad\text{for $2 \le i \le N_K-1$}
        ,
        \\[1ex]
        x_{K, N_K}(t) &= \frac{t}{2} \left( \frac{1}{\lambda_{K}} \right)
        + \frac12 \ln \bigl( 2 \sigma_{N_K-1} \, a_K(0) \bigr) + o(1)
        .
      \end{aligned}
    \end{equation}
    The only $X$-group, if $K=1$:
    \begin{equation}
      \begin{aligned}
        x_{1,1}(t) &=
        \frac{t}{2} \left( \frac{1}{\lambda_1} \right)
        + \frac12 \ln \left( \frac{2 \sigma_1 \, a_1(0)}{1 + C \, \sigma_1} \right)
        + o(1)
        ,
        \\[1ex]
        x_{1,i}(t) &= \frac{t}{2} \left( \frac{1}{\lambda_1} \right) + \frac12 \ln \left( \frac{2 S_i \, a_1(0)}{T_i} \right)
        + o(1)
        ,
        \quad\text{for $2\le i\le N_1-1$}
        ,
        \\[1ex]
        x_{1,N_1}(t) &= \frac{t}{2} \left( \frac{1}{\lambda_1} \right) + \frac12 \ln \bigl( 2 \sigma_{N_1-1} \, a_1(0) \bigr)
        ,
      \end{aligned}
    \end{equation}
    where the formula for $x_{1,N_1}(t)$ is exact and does not need an $o(1)$ term.

    All $Y$-groups except the rightmost one ($1 \le j \le K-1$), if $K \ge 2$:
    \begin{equation}
      \begin{aligned}
        y_{j,1}(t) &=
        \frac{t}{2} \left( \frac{1}{\lambda_{j}} + \frac{1}{\mu_{j}} \right)
        + \frac12 \ln \left( \frac{2 a_{j}(0) \, b_{j}(0) \, \Psi_{[j,K][j,K-1]}}{\lambda_{[j+1,K]} \, \mu_{[j+1,K-1]} \Psi_{[j+1,K][j+1,K-1]}} \right)
        + o(1)
        ,
        \\[1ex]
        y_{j,i}(t) &=
        \frac{t}{2} \left( \frac{1}{\lambda_{j+1}} + \frac{1}{\mu_{j}} \right)
        + \frac12 \ln \left( \frac{2 a_{j+1}(0) \, b_{j}(0) \, \Psi_{[j+1,K][j,K-1]}}{\lambda_{[j+2,K]} \, \mu_{[j+1,K-1]} \, \Psi_{[j+2,K][j+1,K-1]}} \right)
        + o(1)
        ,
        \\
        &
        \quad \text{for $2 \le i \le N_{j}$}
        .
      \end{aligned}
    \end{equation}
    The rightmost $Y$-group:
    \begin{equation}
      \begin{aligned}
	y_{K,1}(t) &=
        \frac{t}{2} \left( \frac{1}{\lambda_K} \right)
        + \frac12 \ln \bigl( 2 \tau_1 \, a_K(0) \bigr)
        + o(1)
        ,
	\\[1ex]
        y_{K,i}(t) &=
        \frac12 \ln(2 S_i)
        + o(1)
        ,
        \quad \text{for $1 \le i \le N_{K}-1$}
        ,
        \\[1ex]
        y_{K,N_{K}}(t) &=
        \frac12 \ln 2 (S_{N_K-1} \, + D \, \sigma_{N_K-1})
        + o(1)
        .
      \end{aligned}
    \end{equation}
  \end{itemize}
\end{theorem}

\begin{remark}
  \label{rem:how-to-read-asymptotics-pos}
  Note that, as $t \to +\infty$, the rightmost peakon in each
  non-singleton group has a different asymptotic velocity than all
  other peakons in the group; the velocity of this rightmost peakon is
  the same as it would be if the group were a singleton.
  
  Note also that the formula for  the positions of the other peakons in the group,
  $x_{j,i}(t)$ or $y_{j,i}(t)$,
  in most cases does not depend on~$i$,
  so that all of them actually approach the \emph{same} line.
  Their asymptotic velocity is the same as a singleton group
  \emph{in the neighbouring left site} would have.
  (For example, the curves $x=y_{j,i}(t)$ with $1 \le i \le N_j^Y-1$
  approach the same line as the curve $x=x_j(t)$
  if the $j$th $X$-group is a singleton, or $x=x_{j,N_j^X}(t)$ otherwise.)
  
  The exceptions occur in the leftmost $X$-group and the leftmost $Y$-group,
  where the peakons except the rightmost one don't approach the same line,
  only \emph{parallel} lines.
  It is obvious that the leftmost $X$-group might be special,
  since it is given by separate formulas to begin with.
  The reason that the leftmost $Y$-group also has special asymptotics
  is that $\detJ_{KK}^{00} =0$, so that the term $\detJ_{jj}^{00}$, which
  normally dominates in the numerator of
  the solution formula~\eqref{eq:even-Y-typical-group-pos}
  for $Y_{j',i}$, is absent for~$Y_{1,i}$.

  As $t \to -\infty$, things are analogous but with right and left reversed.
  In particular, the asymptotic velocities occuring as $t \to -\infty$ are the same as
  the ones occuring as $t \to +\infty$.
  However, the behaviour in the odd case is more complicated,
  as will be described in detail in the sections below,
  and as we already saw in the examples in Section~\ref{sec:examples-groups-odd}.
\end{remark}

\begin{proof}[Proof of Theorem~\ref{thm:asymptotics-positions-even}]
  We will prove the theorem in detail only for $y_{j',i} = y_{K+1-j,i}$, as $t \to +\infty $,
  where $2 \le j \le K-1$ and $1 \le i \le N_{j'}^Y - 1$,
  the case described by the first equation in~\eqref{eq:asy-even-posinf-Y-typical-group}.
  The proofs of the other results are analogous and will be omitted.
  
  Recall the formula~\eqref{eq:even-Y-typical-group-pos} for
  $Y_{j',i} = \tfrac12 \exp 2 y_{j',i}$:
  \begin{equation}
    \label{eq:asymp-proof-Y}
    Y_{j',i} = \frac{\detJ_{jj}^{00} + T_{i} \detJ_{j,j-1}^{00} + S_i \detJ_{j-1,j-1}^{00}}{\detJ_{j-1,j-1}^{11} + T_{i} \detJ_{j-1,j-2}^{11} + S_i \detJ_{j-2,j-2}^{11}}
    ,
    \qquad
    1 \le i \le N_{j'}^Y - 1
    .
  \end{equation}
  As $t \to +\infty$,
  \begin{equation}
    a_1 \gg a_2\gg\dots \gg a_K
    ,\qquad
    b_1 \gg b_2\gg\dots \gg b_{K-1}
    ,
  \end{equation}
  since
  $a_i(t) = a_i(0) \, e^{t/\lambda_i}$ and~$b_j(t) = b_j(0) \, e^{t/\mu_j}$
  and we are assuming
  \begin{equation*}
    0 < \lambda_1 < \lambda_2 < \dots < \lambda_K
    , \qquad
    0 < \mu_1 < \mu_2 < \dots < \mu_{K-1}.
  \end{equation*}
  This means that the dominant terms in the sums $\detJ_{ij}^{rs}$
  (see Definition~\ref{def:heineintegral})
  are the ones with the smallest indices.
  Namely, by Lemma~$9.1$ in \cite{lundmark-szmigielski:2017:GX-dynamics-interlacing},
  \begin{equation}
    \detJ_{ij}^{rs}(t) =
    \Psi_{[i][j]} \,
    \lambda_{[i]}^r \, \mu_{[j]}^s \,
    a_{[i]}(t) \, b_{[j]}(t) \,
    (1+o(1))
    ,
  \end{equation}
  where $[k] = [1,k] = \{1,2,\dots,k \}$.
  Comparing the dominant terms in $\detJ_{jj}^{00}$, $\detJ_{j,j-1}^{00}$
  and~$\detJ_{j-1,j-1}^{00}$, we see that the dominant term in the numerator
  of~\eqref{eq:asymp-proof-Y} is
  $a_{[j]}(t) \, b_{[j]}(t)$, since it is growing exponentially faster than the other ones.
  Similarly the leading dominant term in the denominator is $a_{[j-1]}(t) \, b_{[j-1]}(t)$.
  Factoring out these terms from~\eqref{eq:asymp-proof-Y},
  and substituting $a_k(t) = a_k(0) \, e^{t/\lambda_k}$ and~$b_k(t) = b_k(0) \, e^{t/\mu_k}$,
  we get
  \begin{equation*}
    \begin{aligned}
      Y_{j',i}(t) &
      = \frac{a_{[j]}(t) \, b_{[j]}(t)}{a_{[j-1]}(t) \, b_{[j-1]}(t)}
      \\ & \quad
      \times
      \frac{\Psi_{[j][j]} + \frac{T_i}{b_{j}(t)} \, \Psi_{[j][j-1]} + \frac{S_i}{a_{j}(t) \, b_{j}(t)} \, \Psi_{[j-1][j-1]}}
      {\Psi_{[j-1][j-1]} \lambda_{[j-1]} \, \mu_{[j-1]}
          + \frac{T_i \, \lambda_{[j-1]} \, \mu_{[j-2]}}{b_{j-1}(t)} \Psi_{[j-1][j-2]} + \frac{S_i \lambda_{[j-2]} \, \mu_{[j-2]}}{a_{j-1}(t) \, b_{j-1}(t)} \Psi_{[j-2][j-2]}}
      \\ & \quad
      \times (1 + o(1))
      \\ &
      = a_j(0) \, b_j(0) \, e^{t \bigl( \frac{1}{\lambda_j} + \frac{1}{\mu_j} \bigr)}
      \times
      \frac{\Psi_{[j][j]} + o(1)}{\lambda_{[j-1]} \, \mu_{[j-1]} \, \Psi_{[j-1][j-1]} + o(1)}
      \times (1 + o(1))
      .
    \end{aligned}
  \end{equation*}
  Taking the logarithm, we get the desired formula for
  $y_{j',i} = \tfrac12 \ln(2 Y_{j',i})$
  in~\eqref{eq:asy-even-posinf-Y-typical-group}.
\end{proof}

Next we present the asymptotics of the amplitudes $m_{j,i}$ and~$n_{j,i}$.
We omit the proof, since the calculations are very similar to those
in the proof of Theorem~\ref{thm:asymptotics-positions-even}.
When writing the formulas, we are using that
$\ln m_k = x_k + \ln(\tfrac12 Q_k)$ (etc.),
and we have not cancelled or simplified any terms except for joining
the $t$-contributions from the two parts.

\begin{theorem}
  \label{thm:asymptotics-amplitudes-even}

  The amplitudes for non-singleton groups in the even case,
  with the eigenvalues ordered as in~\eqref{eq:eigenvalues-ordered-even},
  satisfy the following asymptotic formulas.
  
  \begin{itemize}
  \item As $t \to +\infty$:

    All $X$-groups except the leftmost one
    ($j'=K+1-j$ with $1 \le j \le K-1$), if $K \ge 2$:
    \begin{equation}
      \begin{aligned}
        \ln m_{j',i}(t) &= \frac{t}{2} \left( \frac{1}{\lambda_{j+1}} - \frac{3}{\mu_j} \right)
        + \frac12 \ln \left( \frac{2 a_{j+1}(0) \, b_j(0) \, \Psi_{[1,j+1][1,j]}}{\lambda_{[1,j]} \, \mu_{[1,j-1]} \Psi_{[1,j][1,j-1]}} \right)
        \\ & \quad
        +\ln \left( \frac{(\sigma_i - \sigma_{i-1}) \, \mu_{[1,j-1]}^2 \, \Psi_{[1,j][1,j-1]}^2}{2 b_j^2(0) \, \lambda_{[1,j]} \, \Psi_{[1,j][1,j]}^2} \right)
        + o(1)
        ,
        \\[1ex]
        \ln m_{j',N_{j'}}(t) &= \frac{t}{2} \left( \frac{1}{\lambda_j} - \frac{1}{\mu_{j}} \right)
        + \frac12 \ln \left( \frac{2 a_j(0) \, b_{j}(0) \, \Psi_{[1,j][1,j]}}{\lambda_{[1,j-1]} \, \mu_{[1,j-1]} \, \Psi_{[1,j-1][1,j-1]}} \right)
        \\ & \quad
        +\ln \left( \frac{\mu_{[1,j-1]}^2 \, \Psi_{[1,j][1,j-1]}}{2 b_j(0) \, \lambda_{[1,j]} \, \Psi_{[1,j][1,j]}} \right)
        + o(1)
        ,
        \\ &
        \quad \text{for $1 \le i \le N_{j'}-1$}
        .
      \end{aligned}
    \end{equation}
    The leftmost $X$-group:
    \begin{equation}
      \begin{aligned}
        \ln m_{1,1}(t) &= \frac12 \ln \left( \frac{2 \tau_1}{C L} \right)
        + \ln \left( \frac{C M}{2L} \right)
        + o(1)
        ,
        \\[1ex]
        \ln m_{1,i}(t) &= \frac12 \ln \left( \frac{2 S_i}{L M} \right)
        + \ln \left( \frac{(\sigma_i - \sigma_{i-1}) \, M^2}{2 \sigma_i \, \sigma_{i-1} \, L} \right)
        + o(1)
        ,
        \\ &
        \quad \text{for $2 \le i \le N_1-1$}
        ,
        \\[1ex]
        \ln m_{1,N_1}(t) &= \frac{t}{2} \left( \frac{1}{\lambda_{K}} \right)
        + \frac12 \ln \left( \frac{2 \sigma_{N_1-1} \, a_K(0) \, \Psi_{[1,K][1,K-1]}}{\lambda_{[1,K-1]} \, M \,  \Psi_{[1,K-1][1,K-1]}} \right)
        \\ &\quad
        + \ln \left( \frac{M^2}{2 \sigma_{N_1-1} \, L} \right)
        + o(1)
        .
      \end{aligned}
    \end{equation}
    All $Y$-groups except the leftmost and rightmost ones
    ($j'=K+1-j$ with $2 \le j \le K-1$), if $K \ge 3$:
    \begin{equation}
      \begin{aligned}
        \ln n_{j',i}(t) &= \frac{t}{2} \left( \frac{1}{\mu_j} - \frac{3}{\lambda_j} \right)
        + \frac12 \ln \left( \frac{2 a_j(0) \, b_{j}(0) \, \Psi_{[1,j][1,j]}}{\lambda_{[1,j-1]} \mu_{[1,j-1]} \, \Psi_{[1,j-1][1,j-1]}} \right)
        \\ & \quad
        + \ln \left( \frac{(\sigma_i - \sigma_{i-1}) \, \lambda_{[1,j-1]}^2 \, \Psi_{[1,j-1][1,j-1]}^2}{2 a_j^2(0) \, \mu_{[1,j-1]} \, \Psi_{[1,j][1,j-1]}^2} \right)
        + o(1)
        ,
        \\[1ex]
        \ln n_{j',N_{j'}}(t) &= \frac{t}{2} \left( \frac{1}{\mu_{j-1}} - \frac{1}{\lambda_j} \right)
        + \frac12 \ln \left( \frac{2 a_j(0) \, b_{j-1}(0) \, \Psi_{[1,j][1,j-1]}}{\lambda_{[1,j-1]} \, \mu_{[1,j-2]} \Psi_{[1,j-1][1,j-2]} } \right)
        \\ & \quad
        + \ln \left( \frac{\lambda_{[1,j-1]}^2 \Psi_{[1,j-1][1,j-1]}}{2 a_j(0) \, \mu_{[1,j-1]} \, \Psi_{[1,j][1,j-1]}} \right)
        + o(1)
        ,
        \\ &
        \quad \text{for $1 \le i \le N_{j'}-1$}
        .
      \end{aligned}
    \end{equation}
    The rightmost $Y$-group, if $K \ge 2$:
    \begin{equation}
      \begin{aligned}
        \ln n_{K,i}(t) &= \frac{t}{2} \left( \frac{1}{\mu_1} - \frac{3}{\lambda_1} \right)
        + \frac12 \ln \bigl( 2 a_1(0) \, b_1(0) \, \Psi_{\{1\}\{1\}} \bigr)
        \\ & \quad
        + \ln \left( \frac{\sigma_i - \sigma_{i-1}}{2 a_1^2(0)} \right)
        + o(1)
        ,
        \quad \text{for $1 \le i \le N_{K}-1$}
        ,
        \\[1ex]
        \ln n_{K,N_{K}}(t) &= \frac{t}{2} \left( \frac{1}{\mu_1} - \frac{1}{\lambda_1} \right)
        + \frac12 \ln \bigl( 2 a_1(0) \, b_1(0) \, \Psi_{\{1\}\{1\}} \bigr)
        \\ & \quad
        +\ln \left( \frac{1}{2 a_1(0)} \right)
        + o(1)
        .
      \end{aligned}
    \end{equation}
    The leftmost $Y$-group:
    \begin{equation}
      \begin{aligned}
        \ln n_{1,i}(t) &= \frac{t}{2} \left( \frac{-3}{\lambda_{K}} \right)
        + \frac12 \ln \left( \frac{ 2 T_i a_{K}(0) \, \Psi_{[1,K][1,K-1]}}{\lambda_{[1,K-1]} \, M \, \Psi_{[1,K-1][1,K-1]}} \right)
        \\ & \quad
        + \ln \left( \frac{(\sigma_i - \sigma_{i-1}) \, \lambda_{[1,j-1]}^2 \, \Psi_{[1,j-1][1,j-1]}^2}{2 a_j^2(0) \, \mu_{[1,j-1]} \, \Psi_{[1,j][1,j-1]}^2} \right)
        + o(1)
        ,
        \\[1ex]
        \ln n_{1,N_1}(t) &= \frac{t}{2} \left( \frac{1}{\mu_{K-1}} - \frac{1}{\lambda_{K}} \right)
        + \frac12 \ln \left( \frac{2 a_K(0) \, b_{K-1}(0) \, \Psi_{[1,K][1,K-1]}}{\lambda_{[1,K-1]} \, \mu_{[1,K-2]} \, \Psi_{[1,K-1][1,K-2]}} \right)
	\\ & \quad
        +\ln \left( \frac{\lambda_{[1,K-1]}^2 \Psi_{[1,K-1][1,K-1]}}{2 a_K(0) \, M \, \Psi_{[1,K][1,K-1]}} \right)
        + o(1)
        ,
        \\ &
        \quad \text{for $1 \le i \le N_1-1$}
        .
      \end{aligned}
    \end{equation}
    The only $Y$-group, if $K=1$:
    \begin{equation}
      \begin{aligned}
        \ln n_{1,i} &= \frac{t}{2} \left( \frac{-3}{\lambda_1} \right)
        + \frac12 \ln \bigl(2 T_i \, a_1(0) \bigr)
        + \ln \left( \frac{\sigma_i - \sigma_{i-1}}{2 a_1^2(0)} \right)
        + o(1)
        ,
        \\
        &\quad \text{for $1 \le i \le N_1-1$}
        ,
        \\[1ex]
        \ln n_{1,N_1}(t) &= \frac{t}{2} \left( \frac{-1}{\lambda_1} \right)
        + \frac12 \ln \bigl( 2 (T_{N_1-1} + D) \, a_1(0) \bigr)
        +\ln \left( \frac{1}{2 a_1(0)} \right)
        + o(1)
        .
      \end{aligned}
    \end{equation}
    
  \item As $t \to -\infty$:

    All $X$-groups except the leftmost and rightmost ones ($2 \le j \le K-1$), if $K \ge 3$:
    \begin{equation}
      \begin{aligned}
        \ln m_{j,1}(t) &= \frac{t}{2} \left( \frac{1}{\lambda_{j}} - \frac{1}{\mu_{j-1}} \right)
        + \frac12 \ln \left( \frac{ 2 a_{j}(0) \, b_{j-1}(0) \, \Psi_{[j,K][j-1,K-1]}}{\lambda_{[j+1,K]} \mu_{[j,K-1]} \, \Psi_{[j+1,K][j,K-1]}} \right)
        \\ & \quad
        + \ln \left( \frac{\mu^2_{[j,K-1]} \, \Psi_{[j,K][j,K-1]}}{2 b_{j-1}(0) \, \lambda_{[j,K]} \, \Psi_{[j,K][j-1,K-1]}} \right)
        + o(1)
        ,
        \\[1ex]
	\ln m_{j,i}(t) &= \frac{t}{2} \left( \frac{3}{\lambda_{j}} - \frac{1}{\mu_j} \right)
        + \frac12 \ln \left( \frac{2 a_j(0) \, b_j(0) \, \Psi_{[j,K][j,K-1]}}{\lambda_{[j+1,K]} \, \mu_{[j+1,K-1]} \, \Psi_{[j+1,K][j+1,K-1]}} \right)
	\\ & \quad
        + \ln \left( \frac{(\sigma_i - \sigma_{i-1}) \, S_i}{R_i \, R_{i-1}} \right)
        + \ln \left( \frac{\mu_{[j,K-1]} \, \mu_{[j+1,K-1]} \, a_j(0)}{b_j(0) \, \lambda_{[j+1,K]}} \right)
	\\ & \quad
	+ \ln \left( \frac{\Psi_{[j,K][j,K-1]} \, \Psi_{[j+1,K][j+1,K-1]}}{2 \Psi_{[j+1,K][j,K-1]}^2} \right)
        + o(1)
        ,
	\\ &
        \quad \text{for $2 \le i \le N_j-1$}
        ,
	\\[1ex]
	\ln m_{j,N_j}(t) &= \frac{t}{2} \left( \frac{3}{\lambda_{j}} - \frac{1}{\mu_j} \right)
        + \frac12 \ln \left( \frac{2 a_j(0) \, b_j(0) \Psi_{[j,K][j,K-1]}}{\lambda_{[j+1,K]} \, \mu_{[j+1,K-1]} \, \Psi_{[j+1,K][j+1,K-1]}} \right)
	\\ & \quad
	+ \ln \left( \frac{\sigma_{N_j-1}}{R_{N_j-1} } \right)
        + \ln \left( \frac{\mu_{[j,K-1]} \, \mu_{[j+1,K-1]} \, a_j(0)}{b_j(0) \, \lambda_{[j+1,K]}} \right)
	\\ & \quad
	+ \ln \left( \frac{\Psi_{[j,K][j,K-1]} \, \Psi_{[j+1,K][j+1,K-1]}}{2 \Psi_{[j+1,K][j,K-1]}^2} \right)
        + o(1)
        .
      \end{aligned}
    \end{equation}
    The leftmost $X$-group, if $K \ge 2$:
    \begin{equation}
      \begin{aligned}
        \ln m_{1,1}(t) &= \frac{t}{2} \left( \frac{1}{\lambda_1} - \frac{1}{\mu_1} \right)
        + \frac12 \ln \left( \frac{2 a_1(0) \, b_1(0) \, \Psi_{[1,K][1,K-1]}}{\lambda_{[2,K]} \, \mu_{[2,K-1]} \, \Psi_{[2,K][2,K-1]}} \right)
        \\ & \quad
        + \ln \left( \frac{\mu_{[2,K-1]} \, M \, \Psi_{[2,K][2,K-1]}}{2 b_1(0) \, L \, \Psi_{[2,K][1,K-1]}} \right)
        + o(1)
        ,
        \\[1ex]
        \ln m_{1,i}(t) &=
        \frac{t}{2} \left( \frac{3}{\lambda_1} - \frac{1}{\mu_1} \right)
        + \frac12 \ln \left( \frac{ 2 a_1(0) \, b_1(0) \, \Psi_{[1,K][1,K-1]}}{\lambda_{[2,K]} \, \mu_{[2,K-1]} \, \Psi_{[2,K][2,K-1]}} \right)
        \\ & \quad
        + \ln \left( \frac{(\sigma_i - \sigma_{i-1}) \, S_i}{R_i \, R_{i-1}} \right)
        + \ln \left( \frac{a_1(0) \, M \, \mu_{[2,K-1]}}{b_1(0) \, \lambda_{[2,K]}} \right)
        \\ & \quad
        + \ln \left( \frac{\Psi_{[1,K][1,K-1]} \, \Psi_{[2,K][2,K-1]}}{2 \Psi^2_{[2,K][1,K-1]}} \right)
        + o(1)
        ,
        \\ &
        \quad \text{for $2 \le i \le N_1-1$}
        ,
        \\[1ex]
        \ln m_{1,N_1}(t) &=
        \frac{t}{2} \left( \frac{3}{\lambda_1} - \frac{1}{\mu_1} \right)
        + \frac12 \ln \left( \frac{2 a_1(0) \, b_1(0) \, \Psi_{[1,K][1,K-1]}}{\lambda_{[2,K]} \, \mu_{[2,K-1]} \, \Psi_{[2,K][2,K-1]}} \right)
        \\ & \quad
        + \ln \left( \frac{\sigma_{N_1-1}}{R_{N_1-1}} \right)
        + \ln \left( \frac{a_1(0) \, M \, \mu_{[2,K-1]}}{b_1(0) \, \lambda_{[2,K]}} \right)
        \\ & \quad
        + \ln \left( \frac{\Psi_{[1,K][1,K-1]} \, \Psi_{[2,K][2,K-1]}}{2 \Psi^2_{[2,K][1,K-1]}} \right)
        + o(1)
        .
      \end{aligned}
    \end{equation}
    The rightmost $X$-group, if $K \ge 2$:
    \begin{equation}
      \begin{aligned}
        \ln m_{K,1}(t) &=
        \frac{t}{2} \left( \frac{1}{\lambda_{K}} - \frac{1}{\mu_{K-1}} \right)
        + \frac12 \ln \bigl( 2 a_K(0) \, b_{K-1}(0) \, \Psi_{\{K\}\{K-1\}} \bigr)
        \\
        & \quad
        + \ln \left( \frac{1}{2 b_{K-1}(0) \, \lambda_{K} \, \Psi_{\{K\}\{K-1\}}} \right)
        + o(1)
        ,
        \\[1ex]
        \ln m_{K,i}(t) &=
        \frac{t}{2} \left( \frac{3}{\lambda_{K}} \right)
        + \frac12 \ln \left( \frac{2 S_i \, a_K(0)}{T_i} \right)
        + \ln \left( \frac{(\sigma_i-\sigma_{i-1}) \, T_i \, a_K(0)}{2 R_i \, R_{i-1}} \right)
        + o(1)
        ,
        \\ &
        \quad \text{for $2 \le i \le N_K-1$}
        ,
        \\[1ex]
        \ln m_{K, N_K}(t) &=
        \frac{t}{2} \left( \frac{3}{\lambda_{K}} \right)
        + \frac12 \ln \bigl( 2 \sigma_{N_K-1} \, a_K(0) \bigr)
        + \ln \left( \frac{a_K(0)}{2 R_{N_{j'}-1}} \right)
        + o(1)
        .
      \end{aligned}
    \end{equation}
    The only $X$-group, if $K=1$:
    \begin{equation}
      \begin{aligned}
        \ln m_{1,1}(t) &=
        \frac{t}{2} \left( \frac{1}{\lambda_1} \right)
        + \frac12 \ln \left( \frac{2 \sigma_1 \, a_1(0)}{1 + \sigma_1 \, C} \right)
        + \ln \left( \frac{1 + \sigma_1 \, C}{2 \lambda_1 \, \sigma_1} \right)
        + o(1)
        ,
        \\[1ex]
        \ln m_{1,i}(t) &=
        \frac{t}{2} \left( \frac{3}{\lambda_1} \right)
        + \frac12 \ln \left( \frac{2 S_i \, a_1(0)}{T_i} \right)
        +\ln \left( \frac{(\sigma_i - \sigma_{i-1}) \, T_i \, a_1(0)}{2 R_i \, R_{i-1}} \right)
        + o(1)
        ,
        \\ &
        \quad \text{for $2 \le i \le N_1-1$}
        ,
        \\[1ex]
        \ln m_{1,N_1}(t) &=
        \frac{t}{2} \left( \frac{1}{\lambda_1} \right)
        + \frac12 \ln \bigl( 2 \sigma_{N_1-1} \, a_1(0) \bigr)
        + \ln \left( \frac{a_1(0)}{2 R_{N_1-1}} \right)
        + o(1)
        .
      \end{aligned}
    \end{equation}
    All $Y$-groups except the rightmost one
    ($1 \le j \le K-1$), if $K \ge 2$:
    \begin{equation}
      \begin{aligned}
	\ln n_{j,1}(t) &=
        \frac{t}{2} \left( \frac{1}{\mu_{j}} - \frac{1}{\lambda_{j}} \right)
        + \frac12 \ln \left( \frac{2 a_{j}(0) \, b_{j}(0) \, \Psi_{[j,K][j,K-1]}}{\lambda_{[j+1,K]} \, \mu_{[j+1,K-1]} \, \Psi_{[j+1,K][j+1,K-1]}} \right)
	\\ & \quad
	+\ln \left( \frac{\lambda_{[j+1,K]}^2 \, \Psi_{[j+1,K][j,K-1]}}{2 a_j(0) \, \mu_{[j,K-1]}  \, \Psi_{[j,K][j,K-1]}} \right)
        + o(1)
        ,
	\\[1ex]
	\ln n_{j,i}(t) &=
        \frac{t}{2} \left( \frac{3}{\mu_{j}} - \frac{1}{\lambda_{j+1}} \right)
	+ \frac12 \ln \left( \frac{ 2 a_{j+1}(0) \, b_{j}(0) \, \Psi_{[j+1,K][j,K-1]}}{\lambda_{[j+2,K]} \, \mu_{[j+1,K-1]} \, \Psi_{[j+2,K][j+1,K-1]}} \right)
	\\ & \quad
	+ \ln \left( \frac{(\sigma_i - \sigma_{i-1}) \, S_i}{R_i \, R_{i-1}} \right)
	+ \ln \left( \frac{b_j(0) \, \lambda_{[j+1,K]} \, \lambda_{[j+2,K]}}{a_{j+1}(0) \, \mu_{[j+1,K-1]}} \right)
        \\ & \quad
        + \ln \left( \frac{\Psi_{[j+1,K][j,K-1]} \, \Psi_{[j+2,K][j+1,K-1]}}{2 \Psi_{[j+1,K][j+1,K-1]}^2} \right)
        + o(1)
        ,
        \\ &
        \quad \text{for $2 \le i \le N_{j}-1$}
        ,
        \\[1ex]
        \ln n_{j,N_j}(t) &= \frac{t}{2} \left( \frac{3}{\mu_{j}} - \frac{1}{\lambda_{j+1}} \right)
        + \frac12 \ln \left( \frac{2 a_{j+1}(0) \, b_{j}(0) \, \Psi_{[j+1,K][j,K-1]}}{\lambda_{[j+2,K]} \, \mu_{[j+1,K-1]} \, \Psi_{[j+2,K][j+1,K-1]}} \right)
        \\ & \quad
        + \ln \left( \frac{\sigma_{N_j-1}}{R_{N_j-1}} \right)
        + \ln \left( \frac{b_j(0) \, \lambda_{[j+1,K]} \, \lambda_{[j+2,K]}}{a_{j+1}(0) \, \mu_{[j+1,K-1]}} \right)
        \\ & \quad
        + \ln \left( \frac{\Psi_{[j+1,K][j,K-1]} \, \Psi_{[j+2,K][j+1,K-1]}}{2 \Psi_{[j+1,K][j+1,K-1]}^2} \right)
        + o(1)
        .
      \end{aligned}
    \end{equation}
    The rightmost $Y$-group:
    \begin{equation}
      \begin{aligned}
        \ln n_{K,1}(t) &=
        \frac{t}{2} \left( \frac{-1}{\lambda_K} \right)
        + \frac12 \ln \bigl( 2 \tau_1 \, a_K(0) \bigr)
        + \ln \left( \frac{1}{2 a_K(0)} \right)  + o(1)
        ,
        \\
        \ln n_{K,i}(t) &=
        \frac12 \ln(2 S_i)
        + \ln \left( \frac{\sigma_{i} - \sigma_{i-1}}{2 \sigma_{i} \, \sigma_{i-1}} \right)
        + o(1)
        ,
        \quad \text{for $2 \le i \le N_{K}-1$}
        ,
        \\[0.5ex]
        \ln n_{K,N_K}(t) &=
        \frac12 \ln \bigl( 2 (S_{N_K-1} + D \, \sigma_{N_K-1}) \bigr)
        + \ln \left( \frac{1}{2 \sigma_{N_{K}-1}} \right)
        + o(1)
        .
      \end{aligned}
    \end{equation}
  \end{itemize}
\end{theorem}

\begin{remark}
  Note that if we list all the coefficients of $t$ in the asymptotic formulas
  for all $\ln m_{j,i}(t)$
  and all $-\ln n_{j,i}(t)$,
  then they appear in the opposite order when $t \to \infty$
  compared to when $t \to -\infty$.
  For example, for a typical $X$-group as $t \to +\infty$, we have
  \begin{equation*}
    \ln m_{j',i}(t) =
    \frac{t}{2}  \left( \frac{1}{\lambda_{j+1}} - \frac{3}{\mu_{j}} \right)
    + \text{constant} + o(1)
    ,\qquad
    1 \le i \le N^X_{j'}-1
    ,
  \end{equation*}
  which agrees with its mirror partner, a typical $Y$-group, as $t \to -\infty$:
  \begin{equation*}
    -\ln n_{j,i}(t) =
    \frac{t}{2}  \left( \frac{1}{\lambda_{j+1}} - \frac{3}{\mu_{j}} \right)
    + \text{constant} + o(1)
    ,\qquad
    2 \le i \le N^Y_{j}
    .
  \end{equation*}
  This was illustrated in Example~\ref{ex:GX-3+3-allgroups}.
  Again, the behaviour in the odd case (described below) will be a little different.
\end{remark}

\section{Solution formulas for  the odd case ($2K+1$ groups)}
\label{sec:solutions-odd}

In this section we list the solution formulas for the odd case,
where we have $K+1$ groups of $X$-type and $K$ groups of
$Y$-type, since we assume without loss of generality that the first
and last groups are $X$-groups.
We assume that $K \ge 1$, since if there is just a single $X$-groups,
with no $Y$-group, then the dynamics is trivial.
The proofs will be given in Section~\ref{sec:proofs-odd}, and the
asymptotics of these solutions as $t \to \pm \infty$ will be studied
in Section~\ref{sec:asymptotics-odd}.

As we mentioned already in Section~\ref{sec:more-notation}, in the odd case
there are $K+K$ eigenvalues
\begin{equation*}
  0 < \lambda_1 < \lambda_2 < \dots < \lambda_{K}
  , \qquad
  0 < \mu_1 < \mu_2 < \dots < \mu_{K}
  ,
\end{equation*}
which means that all the determinants~$\detJ_{ij}^{rs}$
will have $A=K$ and $B=K$, i.e.,
\begin{equation}
  \detJ_{ij}^{rs} = \detJ[K,K,r,s,i,j]
  .
\end{equation}
The solution formulas look very similar to those for the even case;
except for the number of eigenvalues differing,
there is basically just a shift in the lower indexes of the~$\detJ_{ij}^{rs}$.
However, this has consequences for the asymptotics,
which differ from the even case in some curious ways,
as we saw in the examples in Section~\ref{sec:examples-groups-odd}.

As before, the intent of the abbreviation $j'$ will be to denote
``the $j$th object from the right''.
Therefore, it will be defined slightly differently for $X$-groups
here in the odd case, namely
\begin{equation}
  j' = (K+1)-j+1 = K+2-j
  ,
\end{equation}
while for $Y$-groups we still write
\begin{equation}
  j'=K+1-j
  .
\end{equation}

\subsection{Solutions  for $X$-groups}
\label{sec:solutions-odd-X}

First we state the solution formulas for the $j$th $X$-group from the right
in the odd case,
i.e., group number $j' = K+2-j$.
As usual, the formulas depend on whether the group is a singleton
or a group consisting of $N_{j'}^X \ge 2$ peakons.
Moreover, the solution formulas for the leftmost and rightmost $X$-groups are different
from those for the ``typical'' $X$-groups in the middle.
And as in the even case, the rightmost peakon $X_{j,N_j^X}$ in each group is given by
a separate formula.

\subsubsection{$X$-singletons}
\label{sec:solutions-odd-X-singleton}

The solution formulas for the leftmost $X$-peakon are
\begin{equation}
  \label{eq:odd-X-leftmost-singleton}
  X_1 = \frac{\detJ_{KK}^{00}}{\detJ_{K-1,K-1}^{11} + C \detJ_{K-1,K}^{10}}
  ,\qquad
  Q_1 =
  \frac{M}{L}
  \left( \frac{\detJ_{K-1,K-1}^{11}}{\detJ_{K-1,K}^{10}} + C \right)
  ,
\end{equation}
where $L = \prod_{i=1}^K \lambda_i$ and $M = \prod_{j=1}^{K} \mu_j$.
The rightmost $X$-peakon is given by
\begin{equation}
  \label{eq:odd-X-rightmost-singleton}
  X_{K+1} = \detJ_{11}^{00} + D \detJ_{01}^{00}
  ,\qquad
  Q_{K+1} = \frac{1}{\detJ_{01}^{10}}
  ,
\end{equation}
and the formulas for $X$-peakon number $j'=K+2-j$
where $2 \le j \le K$ are
\begin{equation}
  \label{eq:odd-X-typical-singleton}
  X_{j'} = \frac{\detJ_{j-1,j}^{00}}{\detJ_{j-2,j-1}^{11}}
  ,\qquad
  Q_{j'} = \frac{\detJ_{j-1,j-1}^{01} \detJ_{j-2,j-1}^{11}}{\detJ_{j-1,j}^{10} \detJ_{j-2,j-1}^{10}}
  .
\end{equation}

\subsubsection{All $X$-groups except the leftmost and the rightmost}
\label{sec:solutions-odd-X-typical-group}

Next, we give the formulas for $X$-group number $j'=K+2-j$, in case it
consist of $N_{j'}^X \ge 2$ peakons. Here $2 \le j \le K$; the
leftmost group ($j=K+1$) and the rightmost group ($j=1$) are treated
separately below.
The formulas for the positions are
\begin{equation}
  \label{eq:odd-X-typical-group-pos}
  \begin{aligned}
    X_{j',i} &
    = \frac{\detJ_{jj}^{00} + T_{i} \detJ_{j-1,j}^{00} + S_i \detJ_{j-1,j-1}^{00}}{\detJ_{j-1,j-1}^{11} + T_{i} \detJ_{j-2,j-1}^{11} + S_i \detJ_{j-2,j-2}^{11}}
    ,
    \\[1ex]
    X_{j',N_{j'}} &
    = \frac{\detJ_{j-1,j}^{00} + \sigma_{N_{j'}-1} \detJ_{j-1,j-1}^{00}}{\detJ_{j-2,j-1}^{11} + \sigma_{N_{j'}-1} \detJ_{j-2,j-2}^{11}}
    ,
  \end{aligned}
\end{equation}
and for the amplitudes
\begin{equation}
  \label{eq:odd-X-typical-group-amp}
  \begin{aligned}
    Q_{j',i} &
    = \frac{(\sigma_i - \sigma_{i-1}) \, \detJ_{j-1,j-1}^{01}  \left( \detJ_{j-1,j-1}^{11} + T_i \detJ_{j-2,j-1}^{11} + S_i \detJ_{j-2,j-2}^{11} \right)}
    {\left( \detJ_{j-1,j}^{10} + \sigma_i \detJ_{j-1,j-1}^{10} + R_{i} \detJ_{j-2,j-1}^{10} \right) \left( \detJ_{j-1,j}^{10} + \sigma_{i-1} \detJ_{j-1,j-1}^{10} + R_{i-1} \detJ_{j-2,j-1}^{10} \right)}
    ,
    \\[1ex]
    Q_{j',N_{j'}} &
    = \frac{\detJ_{j-1,j-1}^{01} \left( \detJ_{j-2,j-1}^{11} + \sigma_{N_{j'}-1} \detJ_{j-2,j-2}^{11} \right)}{\detJ_{j-2,j-1}^{10} \left( \detJ_{j-1,j}^{10} + \sigma_{N_{j'}-1} \detJ_{j-1,j-1}^{10} + R_{N_{j'}-1} \detJ_{j-2,j-1}^{10} \right)}
    ,
  \end{aligned}
\end{equation}
where $1 \le i \le N_{j'}-1$.
Recall that the sums $T_i$, $S_i$ and~$R_{i}$ were defined in Definition~\ref{def:T-S-R}.

\subsubsection{The leftmost $X$-group }
\label{sec:solutions-odd-X-leftmost-group}

If the leftmost $X$-group contains more than one peakon,
then the positions are given by
\begin{equation}
  \label{eq:odd-X-leftmost-group-pos}
  \begin{aligned}
    X_{1,1} &
    = \frac{\detJ_{KK}^{00} }{\detJ_{K-1,K-1}^{11}  + \frac{1}{\sigma_1} \detJ_{K-1,K}^{11} + C \, \left( \detJ_{K-1,K}^{10} + \frac{1}{\tau_1} \detJ_{KK}^{10} \right)}
    ,
    \\[0.5ex]
    X_{1,i} &
    = \frac{S_{i} \, \detJ_{KK}^{00}}{\detJ_{KK}^{11} + T_{i} \detJ_{K-1,K}^{11} + S_{i} \detJ_{K-1,K-1}^{11}}
    ,
    \\[0.5ex]
    X_{1,N_1} &
    = \frac{\sigma_{N_1-1} \detJ_{KK}^{00}}{\detJ_{K-1,K}^{11} + \sigma_{N_1-1} \detJ_{K-1,K-1}^{11}}
    ,
  \end{aligned}
\end{equation}
and the amplitudes by
\begin{equation}
  \label{eq:odd-X-leftmost-group-amp}
  \begin{aligned}
    Q_{1,1} &
    =
    \frac{M}{L}
    \left( \frac{\detJ_{K-1,K-1}^{11} + \frac{1}{\sigma_1} \detJ_{K-1,K}^{11}}{\detJ_{K-1,K}^{10} + \frac{1}{\tau_1} \detJ_{KK}^{10}} + C \right)
    ,
    \\[1ex]
    Q_{1,i} &
    = \frac{(\sigma_i - \sigma_{i-1}) \, \detJ_{KK}^{01} \left( \detJ_{KK}^{11} + T_i \detJ_{K-1,K}^{11} + S_i \detJ_{K-1,K-1}^{11} \right)}
    {\left( \sigma_i \detJ_{KK}^{10} + R_{i} \detJ_{K-1,K}^{10} \right)
      \left( \sigma_{i-1} \detJ_{KK}^{10} + R_{i-1} \detJ_{K-1,K}^{10} \right)}
    ,
    \\[1ex]
    Q_{1,N_1} &
    = \frac{\detJ_{KK}^{01} \left( \detJ_{K-1,K}^{11} + \sigma_{N_1-1} \detJ_{K-1,K-1}^{11 } \right)}{\detJ_{K-1,K}^{10} \left( \sigma_{{N_1}-1} \detJ_{KK}^{10} + R_{N_1-1} \detJ_{K-1,K}^{10} \right)}
    ,
  \end{aligned}
\end{equation}
where $2 \le i \le N_1-1$.
As in the even case, not only is the rightmost peakon in this group
given by a separate formula, but also the leftmost peakon.

\subsubsection{The rightmost $X$-group}
\label{sec:solutions-odd-X-rightmost-group}

If the rightmost $X$-group contains more than one peakon,
then the positions are given by
\begin{equation}
  \label{eq:odd-X-rightmost-group-pos}
  \begin{aligned}
    X_{K+1,i} &= \detJ_{11}^{00} + T_{i} \detJ_{01}^{00} + S_i
    ,
    \\[1ex]
    X_{K+1,N_{K+1}} &
    = \detJ_{11}^{00} + \left( T_{N_{K+1}-1} + D \right) \detJ_{01}^{00} + S_{N_{K+1}-1} + D \, \sigma_{N_{K+1}-1}
    ,
  \end{aligned}
\end{equation}
and the amplitudes by
\begin{equation}
  \label{eq:odd-X-rightmost-group-amp}
  \begin{aligned}
    Q_{K+1,i} &
    = \frac{\sigma_i - \sigma_{i-1}}{\left( \detJ_{01}^{10} + \sigma_i \right) \left( \detJ_{01}^{10} + \sigma_{i-1} \right)}
    ,
    \\[1ex]
    Q_{K+1,N_{K+1}} &= \frac{1}{\detJ_{01}^{10} + \sigma_{N_{K+1}-1}}
    ,
  \end{aligned}
\end{equation}
where $1 \le i \le N_{K+1}-1$.

\subsection{Solutions  for  $Y$-groups}
\label{sec:solutions-odd-Y}

In this section we give the solution formulas for $Y$-groups.
In the odd case there are $K$ groups of $Y$-type, between the
$X$-groups, so the $j$th $Y$-group from the right will be number
$j'=K+1-j$, where $1 \le j \le K$.

\subsubsection{$Y$-singletons}
\label{sec:solutions-odd-Y-singleton}

The solution formulas for a $Y$-group which consists of a single peakon are
\begin{equation}
  \label{eq:odd-Y-singleton}
  Y_{j'} = \frac{\detJ_{jj}^{00}}{\detJ_{j-1,j-1}^{11}}
  ,\qquad
  P_{j'} = \frac{\detJ_{j-1,j}^{10} \detJ_{j-1,j-1}^{11}}{\detJ_{jj}^{01} \detJ_{j-1,j-1}^{01}}
  .
\end{equation}

\subsubsection{All $Y$-groups}
\label{sec:solutions-odd-Y-group}

If a $Y$-group contains $N_{j'} \ge 2$ peakons,
where $j'=K+1-j$ with $1 \le j \le K$, then the positions are given by
\begin{equation}
  \label{eq:odd-Y-group-pos}
  \begin{aligned}
    Y_{j',i} &
    = \frac{\detJ_{j,j+1}^{00} + T_{i} \detJ_{jj}^{00} + S_i \detJ_{j-1,j}^{00} }{\detJ_{j-1,j}^{11} + T_{i} \detJ_{j-1,j-1}^{11} + S_i \detJ_{j-2,j-1}^{11}}
    ,
    \\[1ex]
    Y_{j',N_{j'}} &
    = \frac{\detJ_{jj}^{00} + \sigma_{N_{j'}-1} \detJ_{j-1,j}^{00}}{\detJ_{j-1,j-1}^{11} + \sigma_{N_{j'}-1} \detJ_{j-2,j-1}^{11}}
    ,
  \end{aligned}
\end{equation}
and the amplitudes by
\begin{equation}
  \label{eq:odd-Y-group-amp}
  \begin{aligned}
    P_{j',i} &
    = \frac{(\sigma_i - \sigma_{i-1}) \, \detJ_{j-1,j}^{10} \left( \detJ_{j-1,j}^{11} + T_i \detJ_{j-1,j-1}^{11} + S_i \detJ_{j-2,j-1}^{11} \right)}
    {\left( \detJ_{jj}^{01} + \sigma_i \detJ_{j-1,j}^{01} + R_{i} \detJ_{j-1,j-1}^{01} \right)
      \left( \detJ_{jj}^{01} + \sigma_{i-1} \detJ_{j-1,j}^{01} + R_{i-1} \detJ_{j-1,j-1}^{01} \right)}
    ,
    \\[1ex]
    P_{j',N_{j'}} &
    = \frac{\detJ_{j-1,j}^{10} \left( \detJ_{j-1,j-1}^{11} + \sigma_{N_{j'}-1} \detJ_{j-2,j-1}^{11} \right)}
    {\detJ_{j-1,j-1}^{01} \left( \detJ_{jj}^{01} + \sigma_{N_{j'}-1} \detJ_{j-1,j}^{01} + R_{N_{j'}-1} \detJ_{j-1,j-1}^{01} \right)}
    ,
  \end{aligned}
\end{equation}
where $1 \le i \le N_{j'}-1$.

\section{Proofs for the odd case}
\label{sec:proofs-odd}

In this section, we will prove that the solution in the odd
case is given by the formulas in Section~\ref{sec:solutions-odd}.

First we give a lemma which will be used in the proof.
Then we show that killing the rightmost peakon in the interlacing case
with $(K+1)+(K+1)$ peakons, we obtain the solution formulas
for singletons given in Sections~\ref{sec:solutions-odd-X-singleton}
and~\ref{sec:solutions-odd-Y-singleton}.
Then, with a similar argument as in the
proof for the even case, it can be shown that the solution formulas for
the non-singletons groups in the odd case are as stated in
Sections~\ref{sec:solutions-odd-X-typical-group},
\ref{sec:solutions-odd-X-leftmost-group},
\ref{sec:solutions-odd-X-rightmost-group}
and~\ref{sec:solutions-odd-Y-group}.

\begin{lemma}
  \label{lem:heineintegral-killright}
  Let
  \begin{equation}
    \hat{\detJ}_{ij}^{rs} = \detJ[K+1,K,r,s,i,j]
    \quad \text{and} \quad
    \detJ_{ij}^{rs} = \detJ[K,K,r,s,i,j]
    .
  \end{equation}
  If we reparametrize the spectral data as
  \begin{equation}
    \label{eq:substitution-killright}
    \lambda_{K+1} = \frac{\alpha}{\epsilon}
    ,\qquad
    a_{K+1} = Z \, \epsilon^{k}
    ,
  \end{equation}
  then as $\epsilon \to 0$ the determinant $\hat{\detJ}_{ij}^{rs}$ satisfies
  \begin{equation}
    \hat{\detJ}^{rs}_{ij} =
    \detJ_{ij}^{rs}
    + \detJ_{i-1,j}^{rs} \, Z \, \alpha^{q} \, \epsilon^{k-q}
    \, \bigl( 1 + \mathcal{O}(\epsilon) \bigr)
    ,
  \end{equation}
  where
  $q = 2(i-1)-j+r$.
\end{lemma}

\begin{proof}
  The argument is similar to the proof of Lemma~\ref{lem:heineintegral-epsilon}.
  Since we are only redefining~$\lambda_{K+1}$, not~$\mu_{K}$,
  we only split the sum into two cases, according to whether $K+1 \in I$ or not.
\end{proof}

\subsection{How to kill the rightmost $Y$-peakon in the even interlacing configuration}
\label{sec:kill-rightmost}

First we will kill the rightmost $Y$-peakon.
Starting from the $(K+1)+(K+1)$ interlacing
solutions~\eqref{eq:interlacing-solution-positions}
and~\eqref{eq:interlacing-solution-amplitudes},
we know that the solution for the rightmost singleton $Y$-peakon (before the kill) is
\begin{equation}
  \label{eq:kill-rightmost-Y-1}
  \hat{Y}_{K+1} = \hat{\detJ}_{11}^{00} + \hat{D} \, \hat{\detJ}_{10}^{00}
  ,\qquad
  \hat{P}_{K+1} = \frac{1}{\hat{\detJ}_{10}^{00}}
  .
\end{equation}
We will use
$k = -1$,
$\alpha = 1$
and
$Z = D$
in the substitution~\eqref{eq:substitution-killright},
and we also redefine $\hat{C}$ and~$\hat{D}$, as follows:
\begin{equation}
  \label{eq:kill-rightmost-Y-substitutions}
  \lambda_{K+1} = \frac{1}{\epsilon}
  ,\qquad
  a_{K+1} = \frac{D}{\epsilon}
  ,\qquad
  \hat{D} = \frac{\theta \epsilon}{D}
  ,\qquad
  \hat{C} = \frac{C}{\epsilon}
  .
\end{equation}
Then, by Lemma~\ref{lem:heineintegral-killright},
the determinant $\hat{\detJ}_{11}^{00}$ becomes
\begin{equation}
  \label{eq:kill-rightmost-Y-2}
  \begin{aligned}
    \hat{\detJ}_{11}^{00} &=
    \detJ_{11}^{00} + Z \alpha^{q} \detJ_{01}^{00} \epsilon^{k-q} \, \bigl( 1 + \mathcal{O}(\epsilon) \bigr)
    \\ &
    = \detJ_{11}^{00} + D \detJ_{01}^{00}
    \bigl( 1 + \mathcal{O}(\epsilon) \bigr)
    ,
  \end{aligned}
\end{equation}
since $q=-1$ and $k-q=0$.
The other determinant $\hat{\detJ}_{10}^{00}$, with $q= 0$ and $k-q=-1$, becomes
\begin{equation}
  \label{eq:kill-rightmost-Y-3}
  \hat{\detJ}_{10}^{00}= D \epsilon^{-1} \detJ_{00}^{00} \bigl( 1 + \mathcal{O}(\epsilon) \bigr)
  .
\end{equation}
When we insert this into the formula for $\hat{Y}_{K+1}$
and let $\epsilon \to 0$, we get
\begin{equation}
  \label{eq:kill-rightmost-ghost}
  \begin{aligned}
    Y_{\text{ghost}}
    = \lim_{\epsilon \to 0} \hat{Y}_{K+1}
    &
    = \lim_{\epsilon \to 0} \left( \hat{\detJ}_{11}^{00} + \hat{D} \, \hat{\detJ}_{10}^{00} \right)
    \\ &
    = \detJ_{11}^{00} + D \detJ_{01}^{00} + \frac{\theta}{D} \, D \detJ_{00}^{00}
    \\ &
    = \detJ_{11}^{00} + D \detJ_{01}^{00}+ \theta
    .
  \end{aligned}
\end{equation}
And from \eqref{eq:kill-rightmost-Y-1} and~\eqref{eq:kill-rightmost-Y-3}
we see that~$\hat{P}_{K+1}$ tends to zero as $\epsilon \to 0$,
which shows that this peakon really becomes a ghostpeakon.

\subsubsection{What happens to the rightmost $X$-peakon}

Now let us see what happens to the rightmost $X$-peakon,
given by the case $j=1$ in Theorem~\ref{thm:interlacing-solution}
or in~\eqref{eq:even-X-typical-singleton}:
\begin{equation*}
  \hat{X}_{K+1}
  = \frac{\hat{\detJ}_{11}^{00}}{\hat{\detJ}_{00}^{11}}
  = \frac{\hat{\detJ}_{11}^{00}}{1}
  ,\qquad
  \hat{Q}_{K+1}
  = \frac{\hat{\detJ}_{00}^{11} \hat{\detJ}_{10}^{01}}{\hat{\detJ}_{11}^{10} \hat{\detJ}_{00}^{10}}
  = \frac{1 \cdot \hat{\detJ}_{10}^{01}}{\hat{\detJ}_{11}^{10} \cdot 1}
  .
\end{equation*}
From~\eqref{eq:kill-rightmost-Y-2} we have at once
\begin{equation}
  \hat{X}_{K+1}
  = \hat{\detJ}_{11}^{00}
  = \detJ_{11}^{00} + D \detJ_{01}^{00} \bigl( 1 + \mathcal{O}(\epsilon) \bigr)
  ,
\end{equation}
so letting $\epsilon \to 0$ we obtain
$X_{K+1} = \detJ_{11}^{00} + D \detJ_{01}^{00}$,
which is the formula for the position of the rightmost $X$-peakon
given in~\eqref{eq:odd-X-rightmost-singleton}.

Concerning the amplitude, using Lemma~\ref{lem:heineintegral-killright}
with $k=-1$, we have
$q=0$ and $k-q=-1$ for~$\hat{\detJ}_{10}^{01}$,
and the same for~$\hat{\detJ}_{11}^{10}$,
so
\begin{equation}
  Q_{K+1}
  = \lim_{\epsilon \to 0} \hat{Q}_{K+1}
  = \lim_{\epsilon \to 0} \frac{\hat{\detJ}_{10}^{01}}{\hat{\detJ}_{11}^{10}}
  = \lim_{\epsilon \to 0} \frac{\detJ_{00}^{01} D \epsilon^{-1} \bigl( 1 + \mathcal{O}(\epsilon) \bigr)}{\detJ_{01}^{10} D \epsilon^{-1} \bigl( 1 + \mathcal{O}(\epsilon) \bigr)}
  = \frac{\detJ_{00}^{01}}{\detJ_{01}^{10}}
  = \frac{1}{\detJ_{01}^{10}}
  ,
\end{equation}
which also agrees with~\eqref{eq:odd-X-rightmost-singleton}.

\subsubsection{What happens to the leftmost $X$-peakon}
\label{sec:kill-rightmost-what-happens-X1}

The formulas~\eqref{eq:even-X-leftmost-singleton}
for the leftmost $X$-peakon before the killing are
\begin{equation}
  \hat{X}_1 =
  \frac{\hat{\detJ}_{K+1,K}^{00}}{\hat{\detJ}_{K,K-1}^{11} + \hat{C} \, \hat{\detJ}_{KK}^{10}}
  ,
  \qquad
  \hat{Q}_1 =
  \frac{\mu_1 \dotsm \mu_{K}}{\lambda_1 \dotsm \lambda_{K} \lambda_{K+1}}
  \left(
    \frac{\hat{\detJ}_{K,K-1}^{11}}{\hat{\detJ}_{KK}^{10}}
    + \hat{C}
  \right)
  .
\end{equation}
For the determinants
$\hat{\detJ}^{00}_{K+1,K}$ and $\hat{\detJ}^{11}_{K,K-1}$
we have $q=K$ and $k-q=-K-1$,
and for $\hat{\detJ}^{10}_{KK}$ we have $q=K-1$ and $k-q=-K$.
Since $K$ is positive,
the smallest power of $\epsilon$ is going to be the second term
in Lemma~\ref{lem:heineintegral-killright} for all the determinants.
Using $\hat{C} = C / \epsilon$
from~\eqref{eq:kill-rightmost-Y-substitutions}, we therefore obtain
\begin{equation}
  \begin{split}
    X_1
    = \lim_{\epsilon \to 0} \hat{X}_1
    &
    = \lim_{\epsilon \to 0}
    \frac{\detJ_{KK}^{00} D \epsilon^{-K-1} \bigl( 1 + \mathcal{O}(\epsilon) \bigr)}{\detJ_{K-1,K-1}^{11} D \epsilon^{-K-1} \bigl( 1 + \mathcal{O}(\epsilon) \bigr) + \frac{C}{\epsilon} \, \detJ_{K-1,K}^{10} D \epsilon^{-K} \bigl( 1 + \mathcal{O}(\epsilon) \bigr)}
    \\ &
    = \frac{\detJ_{KK}^{00}}{\detJ_{K-1,K-1}^{11} + C \, \detJ_{K-1,K}^{10}}
  ,
  \end{split}
\end{equation}
which is the position formula for the leftmost peakon
in~\eqref{eq:odd-X-leftmost-singleton}.
Similarly for the amplitude formula.

\subsubsection{What happens to the other $X$-peakons}

Having investigated the rightmost and leftmost $X$-peakons,
we now turn to the typical $j$th $X$-peakon from the right
in the $(K+1)+(K+1)$ interlacing configuration,
given by~\eqref{eq:even-X-typical-singleton},
\begin{equation}
  \hat{X}_{j'}
  =
  \frac{ \hat{\detJ}_{jj}^{00}}{\hat{\detJ}_{j-1,j-1}^{11}}
  ,\qquad
  \hat Q_{j'}
  =
  \frac{\hat{\detJ}_{j-1,j-1}^{11} \hat{\detJ}_{j,j-1}^{01}}{\hat{\detJ}_{jj}^{10} \hat{\detJ}_{j-1,j-1}^{10}}
  ,
\end{equation}
for $2 \le j \le K$ and $j'=K+2-j$.

Looking first at the positions,
we get $q=j-2$ and $k-q=1-j$
in Lemma~\ref{lem:heineintegral-killright},
for both $\hat{\detJ}_{jj}^{00}$ and~$\hat{\detJ}_{j-1,j-1}^{11}$.
The condition $j \ge 2$ implies $1-j \le -1$.
Thus the smallest power of~$\epsilon$ is~$-1$,
and only the second term survives in the expressions
in Lemma~\ref{lem:heineintegral-killright},
so
\begin{equation}
  X_{K+2-j}
  = X_{j'}
  = \lim_{\epsilon \to 0} \hat{X}_{j'}
  = \lim_{\epsilon \to 0} \frac{\hat{\detJ}_{jj}^{00}}{\hat{\detJ}_{j-1,j-1}^{11}}
  = \frac{ \detJ_{j-1,j}^{00}}{\detJ_{j-2,j-1}^{11}}
\end{equation}
for $2 \le j \le K$, which agrees with the formula for the positions
given in~\eqref{eq:odd-X-typical-singleton}.
The formula for the amplitudes are proved similarly.

\subsubsection{What happens to the other $Y$-peakons}

Now let us consider the $Y$-peakons (except for the rightmost one that we have killed),
for which we recall the solution formulas~\eqref{eq:even-Y-typical-singleton},
\begin{equation}
  \hat{Y}_{j'}
  =
  \frac{\hat{\detJ}_{j,j-1}^{00}}{\hat{\detJ}_{j-1,j-2}^{11}}
  ,\qquad
  \hat{P}_{j'}
  =
  \frac{\hat{\detJ}_{j-1,j-1}^{10} \hat{\detJ}_{j-1,j-2}^{11}}{\hat{\detJ}_{j,j-1}^{01} \hat{\detJ}_{j-1,j-2}^{01}}
  ,
\end{equation}
for $2 \le j \le K+1$ and $j'=K+2-j$.

Using Lemma~\ref{lem:heineintegral-killright} for the positions, with $k=-1$, we find for
$\hat{\detJ}_{j,j-1}^{00}$ and~$\hat{\detJ}_{j-1,j-2}^{11}$
that $q=j-1$ and $k-q=-j$.
The smallest power of $\epsilon$ is $-j \le -2 < 0$,
so only the second term in Lemma~\ref{lem:heineintegral-killright}
will remain in the numerator and in the denominator, as $\epsilon \to 0$.
Writing $p'=K+1-p$ after the killing, since the number of $Y$-peakons has decreased by one,
we get, with $p=j-1$,
\begin{equation}
  Y_{p'}
  = Y_{K+1-p}
  = Y_{K+2-j}
  = \lim_{\epsilon \to 0} \hat{Y}_{K+2-j}
  = \lim_{\epsilon \to 0} \frac{\hat{\detJ}_{j,j-1}^{00}}{\hat{\detJ}_{j-1,j-2}^{11}}
  = \frac{\detJ_{j-1,j-1}^{00}}{\detJ_{j-2,j-2}^{11}}
  = \frac{\detJ_{pp}^{00}}{\detJ_{p-1,p-1}^{11}}
  ,
\end{equation}
for $1 \le p \le K$.
This proves the position formula given in~\eqref{eq:odd-Y-singleton},
and the amplitude formula is obtained with a similar argument.

\subsection{How to successively kill other peakons}
\label{sec:kill-others-odd}

We now know the solution formulas for the odd interlacing case.
Any non-interlacing odd configuration can be interspersed with auxiliary peakons
to make it interlacing (and odd).
Then we can successively kill off the peakons that we inserted,
from right to left,
to reach the desired configuration,
just as for the even case in Section~\ref{sec:proofs-even}.

At a generic stage of this process we have an odd number of groups,
say $(K+1)+K$,
with ``finished'' groups to the right of the peakon that we are going to kill,
and only singletons to the left.
The effectuate the killing, we let $\epsilon \to 0$
after making substitutions analogous to~\eqref{eq:substitution-kill-peakon-even},
namely
\begin{equation}
  \label{eq:substitution-kill-peakon-odd}
  \begin{aligned}
    \lambda_K &= \frac{\alpha}{\epsilon}
    ,
    &
    \mu_K &= \frac{\beta}{\epsilon}
    ,
    \\
    a_K &= Z \, \epsilon^{k_1}
    ,
    &
    b_K &= W \, \epsilon^{k_2}
    .
  \end{aligned}
\end{equation}
If we want to kill the $\hat{Y}_{K+1-j}$-singleton which constitutes the $j$th $Y$-group from the right
($1 \le j \le K$),
then we choose the parameters as follows:
\begin{equation}
  \label{eq:substitution-kill-Y-oddd}
  \begin{gathered}
    k_1 = j-2
    ,\qquad
    k_2 = j-1
    ,
    \\[1ex]
    \alpha = \frac{\theta}{\tau_1}
    ,\qquad
    \beta=1
    ,\qquad
    Z = \frac{(\alpha+1) \, \tau_1}{\alpha^{j-2}}
    = (\tau_1 + \theta) (\tau_1 / \theta)^{j-2}
    ,\qquad
    W =  \sigma_1
    ,
    \\[1ex]
    \text{where $\tau_1 = \tau_{K+1-j,1}^X$ and $\sigma_1 = \sigma_{K+1-j,1}^X$}
    .
  \end{gathered}
\end{equation}
And to kill the $\hat{X}_{K+1-j}$-singleton which constitutes $X$-group number $j+1$ from the right
($1 \le j \le K-1$),
we use
\begin{equation}
  \label{eq:substitution-kill-X-oddd}
  \begin{gathered}
    k_1 = j-2
    ,\qquad
    k_2 = j
    ,
    \\[1ex]
    \alpha=1
    ,\qquad
    \beta=\frac{\theta}{\tau_1}
    ,\qquad
    Z = \sigma_1
    ,\qquad
    W = \frac{(\beta + 1) \, \tau_1}{\beta^{j}}
    = (\tau_1 + \theta) (\tau_1 / \theta)^{j}
    ,
    \\[1ex]
    \text{where $\tau_1 = \tau_{K-j,1}^Y$ and $\sigma_1 = \sigma_{K-j,1}^Y$}
    .
  \end{gathered}
\end{equation}
There is never any need to kill the rightmost peakon~$\hat{X}_{K+1}$,
since that would just bring us back to even case,
nor the leftmost peakon~$\hat{X}_1$, since by assumption our configurations
always begin with an $X$-peakon.
For redefining the parameters $\hat C$, $\hat{D}$, $\hat{\sigma}_i$ and~$\hat{\tau}_i$
in each step,
we follow rules analogous to those in Section~\ref{sec:proofs-even}.
We omit the detailed verification that this really works as claimed,
since it is very similar to the even case,
and since we also sketch a more direct proof in the next section.

\subsection{An alternative approach}
\label{sec:alt-proof-odd-case}

A slightly different proof for the odd case,
which avoids redoing the work of killing off
peakons one by one, goes as follows.
In order to determine the solution formulas
for a given odd configuration with $(K+1)+K$ groups,
we add a single auxiliary $Y$-peakon on the far right, so that
we obtain an even configuration with $(K+1)+(K+1)$ groups,
for which the complete solution formulas are already known
(proved in Section~\ref{sec:proofs-even}).
Then, if we can kill that rightmost $Y$-peakon,
we will obtain all the formulas for the odd configuration in a single stroke.

This requires reparametrizing not only the spectral parameters before
letting $\epsilon \to 0$,
but also the internal parameters ($\tau$ and~$\sigma$) in every non-singleton group.
It can be verified that
the substitution which does the trick is to transform the spectral variables
just as in~\eqref{eq:kill-rightmost-Y-substitutions}, namely
\begin{equation}
  \lambda_{K+1} = \frac{1}{\epsilon}
  ,\qquad
  a_{K+1} = \frac{D}{\epsilon}
  ,\qquad
  \hat{D} = \frac{\theta \epsilon}{D}
  ,\qquad
  \hat{C} = \frac{C}{\epsilon}
  ,
\end{equation}
with the other spectral variables unchanged,
and to transform all internal parameters as
\begin{equation}
  \label{eq:kill-rightmost-Y-theta-sigma}
  \hat{\tau} = \frac{\tau}{\epsilon^2}
  ,\qquad
  \hat{\sigma} = \sigma \epsilon
  ,
\end{equation}
where $\tau$ stands for any $\tau_{j,i}^X$ or~$\tau_{j,i}^Y$,
and $\sigma$ stands for any $\sigma_{j,i}^X$ or~$\sigma_{j,i}^Y$.

We have already shown in Section~\ref{sec:kill-rightmost}
that this does the right thing for all singletons.
We omit most of the details for the groups,
showing only what becomes of~$\hat{X}_{1,1}$, as an example.
The formula before the killing is~\eqref{eq:even-X-leftmost-group-pos},
\begin{equation}
  \hat{X}_{1,1}
  = \frac{\hat{\detJ}_{K+1,K}^{00}}{\hat{\detJ}_{K,K-1}^{11} + \frac{1}{\hat{\sigma}_1} \hat{\detJ}_{KK}^{11} + \hat{C} \left( \hat{\detJ}_{KK}^{10} + \frac{1}{\hat{\tau}_1} \hat{\detJ}_{K+1,K}^{10} \right)}
  ,
\end{equation}
and in a similar way as in Section~\ref{sec:kill-rightmost-what-happens-X1}
we find
\begin{equation}
  \begin{split}
    X_{1,1}
    &
    = \lim_{\epsilon \to 0} \hat{X}_{1,1}
    \\
    &
    = \lim_{\epsilon \to 0}
    \frac{\detJ_{KK}^{00} D \epsilon^{-K-1} \bigl( 1 + \mathcal{O}(\epsilon) \bigr)}
    {\left[
        \begin{aligned}
          & \detJ_{K-1,K-1}^{11} D \epsilon^{-K-1} \bigl( 1 + \mathcal{O}(\epsilon) \bigr) + \tfrac{1}{\sigma_1 \epsilon} \detJ_{K-1,K}^{11} D \epsilon^{-K} \bigl( 1 + \mathcal{O}(\epsilon) \bigr)
          \\ &
          + \tfrac{C}{\epsilon} \left( \detJ_{K-1,K}^{10} D \epsilon^{-K} \bigl( 1 + \mathcal{O}(\epsilon) \bigr) + \tfrac{\epsilon^2}{\tau_1} \detJ_{KK}^{10} D \epsilon^{-K-2} \bigl( 1 + \mathcal{O}(\epsilon) \bigr) \right)
        \end{aligned}
      \right]}
    \\ &
    = \frac{\detJ_{KK}^{00}}{\detJ_{K-1,K-1}^{11} + \frac{1}{\sigma_1} \detJ_{K-1,K}^{11} + C \left( \detJ_{K-1,K}^{10} + \frac{1}{\tau_1} \detJ_{KK}^{10} \right)}
  ,
  \end{split}
\end{equation}
in agreement with~\eqref{eq:odd-X-leftmost-group-pos}.

\section{Asymptotics  for  the odd case}
\label{sec:asymptotics-odd}

In this section, we study the limiting behavior for the odd case
($2K+1$ groups) of the Geng--Xue peakon solutions as $t \to \pm \infty$.
First, we state the asymptotic formulas for singletons in
Theorem~\ref{thm:asymptotics-singletons-odd}, then the asymptotics for
non-singleton groups in Theorem~\ref{thm:asymptotics-positions-odd}
for the positions
and Theorem~\ref{thm:asymptotics-amplitudes-odd}
for the amplitudes.
We omit the proofs, since the calculations are very similar to those
in Section~\ref{sec:asymptotics-even}
where we studied the even case.

As always, we number the eigenvalues in increasing order:
\begin{equation}
  \label{eq:eigenvalues-ordered-odd}
  0 < \lambda_1 < \lambda_2 < \dots < \lambda_K
  , \qquad
  0 < \mu_1 < \mu_2 < \dots < \mu_{K}
  .
\end{equation}
Recall the notation for the products of all eigenvalues:
\begin{equation*}
  L = \lambda_{[1,K]} = \lambda_1 \dotsm \lambda_K
  ,\qquad
  M = \mu_{[1,K]} = \mu_1 \dotsm \mu_K
  .
\end{equation*}

\begin{theorem}
  \label{thm:asymptotics-singletons-odd}

  Any singletons in a solution with $(K+1)+K$ groups,
  and in particular all peakons in the $(K+1)+K$ interlacing solution,
  where $K \ge 1$ and the eigenvalues are ordered as in~\eqref{eq:eigenvalues-ordered-odd},
  satisfy the following asymptotic formulas.
  
  \begin{itemize}
  \item As $t \to +\infty$:

    All $X$-peakons except the leftmost and rightmost ones
    ($j'=K+2-j$ with $2 \le j \le K$), if $K \ge 2$:
    \begin{equation}
      \begin{aligned}
	x_{j'}(t) &=
        \frac{t}{2} \left( \frac{1}{\lambda_{j-1}} + \frac{1}{\mu_j} \right)
        + \frac12 \ln \left( \frac{2 a_{j-1}(0) \, b_j(0) \, \Psi_{[1,j-1][1,j]}}{\lambda_{[1,j-2]} \, \mu_{[1,j-1]} \, \Psi_{[1,j-2][1,j-1]}} \right)
        + o(1)
	,
	\\[1ex]
        \ln m_{j'}(t) &=
        \frac{t}{2} \left( \frac{1}{\lambda_{j-1}} - \frac{1}{\mu_j} \right)
        + \frac12 \ln \left( \frac{2 a_{j-1}(0) \, b_j(0) \, \Psi_{[1,j-1][1,j]}}{\lambda_{[1,j-2]} \, \mu_{[1,j-1]} \, \Psi_{[1,j-2][1,j-1]}} \right)
	\\ & \quad
	+ \ln \left( \frac{\mu_{[1,j-1]}^2 \, \Psi_{[1,j-1][1,j-1]}}{2 b_j(0) \, \lambda_{[1,j-1]} \, \Psi_{[1,j-1][1,j]}} \right)
	+ o(1)
	.
      \end{aligned}
    \end{equation}
    The leftmost $X$-peakon:
    \begin{equation}
      \begin{aligned}
	x_1(t) &=
        \frac{t}{2} \left( \frac{1}{\lambda_{K}} \right)
        + \frac12 \ln \left( \frac{2 a_{K}(0) \, \Psi_{[1,K][1,K]}}{C \, \lambda_{[1,K-1]} \, \Psi_{[1,K-1][1,K]}} \right)
        + o(1)
	,
	\\[1ex]
        \ln m_1(t) &=
        \frac{t}{2} \left( \frac{1}{\lambda_{K}} \right)
        + \frac12 \ln \left( \frac{2 a_K(0) \, \Psi_{[1,K][1,K]}}{C \, \lambda_{[1,K-1]} \, \Psi_{[1,K-1][1,K]}} \right)
	+ \ln \left( \frac{C M}{2L} \right)
	+ o(1)
	.
      \end{aligned}
    \end{equation}
    The rightmost $X$-peakon:
    \begin{equation}
      \begin{aligned}
	x_{K+1}(t) &=
        \frac{t}{2} \left( \frac{1}{\lambda_1} + \frac{1}{\mu_1} \right)
        + \frac12 \ln \bigl( 2 a_1(0) \, b_1(0) \, \Psi_{\{1\}\{1\}} \bigr)
        + o(1)
        ,
	\\[1ex]
        \ln m_{K+1}(t) &=
        \frac{t}{2} \left( \frac{1}{\lambda_1} - \frac{1}{\mu_1} \right)
        + \frac12 \ln \bigl( 2 a_1(0) \, b_1(0) \, \Psi_{\{1\}\{1\}} \bigr)
	+\ln \left( \frac{1}{2 b_1(0)} \right)
	+ o(1)
	.
      \end{aligned}
    \end{equation}
    All $Y$-peakons ($j'=K+1-j$ with $1 \le j \le K$):
    \begin{equation}
      \begin{aligned}
    	y_{j'}(t) &=
        \frac{t}{2} \left( \frac{1}{\lambda_{j}} + \frac{1}{\mu_j} \right)
        + \frac12 \ln \left( \frac{2 a_{j}(0) \, b_j(0) \, \Psi_{[1,j][1,j]}}{\lambda_{[1,j-1]} \, \mu_{[1,j-1]} \, \Psi_{[1,j-1][1,j-1]}} \right)
        + o(1)
        ,
        \\[1ex]
        \ln n_{j'}(t) &=
        \frac{t}{2} \left( \frac{1}{\mu_j} - \frac{1}{\lambda_{j}} \right)
        + \frac12 \ln \left( \frac{2 a_{j}(0) \, b_j(0) \, \Psi_{[1,j][1,j]}}{\lambda_{[1,j-1]} \, \mu_{[1,j-1]} \, \Psi_{[1,j-1][1,j-1]}} \right)
        \\ & \quad
        + \ln \left( \frac{\lambda_{[1,j-1]}^2 \, \Psi_{[1,j-1][1,j]}}{2 a_j(0) \, \mu_{[1,j]} \, \Psi_{[1,j][1,j]}} \right)
        + o(1)
        .
      \end{aligned}
    \end{equation}

  \item As $t \to -\infty$:

    All $X$-peakons except the leftmost and rightmost ones
    ($2 \le j \le K$), if $K \ge 2$:
    \begin{equation}
      \begin{aligned}
	x_j(t) &=
        \frac{t}{2} \left( \frac{1}{\lambda_{j}} + \frac{1}{\mu_{j-1}} \right)
        + \frac12 \ln \left( \frac{2 a_{j}(0) \, b_{j-1}(0) \, \Psi_{[j,K][j-1,K]}}{\lambda_{[j+1,K]} \, \mu_{[j,K]} \, \Psi_{[j+1,K][j,K]}} \right)
        + o(1)
	,
	\\[1ex]
        \ln m_j(t) &=
        \frac{t}{2} \left( \frac{1}{\lambda_{j}} - \frac{1}{\mu_{j-1}} \right)
        + \frac12 \ln \left( \frac{2 a_{j}(0) \, b_{j-1}(0) \, \Psi_{[j,K][j-1,K]}}{\lambda_{[j+1,K]} \, \mu_{[j,K]} \, \Psi_{[j+1,K][j,K]}} \right)
	\\ & \quad
	+ \ln \left( \frac{\mu_{[j,K]}^2 \, \Psi_{[j,K][j,K]}}{2 b_{j-1}(0) \, \lambda_{[j,K]} \, \Psi_{[j,K][j-1,K]}} \right)
	+ o(1)
	.
      \end{aligned}
    \end{equation}
    The leftmost $X$-peakon:
    \begin{equation}
      \begin{aligned}
	x_1(t) &=
        \frac{t}{2} \left( \frac{1}{\lambda_1} + \frac{1}{\mu_1} \right)
        + \frac12 \ln \left( \frac{2 a_1(0) \, b_1(0) \, \Psi_{[1,K][1,K]}}{\lambda_{[2,K]} \, \mu_{[2,K]} \, \Psi_{[2,K][2,K]}} \right)
        + o(1)
	,
	\\[1ex]
        \ln m_1(t) &=
        \frac{t}{2} \left( \frac{1}{\lambda_1} - \frac{1}{\mu_1} \right)
        + \frac12 \ln \left( \frac{2 a_1(0) \, b_1(0) \, \Psi_{[1,K][1,K]}}{\lambda_{[2,K]} \, \mu_{[2,K]} \, \Psi_{[2,K][2,K]}} \right)
	\\ & \quad
	+ \ln \left( \frac{\mu_{[2,K]} \, M \, \Psi_{[2,K][2,K]}}{2 b_1(0) \, L \, \Psi_{[2,K][1,K]}} \right)
	+ o(1)
	.
      \end{aligned}
    \end{equation}
    The rightmost $X$-peakon:
    \begin{equation}
      \begin{aligned}
	x_{K+1}(t) &=
        \frac{t}{2} \left( \frac{1}{\mu_{K}} \right)
        + \frac12 \ln \bigl( 2 D \, b_{K}(0) \bigr)
        + o(1)
	,
	\\[1ex]
        \ln m_{K+1}(t) &=
        \frac{t}{2} \left( \frac{-1}{\mu_{K}} \right)
        + \frac12 \ln \bigl( 2 D \, b_{K}(0) \bigr)
	+ \ln \left( \frac{1}{2 b_{K}(0)} \right)
	+ o(1)
	.
      \end{aligned}
    \end{equation}
    All $Y$-peakons ($1 \le j \le K$):
    \begin{equation}
      \begin{aligned}
	y_j(t) &=
        \frac{t}{2} \left( \frac{1}{\lambda_{j}} + \frac{1}{\mu_{j}} \right)
        + \frac12 \ln \left( \frac{2 a_{j}(0) \, b_{j}(0) \, \Psi_{[j,K][j,K]}}{\lambda_{[j+1,K]} \, \mu_{[j+1,K]} \, \Psi_{[j+1,K][j+1,K]}} \right)
        + o(1)
	,
	\\[1ex]
        \ln n_j(t) &=
        \frac{t}{2} \left( \frac{1}{\mu_{j}} - \frac{1}{\lambda_{j}} \right)
        + \frac12 \ln \left( \frac{2 a_{j}(0) \, b_{j}(0) \, \Psi_{[j,K][j,K]}}{\lambda_{[j+1,K]} \, \mu_{[j+1,K]} \, \Psi_{[j+1,K][j+1,K]}} \right)
	\\ & \quad
	+ \ln \left( \frac{\lambda_{[j+1,K]}^2 \, \Psi_{[j+1,K][j,K]}}{2 a_{j}(0) \, \mu_{[j,K]} \, \Psi_{[j,K][j,K]}} \right)
	+ o(1)
	.
      \end{aligned}
    \end{equation}
  \end{itemize}
\end{theorem}

\begin{theorem}
  \label{thm:asymptotics-positions-odd}

  The positions for non-singleton groups in the odd case,
  with the eigenvalues ordered as in~\eqref{eq:eigenvalues-ordered-odd},
  satisfy the following asymptotic formulas.
  
  \begin{itemize}
  \item As $t \to +\infty$:

    All $X$-groups except the leftmost and rightmost ones
    ($j'=K+2-j$ with $2 \le j \le K$), if $K \ge 2$:
    \begin{equation}
      \begin{aligned}
	X_{j',i}(t) &=
        \frac{t}{2} \left( \frac{1}{\lambda_{j}} + \frac{1}{\mu_j} \right)
        + \frac12 \ln \left( \frac{2 a_{j}(0) \, b_j(0) \, \Psi_{[1,j][1,j]}}{\lambda_{[1,j-1]} \, \mu_{[1,j-1]} \, \Psi_{[1,j-1][1,j-1]}} \right)
        + o(1)
        ,
	\\ &
        \quad \text{for $1 \le i \le N_{j'}-1$}
        ,
        \\[1ex]
	X_{j',N_{j'}}(t) &=
        \frac{t}{2} \left( \frac{1}{\lambda_{j-1}} + \frac{1}{\mu_{j}} \right)
        + \frac12 \ln \left( \frac{2 a_{j-1}(0) \, b_{j}(0) \, \Psi_{[1,j-1][1,j]}}{\lambda_{[1,j-2]} \, \mu_{[1,j-1]} \, \Psi_{[1,j-2][1,j-1]}} \right)
        + o(1)
        .
      \end{aligned}
    \end{equation}
    The leftmost $X$-group:
    \begin{equation}
      \begin{aligned}
	X_{1,1}(t) &=
        \frac12 \ln \left( \frac{2 \tau_1}{C L} \right)
        + o(1)
        ,
        \\
	X_{1,i}(t) &=
        \frac12 \ln \left( \frac{2S_i}{L M} \right)
        + o(1)
        ,
        \quad
        \text{for $2 \le i \le N_1-1$}
        ,
	\\[1ex]
	X_{1,N_1}(t) &=
        \frac{t}{2} \left( \frac{1}{\lambda_K} \right) + \frac12 \ln \left( \frac{2 \sigma_{N_1-1} \, a_K(0) \, \Psi_{[1,K][1,K]}}{\lambda_{[1,K-1]} \, M \, \Psi_{[1,K-1][1,K]}} \right)
        + o(1)
        .
      \end{aligned}
    \end{equation}
    The rightmost $X$-group:
    \begin{equation}
      X_{K+1,i}(t) =
      \frac{t}{2} \left( \frac{1}{\lambda_1} + \frac{1}{\mu_1} \right)
      + \frac12 \ln \bigl( 2 a_1(0) \, b_1(0) \, \Psi_{\{1\}\{1\}} \bigr)
      + o(1)
      , \quad
      \text{for $1 \le i \le N_{j'}$}
      .
    \end{equation}
    All $Y$-groups except the leftmost one
    ($j'=K+1-j$ with $1 \le j \le K-1$), if $K \ge 2$:
    \begin{equation}
      \begin{aligned}
	Y_{j',i}(t) &=
        \frac{t}{2} \left( \frac{1}{\lambda_{j}} + \frac{1}{\mu_{j+1}} \right)
        + \frac12 \ln \left( \frac{2 a_{j}(0) \, b_{j+1}(0) \, \Psi_{[1,j][1,j+1]}}{\lambda_{[1,j-1]} \, \mu_{[1,j]} \, \Psi_{[1,j-1][1,j]}} \right)
        + o(1)
        ,
	\\ &
        \quad \text{for $1 \le i \le N_{j'}-1$}
        ,
	\\[1ex]
	Y_{j',N_{j'}}(t) &=
        \frac{t}{2} \left( \frac{1}{\lambda_{j}} + \frac{1}{\mu_{j}} \right)
        + \frac12 \ln \left( \frac{2 a_{j}(0) \, b_{j}(0) \, \Psi_{[1,j][1,j]}}{\lambda_{[1,j-1]} \, \mu_{[1,j-1]} \,  \Psi_{[1,j-1][1,j-1]}} \right)
        + o(1)
        .
      \end{aligned}
    \end{equation}
    The leftmost $Y$-group:
    \begin{equation}
      \begin{aligned}
	y_{1,i}(t) &=
        \frac{t}{2} \left( \frac{1}{\lambda_{K}} \right)
        + \frac12 \ln \left( \frac{2 T_i \, a_K(0) \, \Psi_{[1,K][1,K]}}{\lambda_{[1,K-1]} \, M \, \Psi_{[1,K-1][1,K]} } \right)
        + o(1)
        ,
        \\ &
        \quad \text{ for  $1 \le i \le N_1-1$}
        ,
        \\[1ex]
	y_{1,N_1}(t) &=
        \frac{t}{2} \left( \frac{1}{\lambda_{K}} + \frac{1}{\mu_K} \right)
        + \frac12 \ln \left( \frac{2 a_{K}(0) \, b_K(0) \, \Psi_{[1,K][1,K]}}{\lambda_{[1,K-1]} \, M \, \Psi_{[1,K-1][1,K-1]}} \right)
        + o(1)
        .
      \end{aligned}
    \end{equation}

  \item As $t \to -\infty$:

    All $X$-groups except the leftmost and the rightmost ones ($2 \le j \le K$), if $K \ge 2$:
    \begin{equation}
      \begin{aligned}
	x_{j,1}(t) &=
        \frac{t}{2} \left( \frac{1}{\lambda_{j}} + \frac{1}{\mu_{j-1}} \right)
        + \frac12 \ln \left( \frac{2 a_{j}(0) \, b_{j-1}(0) \, \Psi_{[j,K][j-1,K]}}{\lambda_{[j+1,K]} \, \mu_{[j,K]} \, \Psi_{[j+1,K][j,K]}} \right)
        + o(1)
        ,
	\\
	x_{j,i}(t) &=
        \frac{t}{2} \left( \frac{1}{\lambda_{j}} + \frac{1}{\mu_{j}} \right)
        + \frac12 \ln \left( \frac{2 a_{j}(0) \, b_{j}(0) \, \Psi_{[j,K][j,K]}}{\lambda_{[j+1,K]} \, \mu_{[j+1,K]} \, \Psi_{[j+1,K][j+1,K]}} \right)
        + o(1)
        ,
        \\ &
        \quad \text{for $2 \le i \le N_j$}
        .
      \end{aligned}
    \end{equation}
    The leftmost $X$-group:
    \begin{equation}
      \begin{aligned}
	x_{1,i}(t) &=
        \frac{t}{2} \left( \frac{1}{\lambda_1} + \frac{1}{\mu_1} \right)
        + \frac12 \ln \left( \frac{2 a_1(0) \, b_1(0) \, \Psi_{[1,K][1,K]}}{\lambda_{[2,K]} \, \mu_{[2,K]} \, \Psi_{[2,K][2,K]}} \right)
        + o(1)
        ,
	\\ &
        \quad \text{for $1 \le i \le N_1$}
        .
      \end{aligned}
    \end{equation}
    The rightmost $X$-group:
    \begin{equation}
      \begin{aligned}
	x_{K+1,1}(t) &=
        \frac{t}{2} \left( \frac{1}{\mu_{K}} \right)
        + \frac12 \ln \bigl( 2 T_1 \, b_K(0) \bigr)
        + o(1)
        ,
	\\
	x_{K+1,i}(t) &=
        \frac{1}{2} \ln(2 S_i)
        + o(1)
        ,
	\quad
        \text{for $2 \le i \le N_{K+1}-1$}
        ,
	\\
	x_{K+1,N_{K+1}}(t) &=
        \frac{1}{2} \ln \bigl( 2 (S_{N_{K+1}-1} + D \, \sigma_{N_{K+1}-1}) \bigr)
        + o(1)
        .
      \end{aligned}
    \end{equation}
    All $Y$-groups except the rightmost one ($1 \le j \le K-1$), if $K \ge 2$:
    \begin{equation}
      \begin{aligned}
	y_{j,1}(t) &=
        \frac{t}{2} \left( \frac{1}{\lambda_{j}} + \frac{1}{\mu_{j}} \right)
        + \frac12 \ln \left( \frac{2 a_{j}(0) \, b_{j}(0) \, \Psi_{[j,K][j,K]}}{\lambda_{[j+1,K]} \, \mu_{[j+1,K]} \, \Psi_{[j+1,K][j+1,K]}} \right)
        + o(1)
	,
	\\
	y_{j,i}(t) &=
        \frac{t}{2} \left( \frac{1}{\lambda_{j+1}} + \frac{1}{\mu_{j}} \right)
        + \frac12 \ln \left( \frac{2 a_{j+1}(0) \, b_{j}(0) \, \Psi_{[j+1,K][j,K]}}{\lambda_{[j+2,K]} \, \mu_{[j+1,K]} \, \Psi_{[j+2,K][j+1,K]} }
        \right)
        + o(1)
        ,
	\\ &
        \quad \text{for $2 \le i \le N_j$}
        .
      \end{aligned}
    \end{equation}
    The rightmost $Y$-group:
    \begin{equation}
      \begin{aligned}
        y_{K,1}(t) &=
        \frac{t}{2} \left( \frac{1}{\mu_K} +\frac{1}{\lambda_K} \right)
        + \frac12 \ln \left( 2 a_K(0) \, b_k(0) \, \Psi_{\{K\}\{K\}} \right)
        + o(1)
        ,
        \\
        y_{K,i}(t) &=
        \frac{t}{2} \left( \frac{1}{\mu_K} \right)
        + \frac12 \ln \left( \frac{2 S_i \, b_K(0)}{T_i} \right)
        + o(1)
        ,
        \quad
        \text{for $2 \le i \le N_K-1$}
        ,
        \\
        y_{K,N_K}(t) &=
        \frac{t}{2} \left( \frac{1}{\mu_K} \right)
        + \frac12 \ln \bigl( 2 \sigma_{N_K-1} \, b_K(0) \bigr)
        + o(1)
        .
      \end{aligned}
    \end{equation}
  \end{itemize}
\end{theorem}

\begin{theorem}
  \label{thm:asymptotics-amplitudes-odd}

  The amplitudes for non-singleton groups in the odd case,
  with the eigenvalues ordered as in~\eqref{eq:eigenvalues-ordered-odd},
  satisfy the following  asymptotic formulas.
  
  \begin{itemize}
  \item As $t \to +\infty$:

    All $X$-groups except the leftmost and the rightmost ones
    ($j'=K+2-j$ with $2 \le j \le K$), if $K \ge 2$:
    \begin{equation}
      \begin{aligned}
        \ln m_{j',i}(t) &=
        \frac{t}{2} \left( \frac{1}{\lambda_{j}}-\frac{3}{\mu_j} \right)
        + \frac12 \ln \left( \frac{2 a_{j}(0) \, b_j(0) \, \Psi_{[1,j][1,j]}}{\lambda_{[1,j-1]} \, \mu_{[1,j-1]} \, \Psi_{[1,j-1][1,j-1]}} \right)
        \\ & \quad
        + \ln \left( \frac{(\sigma_i - \sigma_{i-1}) \, \mu_{[1,j-1]}^2 \, \Psi_{[1,j-1][1,j-1]}^2}{2 b_j^2(0) \, \lambda_{[1,j-1]} \, \Psi_{[1,j-1][1,j]}^2} \right)
        + o(1)
        ,
        \\ &
        \quad \text{for $1 \le i \le N_{j'}-1$}
        ,
        \\[1ex]
        \ln m_{j',N_{j'}}(t) &=
        \frac{t}{2} \left( \frac{1}{\lambda_{j-1}} - \frac{1}{\mu_{j}} \right)
        + \frac12 \ln \left( \frac{2 a_{j-1}(0) \, b_{j}(0) \, \Psi_{[1,j-1][1,j]}}{\lambda_{[1,j-2]} \, \mu_{[1,j-1]} \, \Psi_{[1,j-2][1,j-1]} } \right)
        \\ & \quad
        + \ln \left( \frac{\mu_{[1,j-1]}^2 \, \Psi_{[1,j-1][1,j-1]}}{2 b_j(0) \, \lambda_{[1,j-1]} \, \Psi_{[1,j-1][1,j]}} \right)
        + o(1)
        .
      \end{aligned}
    \end{equation}
    The leftmost $X$-group:
    \begin{equation}
      \begin{aligned}
        \ln m_{1,1}(t) &=
        \frac12 \ln \left( \frac{2 \tau_1}{C L} \right)
        + \ln \left( \frac{C M}{2L} \right)
        + o(1)
        ,
        \\[1ex]
        \ln m_{1,i}(t) &=
        \frac12 \ln \left( \frac{S_i}{L M} \right)
        + \ln \left( \frac{(\sigma_i - \sigma_{i-1}) \, M^2}{2 \sigma_i \, \sigma_{i-1} \, L} \right)
        + o(1)
        ,
        \\ &
        \quad \text{for $2 \le i \le N_1-1$}
        ,
        \\[1ex]
        \ln m_{1,N_1}(t) &=
        \frac{t}{2} \left( \frac{1}{\lambda_{K}} \right)
        + \frac12 \ln \left( \frac{2 \sigma_{N_1-1} \, a_K(0) \, \Psi_{[1,K][1,K]}}{\lambda_{[1,K-1]} \, M \, \Psi_{[1,K-1][1,K]}} \right)
        \\ & \quad
        +\ln \left( \frac{M^2}{2 \sigma_{N_1-1} \, L} \right)
        + o(1)
        .
      \end{aligned}
    \end{equation}
    The rightmost $X$-group:
    \begin{equation}
      \begin{aligned}
        \ln m_{K+1,i}(t) &=
        \frac{t}{2} \left( \frac{1}{\lambda_1} - \frac{3}{\mu_1} \right)
        + \frac12 \ln \bigl( 2 a_1(0) \, b_1(0) \, \Psi_{\{1\}\{1\}} \bigr)
        \\ & \quad
        +\ln \left( \frac{\sigma_i-\sigma_{i-1}}{2 b_1^2(0)} \right)
        + o(1)
        ,
        \quad
        \text{for $1 \le i \le N_{K+1}-1$}
        ,
        \\
        \ln m_{K+1,N_{K+1}}(t) &=
        \frac{t}{2} \left( \frac{1}{\lambda_1} - \frac{1}{\mu_1} \right)
        + \frac12 \ln \bigl( 2 a_1(0) \, b_1(0) \, \Psi_{\{1\}\{1\}} \bigr)
        \\ & \quad
        + \ln \left( \frac{1}{2 b_1(0)} \right)
        + o(1)
        .
      \end{aligned}
    \end{equation}
    All $Y$-groups except the leftmost one
    ($j'=K+1-j$ with $1 \le j \le K-1$), if $K \ge 2$:
    \begin{equation}
      \begin{aligned}
        \ln n_{j',i}(t) &=
        \frac{t}{2} \left( \frac{1}{\mu_{j+1}} - \frac{3}{\lambda_j} \right)
        + \frac12 \ln \left( \frac{2 a_j(0) \, b_{j+1}(0) \, \Psi_{[1,j][1,j+1]}}{\lambda_{[1,j-1]} \, \mu_{[1,j]} \, \Psi_{[1,j-1][1,j]}} \right)
        \\ & \quad
        + \ln \left( \frac{(\sigma_i - \sigma_{i-1}) \, \lambda_{[1,j-1]}^2 \, \Psi_{[1,j-1][1,j]}^2}{2 a_j^2(0) \, \mu_{[1,j]} \, \Psi_{[1,j][1,j]}^2} \right)
        + o(1)
        ,
        \\[1ex]
        \ln n_{j',N_{j'}}(t) &=
        \frac{t}{2} \left( \frac{1}{\mu_{j}} - \frac{1}{\lambda_j} \right)
        + \frac12 \ln \left( \frac{2 a_j(0) \, b_{j}(0) \, \Psi_{[1,j][1,j]}}{\lambda_{[1,j-1]} \, \mu_{[1,j-1]} \, \Psi_{[1,j-1][1,j-1]} } \right)
        \\ & \quad
        + \ln \left( \frac{\lambda_{[1,j-1]}^2 \, \Psi_{[1,j-1][1,j]}}{2 a_j(0) \, \mu_{[1,j]} \, \Psi_{[1,j][1,j]}} \right)
        + o(1)
        ,
        \quad \text{for $1 \le i \le N_{j'}-1$}
        .
      \end{aligned}
    \end{equation}
    The leftmost $Y$-group:
    \begin{equation}
      \begin{aligned}
        \ln n_{1,i}(t) &=
        \frac{t}{2} \left( \frac{-3}{\lambda_{K}} \right)
        + \frac12 \ln \left( \frac{2 T_i \, a_K(0) \, \Psi_{[1,K][1,K]}}{\lambda_{[1,K-1]} \, M \, \Psi_{[1,K-1][1,K]}} \right)
        \\ & \quad
        + \ln \left( \frac{(\sigma_i - \sigma_{i-1}) \, \lambda_{[1,K-1]}^2 \, \Psi_{[1,K-1][1,K]}^2}{2 a_K^2(0) \, M \, \Psi_{[1,K][1,K]}^2} \right)
        + o(1)
        ,
        \\[1ex]
        \ln n_{1,N_1}(t) &=
        \frac{t}{2} \left( \frac{1}{\mu_K} - \frac{1}{\lambda_{K}} \right)
        + \frac12 \ln \left( \frac{2 a_{K}(0) \, b_K(0) \, \Psi_{[1,K][1,K]}}{\lambda_{[1,K-1]} \, M \, \Psi_{[1,K-1][1,K-1]}} \right)
        \\ & \quad
        + \ln \left( \frac{\lambda_{[1,K-1]}^2 \, \Psi_{[1,K-1][1,K]}}{2 a_K(0) \, M \, \Psi_{[1,K][1,K]}} \right)
        + o(1)
        ,
        \quad \text{for $1 \le i \le N_1-1$}
        .
      \end{aligned}
    \end{equation}

  \item As $t \to -\infty$:

    All $X$-groups except the leftmost and the rightmost ones
    ($2 \le j \le K$), if $K \ge 2$:
    \begin{equation}
      \begin{aligned}
        \ln m_{j,1}(t) &=
        \frac{t}{2} \left( \frac{1}{\lambda_{j}} - \frac{1}{\mu_{j-1}} \right)
        + \frac12 \ln \left( \frac{2 a_{j}(0) \, b_{j-1}(0) \, \Psi_{[j,K][j-1,K]}}{\lambda_{[j+1,K]} \, \mu_{[j,K]} \, \Psi_{[j+1,K][j,K]}} \right)
        \\ & \quad
        +\ln \left( \frac{\mu^2_{[j,K]} \, \Psi_{[j,K][j,K]}}{2 b_{j-1}(0) \, \lambda_{[j,K]} \, \Psi_{[j,K][j-1,K]}} \right)
        + o(1)
        ,
        \\[1ex]
        \ln m_{j,i}(t) &=
        \frac{t}{2} \left( \frac{3}{\lambda_{j}} - \frac{1}{\mu_{j}} \right)
        + \frac12 \ln \left( \frac{2 a_{j}(0) \, b_{j}(0) \, \Psi_{[j,K][j,K]}}{\lambda_{[j+1,K]} \, \mu_{[j+1,K]} \, \Psi_{[j+1,K][j+1,K]}} \right)
        \\ & \quad
        + \ln \left( \frac{(\sigma_i -\sigma_{i-1}) \, S_i}{R_i \, R_{i-1}} \right)
        + \ln \left( \frac{a_j(0) \, \mu_{[j,K]} \, \mu_{[j+1,K]}}{b_j(0) \, \lambda_{[j+1,K]}} \right)
        \\ & \quad
        + \ln \left( \frac{\Psi_{[j,K][j,K]} \, \Psi_{[j+1,K][j+1,K]}}{2 \Psi_{[j+1,K][j,K]}^2} \right)
        + o(1)
        ,
        \\ &
        \quad \text{for $2 \le i \le N_j-1$}
        ,
        \\
        \ln m_{j,N_j}(t) &=
        \frac{t}{2} \left( \frac{3}{\lambda_{j}} - \frac{1}{\mu_j} \right)
        + \frac12 \ln \left( \frac{2 a_j(0) \, b_j(0) \, \Psi_{[j,K][j,K]}}{\lambda_{[j+1,K]} \, \mu_{[j+1,K]} \, \Psi_{[j+1,K][j+1,K]}} \right)
        \\ & \quad
        + \ln \left( \frac{\sigma_{N_j-1}}{R_{N_j-1} } \right)
        + \ln \left( \frac{a_j(0) \, \mu_{[j,K]} \, \mu_{[j+1,K]}}{b_j(0) \, \lambda_{[j+1,K]}} \right)
        \\ & \quad
        + \ln \left( \frac{\Psi_{[j,K][j,K]} \, \Psi_{[j+1,K][j+1,K]}}{2 \Psi_{[j+1,K][j,K]}^2} \right)
        + o(1)
        .
      \end{aligned}
    \end{equation}
    The leftmost $X$-group:
    \begin{equation}
      \begin{aligned}
        \ln m_{1,1}(t) &=
        \frac{t}{2} \left( \frac{1}{\lambda_1} - \frac{1}{\mu_1} \right)
        + \frac12 \ln \left( \frac{2 a_1(0) \, b_1(0) \, \Psi_{[1,K][1,K]}}{\lambda_{[2,K]} \ \mu_{[2,K]} \, \Psi_{[2,K][2,K]}} \right)
        \\ & \quad
        + \ln \left(  \frac{M \, \mu_{[2,K]} \Psi_{[2,K][2,K]}}{b_1(0) \, L \, \Psi_{[2,K][1,K]}} \right)
        + o(1)
        ,
        \\[1ex]
        \ln m_{1,i}(t) &=
        \frac{t}{2} \left( \frac{3}{\lambda_1} - \frac{1}{\mu_1} \right)
        + \frac12 \ln \left( \frac{2 a_1(0) \, b_1(0) \, \Psi_{[1,K][1,K]}}{\lambda_{[2,K]} \, \mu_{[2,K]} \, \Psi_{[2,K][2,K]}} \right)
        \\ & \quad
        + \ln \left( \frac{(\sigma_i-\sigma_{i-1}) \, S_i}{R_i \, R_{i-1}} \right)
        + \ln \left( \frac{a_1(0) \, M \, \mu_{[2,K]}}{b_1(0) \, \lambda_{[2,K]}} \right)
        \\ & \quad
        + \ln \left( \frac{\Psi_{[1,K][1,K]} \, \Psi_{[2,K][2,K]}}{\Psi_{[2,K][1,K]}^2} \right)
        + o(1)
        ,
        \quad \text{for $2 \le i \le N_1-1$}
        ,
        \\[1ex]
        \ln m_{1,N_1}(t) &=
        \frac{t}{2} \left( \frac{3}{\lambda_1} - \frac{1}{\mu_1} \right)
        + \frac12 \ln \left( \frac{2 a_1(0) \, b_1(0) \, \Psi_{[1,K][1,K]}}{\lambda_{[2,K]} \, \mu_{[2,K]} \, \Psi_{[2,K][2,K]}} \right)
        \\ & \quad
        + \ln \left( \frac{\sigma_{N_1-1}}{R_{N_1-1}} \right)
        + \ln \left( \frac{a_1(0) \, M \, \mu_{[2,K]}}{b_1(0) \, \lambda_{[2,K]}} \right)
        \\ & \quad
        + \ln \left( \frac{\Psi_{[1,K][1,K]} \, \Psi_{[2,K][2,K]}}{\Psi_{[2,K][1,K]}^2} \right)
        + o(1)
        .
      \end{aligned}
    \end{equation}
    The rightmost $X$-group:
    \begin{equation}
      \begin{aligned}
        \ln m_{K+1,1}(t) &=
        \frac{t}{2} \left( \frac{-1}{\mu_{K}} \right)
        + \frac12 \ln \bigl( 2 \tau_1 \, b_K(0) \bigr)
        + \ln \left( \frac{1}{2 b_K(0)} \right)
        + o(1)
        ,
        \\
        \ln m_{K+1,i}(t) &=
        \frac{1}{2} \ln(2 S_i)
        + \ln \left( \frac{\sigma_i - \sigma_{i-1}}{2 \sigma_i \, \sigma_{i-1}} \right)
        + o(1)
        ,
        \quad \text{for $1 \le i \le N_{K+1}-1$}
        ,
        \\
        \ln m_{K+1,N_{K+1}}(t) &=
        \frac{1}{2} \ln \bigl( 2 (S_{N_{K+1}-1} + D \, \sigma_{N_{K+1}-1}) \bigr)
        + \ln \left( \frac{1}{2 \sigma_{N_{K+1}-1}} \right)
        + o(1)
        .
      \end{aligned}
    \end{equation}
    All $Y$-groups except the rightmost one ($1 \le j \le K-1$), if $K \ge 2$:
    \begin{equation}
      \begin{aligned}
        \ln n_{j,1}(t) &=
        \frac{t}{2} \left( \frac{1}{\mu_{j}} - \frac{1}{\lambda_{j}} \right)
        + \frac12 \ln \left( \frac{2 a_{j}(0) \, b_{j}(0) \, \Psi_{[j,K][j,K]}}{\lambda_{[j+1,K]} \, \mu_{[j+1,K]} \, \Psi_{[j+1,K][j+1,K]} } \right)
        \\ & \quad
        + \ln \left( \frac{\lambda_{[j+1,K]}^2 \, \Psi_{[j+1,K][j,K]}}{2 a_j(0) \, \mu_{[j,K]} \, \Psi_{[j,K][j,K]}} \right)
        + o(1)
        ,
        \\
        \ln n_{j,i}(t) &=
        \frac{t}{2} \left( \frac{3}{\mu_{j}} - \frac{1}{\lambda_{j+1}} \right)
        + \frac12 \ln \left( \frac{2 a_{j+1}(0) \, b_{j}(0) \, \Psi_{[j+1,K][j,K]}}{\lambda_{[j+2,K]} \, \mu_{[j+1,K]} \Psi_{[j+2,K][j+1,K]} } \right)
        \\ & \quad
        + \ln \left( \frac{(\sigma_i - \sigma_{i-1}) \, S_i}{R_i \, R_{i-1}} \right)
        + \ln \left( \frac{b_j(0) \, \lambda_{[j+1,K]} \, \lambda_{[j+2,K]}}{a_{j+1}(0) \, \mu_{[j+1,K]}} \right)
        \\ & \quad
        + \ln \left( \frac{\Psi_{[j+1,K][j,K]} \, \Psi_{[j+2,K][j+1,K]}}{2 \Psi_{[j+1,K][j+1,K]}^2} \right)
        + o(1)
        ,
        \\ &
        \quad \text{for $2 \le i \le N_j-1$}
        ,
        \\
        \ln n_{j,N_j}(t) &=
        \frac{t}{2} \left( \frac{3}{\mu_{j}} - \frac{1}{\lambda_{j+1}} \right)
        + \frac12 \ln \left( \frac{2 a_{j+1}(0) \, b_{j}(0) \, \Psi_{[j+1,K][j,K]}}{\lambda_{[j+2,K]} \, \mu_{[j+1,K]} \, \Psi_{[j+2,K][j+1,K]}} \right)
        \\ & \quad
        + \ln \left( \frac{\sigma_{N_j-1}}{R_{N_j-1}} \right)
        + \ln \left( \frac{b_j(0) \, \lambda_{[j+1,K]} \, \lambda_{[j+2,K]}}{a_{j+1}(0) \, \mu_{[j+1,K]}} \right)
        \\ & \quad
        + \ln \left( \frac{\Psi_{[j+1,K][j,K]} \, \Psi_{[j+2,K][j+1,K]}}{2 \Psi_{[j+1,K][j+1,K]}^2} \right)
        + o(1)
        .
      \end{aligned}
    \end{equation}
    The rightmost $Y$-group:
    \begin{equation}
      \begin{aligned}
        \ln n_{K,1}(t) &=
        \frac{t}{2} \left( \frac{1}{\mu_K} - \frac{1}{\lambda_K} \right)
        + \frac12 \ln \left( 2 a_K(0) \, b_k(0) \, \Psi_{\{K\}\{K\}} \right)
        \\ & \quad
        + \ln \left( \frac{1}{2 a_K(0) \, \mu_K \, \Psi_{\{K\}\{K\}}} \right)
        + o(1)
        ,
        \\
        \ln n_{K,i}(t) &=
        \frac{t}{2} \left( \frac{3}{\mu_K} \right)
        + \frac12 \ln \left( \frac{2 S_i \, b_K(0)}{T_i} \right)
        + \ln \left( \frac{(\sigma_i - \sigma_{i-1}) \, T_i \, b_K(0)}{2 R_i \, R_{i-1}} \right)
        + o(1)
        ,
        \\ &
        \quad \text{for $2 \le i \le N_K-1$}
        ,
        \\
        \ln n_{K,N_K}(t) &=
        \frac{t}{2} \left( \frac{3}{\mu_K} \right)
        + \frac12 \ln \bigl( 2 \sigma_{N_K-1} \, b_K(0) \bigr)
        + \ln \left( \frac{b_K(0)}{2 R_{ N_K-1}} \right)
        + o(1)
        .
      \end{aligned}
    \end{equation}
  \end{itemize}
\end{theorem}

\section{Effective position and amplitude of a peakon group}
\label{sec:effective}

In this section we will explain a phenomenon that we have seen repeatedly throughout this paper,
namely that the solution formulas for the position and amplitude of a singleton peakon group
within a non-interlacing solution
are identical to the solution formulas for the corresponding peakon in the interlacing solution
(with the same spectral data).
In fact, we will show that any peakon group has an \emph{effective position}
and an \emph{effective amplitude} which behave like the position and amplitude
of the corresponding peakon in the interlacing solution, and if the group is a singleton,
then the effective position and amplitude of that group are just the actual position and amplitude
of that single peakon.

To motivate the definitions that follow, let us recall from the paper by
Lundmark and Szmigielski~\cite{lundmark-szmigielski:2016:GX-inverse-problem}
the two kinds of ``jump matrices'' that appear in the study of the spectral problems
connected with the Lax pairs for the Geng--Xue equation:
\begin{equation}
  \label{eq:jump-TS}
  S(x,m,\lambda)
  =
  \begin{pmatrix}
    1 & 0 & 0 \\
    m e^{x} & 1 & \lambda m e^{-x} \\
    0 & 0 & 1
  \end{pmatrix}
  ,
  \qquad
  T(x,m,\lambda)
  =
  \begin{pmatrix}
    1 & -2 \lambda m e^{-x} & 0 \\
    0 & 1 & 0 \\
    0 &  2 m e^{x} & 1
  \end{pmatrix}
  .
\end{equation}
These matrices are used for defining the spectral data for a given interlacing
peakon configuration.
As a simple example, we will consider the $2+2$ interlacing case $x_1 < y_1 < x_2 < y_2$,
and show how the eight peakon variables $\{ x_k, m_k, y_k, n_k \}_{k=1}^2$
define the eight spectral variables
$\lambda_1$, $\lambda_2$, $a_1$, $a_2$, $\mu_1$, $b_1$, $C$ and~$D$.
The equation
\begin{equation}
  \label{eq:ABC-2+2-interlacing}
  \begin{pmatrix}
    A(\lambda) \\
    B(\lambda) \\
    C(\lambda)
  \end{pmatrix}
  =
  T(y_2, n_2, \lambda)
  \,
  S(x_2, m_2, \lambda)
  \,
  T(y_1, n_1, \lambda)
  \,
  S(x_1, m_1, \lambda)
  \,
  \begin{pmatrix}
    1 \\
    0 \\
    0
  \end{pmatrix}
  ,
\end{equation}
defines the polynomials
\begin{equation}
  \begin{aligned}
    A(\lambda) &
    = 1 - 2 \lambda \left( m_1 n_1 e^{x_1- y_1} + m_1 n_2 e^{x_1-y_2} + m_2 n_2 e^{x_2-y_2} \right) 
    \\
    \quad &
    + (2\lambda)^2 \left( m_1 n_1 e^{x_1- y_1} \left( 1 - e^{2(y_1-x_2)} \right) m_2 n_2 e^{x_2-y_2} \right)
    ,
  \end{aligned}
\end{equation}
which turns out to be time-independent (its coefficients are constants of motion),
and 
\begin{equation}
  B(\lambda)
  =
  \left( m_1 e^{x_1} + m_2 e^{x_2} \right)
  - 2 \lambda m_1 n_1 e^{x_1- y_1} \left( 1 - e^{2(y_1-x_2)} \right) m_2 e^{x_2}
  ,
\end{equation}
which depends on time in a known way.
The eigenvalues $0 < \lambda_1 < \lambda_2$ are then defined as the zeros of~$A(\lambda)$,
which implies that they are constant in time,
while $a_1$ and~$a_2$ are defined as the residues in the partial fraction
decomposition of the so-called Weyl function
\begin{equation}
  -\frac{B(\lambda)}{A(\lambda)} =
  \frac{a_1}{\lambda - \lambda_1} + \frac{a_2}{\lambda - \lambda_2}
  .
\end{equation}
From the known time-dependence of $B(\lambda)$
one can show that $\dot a_k = a_k / \lambda_k$,
which gives $a_k(t) = a_k(0) \, e^{t / \lambda_k}$ since $\lambda_k$ is constant.
Similarly, from the equation
\begin{equation}
  \begin{pmatrix}
    \tilde A(\lambda) \\
    \tilde B(\lambda) \\
    \tilde C(\lambda)
  \end{pmatrix}
  =
  S(y_2, n_2, \lambda)
  \,
  T(x_2, m_2, \lambda)
  \,
  S(y_1, n_1, \lambda)
  \,
  T(x_1, m_1, \lambda)
  \,
  \begin{pmatrix}
    1\\
    0\\
    0
  \end{pmatrix}
  ,
\end{equation}
where the roles of the jump matrices are interchanged,
we get the polynomials
\begin{equation}
  \begin{aligned}
    \tilde A(\lambda) &= 1 - 2 \lambda n_1 m_2 e^{ y_1-x_2}
    ,
    \\
    \tilde B(\lambda) &= \left( n_1 e^{y_1} + n_2 e^{y_2} \right)
    - 2 \lambda m_2 n_1 e^{y_1-x_2} \left( 1 - e^{2(x_2-y_2)} \right) n_2 e^{y_2} 
    ,
  \end{aligned}
\end{equation}
where the eigenvalue~$\mu_1$ is defined as the zero of
time-independent polynomial~$\tilde A(\lambda)$, and $b_1$ and~$D$ are given by
\begin{equation}
  -\frac{\tilde B(\lambda)}{\tilde A(\lambda)} = -D + \frac{b_1}{\lambda - \mu_1}
  ,
\end{equation}
from which one can show that $\dot b_1 = b_1 / \mu_1$
and that
\begin{equation}
  D = \lim_{\lambda\to\infty} \frac{\tilde B(\lambda)}{\tilde A(\lambda)}
  = n_2 e^{y_2} \left( 1 - e^{2(x_2-y_2)} \right)
\end{equation}
is constant.
This constant of motion~$D$ was denoted by~$b_{\infty}$ in the papers by
Lundmark and Szmigielski~\cite{lundmark-szmigielski:2016:GX-inverse-problem, lundmark-szmigielski:2017:GX-dynamics-interlacing}.
The remaining constant parameter~$C$ is given by
\begin{equation}
  C = 2 b_{\infty}^* \frac{\lambda_1 \lambda_2}{\mu_1}
  ,
\end{equation}
where the constant of motion
\begin{equation}
  b_{\infty}^* 
  =
  m_1  e^{-x_1} \left( 1 - e^{2(x_1-y_1)} \right)
\end{equation}
comes from the so-called adjoint spectral problem,
or from the symmetry of the setup.

If we consider instead a non-interlacing configuration, say
\begin{equation*}
  x_1 < \underbrace{ y_{1,1} < y_{1,2} < y_{1,3} } < x_2 < y_2
  ,
\end{equation*}
and try do to the same thing, then
\begin{equation}
  \label{eq:ABC-2+2-noninterlacing}
  \begin{split}
    \begin{pmatrix}
      A(\lambda) \\
      B(\lambda) \\
      C(\lambda)
    \end{pmatrix}
    &
    =
    T(y_2, n_2, \lambda)
    \,
    S(x_2, m_2, \lambda)
    \\ &
    \times
    \underbrace{
      T(y_{1,3}, n_{1,3}, \lambda)
      \,
      T(y_{1,2}, n_{1,2}, \lambda)
      \,
      T(y_{1,1}, n_{1,1}, \lambda)
    }
    \,
    S(x_1, m_1, \lambda)
    \,
    \begin{pmatrix}
      1 \\
      0 \\
      0
    \end{pmatrix}
    ,
  \end{split}
\end{equation}
where the indicated product of the jump matrices for the three adjacent $Y$-peakons is
\begin{equation}
  \begin{split}
    &
    T(y_{1,3}, n_{1,3}, \lambda)
    \,
    T(y_{1,2}, n_{1,2}, \lambda)
    \,
    T(y_{1,1}, n_{1,1}, \lambda)
    \\[1ex]
    & \qquad
    =
    \begin{pmatrix}
      1 & -2 \lambda n_{1,1} e^{-y_{1,1}} - 2 \lambda n_{1,2} e^{-y_{1,2}} - 2 \lambda n_{1,3} e^{-y_{1,3}} & 0 \\
      0 & 1 & 0 \\
      0 &  2 n_{1,1} e^{y_{1,1}} + 2 n_{1,2} e^{y_{1,2}} + 2 n_{1,3} e^{y_{1,3}} & 1
    \end{pmatrix}
    ,
  \end{split}
\end{equation}
which happens to be of the form 
$T(\tilde y_1, \tilde n_1, \lambda)$, where
\begin{equation}
  \label{eq:effective-mass-and-position-y-example}
  \tilde n_1 e^{\tilde y_1}
  = \sum_{i=1}^3 n_{1,i} e^{y_{1,i}}
  ,\qquad
  \tilde n_1 e^{-\tilde y_1}
  = \sum_{i=1}^3 n_{1,i} e^{-y_{1,i}}
  .
\end{equation}
Thus, the polynomials $A(\lambda)$ and~$B(\lambda)$ from~\eqref{eq:ABC-2+2-noninterlacing}
will not be able to resolve the individual
positions $y_{1,i}$ and amplitudes $n_{1,i}$ in the $Y_1$-group,
but only the specific combinations $\tilde y_1$ and~$\tilde n_1$
defined by~\eqref{eq:effective-mass-and-position-y-example},
playing exactly the roles that $y_1$ and~$n_1$ did in the interlacing
case~\eqref{eq:ABC-2+2-interlacing}.
And the same thing happens for $\tilde A(\lambda)$ and~$\tilde B(\lambda)$,
since also
\begin{equation}
    S(y_{1,3},n_{1,3},\lambda)
    \,
    S(y_{1,2},n_{1,2},\lambda)
    \,
    S(y_{1,1},n_{1,1},\lambda)
    =
    S(\tilde y_1,\tilde n_1,\lambda)
    ,
\end{equation}
as is easily verified.
Thus, the definition of the spectral variables will be just like in the interlacing case,
except that $y_1$ and~$n_1$ are replaced by the quantities $\tilde y_1$ and~$\tilde n_1$,
which therefore act as the \emph{effective} position and amplitude of the $Y_1$-group
as a whole, as far as the spectral data are concerned.
The solution formulas for the interlacing case are nothing but the inverse spectral map
from the spectral variables back to the peakon variables,
so what they give us in this non-interlacing case is
$x_1(t)$, $\tilde y_1(t)$, $x_2(t)$, $y_2(t)$
and $m_1(t)$, $\tilde n_1(t)$, $m_2(t)$, $n_2(t)$.
In particular, they provide the formulas for all the singletons in the non-interlacing
solution.
(But the formulas for the individual variables
$y_{1,i}$ and~$n_{1,i}$ must be obtained in some other way.)

Let us now give an independent verification that things always work like this,
using the peakon ODEs directly, rather than the setup coming from the Lax pairs.

\begin{definition}
  The \textbf{effective position} $\tilde x_j$ and
  the \textbf{effective amplitude} $\tilde m_j$ of
  the $j$th $X$-group are defined by
  \begin{equation}
    \label{eq:effective-mass-and-position-X}
    \tilde m_j e^{\tilde x_j}
    = \sum_{i=1}^{N_j^X} m_{j,i} e^{x_{j,i}}
    ,\qquad
    \tilde m_j e^{-\tilde x_j}
    = \sum_{i=1}^{N_j^X} m_{j,i} e^{-x_{j,i}}
    .
  \end{equation}
  Similarly, for the $j$th $Y$-group,
  $\tilde y_j$ and $\tilde n_j$ are defined by
  \begin{equation}
    \label{eq:effective-mass-and-position-Y}
    \tilde n_j e^{\tilde y_j}
    = \sum_{i=1}^{N_j^Y} n_{j,i} e^{y_{j,i}}
    ,\qquad
    \tilde n_j e^{-\tilde y_j}
    = \sum_{i=1}^{N_j^Y} n_{j,i} e^{-y_{j,i}}
    .
  \end{equation}
\end{definition}

\begin{remark}
  In this definition, we are tacitly assuming (as always in this article)
  that all amplitudes $m_{j,i}$ and~$n_{j,i}$ are positive.
  This ensures that the defining system
  \begin{equation*}
    \tilde m e^{\tilde x} = A > 0
    ,\qquad
    \tilde m e^{-\tilde x} = B > 0
  \end{equation*}
  can be solved for the quantities being defined,
  \begin{equation*}
    \tilde x = \tfrac12 \ln(A/B)
    ,\qquad
    \tilde m = \sqrt{AB} > 0
    .
  \end{equation*}
  Without this positivity requirement, there may not even exist a real solution
  for $\tilde x$ and~$\tilde m$.
  Similarly for $\tilde y$ and~$\tilde n$, of course.
\end{remark}

\begin{proposition}
  The effective position of a group lies in the convex hull of the actual positions:
  \begin{equation}
    x_{j,1} < \tilde x_j < x_{j,N_j^X}
    ,\qquad
    y_{j,1} < \tilde y_j < y_{j,N_j^Y}
    .
  \end{equation}
\end{proposition}

\begin{proof}
  It is enough to show this for $\tilde x_j$,
  since the proof for $\tilde y_j$ is identical.
  Omitting the group index~$j$ for simplicity,
  we find since all $m_i>0$ and $x_1 < \dots < x_N$ that
  \begin{equation*}
    \tilde m e^{\tilde x}
    = \sum_{i=1}^N m_i e^{x_i}
    \in \bigl(e^{x_1} M, e^{x_N} M \bigr)
    ,\qquad
    \frac{\tilde m}{e^{\tilde x}}
    = \sum_{i=1}^N \frac{m_i}{e^{x_i}}
    \in \biggl( \frac{M}{e^{x_N}}, \frac{M}{e^{x_1}} \biggr)
    ,
  \end{equation*}
  where $M = \sum_{i=1}^N m_i > 0$, and hence
  \begin{equation*}
    e^{2 \tilde x}
    = \tilde m e^{\tilde x} \cdot \frac{e^{\tilde x}}{\tilde m}
    \in
    \biggl( e^{x_1} M \cdot \frac{e^{x_1}}{M}, e^{x_N} M \cdot \frac{e^{x_N}}{M} \biggr)
    = \bigl( e^{2x_1}, e^{2x_N} \bigr)
    .
  \end{equation*}
\end{proof}

\begin{theorem}
  For any peakon configuration, the effective positions and masses of the groups
  satisfy the ODEs for an interlacing peakon configuration.
\end{theorem}

\begin{proof}
  Recall from~\eqref{eq:GX-peakon-ode-new-notation} that the peakon ODEs are
  \begin{equation*}
    \begin{aligned}
      \dot x_{j,i} &= uv |_{x_{j,i}}
      ,&
      \dot m_{j,i} &= m_{j,i} \, (uv_x - u_xv)|_{x_{j,i}}
      ,\\
      \dot y_{j,i} &= uv |_{y_{j,i}}
      ,&
      \dot n_{j,i} &= n_{j,i} \, (vu_x - v_xu)|_{y_{j,i}}
      .
    \end{aligned}
  \end{equation*}
  Fix some~$p$, and define $\tilde u(x,t)$ by replacing $X$-group
  number $p$ in $u(x,t)$ with a singleton having the position~$\tilde x_p$
  and the amplitude~$\tilde m_p$.
  Then, for all the $X$-groups with $p\neq j$, and for all Y-groups,
  the ODEs above are unchanged if we replace $u$ with $\tilde u$.
  Indeed, we have
  \begin{equation*}
    \begin{aligned}
      u(x,t) &= \dotsb + \sum_{i=1}^{N} m_{p,i} e^{-\abs{x-x_{p,i}}} + \dotsb
      \qquad
      (\text{where $N=N_p^X$})
      ,\\
      \tilde u(x,t) &= \dotsb + \tilde m_{p} e^{-\abs{x-\tilde x_{p}}} + \dotsb
      ,
    \end{aligned}
  \end{equation*}
  where the dots denote terms which are identical in both functions,
  so if we evaluate at some $x<x_{p,1}$,
  remembering that $x_{p,1} < \tilde x_p < x_{p,N}$,
  we get
  \begin{equation*}
    \begin{aligned}
      u(x,t) &= \dotsb + \sum_{i=1}^{N} m_{p,i} e^{x-x_{p,i}} + \dotsb
      ,\\
      \tilde u(x,t) &= \dotsb + \tilde m_{p} e^{x-\tilde x_{p}} + \dotsb
      ,
    \end{aligned}
  \end{equation*}
  which is clearly the same thing (by the definition of $\tilde x_p$ and~$\tilde m_p$),
  and similarly we find for $x>x_{p,N}$ that
  \begin{equation*}
    \begin{aligned}
      u(x,t) &= \dotsb + \sum_{i=1}^{N} m_{p,i} e^{x_{p,i}-x} + \dotsb
      ,\\
      \tilde u(x,t) &= \dotsb + \tilde m_{p} e^{\tilde x_{p}-x} + \dotsb
      ,
    \end{aligned}
  \end{equation*}
  which are also equal.
  Thus,
  \begin{equation}
    u(x,t)=\tilde u(x,t)
    \qquad
    \text{for $x < x_{p,1}$ and for $x > x_{p,N}$}
    ,
  \end{equation}
  and in particular the functions $u$ and~$\tilde u$
  agree when evaluated at some $x_{j,i}$ with $j \neq p$
  or at some $y_{j,i}$,
  and likewise for their derivatives $u_x$ and~$\tilde u_x$.
  So the right-hand sides of the peakon ODEs (for all groups except the $X_p$-group)
  are unchanged, as we claimed.
  
  Next, we show that $\tilde x_p$ and $\tilde m_p$
  satisfy the correct singleton ODEs, namely
  \begin{equation*}
    \dot{\tilde x}_{p} = \tilde u v |_{\tilde x_{p}}
    ,\qquad
    \dot{\tilde m}_{p} = \tilde m_{p} \, (\tilde u v_x - 2 \tilde u_x v)|_{\tilde x_{p}}
    ,
  \end{equation*}
  or, equivalently,
  \begin{equation}
    \label{eq:x-m-tilde-correct-ODEs}
    \begin{aligned}
      \tfrac{d}{dt} \tilde m_{p} e^{\tilde x_{p}}
      &
      = \tilde m_{p} e^{\tilde x_{p}} \bigl( \tilde u (v_x+v) - 2 \tilde u_x v \bigr)|_{\tilde x_{p}}
      ,\\
      \tfrac{d}{dt} \tilde m_{p} e^{-\tilde x_{p}}
      &
      = \tilde m_{p} e^{-\tilde x_{p}} \bigl( \tilde u (v_x-v) - 2 \tilde u_x v \bigr)|_{\tilde x_{p}}
      .
    \end{aligned}
  \end{equation}
  What we are assuming is that the dynamics
  is induced from the noninterlacing configuration,
  \begin{equation*}
    \begin{aligned}
      \tfrac{d}{dt} \tilde m_{p} e^{\tilde x_{p}}
      &
      = \tfrac{d}{dt} \sum_{i=1}^N m_{p,i} e^{x_{p,i}}
      = \sum_{i=1}^N m_{p,i} e^{x_{p,i}} \bigl( u (v_x + v) - 2 u_x v \bigr)|_{x_{p,i}}
      ,
      \\
      \tfrac{d}{dt} \tilde m_{p} e^{-\tilde x_{p}}
      &
      = \tfrac{d}{dt} \sum_{i=1}^N m_{p,i} e^{-x_{p,i}}
      = \sum_{i=1}^N m_{p,i} e^{-x_{p,i}} \bigl( u (v_x - v) - 2 u_x v \bigr)|_{x_{p,i}}
      ,
    \end{aligned}
  \end{equation*}
  and what we need to show is that these expressions agree
  with~\eqref{eq:x-m-tilde-correct-ODEs},
  i.e.,
  \begin{equation}
    \label{eq:x-m-tilde-to-check}
    \begin{aligned}
      \sum_{i=1}^N m_{p,i} e^{x_{p,i}} \bigl( u (v_x + v) - 2 u_x v \bigr)|_{x_{p,i}}
      &
      = \tilde m_{p} e^{\tilde x_{p}} \bigl( \tilde u (v_x + v) - 2 \tilde u_x v \bigr)|_{\tilde x_{p}}
      ,
      \\
      \sum_{i=1}^N m_{p,i} e^{-x_{p,i}} \bigl( u (v_x - v) - 2 u_x v \bigr)|_{x_{p,i}}
      &
      = \tilde m_{p} e^{-\tilde x_{p}} \bigl( \tilde u (v_x - v) - 2 \tilde u_x v \bigr)|_{\tilde x_{p}}
      .
    \end{aligned}
  \end{equation}
  If for simplicity we write just
  \begin{equation*}
    x_i = x_{p,i}
    ,\qquad
    m_i = x_{p,i}
    ,\qquad
    \tilde x = \tilde x_p
    ,\qquad
    \tilde m = \tilde m_p
    ,
  \end{equation*}
  then for $x$ in the relevant range (the $p$th $X$-group) the functions
  $u$, $\tilde u$ and~$v$ have the form
  \begin{equation*}
    \begin{aligned}
      u &= A e^{-x} + \sum_{r=1}^N m_r e^{-\abs{x-x_r}} + B e^x
      ,\\
      \tilde u &= A e^{-x} + \tilde m e^{-\abs{x-\tilde x}} + B e^x
      ,\\
      v &= C e^{-x} + D e^x
      ,
    \end{aligned}
  \end{equation*}
  so
  \begin{equation*}
    \begin{aligned}
      u(x_i) &= A e^{-x_i} + e^{-x_i} \underbrace{\sum_{r<i} m_r e^{x_r}}_{=: E_i} + m_i + e^{x_i} \underbrace{\sum_{r>i} m_r e^{-x_r}}_{=: F_i} + B e^{x_i}
      \\
      &= (A+E_i) e^{-x_i} + m_i + (B+F_i) e^{x_i}
      ,\\
      u_x(x_i)
      &= -(A+E_i) e^{-x_i} + (B+F_i) e^{x_i}
      ,\\
      v(x_i) &= C e^{-x_i} + D e^{x_i}
      ,\\
      v_x(x_i) &= -C e^{-x_i} + D e^{x_i}
      .
    \end{aligned}
  \end{equation*}
  Thus, the right-hand side of the first equation in~\eqref{eq:x-m-tilde-to-check} is
  \begin{equation*}
    \begin{aligned}
      &
      \tilde m e^{\tilde x} (\tilde u(v_x+v) - 2 \tilde u_xv)|_{\tilde x}
      \\ &
      = \tilde m e^{\tilde x}
      \Bigl(
      (A e^{-\tilde x} + \tilde m + B e^{\tilde x}) \, 2D e^{\tilde x}
      + 2 \bigl( A e^{-\tilde x} - B e^{\tilde x} \bigr)
      \bigl( C e^{-\tilde x} + D e^{\tilde x} \bigr)
      \Bigr)
      \\ &
      = 2 \tilde m e^{\tilde x} (2AD-BC)
      + 2 D (\tilde m e^{\tilde x})^2
      + 2 AC \tilde m e^{-\tilde x}
      ,
    \end{aligned}
  \end{equation*}
  while the left-hand side is
  \begin{equation*}
    \begin{split}
      & 
      \sum_{i=1}^N m_{i} e^{x_{i}} (u(v_x+v) - 2 u_xv)|_{x_{i}}
      \\ &
      = \sum_{i=1}^N m_i e^{x_i}
      \biggl(
      \bigl( (A+E_i) e^{-x_i} + m_i + (B+F_i) e^{x_i} \bigr) \, 2D e^{x_i}
      \\
      & \qquad\qquad\qquad
      + 2 \bigl( (A+E_i) e^{-x_i} - (B+F_i) e^{x_i} \bigr)
      \bigl( C e^{-x_i} + D e^{x_i} \bigr)
      \biggr)
      \\ &
      = 2 \sum_{i=1}^N m_i e^{x_i}
      \biggl(
      2D (A+E_i) + D m_i e^{x_i}
      + C(A+E_i) e^{-2x_i}
      - C(B+F_i)
      \biggr)
      \\ &
      = 2 \tilde m e^{\tilde x} (2AD-BC)
      + 2D \underbrace{\sum_{i=1}^N m_i e^{x_i} \biggl( m_i e^{x_i} + 2 \sum_{r<i} m_r e^{x_r} \biggr)}_{= \biggl( \sum_{i=1}^N m_i e^{x_i} \biggr)^2 = (\tilde m e^{\tilde x})^2}
      \\
      & \quad
      + 2 \tilde m e^{-\tilde x} AC
      + 2C \underbrace{\Biggl(
        \sum_{i=1}^N m_i e^{-x_i} \biggl( \sum_{r<i} m_r e^{x_r} \biggr)
        - \sum_{i=1}^N m_i e^{x_i} \biggl( \sum_{r>i} m_r e^{-x_r} \biggr)
        \Biggr)}_{=0}
      ,
    \end{split}
  \end{equation*}
  so they are equal.
  The second equation in~\eqref{eq:x-m-tilde-to-check} is
  proved similarly.

  Interchanging $u$ and~$v$ in these calculations
  shows that the same holds for $Y$-groups.
  The conclusion is that if we successively replace all the $X$-groups and $Y$-groups
  with their corresponding effective positions and amplitudes,
  the interlacing configuration which remains in the end will satisfy the interlacing peakon ODEs.
\end{proof}

\begin{remark}
  As a further independent verification,
  one may also check that if $x_{j,i}$ and~$m_{j,i}$
  are given by our solution formulas, with expressions
  in terms of the determinants $\detJ_{ij}^{rs}$,
  then the sums in the definition of $\tilde x_k$ and~$\tilde m_k$
  will simplify in such a way that
  these quantities will indeed agree with the corresponding expressions from
  the solution formulas for the interlacing case.
  Here we omit these somewhat lenghty calculations, which involve
  an induction on the number of peakons in the group,
  together with determinant manipulations based on ``Lewis Carroll's identity'',
  similar to those in Section~A.3 in~\cite{lundmark-szmigielski:2016:GX-inverse-problem}.
\end{remark}

\section{Absence of collisions}
\label{sec:no-collisions}

This section is devoted to the proof of the following theorem:

\begin{theorem}
  \label{thm:no-collisions}
  Collisions cannot take place for a pure peakon configuration in the Geng--Xue equation.
  In other words, if all amplitudes are positive,
  then the strict ordering of the positions of the peakons
  is preserved for all~$t \in \R$.
\end{theorem}

\begin{proof}
  For simplicity, we will illustrate the general pattern with an example.
  Suppose we have a configuration with $5+5$ groups,
  where two $Y$-singletons border an $X_3$-group containing six peakons.
  Then the solution formulas for those positions have the form
  \begin{equation}
    \label{eq:X-group-structure-no-collisions}
    \begin{aligned}
      Y_2 &= \frac{\alpha}{\beta}
      \\[2ex]
      X_{3,1} &= \frac{\alpha + \tau_1 \gamma}{\beta + \tau_1 \delta}
      \\
      X_{3,2} &= \frac{\alpha + (\tau_1 + \tau_2) \gamma + \tau_2 \sigma_1 \epsilon}{\beta + (\tau_1 + \tau_2) \delta + \tau_2 \sigma_1 \phi}
      \\
      X_{3,3} &= \frac{\alpha + (\tau_1 + \tau_2 + \tau_3) \gamma + (\tau_2 \sigma_1 + \tau_3 \sigma_2) \epsilon}{\beta + (\tau_1 + \tau_2 + \tau_3) \delta + (\tau_2 \sigma_1 + \tau_3 \sigma_2) \phi}
      \\
      X_{3,4} &= \frac{\alpha + (\tau_1 + \tau_2 + \tau_3 + \tau_4) \gamma + (\tau_2 \sigma_1 + \tau_3 \sigma_2 + \tau_4 \sigma_3) \epsilon}{\beta + (\tau_1 + \tau_2 + \tau_3 + \tau_4) \delta + (\tau_2 \sigma_1 + \tau_3 \sigma_2 + \tau_4 \sigma_3) \phi}
      \\
      X_{3,5} &= \frac{\alpha + (\tau_1 + \tau_2 + \tau_3 + \tau_4 + \tau_5) \gamma + (\tau_2 \sigma_1 + \tau_3 \sigma_2 + \tau_4 \sigma_3 + \tau_5 \sigma_4) \epsilon}{\beta + (\tau_1 + \tau_2 + \tau_3 + \tau_4 + \tau_5) \delta + (\tau_2 \sigma_1 + \tau_3 \sigma_2 + \tau_4 \sigma_3 + \tau_5 \sigma_4) \phi}
      \\
      X_{3,6} &= \frac{\gamma + \sigma_5 \epsilon}{\delta + \sigma_5 \phi}
      \\[2ex]
      Y_3 &= \frac{\epsilon}{\phi}
      ,
    \end{aligned}
  \end{equation}
  where
  \begin{equation*}
    \alpha = \detJ_{43}^{00}
    ,\quad
    \beta = \detJ_{32}^{11}
    ,\quad
    \gamma =\detJ_{33}^{00}
    ,\quad
    \delta = \detJ_{22}^{11}
    ,\quad
    \epsilon = \detJ_{32}^{00}
    ,\quad
    \phi = \detJ_{21}^{11}
    ,
  \end{equation*}
  and where we know from the interlacing case,
  where the solution formula is $X_3 = \gamma/\delta$,
  that
  \begin{equation*}
    \frac{\alpha}{\beta} < \frac{\gamma}{\delta} < \frac{\epsilon}{\phi}
    .
  \end{equation*}
  By assumption, $\alpha$, $\beta$, $\gamma$, $\delta$, $\epsilon$, $\phi$, $\tau_i$ and $\sigma_i$
  are all positive,
  and $\sigma_1 < \dots < \sigma_5$.
  
  The basic proposition that we will use is if $a/b < c/d$ with
  $a$, $b$, $c$, $d$ positive, then the function
  \begin{equation*}
    f(x)
    = \frac{a+xc}{b+xd}
    = \frac{c}{d} - \underbrace{b \left( \frac{c}{d} - \frac{a}{b} \right)}_{> 0} \cdot \underbrace{\frac{1}{b+dx}}_{\text{decr.}}
  \end{equation*}
  is \textbf{increasing} in the interval $x > -b/d$, and in particular for $x \ge 0$.
  Thus, $f$ increases from $f(0)=a/b$ to $\lim_{x \to \infty} f(x) = c/d$:
  \begin{equation}
    \label{eq:little-proposition}
    \frac{a}{b} < \frac{a+xc}{b+xd} < \frac{c}{d}
    ,\qquad
    \text{if $x > 0$}    
    .
  \end{equation}
  This shows at once that
  \begin{equation}
    \underbrace{\frac{\alpha}{\beta}}_{=Y_2}
    <
    \underbrace{\frac{\alpha+\tau_1 \gamma}{\beta+\tau_1 \delta}}_{=X_{3,1}}
    \qquad\text{and}\qquad
    \underbrace{\frac{\gamma+\sigma_5 \epsilon}{\delta+\sigma_5 \phi}}_{=X_{3,6}}
    <
    \underbrace{\frac{\epsilon}{\phi}}_{=Y_3}
    .
  \end{equation}
  (And we also see again that $X_{3,1}$ and~$X_{3,6}$ lie on either side of
  the singleton $X_3 = \gamma/\delta$,
  as we already proved in Section~\ref{sec:effective}.)
  Next, \eqref{eq:little-proposition} also shows that
  \begin{equation*}
    \frac{\alpha+\tau_1 \gamma}{\beta+\tau_1 \delta}
    < \frac{\gamma}{\delta}
    < \frac{\gamma+\sigma_1 \epsilon}{\delta+\sigma_1 \phi}
    ,
  \end{equation*}
  and we can apply~\eqref{eq:little-proposition} to the outer members of this inequality,
  with $x=\tau_2$, to obtain that
  \begin{equation}
    \underbrace{\frac{\alpha+\tau_1 \gamma}{\beta+\tau_1 \delta}}_{=X_{3,1}}
    <
    \underbrace{\frac{\alpha+\tau_1 \gamma + \tau_2 (\gamma+\sigma_1 \epsilon)}{\beta+\tau_1 \delta + \tau_2 (\delta+\sigma_1 \phi)}}_{=X_{3,2}}
    < \frac{\gamma+\sigma_1 \epsilon}{\delta+\sigma_1 \phi}
    .
  \end{equation}
  Then also $X_{3,2} < \frac{\gamma+\sigma_2 \epsilon}{\delta+\sigma_2 \phi}$,
  since $\sigma_1 < \sigma_2$ and the right-hand side increases with~$\sigma$,
  so we can apply~\eqref{eq:little-proposition} to those two ratios, with $x=\tau_3$,
  to obtain
  \begin{equation}
    \underbrace{\frac{\alpha+\tau_1 \gamma + \tau_2 (\gamma+\sigma_1 \epsilon)}{\beta+\tau_1 \delta + \tau_2 (\delta+\sigma_1 \phi)}}_{=X_{3,2}}
    <
    \underbrace{\frac{\alpha+\tau_1 \gamma + \tau_2 (\gamma+\sigma_1 \epsilon) + \tau_3 (\gamma+\sigma_2 \epsilon)}{\beta+\tau_1 \delta + \tau_2 (\delta+\sigma_1 \phi) + \tau_3 (\delta+\sigma_2 \phi)}}_{=X_{3,3}}
    < \frac{\gamma+\sigma_2 \epsilon}{\delta+\sigma_2 \phi}
    .
  \end{equation}
  Continuing in the same manner, we find
  \begin{equation}
    X_{3,3} < X_{3,4} < \frac{\gamma+\sigma_3 \epsilon}{\delta+\sigma_3 \phi}
  \end{equation}
  and
  \begin{equation}
    X_{3,4} < X_{3,5} < \frac{\gamma+\sigma_4 \epsilon}{\delta+\sigma_4 \phi}
    < \underbrace{\frac{\gamma+\sigma_5 \epsilon}{\delta+\sigma_5 \phi}}_{= X_{3,6}}
    .
  \end{equation}
  Thus,
  \begin{equation}
    Y_2 < X_{3,1} < X_{3,2} < X_{3,3} < X_{3,4} < X_{3,5} < X_{3,6} < Y_3
    ,
  \end{equation}
  as desired.

  It is also easy to show that two (typical) adjacent non-singleton groups cannot overlap,
  since the position of the rightmost peakon in the first group and
  the position of the leftmost peakon in the second group are both given by
  expressions of the same form
  \begin{equation*}
    \frac{\alpha + x \gamma}{\beta + x \delta}
  \end{equation*}
  where $x$ is the last $\sigma$ in the first group,
  or the first $\tau$ in the second group,
  which satisfy $\sigma < \tau$ by assumption.

  All other cases (involving the outermost groups) can be checked in a similar manner,
  using the constraints in Section~\ref{sec:more-notation} whenever necessary.
\end{proof}

\section{Characteristic curves}
\label{sec:characteristic-curves}

Let us define the \emph{characteristic curves} (or \emph{characteristics}) for a given solution
$\bigl( u(x,t), \, v(x,t) \bigr)$ of the Geng--Xue equation
as the solutions $x = \xi(t)$ of the ODE
\begin{equation}
  \label{eq:characteristic-curves}
  \frac{d \xi}{dt}(t) = u \bigl( \xi(t),t \bigr) \, v \bigl( \xi(t),t \bigr)
  ,
\end{equation}
or $\dot\xi = u(\xi) \, v(\xi)$ for short.
The trajectories of the peakons, $x = x_{j,i}(t)$ and $x = y_{j,i}(t)$,
are particular characteristic curves, according to the peakon
ODEs~\eqref{eq:GX-peakon-ode-new-notation}.
Some of the characteristic curves \emph{between} the peakons are obtained
as a byproduct of our proofs in the form of ``ghostpeakons'',
as we noted already in Example~\ref{ex:proof-technique};
see Figure~\ref{fig:proof-example-characteristics}
in particular.
And the remaining characteristic curves can also be found by taking suitable limits,
as we shall see.

We will summarize the formulas for all the characteristics in this sections,
because they may shed some light on the structure of the solution formulas for the
positions of the peakons.

For a typical non-singleton group (not the leftmost or rightmost group),
the structure of the solution formulas for the positions is
\begin{equation*}
  \text{$X_{j,i}$ or $Y_{j,i}$} =
  \begin{cases}
    \dfrac{\alpha + T_i \gamma + S_i \epsilon}{\beta + T_i \delta + S_i \phi}
    ,&
    1 \le i \le N-1
    ,
    \\[2ex]
    \dfrac{\gamma + \sigma_{N-1} \epsilon}{\delta + \sigma_{N-1} \phi}
    ,&
    i = N
    ,
  \end{cases}
\end{equation*}
where,
like in~\eqref{eq:X-group-structure-no-collisions},
the letters $\alpha$, $\beta$, $\gamma$, $\delta$, $\epsilon$, $\phi$ symbolize
certain determinants~$\detJ_{ab}^{rs}$
which are independent of~$i$.
The exact choice of indices for these determinants depends on~$j$,
and on whether the group is an $X$-group or a $Y$-group,
and on whether we are in the even or odd case;
for details see the formulas
\eqref{eq:even-X-typical-group-pos},
\eqref{eq:even-Y-typical-group-pos},
\eqref{eq:odd-X-typical-group-pos}
and~\eqref{eq:odd-Y-group-pos},
which all have the above structure in common.
What we directly get by examining the ghostpeakon formulas from the proofs in
Sections \ref{sec:proofs-even} and~\ref{sec:proofs-odd},
such as equation~\eqref{eq:proof-Y-ghost} and its remnants as
further peakons are killed,
is that the family of characteristic curves $x = \xi(t;\theta)$
between peakons number $i$ and $i+1$ in such a group is given by
\begin{equation}
  \label{eq:char-within-typical}
  \Xi := \tfrac12 e^{2 \xi}
  = \frac{\alpha + (T_i + \theta) \gamma + (S_i + \theta \sigma_i) \epsilon}{\beta + (T_i + \theta) \delta + (S_i + \theta \sigma_i) \phi}
  ,
\end{equation}
where the variable~$\theta$ which indexes the family is allowed to vary in the range
$0 < \theta < \tau_{i+1}$ for $1 \le i \le N-2$
and $0 < \theta < \infty$ for $i = N-1$.
Note that in the limit as $\theta$ tends to its lower or upper bounding value,
the characteristic curve converges to the neighbouring peakon curve
to its left or right, respectively.
In particular, when $\theta \to \infty$ for $i=N-1$,
this explains nicely how the ``exceptional'' formula for
the rightmost peakon in the group actually fits naturally into the pattern
formed by the other formulas.

We also find for the rightmost group ($Y_K$ in the even case,
$X_{K+1}$ in the odd case),
where the solution formulas
\eqref{eq:even-Y-rightmost-group-pos}
and~\eqref{eq:odd-X-rightmost-group-pos}
have the structure
\begin{equation*}
  \text{$Y_{K,i}$ or $X_{K+1,i}$} =
  \begin{cases}
    \alpha + T_i \beta + S_i
    ,&
    1 \le i \le N-1
    ,
    \\
    \alpha + (T_{N-1} + D) \beta + S_i + D \sigma_{N-1}
    ,&
    i = N
    ,
  \end{cases}
\end{equation*}
that the characteristics between peakons number $i$ and $i+1$ are given by
\begin{equation}
  \label{eq:char-within-rightmost}
  \Xi = \alpha + (T_i + \theta) \beta + S_i + \theta \sigma_i
  ,
\end{equation}
where
$0 < \theta < \tau_{i+1}$ for $1 \le i \le N-2$
and $0 < \theta < D$ for $i = N-1$.

For the leftmost group, if we let
\begin{equation}
  \theta_0 = \frac{\tau_1 M}{\sigma_1 C}
  ,\quad
  \text{where $M = \prod_j \mu_j$}
  ,
\end{equation}
and rewrite the formula for~$X_{1,1}$
from \eqref{eq:even-X-leftmost-group-pos}
or~\eqref{eq:odd-X-leftmost-group-pos}
a little,
the structure is
\begin{equation}
  \label{eq:leftmost-group-rewritten-structure}
  X_{1,i} =
  \begin{cases}
    \dfrac{(S_1 + \theta_0 \sigma_1) \delta}{\alpha + (T_1 + \theta_0) \beta + (S_1 + \theta_0 \sigma_1) \gamma}
    ,&
    i = 1
    ,
    \\[2ex]
    \dfrac{S_i \delta}{\alpha + T_i \beta + S_i \gamma}
    ,&
    2 \le i \le N-1
    ,
    \\[2ex]
    \dfrac{\sigma_{N-1} \delta}{\beta + \sigma_{N-1} \gamma}
    ,&
    i = N
    ,
  \end{cases}
\end{equation}
and the characteristics between peakons number $i$ and $i+1$ are given by
\begin{equation}
  \label{eq:char-within-leftmost}
  \Xi =
  \frac{(S_i + \theta \sigma_i) \delta}{\alpha + (T_i + \theta) \beta + (S_i + \theta \sigma_i) \gamma}
  ,
\end{equation}
where (if $N \ge 3$)
$\theta_0 < \theta < \tau_2$ for $i = 1$,
$0 < \theta < \tau_{i+1}$ for $2 \le i \le N-2$
and $0 < \theta < \infty$ for $i = N-1$.
This requires that $\theta_0 < \tau_2$, which is exactly what
the constraint~\eqref{eq:constraint-C-general} says.
If $N=2$, so that the middle case $2 \le i \le N-1$ is absent
from~\eqref{eq:leftmost-group-rewritten-structure},
then the range of $\theta$ is instead
$\theta_0 < \theta < \infty$ (for $i=1$, i.e., between the only two peakons in the group).

\bigskip

For the remaining characteristics,
that we have not already obtained as byproducts,
there is a bit more work left to do,
but they can be obtained via much simpler substitutions and limits than the ghostpeakons
so far.

To obtain the characteristics between two groups,
add an extra peakon to the left group by increasing its $N$ to $N+1$,
set $\sigma_N = \theta$ in the solution formula for the rightmost peakon in that group
(letting $\theta$ inherit whatever constraints $\sigma_N$ had),
and let $\tau_N \to \infty$.
One can check that this will bring the amplitude of the auxiliary peakon to zero,
leaving a ghostpeakon between the groups.
For example, between a typical $X$-singleton
and a typical $Y$-singleton, given by formulas of the structure
\begin{equation*}
  X_j = \frac{\alpha}{\beta}
  ,\qquad
  Y_j = \frac{\gamma}{\delta}
  ,
\end{equation*}
this produces the characteristics
\begin{equation}
  \label{eq:char-between-XY}
  \Xi = \frac{\alpha + \theta \gamma}{\beta + \theta \delta}
  ,
\end{equation}
where $0 < \theta < \infty$.
Between a typical non-singleton $X$-group and a typical non-singleton $Y$-group,
whose neighbouring members are given by
\begin{equation*}
  X_{j,N} = \frac{\alpha + \sigma_{N-1}^X \gamma}{\beta + \sigma_{N-1}^X \delta}
  ,\qquad
  Y_{j,1} = \frac{\alpha + \tau_1^Y \gamma}{\beta + \tau_1^Y \delta}
  ,
\end{equation*}
we get again~\eqref{eq:char-between-XY}, but now with
$\sigma_{N-1}^X < \theta < \tau_1^Y$
(which is possible because of the constraint~\eqref{eq:constraint-last-sigma-first-tau}
that the last $\sigma$ in one group
must be less than the first~$\tau$ in the next group
whenever two non-singleton groups are adjacent).
And if one of the groups is a singleton and the other one isn't,
then once again the characteristics are given by~\eqref{eq:char-between-XY},
but with
$\sigma_{N-1}^X < \theta < \infty$
or
$0 < \theta < \tau_1^Y$.

Characteristics between a typical $Y$-group and a typical $X$-group follow the analogous pattern.

\begin{remark}
  These results reveal one thing that happens when going
  from the interlacing to the non-interlacing case:
  when a typical singleton (not the leftmost or rightmost) in the interlacing case
  is replaced with a group of several peakons,
  the outermost peakons in that group will travel along
  characteristic curves for the interlacing solution.
  The parameter~$\tau_1$ picks out the $\theta$-value of the characteristic
  that the leftmost peakon in the group will follow
  (in the family of characteristics between the singleton in question and its left neighbour),
  and $\sigma_{N-1}$ does the same for the rightmost peakon in the group
  (in the family of characteristics towards the right neighbour).
  However, what happens ``inside'' the group is more complicated, and cannot
  be determined from the characteristics of the interlacing solution.
\end{remark}

Between the second rightmost group and the rightmost group ($Y_K$/$X_{K+1}$ in the even/odd case),
the same formula~\eqref{eq:char-between-XY}
for the characteristics holds, with the additional caveat
that there is an upper bound $\theta < D$ instead of $\theta < \infty$
if the rightmost group is a singleton.
Similarly, between the $X_1$-group and the $Y_2$-group,
the formula is the same, but with a lower bound $M/C < \theta$
instead of $0 < \theta$ if $N_1^X = 1$, in which case we
define $\alpha$, $\beta$, $\gamma$, $\delta$ by rewriting the solution
formula~\eqref{eq:even-X-leftmost-singleton} for the even case as
\begin{equation}
  \label{eq:leftmost-singleton-rewritten-structure}
  X_1
  = \frac{\detJ_{K,K-1}^{00}}{\detJ_{K-1,K-2}^{11} + C \, \detJ_{K-1,K-1}^{10}}
  = \frac{0 + \frac{M}{C} \, \detJ_{K,K-1}^{00}}{\detJ_{K-1,K-1}^{11} + \frac{M}{C} \, \detJ_{K-1,K-2}^{11}}
  = \frac{\alpha + \frac{M}{C} \, \gamma}{\beta + \frac{M}{C} \, \delta}
  ,
\end{equation}
Thus, if $N_1^X = 1$ and $N_1^Y \ge 2$, with
\begin{equation}
  Y_{1,1} = \frac{0 + \tau_1 \detJ_{K,K-1}^{00}}{\detJ_{K-1,K-1}^{11} + \tau_1 \detJ_{K-1,K-2}^{11}}
  = \frac{\alpha + \tau_1 \gamma}{\beta + \tau_1 \delta}
\end{equation}
by~\eqref{eq:even-Y-typical-group-pos},
then the range of~$\theta$ is
$\frac{M}{C} < \theta < \tau_{1,1}^Y$,
which is possible because of the constraint~\eqref{eq:constraint-C-simpler}.
Similarly with~\eqref{eq:odd-X-leftmost-singleton}
and~\eqref{eq:odd-Y-group-pos} for the odd case.

Finally, we have the characteristic curves on the far right or far left.
To obtain a ghostpeakon to the right of the rightmost peakon,
add an extra peakon by changing $N$ to $N+1$ in the rightmost group,
let
\begin{equation}
  \hat{\sigma}_N = \frac{1}{\epsilon}
  ,\qquad
  \hat{\tau}_N = D
  ,\qquad
  \hat{D} = \theta \epsilon
  ,
\end{equation}
and let $\epsilon \to 0$.
The result is that if the rightmost group is a singleton,
\begin{equation*}
  \text{$X_{K+1}$ or $Y_{K}$} = \alpha + D \beta
  ,
\end{equation*}
then the characteristics on its right
(as we incidentally already knew in the odd case
from~\eqref{eq:kill-rightmost-ghost})
are given by
\begin{equation}
  \label{eq:char-far-right-singleton}
  \Xi = \alpha + D \beta + \theta
  ,\qquad
  0 < \theta < \infty
  ,
\end{equation}
and if it's a group with $N \ge 2$ peakons, the rightmost of which is
\begin{equation*}
  \text{$X_{K+1,N}$ or $Y_{K,N}$} = \alpha + (T_{N-1} + D) \beta + S_i + D \sigma_{N-1}
  ,
\end{equation*}
then
\begin{equation}
  \label{eq:char-far-right-group}
  \Xi =
  \alpha + (T_{N-1} + D) \beta + S_i + D \sigma_{N-1} + \theta
  ,\qquad
  D < \theta < \infty
  .
\end{equation}
In other words, the formula for $\Xi$ is simply obtained by adding~$\theta$
to the formula for the rightmost $X$ or~$Y$.

For the characteristics to the left of the leftmost peakon,
it's easiest to use symmetry: add~$1/\theta$ to the formula for $1/X_1$ to get~$1/\Xi$.
Thus, if the leftmost group is a singleton
(cf.~\eqref{eq:leftmost-singleton-rewritten-structure}),
\begin{equation*}
  X_1 = \frac{\frac{M}{C} \, \gamma}{\beta + \frac{M}{C} \, \delta}
  ,
\end{equation*}
then the characteristics on its left are given by
\begin{equation}
  \label{eq:char-far-left-singleton}
  \Xi = \frac{1}{\dfrac{\beta + \frac{M}{C} \, \delta}{\frac{M}{C} \, \gamma} + \dfrac{1}{\theta}}
  = \frac{\theta \frac{M}{C} \, \gamma}{\theta \left( \beta + \frac{M}{C} \, \delta \right) + \frac{M}{C} \, \gamma}
  ,\qquad
  0 < \theta < \infty
  ,
\end{equation}
and if it's a group with $N \ge 2$ peakons, the leftmost of which is
(cf.~\eqref{eq:leftmost-group-rewritten-structure})
\begin{equation*}
  X_{1,1}
  = \frac{(S_1 + \theta_0 \sigma_1) \delta}{\alpha + (T_1 + \theta_0) \beta + (S_1 + \theta_0 \sigma_1) \gamma}
  = \frac{\theta_0 \sigma_1 \delta}{\alpha + (\tau_1 + \theta_0) \beta + \theta_0 \sigma_1 \gamma}
  ,
\end{equation*}
then
\begin{equation}
  \label{eq:char-far-left-group}
  \Xi = \frac{1}{\dfrac{\alpha + (\tau_1 + \theta_0) \beta + \theta_0 \sigma_1 \gamma}{\theta_0 \sigma_1 \delta} + \dfrac{1}{\theta}}
  = \frac{\theta \, \theta_0 \sigma_1 \delta}{\theta \, \bigl( \alpha + (\tau_1 + \theta_0) \beta + \theta_0 \sigma_1 \gamma \bigr) + \theta_0 \sigma_1 \delta}
  ,\qquad
  0 < \theta < \infty
  .
\end{equation}
These results can also be obtained by
considering a leftmost group with $N+1$ peakons
and making the following somewhat elaborate substititions before letting $\epsilon \to 0$:
if $N = 1$, let
\begin{equation}
  \hat{C} = \frac{\epsilon}{\theta L}
  ,\qquad
  \hat{\tau}_1 = \epsilon
  ,\qquad
  \hat{\sigma}_1 = \frac{M}{C}
  ,
\end{equation}
and if $N \ge 2$, let
\begin{equation}
  \begin{aligned}
    \hat{C} &= \left( \frac{C}{\tau_1} + \frac{1}{\theta L} \right) \epsilon
    ,\\
    \hat{\tau}_1 &= \epsilon
    ,\\
    \hat{\tau}_2 &= \tau_1 + \frac{M \tau_1}{C \sigma_1}
    ,\\
    \hat{\tau}_3 &= \tau_2 - \frac{M \tau_1}{C \sigma_1}
    \qquad
    \text{(if $N \ge 3$)}
    ,\\
    \hat{\tau}_i &= \tau_{i-1}
    ,\qquad
    \text{for $4 \le i \le N$}
    \qquad
    \text{(if $N \ge 4$)}
    ,\\
    \hat{\sigma}_1 &= \frac{1}{\frac{1}{\sigma_1} + \frac{C}{M}}
    ,\\
    \hat{\sigma}_i &= \sigma_{i-1}
    ,\qquad
    \text{for $2 \le i \le N$}
    .
  \end{aligned}
\end{equation}

\phantomsection
\addcontentsline{toc}{section}{Acknowledgements}
\section*{Acknowledgments}
We thank Krzysztof Marciniak for many useful comments.

\small
\bibliographystyle{GX-noninterlacing}
\bibliography{GX-noninterlacing}

\begin{thebibliography}{15}
\providecommand{\enquote}[1]{``#1''}
\providecommand{\url}[1]{{\tt #1}}
\providecommand{\href}[2]{#2}
\setlength{\itemsep}{0pt}


\bibitem[1]{beals-sattinger-szmigielski:2000:moment}
Richard Beals, David H. Sattinger, and Jacek Szmigielski (2000). \href{http://dx.doi.org/10.1006/aima.1999.1883}{Multipeakons and the classical moment problem}. \textit{Adv. Math.} \textbf{154}(2):\discretionary{}{}{}229--257. \href{http://www.ams.org/mathscinet-getitem?mr=1784675}{MR1784675 (2001h:37151)}. arXiv:\href{http://arxiv.org/abs/solv-int/9906001}{solv-int/9906001}.

\bibitem[2]{camassa-holm:1993:CH-orginal-paper}
Roberto Camassa and Darryl D. Holm (1993). \href{http://dx.doi.org/10.1103/PhysRevLett.71.1661}{An integrable shallow water equation with peaked solitons}. \textit{Phys. Rev. Lett.} \textbf{71}(11):\discretionary{}{}{}1661--1664. \href{http://www.ams.org/mathscinet-getitem?mr=1234453}{MR1234453}. arXiv:\href{http://arxiv.org/abs/patt-sol/9305002}{patt-sol/9305002}.

\bibitem[3]{camassa-holm-hyman:1994:CH-new-integrable}
Roberto Camassa, Darryl D. Holm, and James M. Hyman (1994). \href{http://dx.doi.org/10.1016/S0065-2156(08)70254-0}{A new integrable shallow water equation}. \textit{Advances in Applied Mechanics} \textbf{31}:\discretionary{}{}{}1--33.

\bibitem[4]{degasperis-holm-hone:2002:new-integrable-equation-DP}
Antonio Degasperis, Darryl D. Holm, and Andrew N. W. Hone (2002). \href{http://dx.doi.org/10.1023/A:1021186408422}{A new integrable equation with peakon solutions}. \textit{Theor. Math. Phys.} \textbf{133}(2):\discretionary{}{}{}1463--1474. \href{http://www.ams.org/mathscinet-getitem?mr=2001531}{MR2001531}. arXiv:\href{http://arxiv.org/abs/nlin/0205023}{nlin/0205023 [nlin.SI]}.  Proceedings of NEEDS 2001 (Cambridge, UK, July 24--31, 2001).

\bibitem[5]{degasperis-procesi:1999:asymptotic-integrability}
Antonio Degasperis and Michela Procesi (1999). Asymptotic integrability. In A. Degasperis and G. Gaeta, editors, \textit{Symmetry and Perturbation Theory ({Rome}, 1998)}, pp.~23--37. World Scientific Publishing, River Edge, NJ. \href{http://www.ams.org/mathscinet-getitem?mr=1844104}{MR1844104}.

\bibitem[6]{dong-zhou:2018:interlacing-peakons-two-component-CH}
Fengfeng Dong and Lingjun Zhou (2018). \href{http://dx.doi.org/10.1080/14029251.2018.1452674}{Inverse spectral problem and peakons of an integrable two-component {Camassa}--{Holm} system}. \textit{J. Nonlinear Math. Phys.} \textbf{25}(2):\discretionary{}{}{}290--308. \href{http://www.ams.org/mathscinet-getitem?mr=3776562}{MR3776562}.

\bibitem[7]{geng-xue:2009:GX-peakon-equation-cubic-nonlinearity}
Xianguo Geng and Bo Xue (2009). \href{http://dx.doi.org/10.1088/0951-7715/22/8/004}{An extension of integrable peakon equations with cubic nonlinearity}. \textit{Nonlinearity} \textbf{22}(8):\discretionary{}{}{}1847--1856. \href{http://www.ams.org/mathscinet-getitem?mr=2525813}{MR2525813 (2010i:37160)}.

\bibitem[8]{hone-lundmark-szmigielski:2009:novikov}
Andrew N. W. Hone, Hans Lundmark, and Jacek Szmigielski (2009). \href{http://dx.doi.org/10.4310/DPDE.2009.v6.n3.a3}{Explicit multipeakon solutions of {Novikov's} cubically nonlinear integrable {Camassa}--{Holm} type equation}. \textit{Dyn. Partial Differ. Equ.} \textbf{6}(3):\discretionary{}{}{}253--289. \href{http://www.ams.org/mathscinet-getitem?mr=2569508}{MR2569508 (2010i:37172)}. arXiv:\href{http://arxiv.org/abs/0903.3663}{0903.3663 [nlin.SI]}.

\bibitem[9]{hone-wang:2008:cubic-nonlinearity}
Andrew N. W. Hone and Jing Ping Wang (2008). \href{http://dx.doi.org/10.1088/1751-8113/41/37/372002}{Integrable peakon equations with cubic nonlinearity}. \textit{J. Phys. A: Math. Theor.} \textbf{41}(37):\discretionary{}{}{}372002 (10 pages). \href{http://www.ams.org/mathscinet-getitem?mr=2430566}{MR2430566 (2009i:35311)}. arXiv:\href{http://arxiv.org/abs/0805.4310}{0805.4310 [nlin.SI]}.

\bibitem[10]{lundmark-shuaib:2018p:ghostpeakons}
Hans Lundmark and Budor Shuaib (2018). Ghostpeakons and characteristic curves for the {Camassa}--{Holm}, {Degasperis}--{Procesi} and Novikov equations. arXiv:\href{http://arxiv.org/abs/1807.01910}{1807.01910 [nlin.SI]}.

\bibitem[11]{lundmark-szmigielski:2005:DPlong}
Hans Lundmark and Jacek Szmigielski (2005). \href{http://dx.doi.org/10.1155/IMRP.2005.53}{Degasperis--{Procesi} peakons and the discrete cubic string}. \textit{Int. Math. Res. Pap.} \textbf{2005}(2):\discretionary{}{}{}53--116. \href{http://www.ams.org/mathscinet-getitem?mr=2150256}{MR2150256}. arXiv:\href{http://arxiv.org/abs/nlin/0503036}{nlin/0503036 [nlin.SI]}.

\bibitem[12]{lundmark-szmigielski:2016:GX-inverse-problem}
Hans Lundmark and Jacek Szmigielski (2016). \href{http://dx.doi.org/10.1090/memo/1155}{An inverse spectral problem related to the {Geng}--{Xue} two-component peakon equation}. \textit{Mem. Amer. Math. Soc.} \textbf{244}(1155):\discretionary{}{}{}viii+87 pages. \href{http://www.ams.org/mathscinet-getitem?mr=3545110}{MR3545110}. arXiv:\href{http://arxiv.org/abs/1304.0854}{1304.0854 [math.SP]}.

\bibitem[13]{lundmark-szmigielski:2017:GX-dynamics-interlacing}
Hans Lundmark and Jacek Szmigielski (2017). \href{http://dx.doi.org/10.1093/integr/xyw014}{Dynamics of interlacing peakons (and shockpeakons) in the {Geng}--{Xue} equation}. \textit{J. Integrable Syst.} \textbf{2}(1):\discretionary{}{}{}xyw014 (65 pages). \href{http://www.ams.org/mathscinet-getitem?mr=3682465}{MR3682465}. arXiv:\href{http://arxiv.org/abs/1605.02805}{1605.02805 [nlin.SI]}.

\bibitem[14]{novikov:2009:generalizations-of-CH}
Vladimir Novikov (2009). \href{http://dx.doi.org/10.1088/1751-8113/42/34/342002}{Generalizations of the {Camassa}--{Holm} equation}. \textit{J. Phys. A: Math. Theor.} \textbf{42}(34):\discretionary{}{}{}342002 (14 pages). \href{http://www.ams.org/mathscinet-getitem?mr=2530232}{MR2530232 (2011b:35466)}. arXiv:\href{http://arxiv.org/abs/0905.2219}{0905.2219 [nlin.SI]}.

\bibitem[15]{xia-qiao:2015:two-component-CH-with-peakons}
Baoqiang Xia and Zhijun Qiao (2015). \href{http://dx.doi.org/10.1098/rspa.2014.0750}{A new two-component integrable system with peakon solutions}. \textit{Proc. R. Soc. A.} \textbf{471}(2175):\discretionary{}{}{}20140750 (20 pages). \href{http://www.ams.org/mathscinet-getitem?mr=3326340}{MR3326340}. arXiv:\href{http://arxiv.org/abs/1211.5727}{1211.5727 [nlin.SI]}.


\end{thebibliography}

\end{document}